\documentclass[a4paper,UKenglish,cleveref, autoref, thm-restate,numberwithinsect]{lipics-v2021}

\usepackage[T1]{fontenc}
\usepackage{graphicx}

\usepackage{graphicx}
\usepackage[utf8]{inputenc}
\usepackage[T1]{fontenc}
\usepackage{subcaption}
\usepackage{longtable}
\usepackage{rotating}
\usepackage{textcomp}
\usepackage{version}
\usepackage{hyperref}
\usepackage{amsmath}
\usepackage{amsfonts}
\usepackage{amssymb,dsfont}
\usepackage{xcolor}
\usepackage[inline, shortlabels]{enumitem}
\usepackage{appendix}
\usepackage{tikz}
\usetikzlibrary{automata,positioning}
\usetikzlibrary{arrows}
\usetikzlibrary{shapes,snakes}
\usetikzlibrary{fit}
\usepackage{caption}
\usepackage{subcaption}
\usepackage{MnSymbol,wasysym}
        \usepackage{ stmaryrd }
\usepackage{environ}
\usepackage{float}
\usepackage{tabularx}
\usepackage{environ}
\usepackage{lipsum}
\usepackage{circuitikz}
\usetikzlibrary{calc}
\usepackage{amsthm}
\usepackage{comment}

\usepackage{cite}
\nolinenumbers

\usetikzlibrary{shapes.geometric,arrows}

\newcommand{\PP}{\ensuremath{P}}
\newcommand{\CC}{\ensuremath{\mathcal{C}}}
\newcommand{\II}{\ensuremath{\mathcal{I}}}
\newcommand{\qinit}{\ensuremath{q_{in}}}
\newcommand{\set}[1]{\{#1\}}
\renewcommand{\Vert}[1]{\ensuremath{\textsf{V}(#1)}}
\newcommand{\Edges}[1]{\ensuremath{\textsf{E}(#1)}}
\newcommand{\Neigh}[1]{\ensuremath{\textsf{N}(#1)}}
\newcommand{\NeighG}[2]{\ensuremath{\textsf{N}_{#1}(#2)}}

\newcommand{\trans}{\ensuremath{\rightarrow}}
\newcommand{\transup}[1]{\ensuremath{\xrightarrow{#1}}}
\newcommand{\nat}{\mathbb{N}}
\newcommand{\Cover}{\ensuremath{\textsc{Cover}}}
\newcommand{\CoverLine}{\ensuremath{{\textsc{Cover}[\Lines]}}}
\newcommand{\CoverTree}{\ensuremath{{\textsc{Cover}[\Trees]}}}
\newcommand{\CoverStar}{\ensuremath{{\textsc{Cover}[\Stars]}}}
\newcommand{\VASSCover}{\ensuremath{\textsc{VASSControlReach}}}
\newcommand{\counter}{\textsf{x}}
\newcommand{\othercounter}{\textsf{y}}
\newcommand{\Loc}{\text{Loc}}
\newcommand{\TransM}{\text{Trans}}
\newcommand{\transRelM}[1]{\xrightarrow{#1}}
\newcommand{\ellinit}{\ensuremath{\ell_{in}}}
\newcommand{\inc}[1]{\ensuremath{\text{inc(}#1\text{)}}}
\newcommand{\dec}[1]{\ensuremath{\text{dec(}#1\text{)}}}
\newcommand{\test}[1]{\ensuremath{\text{test(}#1\text{)}}}

\newcommand{\zerostate}{\textsf{z}}
\newcommand{\incmess}[2]{\ensuremath{\text{todo}_{{#1}+1}^{#2}}}
\newcommand{\decmess}[2]{\ensuremath{\text{todo}_{{#1}-1}^{#2}}}
\newcommand{\testmess}[2]{\ensuremath{\text{test}_{#1}^{#2}}}
\newcommand{\overincmess}[2]{\ensuremath{{\text{done}_{{#1}+1}}^{#2}}}
\newcommand{\overdecmess}[2]{\ensuremath{{\text{done}_{{#1}-1}}^{#2}}}
\newcommand{\okincmess}[2]{\ensuremath{\overline{\text{todo}_{{#1}+1}}^{#2}}}
\newcommand{\okdecmess}[2]{\ensuremath{\overline{\text{todo}_{{#1}-1}}^{#2}}}

\newcommand{\okoverincmess}[2]{\ensuremath{\overline{\text{done}_{{#1}+1}}^{#2}}}
\newcommand{\okoverdecmess}[2]{\ensuremath{\overline{\text{done}_{{#1}-1}}^{#2}}}

\newcommand{\AOk}{\ensuremath{\textsf{ok}}}
\newcommand{\AOp}{\ensuremath{\textsf{op}}}

\newcommand{\todoinc}[2]{\ensuremath{\text{todo}_{{#1}+1}^{#2}}}
\newcommand{\tododec}[2]{\ensuremath{\text{todo}_{{#1}-1}^{#2}}}
\newcommand{\doneinc}[2]{\ensuremath{{\text{done}_{{#1}+1}}^{#2}}}

\newcommand{\ovtest}[2]{\ensuremath{\overline{\text{test}_{{#1}}}^{#2}}}
\newcommand{\ovdoneinc}[2]{\ensuremath{\overline{\text{done}_{{#1}+1}}^{#2}}}
\newcommand{\ovdonedec}[2]{\ensuremath{\overline{\text{done}_{{#1}-1}}^{#2}}}
\newcommand{\ovtodoinc}[2]{\ensuremath{\overline{\text{todo}_{{#1}+1}}^{#2}}}
\newcommand{\ovtododec}[2]{\ensuremath{\overline{\text{todo}_{{#1}-1}}^{#2}}}
\newcommand{\doneop}[2]{\text{done}_{#1}^{#2}}

\newcommand{\ovanOp}[2]{\ensuremath{\overline{#1}^{#2}}}

\newenvironment{proofof}[1]{\textit{Proof of #1.}}{\hfill$\square$\\ }
\newcommand{\Lines}{\ensuremath{\textsf{Lines}}}
\newcommand{\Trees}{\ensuremath{\textsf{Trees}}}
\newcommand{\Stars}{\ensuremath{\textsf{Stars}}}
\newcommand{\Topo}{\ensuremath{\textsf{Graphs}}}

\newcommand{\firstBroadcast}[2]{\ensuremath{b_{\textsf{first}}(#1,#2)}}
\newcommand{\lastBroadcast}[2]{\ensuremath{t_{\textsf{last}}(#1,#2)}}

\makeatletter
\newcommand{\problemtitle}[1]{\gdef\@problemtitle{#1}}
\newcommand{\probleminput}[1]{\gdef\@probleminput{#1}}
\newcommand{\problemquestion}[1]{\gdef\@problemquestion{#1}}
\newcommand{\problemquestionline}[1]{\gdef\@problemquestionline{#1}}
\NewEnviron{decproblem}{
	\problemtitle{}\probleminput{}\problemquestion{}\problemquestionline{}
	\BODY
	\par\addvspace{.2\baselineskip}
	\noindent
	\fbox{
		
		\centering
		\begin{tabular}{rl}
			\multicolumn{2}{l}{ \@problemtitle} \\
			\hline
			\rule{0pt}{3ex}
			\normalsize \textbf{Input:} &  \@probleminput \\
			\normalsize \textbf{Question:} &  \@problemquestion\\
			\normalsize & \@problemquestionline
		\end{tabular}
	}
	\par\addvspace{.2\baselineskip}
	
}

\newcommand{\LineTikzLong}[9]{%
	\node[rounded rectangle, draw] (v1) [yshift = #1] {\footnotesize{$v_1:$}\small\ensuremath{#2}} ;
	\node[rounded rectangle, draw] (v3) [right of =v1, xshift = 25] {\footnotesize$v_{N-2}:$\small\ensuremath{#3}};
	\node[rounded rectangle, draw] (v4) [right of =v3] {\footnotesize$v_{N-1}:$\small\ensuremath{#4}};
	\node[rounded rectangle, draw] (v5) [right of = v4] {\footnotesize$v_N:$\small\ensuremath{#5}};
	\node[rounded rectangle, draw] (v6) [right = of v5] {\footnotesize$v_{N+1}:$\small\ensuremath{#6}};
	\node[rounded rectangle, draw] (v7) [right = of v6] {\footnotesize$v_{N+2}:$\small\ensuremath{#7}};
	\node[rounded rectangle, draw] (v9) [right = of v7, xshift = 25] {\footnotesize$v_\ell:$\small\ensuremath{#8}};
	
	\path (v1) -- node[auto=false]{\ldots} (v3);
	\path (v7) -- node[auto=false]{\ldots} (v9);
	
	\path[-] 
	(v3) edge node {} (v4)
	(v4) edge node {} (v5) 
	(v5) edge node {} (v6) 
	(v6) edge node {} (v7) 
	;
}

\newcommand{\LineTikz}[7]{%
	\node[rounded rectangle, draw] (v0) [yshift = #1] {$v_0:$\ensuremath{#2}} ;
	\node[rounded rectangle, draw] (v1) [right of =v0] {$v_1:$\ensuremath{#3}};
	\node[rounded rectangle, draw] (v2) [right of =v1] {$v_2:$\ensuremath{#4}};
	\node[rounded rectangle, draw] (v0b) [right of =v2] {$v_3:$\ensuremath{#5}};
	\node[rounded rectangle, draw] (v1b) [right of = v0b] {$v_4:$\ensuremath{#6}};
	\node[rounded rectangle, draw] (v2t) [right = of v1b] {$v_5:$\ensuremath{#7}};

	\path[-] 
	(v0) edge node {} (v1)
	(v1) edge node {} (v2)
	(v2) edge node {} (v0b)
	(v0b) edge node {} (v1b) 
	(v1b) edge node {} (v2t) 
	;
}

\newcommand{\LineTikzFinal}[7]{%
	\node[rounded rectangle, draw] (v0) [yshift = #1] {$v_0:$\ensuremath{#2}} ;
	\node[rounded rectangle, draw] (v1) [right of =v0] {$v_1:$\ensuremath{#3}};
	\node[rounded rectangle, draw] (v2) [right of =v1] {$v_2:$\ensuremath{#4}};
	\node[rounded rectangle, draw] (v0b) [right of =v2] {$v_3:$\ensuremath{#5}};
	\node[rounded rectangle, draw] (v1b) [right of = v0b] {$v_4:$\ensuremath{#6}};
	\node[rounded rectangle, draw] (v2t) [right = of v1b, fill = green!20] {$v_5:q_f$};

	\path[-] 
	(v0) edge node {} (v1)
	(v1) edge node {} (v2)
	(v2) edge node {} (v0b)
	(v0b) edge node {} (v1b) 
	(v1b) edge node {} (v2t) 
	;
}

\newcommand{\LineTikzFrOne}[7]{%
	\node[rounded rectangle, draw] (v0) [yshift = #1] {$v_0:$\ensuremath{#2}} ;
	\node[rounded rectangle, draw, fill=red!20] (v1) [right of =v0] {$v_1:\frownie$};
	\node[rounded rectangle, draw] (v2) [right of =v1] {$v_2:$\ensuremath{#4}};
	\node[rounded rectangle, draw] (v0b) [right of =v2] {$v_3:$\ensuremath{#5}};
	\node[rounded rectangle, draw] (v1b) [right of = v0b] {$v_4:$\ensuremath{#6}};
	\node[rounded rectangle, draw] (v2t) [right = of v1b] {$v_5:$\ensuremath{#7}};

	\path[-] 
	(v0) edge node {} (v1)
	(v1) edge node {} (v2)
	(v2) edge node {} (v0b)
	(v0b) edge node {} (v1b) 
	(v1b) edge node {} (v2t) 
	;
}

\newcommand{\LineTikzZoomOne}[7]{%
	\node[rounded rectangle, draw, inner sep = 2] (v0) [yshift = #1] {$v_0:$\ensuremath{#2}} ;
	\node[rounded rectangle, draw, inner sep = 2] (v1) [right of =v0, xshift = -5] {$v_1:\ensuremath{#3}$};
	\node[rounded rectangle, draw, inner sep = 2] (v2) [right of =v1] {$v_2:$\ensuremath{#4}};
	\node[rounded rectangle, inner sep = 2] (v0b) [right of =v2,xshift = -25] {};

	\path[-] 
	(v0) edge node {} (v1)
	(v1) edge node {} (v2)
	(v2) edge node {} (v0b)
	;

}

\newcommand{\LineTikzZoomFive}[7]{%
	\node[rounded rectangle, inner sep = 2] (v0) [yshift = #1] {} ;
	\node[rounded rectangle, draw, inner sep = 2] (v1) [right of =v0, xshift = -25] {$v_1:\ensuremath{#3}$};
	\node[rounded rectangle, draw, inner sep = 2] (v2) [right of =v1] {$v_2:$\ensuremath{#4}};
	\node[rounded rectangle, inner sep = 2] (v0b) [right of =v2,xshift = -25] {};

	\path[-] 
	(v0) edge node {} (v1)
	(v1) edge node {} (v2)
	(v2) edge node {} (v0b)
	;
	
}

\newcommand{\LineTikzZoomFour}[7]{%
	\node[rounded rectangle, draw, inner sep = 2] (v0) [yshift = #1] {$v_0:$\ensuremath{#2}} ;
	\node[rounded rectangle, draw, inner sep = 2] (v1) [right of =v0, xshift = -5] {$v_1:\ensuremath{#3}$};
	\node[rounded rectangle, inner sep = 2] (v2) [right of =v1, xshift = -25] {};

	\path[-] 
	(v0) edge node {} (v1)
	(v1) edge node {} (v2)
	;
}

\newcommand{\LineTikzZoomFiveFrOne}[7]{%
	\node[rounded rectangle, inner sep = 1] (v0) [yshift = #1] {} ;
	\node[rounded rectangle, draw, fill=red!20, inner sep = 1, xshift = -25] (v1) [right of =v0] {$v_1:\frownie$};
	\node[rounded rectangle, draw, inner sep = 1] (v2) [right of =v1] {$v_2:$\ensuremath{#4}};
	\node[rounded rectangle, inner sep = 1] (v0b) [right of =v2,xshift = -25] {};

	\path[-] 
	(v0) edge node {} (v1)
	(v1) edge node {} (v2)
	(v2) edge node {} (v0b)
	;
}

\newcommand{\LineTikzZoomFourFrOne}[7]{%
	\node[rounded rectangle, draw, inner sep = 1] (v0) [yshift = #1] {$v_0:$\ensuremath{#2}} ;
	\node[rounded rectangle, draw, fill=red!20, inner sep = 1, xshift = -5] (v1) [right of =v0] {$v_1:\frownie$};
	\node[rounded rectangle,  inner sep = 1] (v2) [right of =v1, xshift = -25] {};

	\path[-] 
	(v0) edge node {} (v1)
	(v1) edge node {} (v2)
	;
}

\newcommand{\LineTikzZoomTwo}[7]{%
	\node[rounded rectangle, inner sep = 1] (v1) [yshift = #1] {};

	\node[rounded rectangle, draw, inner sep = 2] (v2) [right of =v1, xshift = -20] {$v_2:$\ensuremath{#4}};
	\node[rounded rectangle, draw, inner sep = 2] (v0b) [right of =v2, xshift = -5] {$v_3:$\ensuremath{#5}};
	\node[rounded rectangle, draw, inner sep = 2] (v1b) [right of = v0b, xshift =-5] {$v_4:$\ensuremath{#6}};
	\node[rounded rectangle, draw, inner sep = 2] (v2t) [right = of v1b, xshift = -10] {$v_5:$\ensuremath{#7}};

	\path[-] 
	(v1) edge node {} (v2)
	(v2) edge node {} (v0b)
		(v0b) edge node {} (v1b) 
		(v1b) edge node {} (v2t) 
	;
}

\newcommand{\AConfiguration}[6]{%
	\node (v21) [draw,xshift = #5, yshift = #6, rounded rectangle, fill=orange!0] {$v_1:$ \ensuremath{#1}};
	\node (v11) [draw, xshift = #5 +22, yshift = #6 -20, rounded rectangle, fill=cyan!0] {$v_2:$ \ensuremath{#2}};
	\node (v31) [draw, xshift = #5 -22, yshift = #6 -20,  rounded rectangle, fill=purple!0] {$v_3:$ \ensuremath{#3}};

	\path[-] 
	(v21) edge node {} (v11)
	(v21) edge node {} (v31)
	(v11) edge [opacity = #4] node {} (v31)  
	
	;
}

\newcommand{\idleState}{\textsf{idl}}
\newcommand{\idleStateAbrev}{\textsf{idl}}
\newcommand{\execState}{\textsf{ex}}
\newcommand{\execStateAbrev}{\textsf{ex}}
\newcommand{\restState}{\textsf{hlt}}
\newcommand{\restStateAbrev}{\textsf{hlt}}
\newcommand{\repState}{\textsf{tr}}
\newcommand{\repStateAbrev}{\textsf{tr}}
\newcommand{\checkState}{\textsf{ch}}
\newcommand{\checkStateAbrev}{\textsf{ch}}

\newcommand{\bprint}[1]{\mathbf{bprint}(#1)}

\usepackage[textsize=small]{todonotes}
\newcommand{\lug}[1]{}
\newcommand{\nas}[1]{}
\newcommand{\ars}[1]{}
\newcommand{\lugtext}[1]{}
\newcommand{\nastext}[1]{}
\newcommand{\arstext}[1]{}
\newcommand{\provisoire}[1]{}

\title{Phase-Bounded Broadcast Networks over Topologies of Communication}

\author{Lucie Guillou}{IRIF, CNRS, Universit\'e Paris Cit\'e,
France}{guillou@irif.fr}{https://orcid.org/0000-0002-6101-2895}{ANR project PaVeDyS (ANR-23-CE48-0005)}
 
\author{Arnaud Sangnier}{DIBRIS, Università di Genova, Italy}{arnaud.sangnier@unige.it}{https://orcid.org/0000-0002-6731-0340}{}
\author{Nathalie Sznajder}{LIP6, CNRS, Sorbonne Universit\'e,
France}{nathalie.sznajder@lip6.fr}{https://orcid.org/0000-0002-4199-2443}{}

\authorrunning{L. Guillou and A. Sangnier and N. Sznajder}
\Copyright{L. Guillou and A. Sangnier and N. Sznajder}



\EventEditors{}
\EventLongTitle{}
\EventShortTitle{}
\EventAcronym{}
\EventYear{}
\EventDate{}
\EventLocation{}
\EventLogo{}
\SeriesVolume{}
\ArticleNo{}

\newif\iflong
\longfalse 
\newif\ifshort
\shorttrue
\newcommand{\Iflong}[1]{\iflong#1\fi}
\newcommand{\Ifshort}[1]{\ifshort#1\fi}

\begin{document}
\maketitle

%

\begin{abstract}
We study networks of processes that all execute the same finite state protocol  and that communicate through broadcasts. The processes are organized
in a graph (a \emph{topology}) and only the neighbors of a process in this graph can receive its broadcasts. 
The coverability problem asks, given a protocol and a state of the protocol, whether there is a topology for the processes such that one of them (at least) 
reaches the given state. This problem is undecidable~\cite{DelzannoSZ10}. We study here an under-approximation of the problem where processes
alternate a bounded number of times $k$ between phases of broadcasting and phases of receiving messages. We show that, if the problem remains undecidable when
$k$ is greater than 6, it becomes decidable for $k=2$, and $\textsc{ExpSpace}$-complete for $k=1$. Furthermore, we show that if we restrict ourselves to line topologies,
the problem is in $P$ for $k=1$ and $k=2$. 
\end{abstract}
\section{Introduction}

\emph{Verifying networks with an unbounded number of entities.} Ensuring safety properties for concurrent and distributed systems is a challenging task, since all possible interleavings must be taken into account; hence, even if each entity has a finite state behavior, the verification procedure has to deal with the state explosion problem. Another level of difficulty arises when dealing with distributed protocols designed for an unbounded number of entities. In that case, the safety verification problem consists in ensuring the safety of the system, for any number of participants. Here, the difficulty comes from the infinite number of possible instantiations of the network. In their seminal paper\cite{german92}, German and Sistla propose a formal model to represent and analyze such networks: in this work, all the processes 
in the network execute the same protocol, given by a finite state automaton, and they communicate thanks to pairwise synchronized rendez-vous. The authors study the parameterized coverability problem, which asks whether there exists an initial number of processes that allow an execution leading to a configuration in which (at least) one process is in an error state (here the parameter is the number of processes). They show that it is decidable in polynomial time. 
 Later on, different variations of this model have been considered, by modifying the communication means: token-passing mechanism\cite{clarke-verification-concur04,aminof-param-vmcai14}, communication through shared register \cite{esparza-parameterized-jacm16,durand-model-fmsd17}, non-blocking rendez-vous mechanism \cite{guillou-safety-concur23}, or adding a broadcast mechanism to send a message to all the entities \cite{esparza-verification-lics99}. The model of population protocol proposed in \cite{angluin-computation-podc04}  and for which verification methods have been developed recently in \cite{esparza-verification-acta17,esparza-complexity-discomp21} belongs also to this family of systems. In this latter model, the properties studied are different, and more complex than safety
 conditions.

\smallskip
\noindent\emph{Broadcast networks working over graphs.} In \cite{DelzannoSZ10}, Delzanno et. al propose a new model of parameterized network in which each process communicates \emph{with its neighbors} by broadcasting messages. The neighbors of an entity are given thanks to a graph: the \emph{communication topology}. This model was inspired by ad hoc networks, where nodes communicate with each other thanks to radio communication. The difficulty in proving safety properties for this new model lies in the fact that one has to show that the network is safe for all possible numbers of processes and all possible communication topologies. So the verification procedure not only looks for the number of entities, but also for a graph representing the relationship of the neighbours to show unsafe execution. As mentioned earlier, it is not the first work to propose a parameterized network with broadcast communication; indeed the parameterized coverability problem in networks with broadcast is decidable \cite{esparza-verification-lics99} and non-primitive recursive \cite{schmitz-power-concur13} when the communication topology is complete (each entity is a neighbor of all the others). However, when there is no restriction on the allowed communication topologies the problem becomes undecidable \cite{DelzannoSZ10} but decidability can be regained by providing a bound on the length of all simple paths in allowed topologies \cite{DelzannoSZ10}. This restriction has then been extended in \cite{DelzannoSZ11} to allow also cliques in the model. However, with this restriction, the complexity of parameterized coverability is non-primitive recursive \cite{DelzannoSZ11}.

\smallskip
\noindent\emph{Bounding the number of phases.} When dealing with infinite-state systems with an undecidable safety verification problem, one option consists in looking at under-approximations of the global behavior, restricting the attention to a subset of executions. If proving whether the considered subset of executions is safe is a decidable problem, this technique leads to a sound but incomplete method for safety verification. Good under-approximation candidates are the ones that can be extended automatically to increase the allowed behavior. For instance, it is known that safety verification of finite systems equipped with integer variables that can be incremented, decremented, or tested to zero is undecidable \cite{Minsky67}, but if one considers only executions in which, for each counter, the number of times the execution alternates between an increasing mode and a decreasing mode is bounded by a given value, then safety verification becomes decidable \cite{ibarra-reversal-acm78}. Similarly, verifying concurrent programs manipulating stacks is undecidable \cite{ramalingam-context-acm00} but decidability can be regained by bounding the number of allowed context switches (a context being a consecutive sequence of transitions performed by the same thread) \cite{qadeer-context-tacas05}. Context-bounded analysis has also been applied to concurrent programs with stacks and dynamic creation of threads~\cite{ABQ2011}. Another type of underapproximation analysis has been conducted by \cite{TorreMP} (and by \cite{BLS18} in another context), by considering bounded round-robin schedules of processes. Inspired by this work, we propose here to look at executions of broadcast networks over communication topologies where, for each process, the number of alternations between phases where it broadcasts messages and phases where it receives messages is bounded. We call such protocols $k$-phase-bounded protocols where $k$ is the allowed number of alternations.

\smallskip
\noindent \emph{Our contributions.} We study the parameterized coverability problem for broadcast networks working over communication topologies. We first show in Section 2 that it is enough to consider only tree topologies. This allows us to ease our presentation in the sequel and is also an interesting result by itself. In Section 3, we prove that the coverability problem is still undecidable when considering $k$-phase-bounded broadcast protocols with $k$ greater than 6.  The undecidability proof relies on a technical reduction from the halting problem for two counter Minsky machines. We then show in Sections 4 and 5 that if the number of alternations is smaller or equal to $2$, then decidability can be regained. More precisely, we show that for $1$-phase-bounded protocols, we can restrict our attention to tree topologies of height 1, which 
provides an \textsc{ExpSpace}-algorithm for the coverability problem.  To solve this problem in the case of $2$-phase-bounded protocols, we prove that we can bound the height of the considered tree and rely on the result of \cite{DelzannoSZ10} which states that the coverability problem for broadcast networks is decidable when considering topologies where the length of all simple paths is bounded. We furthermore show that if we consider line topologies then the coverability problem restricted to $1$- and $2$-phase-bounded protocols can be solved in polynomial time.\\
\noindent Due to lack of space, omitted proofs and reasonings can be found in Appendix.
\section{Preliminaries}
Let $A$ be a countable set, we denote $A^\ast$ as the set of finite sequences of elements taken in $A$. Let $w \in A^\ast$, the length of $w$ is defined as the number of elements in the sequence $w$ and is denoted $|w|$. For a sequence $w =a_1\cdot a_2 \cdots a_k \in A^+$, we denote by $w[-1]$ the sequence $a_1 \cdot a_2 \cdots a_{k-1}$. Let $\ell, n\in \nat$ with $\ell \leq n$, we denote by $[\ell, n]$ the set of integers $\set{\ell, \ell+1, \dots, n}$.

\subsection{Networks of processes}

We study networks of processes where each process executes the same protocol given as a finite-state automaton. Given a finite set of messages $\Sigma$, a transition of the protocol can be labelled by three types of actions: (1) the broadcast of a message $m \in \Sigma$ with label $!!m$, (2) the reception of a message $m \in \Sigma$ with label $?m$ or (3) an internal action with a special label $\tau \nin \Sigma$. Processes are organised according to a topology which gives for each one of them its set of neighbors. When a process broadcasts a message $m \in \Sigma$, the only processes that can receive $m$ are its neighbors, and the ones having an output action $?m$ have to receive it. Furthermore, the topology remains fixed during an execution.

Let $\Sigma$ be a finite alphabet. In order to refer to the different types of actions, we write $!!\Sigma$ for the set $\set{!!m \mid m \in \Sigma}$ and $?\Sigma$ for $\set{?m \mid m \in \Sigma}$.

\begin{definition}
	A \emph{Broadcast Protocol} is a tuple $\PP = (Q, \Sigma, \qinit, \Delta)$ such that $Q$ is a finite set of states, $\Sigma$ is a finite alphabet of messages, $\qinit$ is an initial state and $\Delta \subseteq Q \times (!!\Sigma \times ?\Sigma \cup \set{\tau}) \times Q$ is a finite set of transitions.
\end{definition}
We depict an example of a broadcast protocol in \cref{fig:bp:examples}. 
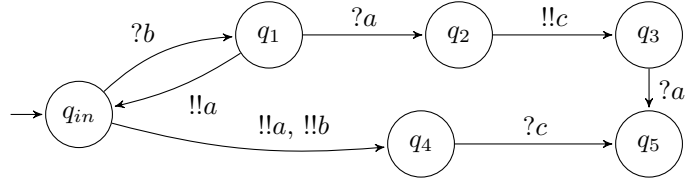
\begin{figure}
	\begin{center}
		\tikzset{box/.style={draw, minimum width=4em, text width=4.5em, text centered, minimum height=17em}}

%
%
%
%
%

\begin{tikzpicture}[-, >=stealth', shorten >=1pt,node distance=2.5cm,on grid,auto, initial text = {}] 
	\node[] (q0) {};
	\node[state,initial] (q0b) [above  of = q0, yshift = -60, xshift =0] {$\qinit$};
	\node[state] (q1) [right of = q0, yshift = 0, xshift =2cm] {$q_4$};
	\node[state] (q2) [right  of = q1, yshift = 0, xshift = 0.5cm] {$q_5$};
	\node[] (q3) [above  of = q0, yshift = -30, xshift =0] {};
	\node[state] (q4) [right of = q3, yshift = -0, xshift = 0] {$q_1$};
	\node[state] (q5) [right  of = q4, xshift = -0]{$q_2$};
	\node[state] (q6) [right  of = q5, xshift = -0]{$q_3$};
	
	\node [] (p1) [left of = q4, yshift = -18] {};
	\node[] (p2) [left of = q1, yshift = 0, xshift =-2cm] {};

	\path[->] 
	(q0b) edge [bend right =10] node {$!!a$, $!!b$} (q1) 
	(q1) edge node {$?c$} (q2) 
	(q0b) edge [bend left =20] node {$?b$} (q4)
	(q4) edge [bend left = 10] node {$!!a$} (q0b)
	(q4) edge node {$?a$} (q5)
	(q5) edge [] node {$!!c$} (q6)
	(q6) edge [] node {$?a$} (q2)
	;
\end{tikzpicture}
	\end{center}
	\caption{Example of a broadcast protocol denoted $\PP$}
	\label{fig:bp:examples}
\end{figure}
Processes are organised according to a topology, defined formally as follows.
\begin{definition}
	A \emph{topology} is an undirected graph, i.e. a tuple $\Gamma = (V, E)$ such that $V$ is a finite set of vertices, and $E \subseteq V\times V$ is a finite set of edges such that $(u,v) \in E$ implies $(v,u) \in E$ for all $(u, v) \in V^2$, and for all $u \in V$, $(u, u) \nin E$ (there is no self-loop).
\end{definition}
We will use $\Vert{\Gamma}$ and $\Edges{\Gamma}$ to denote the set of vertices and edges of $\Gamma$ respectively, namely $V$ and $E$. For $v \in V$, we will denote $\NeighG{\Gamma}{v}$ the set $\set{u \mid (v, u) \in E}$. When the context is clear, we will write $\Neigh{v}$. For $u,v \in \Vert{\Gamma}$, we denote $\langle v, u \rangle$ for the two pairs $(v, u), (u,v)$. We name \Topo\ the set of topologies.
In this work, we will also be interested in some families of topologies: line and tree topologies.
A topology $\Gamma = (V,E)$ is a \emph{tree topology} if $V$ is a set of words of $\nat^\ast$ which is prefix closed with $\epsilon \in V$, and if $E = \set{\langle w[-1], w\rangle  \mid w \in V \cap \nat^+}$. This way, the \emph{root} of the tree is the unique vertex $\epsilon \in V$ and a node $w \in V \cap \nat^+$ has a unique parent $w[-1]$. The \emph{height} of the tree is $\max\set{n\in\nat\mid |w| = n}$. 
We denote by $\Trees$ the set of tree topologies. 
A topology $\Gamma = (V, E)$ is a \emph{line topology} if $V$ is such that $V = \set{v_1, \dots, v_n}$ for some $n \in \nat$ and $E = \set{\langle v_i, v_{i+1} \rangle \mid 1 \leq i < n}$. We denote by $\Lines$ the set of line topologies. 


\emph{Semantics.}
A \emph{configuration} $C$ of a broadcast protocol $\PP = (Q,\Sigma,\qinit, \Delta)$\ is a tuple $(\Gamma, L)$ where $\Gamma$ is a topology, and $L : \Vert{\Gamma} \rightarrow Q$ is a labelling function associating to each vertex $v$ of the topology its current state of the protocol. 
In the sequel, we will sometimes call processes or nodes the vertices of $\Gamma$.
A configuration $C$ is \emph{initial} if $L(v) = \qinit$ for all $v\in \Vert{\Gamma}$. We let $\CC_\PP$ be the set of all configurations of \PP, and $\II_\PP$ the set of all initial configurations. When $\PP$ is clear from the context, we may drop the subscript and simply use $\CC$ and $\II$. 
Given a protocol $\PP = (Q,\Sigma,\qinit, \Delta)$, and a state $q \in Q$, we let $R(q)=\set{m \in \Sigma \mid \exists q' \in Q, (q, ?m, q') \in \Delta}$ be the
set of messages that can be received when in the state $q$.

Consider $\delta = (q, \alpha, q') \in \Delta$ a transition of $\PP$, and $C = (\Gamma, L)$ and $C' = (\Gamma', L')$ two configurations of $\PP$, and let $v \in \Vert{\Gamma}$ be a vertex.
The transition relation ${\transup{v,\delta}} \in \CC \times \CC$ is defined as follows: we have $C \transup{v,\delta} C'$ if and only if $\Gamma = \Gamma'$, and one of the following conditions holds:
\begin{itemize}
	\item $\alpha = \tau$ and $L(v) = q$, $L'(v) = q'$ and $L'(u) = L(u)$  for all $u \in \Vert{\Gamma} \setminus \set{v}$: vertex $v$ performs an internal action;
	\item $\alpha = !!m$ and $L(v) = q$, $L'(v) = q'$ (vertex $v$ performs a broadcast), and for each process $u \in \Neigh{v}$ neighbor of $v$, either $(L(u), ?m, L'(u)) \in \Delta$ (vertex $u$ receives message $m$ from $v$), or $m \nin R(L(u))$ and $L(u) = L'(u)$ (vertex $u$ is not in a state in which it can receive $m$ and stays in the same state). Furthermore, $L'(w) = L(w)$ for all other vertices $w \in \Vert{\Gamma} \setminus (\set{v} \cup \Neigh{v})$ (vertex $w$ does not change state).
\end{itemize}
We write $C \transup{} C'$ whenever there exists $v \in \Vert{\Gamma}$ and $\delta \in \Delta$ such that $C \transup{v,\delta} C'$. We denote by $\trans^\ast$ [resp. $\trans^+$] for the reflexive and transitive closure [resp. transitive] of $\trans$.
An \emph{execution} of $\PP$ is a sequence of configurations $C_0,\dots, C_n\in\CC_\PP$ such that for all $0\leq i <n$, $C_i\trans C_{i+1}$. 

\begin{example}\label{ex:execution}
	We depict in~\cref{fig:execution-example} an execution of protocol $\PP$ (from \cref{fig:bp:examples}): it starts with an initial configuration with three processes $v_1, v_2, v_3$, organised as a clique (each vertex is a neighbour of the two others), each on the initial state $\qinit$.
	More formally, $\Gamma = (V, E)$ with $V = \set{v_1, v_2, v_3}$ and $E = \set{\langle v_1, v_2\rangle, \langle v_2, v_3\rangle, \langle v_1, v_3\rangle }$. 
	From the initial configuration, the following chain of events happens: $C_0\transup{v_1,(\qinit, !!b, q_4)} C_1\transup{v_2, (q_1, !!a, \qinit)}C_2\transup{v_3,
	(q_2, !!c, q_3)}C_3$.
\end{example}


\begin{figure}
	\begin{center}
		\begin{tikzpicture}[->, >=stealth', shorten >=1pt,node distance=1.2cm,on grid,auto, initial text = {}] 
	\AConfiguration{\qinit}{\qinit}{\qinit}{1}{0}{0}
	\draw [-stealth](1.6,-0.5) -- node {$v_1$} (2,-0.5) ;
	\AConfiguration{q_4}{q_1}{q_1}{1}{3.6cm}{0}
	\draw [-stealth](5.2,-0.5) -- node {$v_2$} (5.6,-0.5) ;
	\AConfiguration{q_4}{\qinit}{q_2}{1}{7.2cm}{0}
	\draw [-stealth](8.8,-0.5) -- node {$v_3$} (9.2,-0.5) ;
	\AConfiguration{q_5}{\qinit}{q_3}{1}{10.8cm}{0cm}
\end{tikzpicture}
	\end{center}
	\caption{Example of an execution of protocol $\PP$ (\cref{fig:bp:examples}).}
	\label{fig:execution-example}
\end{figure}

\subsection{Verification problem}\label{subsec:verification-problems}
In this work, we focus on the \emph{coverability problem} which consists in ensuring a safety property: we want to check that, no matter the number of processes in the network, nor the topology in which the processes are organised, a specific error state can never be reached. 


The coverability problem over a family of topologies $\mathcal{S} \in \set{\Topo, \Trees, \Lines}$ is stated as follows:

\begin{decproblem}
	\problemtitle{$\Cover[\mathcal{S}]$~}
	\probleminput{A broadcast protocol $\PP$ and
		a  state $q_f \in Q$;} 
	\problemquestion{Is there $\Gamma \in \mathcal{S}$, $C = (\Gamma, L)\in \II_\PP$ and $C' =(\Gamma, L')\in \CC_\PP$ and $v\in \Vert{\Gamma}$ such that }\problemquestionline{$C \trans^\ast C'$ and $L'(v) = q_f$?}
\end{decproblem}
For a family $\mathcal{S}$, if indeed there exist $C = (\Gamma, L)$ and $C'=(\Gamma, L')$ such that $C\trans^\ast C'$ and $L'(v)=q_f$ for some $v\in\Vert{\Gamma}$, 
we say that $q_f$ is \emph{coverable} (in $\PP$) with $\Gamma$. We also say that the execution $C\trans^\ast C'$ \emph{covers} $q_f$. For short, we write $\Cover$ instead of $\Cover[\Topo]$.
Observe that $\Cover$ is a generalisation of $\CoverTree$ which is itself a generalisation of $\CoverLine$.
\nas{on le dit deja dans l'intro tout ça. On enlève?}In \cite{DelzannoSZ10}, the authors proved that the three problems are undecidable, and they later showed in \cite{DelzannoSZ11}~that the undecidability of \Cover\ still holds when restricting the problem to families of topologies with bounded diameter.

However, in \cite{DelzannoSZ10}, the authors show that $\Cover$ becomes decidable when searching for an execution covering $q_f$ with a $K$-bounded path topology for some $K \in \nat$, i.e. for a topology in which all simple paths between any pair of vertices $v_1, v_2 \in V$ have a length bounded by $K$.
In \cite{DelzannoSZ11}, it is also shown that \Cover~is Ackermann-hard when searching for an execution covering $q_f$ with a topology where all maximal cliques are connected by paths of bounded length.
\Iflong{\begin{example}
	In the broadcast protocol $\PP$ displayed in \cref{fig:bp:examples}, $q_5$ is coverable as shown in \cref{ex:execution}.\lug{en iflong?}
\end{example}}
\Ifshort{
We establish the first result.
\begin{theorem}\label{thm:Cover-CoverTree-equivalent}
$\Cover[\Topo]$ and $\CoverTree$ are equivalent. 
\end{theorem}

Indeed, if it is obvious that when a state is coverable with a tree topology, it is coverable with a topology from $\Topo$, we can show that whenever a state is 
coverable, it is coverable with a tree topology. If a set $q_f$ of a protocol $\PP$ is coverable with a topology $\Gamma\in\Topo$, let $\rho=C_0 \trans \cdots \trans C_n=(\Gamma, L_n)$ be an execution covering $q_f$, and a vertex $v_f \in \Vert{\Gamma}$ such that $L_n(v_f) = q_f$.
	We can build an execution covering $q_f$ with a tree topology $\Gamma'$ where the root reaches $q_f$. Actually, $\Gamma'$ is the unfolding
	of $\Gamma$ in a tree of height $n$. 
}


\section{Phase-Bounded Protocols}
As $\Cover[\Topo]$, $\CoverTree$ and $\CoverLine$ are undecidable in the general case, we investigate a restriction on broadcast protocols: phase-bounded protocols.

For $k \in \nat$, a $k$-phase-bounded protocol is a protocol that ensures that each process alternates at most $k$ times between phases of broadcasts and phases of receptions.
Before giving our formal definition of a phase-bounded protocol, we motivate this restriction. 

Phase-bounded protocols can be seen as a semantic restriction of general protocols in which each process can only switch a bounded number of times between phases where it receives messages and phases where it broadcasts messages. When, usually, restricting the behavior of processes immediately yields an
underapproximation of the reachable states, we highlight in \Ifshort{\cref{sec:cover-cover-phase-bounded}}\Iflong{here}\ the fact that preventing messages from being received can in fact lead to new reachable states. Actually, the reception of a message is something that is not under the control of a process. If another process broadcasts a message, a faithful behavior of the system is that all the processes that can receive it indeed do so, no matter in which phase they are in their own execution. Hence, in a restriction that attempts to limit the number of switches 
	between broadcasting and receiving phases, one should not prevent a reception to happen. 
This motivates our definition of phase-bounded protocols, in which a process in its last broadcasting phase, can still receive messages. A $k$-unfolding of a protocol $\PP$ is then a protocol in which we duplicate the vertices by annotating them with the type and the number of phase ($b$ or $r$ for broadcast or 
reception and an integer between $0$ and $k$ for the number). 
\begin{figure}
	\begin{center}
		\tikzset{box/.style={draw, minimum width=4em, text width=4.5em, text centered, minimum height=17em}}

\begin{tikzpicture}[-, >=stealth', shorten >=1pt,node distance=2.5cm,on grid,auto, initial text = {}] 
	\node[] (q0) {};
	\node[state,initial] (q0b) [above  of = q0, yshift = -60, xshift =0] {$\qinit^0$};
	\node[state] (q1) [right of = q0, yshift = -0, xshift =0.4cm] {$q_4^{b,1}$};
	\node[state] (q2) [right  of = q1, yshift = 0, xshift = 4.5cm] {$q_5^{r,2}$};
	\node[] (q3) [above  of = q0, yshift = -30, xshift =0] {};
	\node[state] (q4) [right of = q3, yshift = 5, xshift = -20] {$q_1^{r,1}$};
	
	\node[state] (q5) [right  of = q4, xshift = -20, yshift = -10]{$q_2^{r,1}$};
	\node[state] (q6) [right  of = q5, xshift = 20, yshift = -15]{$q_3^{b,2}$};
	
	\node[state] (qin2) [above of = q5, yshift = -60, xshift = 50] {$\qinit^{b,2}$};
	\node[state] (q42) [right of = qin2, yshift = -0, xshift = 0] {$q_4^{b,2}$};
	\node[state] (q12) [right of = qin2, yshift = 15, xshift = 60] {$q_1^{r,2}$};
	\node[state] (q22) [right of = q12] {$q_2^{r,2}$};
	
	\node [] (p1) [left of = q4, yshift = -18] {};
	\node[] (p2) [left of = q1, yshift = 0, xshift =-2cm] {};

	\node[dashed, draw, fill = orange, fill opacity = 0, text opacity = 1, fit=(q0b), text height=0.06 \columnwidth, label ={below:Phase 0}] () {};
	\node[dashed, draw, fill = green, fill opacity = 0, text opacity = 1, fit=(q1) (q4) (q5), text height=0.06 \columnwidth, label ={below:Phase 1}, inner xsep = 8, inner ysep=10] () [yshift = 5]{};
	\node[dashed, draw, fill = cyan, fill opacity = 0, text opacity = 1, fit=(q2) (q6) (qin2) (q42) (q12) (q22), text height=0.06 \columnwidth, label ={below:Phase 2}, inner ysep=8] () {};

	\node[fill=violet, fill opacity = 0.2, fit=(q1), inner sep =1] () {};
	\node[fill=green, fill opacity = 0.2, fit=(q4) (q5), inner sep =1] () {};
	
	\node[fill=violet, fill opacity = 0.2, fit=(q42) (qin2) (q6), inner sep =1] () {};
	\node[fill=green, fill opacity = 0.2, fit=(q12) (q22) (q2), inner sep =1] () {};

	\path[->] 
	(q0b) edge [bend right =10] node [xshift =-5] {$!!a$, $!!b$} (q1) 
	(q1) edge node [below] {$?c$} (q2) 
	(q0b) edge [bend left =20] node {$?b$} (q4)
	(q4) edge [bend right = 10] node [xshift = -5]{$?a$} (q5)
	(q5) edge [bend right = 10] node [below, xshift = -5, yshift = 2] {$!!c$} (q6)
	(q6) edge [bend left = 0, thick]  node {$?a$} (q2)
	
	(q4) edge [bend left = 20] node {$!!a$} (qin2)
	(qin2) edge [bend left = 20, thick]  node {$?b$} (q12)
	(qin2) edge [] node [yshift =-5]{$!!a, !!b$} (q42)
	(q12) edge [] node {$?a$} (q22)
	(q42) edge [bend left = 20, thick]  node {$?a$} (q2)
	
	;
\end{tikzpicture}
	\end{center}
	\caption{$\PP_2$: the 2-unfolding of protocol \PP\ (\cref{fig:bp:examples}).}
	\label{fig:unfolding:examples}
\end{figure}
\begin{example}
	\cref{fig:unfolding:examples} pictures the 2-unfolding of protocol \PP\ (\cref{fig:bp:examples}). Observe that from state $q_4^{b,2}$, which is a broadcast state,
	it is still possible to receive message $a$ and go to state $q_5^{r,2}$. However, it is not possible to send a message from $q_5^{r,2}$ (nor from any reception state of phase 2). 
\end{example}
We show in Appendix that this definition of unfolding can be used as an underapproximation for \Cover.
\Iflong{\section{Phase-bounded protocols as an under-approximation}\label{sec:cover-cover-phase-bounded}
	
	\begin{figure}
		\begin{minipage}[c]{0.5\columnwidth}
			\resizebox*{1.0\columnwidth}{!}{
				\tikzset{box/.style={draw, minimum width=4em, text width=4.5em, text centered, minimum height=17em}}

\begin{tikzpicture}[-, >=stealth', shorten >=1pt,node distance=2cm,on grid,auto, initial text = {}] 
	\node[state,initial] (q0) [] {$\qinit$};
	\node[state] (q1) [right of = q0, yshift = 15, xshift =0] {$q_1$};
	\node[state] (q2) [right  of = q1, xshift =0] {$q_2$};
	\node[state] (q3) [right of = q2, xshift =0] {$q_3$};
	\node[state] (q4) [right of = q0, yshift = -15, xshift = 0] {$q_4$};
	\node[state] (q5) [right  of = q4, xshift =0] {$q_5$};
	\node[state] (q6) [right of = q5, xshift =0] {$q_6$};
	
	\node[state] (p) [below of = q4, yshift= 20] {$p$};

	\path[->] 
	(q0) edge node {$!!c$} (q1)
	(q0) edge node {$?c$} (q4)
	(q0) edge [bend right] node {$?m$} (p)
	
	(q1) edge node {$!!m$} (q2)
	
	(q4) edge node {$!!b$} (q5)
	(q4) edge node {$?m$} (p)
	(q2) edge node {$?a$} (q3)
	(q5) edge [bend left] node {$?m$} (p)
	(q5) edge node {$!!a$} (q6)
	;
\end{tikzpicture}
			}
			
			\subcaption{An example of a broadcast protocol $\overline{\PP}$\label{fig:bp:example-2}}
		\end{minipage}
		\hfill
		\begin{minipage}[c]{0.5\columnwidth}
			\resizebox*{1.0\columnwidth}{!}
			{\tikzset{box/.style={draw, minimum width=4em, text width=4.5em, text centered, minimum height=17em}}

\begin{tikzpicture}[-, >=stealth', shorten >=1pt,node distance=2cm,on grid,auto, initial text = {}] 
	\node[state,initial] (q0) [] {$\qinit$};
	\node[state, ellipse, inner sep= 0.01] (q1) [right of = q0, yshift = 15, xshift =10] {\footnotesize$(q_1,b,1)$};
	\node[state, ellipse, inner sep= 0.01] (q2) [right  of = q1, xshift =10] {\footnotesize$(q_2,b,1)$};
	\node[state, ellipse,inner sep= 0.01] (q3) [right of = q2, xshift =10] {\footnotesize$(q_3,r,2)$};
	\node[state, ellipse,inner sep= 0.01] (q4) [right of = q0, yshift = -15, xshift = 10] {\footnotesize$(q_4,r,1)$};
	\node[state, ellipse,inner sep= 0.01] (q5) [right  of = q4, xshift =10] {\footnotesize$(q_5,b,2)$};
	\node[state,ellipse,inner sep= 0.01] (q6) [right of = q5, xshift =10] {\footnotesize$(q_6,b,2)$};
	
	\node[state, ellipse,inner sep= 0.01] (p) [below of = q4, yshift = 20] {$(p,r,1)$};

	\path[->] 
	(q0) edge node {$!!c$} (q1)
	(q0) edge node {$?c$} (q4)
	(q0) edge [bend right] node {$?m$} (p)
	
	(q1) edge node {$!!m$} (q2)
	
	(q4) edge node {$!!b$} (q5)
	(q4) edge node {$?m$} (p)
	(q2) edge node {$?a$} (q3)
	(q5) edge node {$!!a$} (q6)
	;
\end{tikzpicture}}
			\subcaption{Unfolding of $\overline{\PP}$ limited to 2 phases: $(q,b,i)$ means that the state $q$ is reached in a broadcast phase when $(q,r,i)$ means that it is reached in a reception phase. The number of the phase is given by $i$.}\label{fig:bp:example-2-unfold}
		\end{minipage}
		
		\caption{A protocol and its bounded unfolding showing that it may not be an underapproximation.}
	\end{figure}

	Phase-bounded protocols can be seen as a semantic restriction of general protocols in which each process can only switch a bounded number of times between phases where it receives messages and phases where it can send messages. When, usually, restricting the behavior of processes immediately yields an
	underapproximation of the reachable states, we highlight here the fact that preventing messages from being received can in fact lead to new reachable states. 
	This motivates our definition of phase-bounded protocols, in which a process always ends in a reception phase. 
	
	Indeed, consider the protocol pictured on~\cref{fig:bp:example-2}. The state $q_3$ is not coverable: to cover $q_3$, a node $v_1$ needs to receive message $a$ when it is on state $q_2$ from one of its neighbor $v_2$. By construction, $v_1$ broadcasts message $m$ before reaching $q_2$.
	Vertex $v_2$ can only broadcast message $a$ when it is on state $q_5$. To reach $q_5$, vertex $v_2$ visited exactly states $\qinit, q_4$ and $q_5$. From each of those states, there is an outgoing reception transition labelled with $m$ going to $p$. Hence, at any moment of the execution, the broadcast of message $m$ by vertex $v_1$ would have brought $v_2$ in $p$, preventing the broadcast of $a$.
	However, a naive 2-phase bounded unfolding of this protocol, in which we limit the number of phases of sending and reception to 2, is illustrated in~\cref{fig:bp:example-2-unfold}. In 
	this protocol, $(q_3,r,2)$ (hence $q_3$) is coverable:
	consider $\Gamma = (\set{v_1, v_2}, \set{\langle v_1, v_2\rangle})$ and the following execution:
	\begin{align*}
		&(\Gamma, \set{v_1 \mapsto \qinit, v_2\mapsto \qinit}) \trans (\Gamma, \set{v_1 \mapsto (q_1,b,1), v_2\mapsto (q_4,r,1)}) \trans (\Gamma, \set{v_1 \mapsto (q_1,b,1),\\
			& v_2\mapsto (q_5,b,2)}) \trans 
		(\Gamma, \set{v_1 \mapsto (q_2,b,1), v_2\mapsto (q_5,b,2)}) \trans(\Gamma, \set{v_1 \mapsto (q_3,r,2), v_2\mapsto (q_6,b,2)}).
	\end{align*}
	
	In fact, in state $(q_5,b,2)$ a process is not allowed to switch anymore, hence the transition allowing to receive message $m$ has been removed. Doing so, 
	we have made state $q_3$ coverable. This shows that this type of bounded semantics does not give an underapproximation of the coverable states, in spite of what was expected. 
	
	Actually, the reception of a message is something that is not under the control of a process. If another process broadcasts a message, a faithful behavior of the system is that all the processes that can receive it indeed do so, no matter in which phase they are in their own execution. Hence, in a restriction that attempts to limit the number of switches 
	between sending and receiving phases, one should not prevent a reception to happen. This motivates our definition of a phase-bounded protocol.

	Let $\PP = (Q, \Sigma, \qinit, \Delta)$ be a broadcast protocol, and $k \in \nat$. 
	
	We define the $k$-unfolding of $\PP$ denoted by $\PP_k$ as the following protocol: 
	$\PP_k =(Q_k, \Sigma, \qinit, \Delta_k)$ with 
	$Q_k= \{q^0\mid q\in Q\}\cup \{q^{b,j}, q^{r,j}\mid q\in Q, 1\leq j\leq k\}$. To ease the notations we let $q^{r,0}=q^0$, $q^{b,0}=q^0 = $ for all $q^0 \in Q_0$.
	\begin{align*}
		\Delta_k & =  \set{(q^0, \tau, p^0) \mid (q, \tau, p) \in \Delta}\\
		&\cup		
		\set{(q^{r,j}, \alpha, p^{r,j}) \mid 1 \leq j \leq k \text{ and }(q, \alpha, p) \in \Delta \text{ and } \alpha \in \set{\tau} \cup ?\Sigma}\\
		& \cup 
		\set{(q^{b,j}, \alpha, p^{b,j}) \mid 1 \leq j \leq k \text{ and }(q, \alpha, p) \in \Delta \text{ and } \alpha \in \set{\tau} \cup !!\Sigma}\\
		& \cup 
		\set{(q^{r,j}, !!m, p^{b,j+1}) \mid 0 \leq j < k \text{ and }(q, !!m, p) \in \Delta}\\
		& \cup 
		\set{(q^{b,j}, ?m, p^{r,j+1}) \mid 0 \leq j < k \text{ and }(q, ?m, p) \in \Delta}\\
		& \cup 
		\set{(q^{b,k}, ?m, p^{r,k}) \mid (q, ?m, p) \in \Delta}.
	\end{align*}
	
	The last case of the definition of $\Delta_k$ implements the fact that we never prevent a reception from occurring, even in the last phase. 
	\begin{figure}
		\begin{center}
			\tikzset{box/.style={draw, minimum width=4em, text width=4.5em, text centered, minimum height=17em}}

\begin{tikzpicture}[-, >=stealth', shorten >=1pt,node distance=2.5cm,on grid,auto, initial text = {}] 
	\node[] (q0) {};
	\node[state,initial] (q0b) [above  of = q0, yshift = -60, xshift =0] {$\qinit^0$};
	\node[state] (q1) [right of = q0, yshift = -0, xshift =0.4cm] {$q_4^{b,1}$};
	\node[state] (q2) [right  of = q1, yshift = 0, xshift = 4.5cm] {$q_5^{r,2}$};
	\node[] (q3) [above  of = q0, yshift = -30, xshift =0] {};
	\node[state] (q4) [right of = q3, yshift = 5, xshift = -20] {$q_1^{r,1}$};
	
	\node[state] (q5) [right  of = q4, xshift = -20, yshift = -10]{$q_2^{r,1}$};
	\node[state] (q6) [right  of = q5, xshift = 20, yshift = -15]{$q_3^{b,2}$};
	
	\node[state] (qin2) [above of = q5, yshift = -60, xshift = 50] {$\qinit^{b,2}$};
	\node[state] (q42) [right of = qin2, yshift = -0, xshift = 0] {$q_4^{b,2}$};
	\node[state] (q12) [right of = qin2, yshift = 15, xshift = 60] {$q_1^{r,2}$};
	\node[state] (q22) [right of = q12] {$q_2^{r,2}$};
	
	\node [] (p1) [left of = q4, yshift = -18] {};
	\node[] (p2) [left of = q1, yshift = 0, xshift =-2cm] {};

	\node[dashed, draw, fill = orange, fill opacity = 0, text opacity = 1, fit=(q0b), text height=0.06 \columnwidth, label ={below:Phase 0}] () {};
	\node[dashed, draw, fill = green, fill opacity = 0, text opacity = 1, fit=(q1) (q4) (q5), text height=0.06 \columnwidth, label ={below:Phase 1}, inner xsep = 8, inner ysep=10] () [yshift = 5]{};
	\node[dashed, draw, fill = cyan, fill opacity = 0, text opacity = 1, fit=(q2) (q6) (qin2) (q42) (q12) (q22), text height=0.06 \columnwidth, label ={below:Phase 2}, inner ysep=8] () {};

	\node[fill=violet, fill opacity = 0.2, fit=(q1), inner sep =1] () {};
	\node[fill=green, fill opacity = 0.2, fit=(q4) (q5), inner sep =1] () {};
	
	\node[fill=violet, fill opacity = 0.2, fit=(q42) (qin2) (q6), inner sep =1] () {};
	\node[fill=green, fill opacity = 0.2, fit=(q12) (q22) (q2), inner sep =1] () {};

	\path[->] 
	(q0b) edge [bend right =10] node [xshift =-5] {$!!a$, $!!b$} (q1) 
	(q1) edge node [below] {$?c$} (q2) 
	(q0b) edge [bend left =20] node {$?b$} (q4)
	(q4) edge [bend right = 10] node [xshift = -5]{$?a$} (q5)
	(q5) edge [bend right = 10] node [below, xshift = -5, yshift = 2] {$!!c$} (q6)
	(q6) edge [bend left = 0, thick]  node {$?a$} (q2)
	
	(q4) edge [bend left = 20] node {$!!a$} (qin2)
	(qin2) edge [bend left = 20, thick]  node {$?b$} (q12)
	(qin2) edge [] node [yshift =-5]{$!!a, !!b$} (q42)
	(q12) edge [] node {$?a$} (q22)
	(q42) edge [bend left = 20, thick]  node {$?a$} (q2)
	
	;
\end{tikzpicture}
		\end{center}
		\caption{$\PP_2$: the 2-unfolding of protocol \PP\ (\cref{fig:bp:examples}).}
		\label{fig:unfolding:examples}
	\end{figure}

	\begin{example}
		\cref{fig:unfolding:examples} pictures the 2-unfolding of protocol \PP\ (\cref{fig:bp:examples}). Observe that from state $q_4^{b,2}$, which is a broadcast state,
		it is still possible to receive message $a$ and go to state $q_5^{r,2}$. However, it is not possible to send a message from $q_5^{r,2}$ (nor from any reception state of phase 2). 
	\end{example}
	
	The following lemma establishes that this definition of unfolding can be used as an underapproximation for $\Cover$. 
	Formal proof can be found on \cref{app:cover-cover-phase-bounded}.
	\begin{lemma}\label{lemma:cover-general-protocol-phase-bounded-protocol}
		$q_f$ can be covered in $\PP$ if and only if there exist $k\in\nat$, $y\in\set{r,b}$ and $0\leq j\leq k$ such that $(q_f,y,j)$ can be covered in $\PP_k$.
\end{lemma}
}
In the remaining of the paper, we study the verification problems introduced in~\cref{subsec:verification-problems} when considering phase-bounded behaviors. 
We turn this restriction into a syntactic one over the protocol, defined as follows. 

\begin{definition}
	Let $k \in \nat$. A broadcast protocol $\PP = (Q, \Sigma, \qinit, \Delta)$ is \emph{$k$-phase-bounded} if $Q$ can be partitioned into $2k+1$ sets $\mathcal{Q} = \set{Q_0, Q_1^b, Q_1^r, \dots Q_{k}^b, Q_{k}^r}$, such that $\qinit \in Q_0$ and for all $(q, \alpha, q') \in \Delta$ one of the following conditions holds:
	\begin{enumerate}
		\item there exist $0\leq i \leq k$ and $\beta \in \set{r,b}$ such that $q,q' \in Q_i^\beta$ and $\alpha = \tau$ (for ease of notation, we take $Q_0= Q_0^b= Q_0^r$);
		\item there exists $1\leq i \leq k$ such that $q, q' \in Q_i^{b}$ and $\alpha \in !!\Sigma$;
		\item there exists $1\leq i \leq k$ such that $q, q' \in Q_i^{r}$ and $\alpha \in ?\Sigma$;
		\item there exists $0\leq i < k$ such that $q \in Q_i^{b}$, $q' \in Q_{i+1}^r$ and $\alpha \in ?\Sigma$;
		\item there exists $0\leq i < k$ such that $q \in Q_i^{r}$, $q' \in Q_{i+1}^b$ and $\alpha \in !!\Sigma$;
		\item $q \in Q_k^b$, $q'\in Q_k^r$ and $\alpha \in ?\Sigma$
	\end{enumerate}
	
\end{definition}
A protocol \PP\ is phase-bounded if there exists $k \in \nat$ such that \PP\ is $k$-phase-bounded. 
\begin{example}
	Observe that the protocol $\PP$ displayed in \cref{fig:bp:examples}\ is not phase-bounded: by definition, it holds that $Q_0 = \set{\qinit}$, and $q_1 \in Q_1^r$ (because of the transition $(\qinit, ?b, q_1)$). As a consequence $\qinit \in Q_2^b$, because of the transition $(q_1, !!a, \qinit)$. This contradicts the fact that $Q_2^b \cap Q_0 = \emptyset$. Intuitively, $\PP$ does not ensure that every vertex alternates at most a bounded number of times between receptions and broadcasts, in particular, for any integer $k \in \nat$, it might be that there exists an execution where a process alternates $k+1$ times between reception of a message $b$ from 
	state $\qinit$, and broadcast of a message $a$ from state $q_1$. 
	Removing the transition $(q_1, !!a, \qinit)$ from $\PP$ would give a 2-phase-bounded protocol $\PP'$: $Q_0 = \set{\qinit}$, $Q_1^r = \set{q_1, q_2}$, $Q_1^b = \set{q_4}$, $Q_2^b= \set{q_3}$ and $Q_2^r = \set{q_5}$. 
	\Iflong{
	Note that the execution presented in \cref{fig:execution-example}\ is not an execution of $\PP'$ as process $v_2$ takes transition $(q_1, !!a, \qinit)$. However, $q_5$ is still coverable, but with a different topology. An execution covering $q_5$ is displayed in \cref{fig:execution-example-2}. Note that this time $v_3$ and $v_2$ are not neighbors, this is necessary as otherwise the first broadcast of $v_2$ would bring $v_3$ in $q_1$. \nas{j'ai pas compris ce dernier paragraphe. Je propose de supprimer}\lug{c'était pour donner une execution de $\PP'$ qui couvre $q_f$ (ce ne peut pas etre celle de fig 1 car un noeud prend la transition hachée), si vous voulez l'enlever, on a plus besoin de la figure 4 et ça fait gagner un peu de place}
}
\end{example}
\Iflong{
\begin{figure}
	\begin{center}
		\begin{tikzpicture}[->, >=stealth', shorten >=1pt,node distance=1.2cm,on grid,auto, initial text = {}] 
	\AConfiguration{\qinit}{\qinit}{\qinit}{0}{0}{0}
	\draw [-stealth](1.6,-0.5) -- node {$v_2$} (2,-0.5) ;
	\AConfiguration{q_1}{q_4}{\qinit}{0}{3.6cm}{0}
	\draw [-stealth](5.2,-0.5) -- node {$v_3$} (5.6,-0.5) ;
	\AConfiguration{q_2}{q_4}{q_4}{0}{7.2cm}{0}
	\draw [-stealth](8.8,-0.5) -- node {$v_1$} (9.2,-0.5) ;
	\AConfiguration{q_3}{q_5}{q_5}{0}{10.8cm}{0cm}
\end{tikzpicture}
	\end{center}
	\caption{Example of an execution of protocol $\PP'$ (\cref{fig:bp:examples}).}
	\label{fig:execution-example-2}
\end{figure}
}
The following table summarizes our results (PB stands for phase-bounded).
\begin{center}
	\begin{tabular}{|l|c|c|c|}
		\hline
		& \textbf{1-PB Protocols} & \textbf{2-PB Protocols} & \textbf{PB Protocols} \\
		\hline
		$\CoverLine$  & \multicolumn{2}{c|}{$\in P$ (\cref{subsec:coverLine-inP})}&  Undecidable  ($k \geq 4$) (Sec \ref{sec:undecidable})\\
		\hline
		$\Cover[\Topo]$ & $\textsc{ExpSpace}$-complete&  Decidable & Undecidable  ($k \geq 6$) \\
		$\CoverTree$ &(Section \ref{sec:about-1pb})&(Section \ref{subsec:proof-dec-2phases-covertree})&(Section \ref{sec:undecidable})\\
		\hline
	\end{tabular}
\end{center}

\Iflong{
  \section{\Cover\ and \CoverTree\ are equivalent}\label{subsec:Cover-CoverTree-equivalent}

We start by proving that \Cover\ and \CoverTree\ are actually equivalent. If it is obvious that when a state is coverable with a tree topology, it is coverable by a topology, we show in this section that whenever a state is coverable, it is coverable with a tree topology. Observe that this
result holds for any broadcast protocol, and not only phase-bounded ones. 

Let $\PP = (Q, \Sigma, \qinit, \Delta)$ be a broadcast protocol, and $q_f\in Q$. Let $\rho=C_0 \trans \cdots \trans C_n$ with $C_i = (\Gamma, L_i)\in\CC_\PP$ for all $0 \leq i \leq n$, and a vertex $v_f \in \Vert{\Gamma}$ such that $L_n(v_f) = q_f$.
We will build an execution covering $q_f$ with a tree topology $\Gamma'$, rooted in $v_f$, the node that reaches $q_f$. 
$\Gamma'$ is actually an \emph{unfolding} of the topology $\Gamma$.

%
%

We first define inductively the set of nodes $V'\subseteq \nat^\ast$, along with a labelling function $\lambda$ which associates to each node $v' \in V'$ a node $v \in \Vert{\Gamma}$. 
\begin{enumerate}
	\item $\epsilon\in V'$ and $\lambda(\epsilon) = v_f$;\label{enumerate:def-tree:epsilon}
	\item Let $\NeighG{\Gamma}{v_f} = \set{v_1, \dots , v_k}$. For all $1 \leq i \leq k$, $i \in V'$ and $\lambda(i) = v_i$;\label{enumerate:def-tree:depth-1}
	\item Let $w \cdot x \in V'$, with $w \in \nat^\ast$ such that $|w| < n-1$, and $x \in \nat$. Let $\NeighG{\Gamma}{\lambda(w\cdot x)} \setminus \set{\lambda(w)} = \set{v_1, \dots, v_k}$. Then, for all $1 \leq i \leq k$, $w\cdot x \cdot i \in V'$ and $\lambda'(w\cdot x \cdot i) = v_i$.\label{enumerate:def-tree:induction} 
\end{enumerate}
Finally, define $E' = \set{\langle w, w\cdot x\rangle  \mid w \in V', w\cdot x \in V'}$.
Note that $\epsilon \in V'$ and $V'\subseteq \nat^\ast$ and  is prefix closed.
Furthermore, by construction, for all $w \in V'$, $|w| \leq n$, and for all $w \in V'$, in fact $w\in \set{1,\dots, d}^\ast$ 
where $d$ is the maximal degree of $\Gamma$. 
Hence, $V'$ is a finite set and  $\Gamma'=(V',E')$ is a tree topology. 

The way we built $\Gamma'$ ensures that each node $v\in V'$ (except the leaves) enjoys the same set of neighbours than $\lambda(v)\in V$. This is formalised in the
following lemma whose proof can be found in~\cref{appendix:covertree-cover-equivalent}.
\begin{lemma}\label{lemma:Covertree-Cover:well-formed-tree}
For all $u\in V'$, for all $u'\in \NeighG{\Gamma'}{u}$, $\lambda(u')\in \NeighG{\Gamma}{\lambda(u)}$. Then, we let $f_u:\NeighG{\Gamma'}{u} \rightarrow \NeighG{\Gamma}{\lambda(u)}$ defined by $f_u(u')=\lambda(u')$. If $|u|< n$, $f_u$
is a bijection.
\end{lemma}

From the execution $\rho = C_0\rightarrow \dots C_n$ we will build a similar execution on $\Gamma'$. The idea is that for each step of the execution $\rho$, 
for each node $v\in V$, all the nodes in $\Gamma'$ that are labelled by $v$ will behave in the same way. This is possible thanks to~\cref{lemma:Covertree-Cover:well-formed-tree}.
Observe though that the leaves might no be able to behave as expected because they might not have the same set of neighbours than the node they are labelled by, so they might not be able to receive some
broadcast message. However, we can ensure some weaker version of correctness, defined as follows.
Let $0 \leq h \leq n$ and $C= (\Gamma, L)$ be a configuration. We say that a configuration $C' = (\Gamma', L')$ is \emph{$h$-correct} for $C$ if for all $u\in\Vert{\Gamma'}$, if $|u|\leq h$ then $L'(u)=L(\lambda(u))$.

%
%
The following lemma gives the main ingredient that allows to mimick the execution $\rho$ on $\Gamma'$ (formal proof can be found on \cref{appendix:covertree-cover-equivalent}).

\begin{lemma}\label{lemma:Covertree-Cover:induction-step}
	Let $C_1 = (\Gamma, L_1)$, $C_2 = (\Gamma, L_2)\in \CC_\PP$ such that $C_1 \trans C_2$ and let $C'_1 = (\Gamma', L'_1)\in\CC_\PP$ and $0< h \leq n$ such that $C'_1$ is $h$-correct for $C_1$.
	There exists $C'_{2}\in \CC_\PP$ such that $C'_1 \trans^\ast C'_{2}$, and $C'_2$ is $(h-1)$-correct for $C_2$.	
\end{lemma}

We build now an execution covering $q_f$ with $\Gamma'$: let $C'_0=(\Gamma', L'_0)$ defined by $L'_0(v)=\qinit$ for all $v\in\Vert{\Gamma'}$. Obviously,
$C'_0$ is $n$-correct. By~\cref{lemma:Covertree-Cover:induction-step}, there exists a sequence of configurations $(C'_i)_{1\leq i \leq n}$ such that
for all $1\leq i\leq n$, $C'_{i-1}\trans^\ast C'_{i}$ and $C'_i$ is $(n-i)$-correct for $C_i$. Hence, $C'_0\trans^\ast C'_n$ and $C'_n$ is $0$-correct for $C_n$ and
$L'_n(\epsilon)=L_n(\lambda(\epsilon))=L_n(v_f)=q_f$. 

This allows to prove the following result. 
\begin{theorem}\label{thm:Cover-CoverTree-equivalent}
$\Cover$ and $\CoverTree$ are equivalent. 
\end{theorem}

}
\section{Undecidability Results}\label{sec:undecidable}



We prove that $\Cover$ restricted to $k$-phase-bounded protocols (with $k \geq 6)$ is undecidable  by a reduction from the halting problem of a Minksy machine~\cite{Minsky67}: a Minsky machine is a finite-state 
machine (whose states are called locations) with two counters, $\textsf{x}_1$ and $\textsf{x}_2$ (two variables that take their values in $\nat$). Each transition of the machine is associated with an instruction: increment one of the counters, decrement one of the counters or test if one of the counters is equal to 0. The halting problem asks whether there is an execution that
ends in the halting location. 
In a first step, the protocol will enforce the selection of a line of nodes from the topology. All other nodes will be inactive.  
In a second step, the first node of the line (that we call the head) visits the different states of the machine during an execution, while all other nodes 
(except the last one) simulate counters' values: they are either in a state representing value 0, or a state representing $\textsf{x}_1$ (respectively $\textsf{x}_2$). 
The number of processes on states representing $\textsf{x}_1$ gives the actual value of $\textsf{x}_1$ in the execution.  The last node (called the tail) checks that everything happens as expected. When the head has reached the halting location of the machine, it broadcasts a message which is received and forwarded by each node of the line until the tail receives it and reaches the final state to cover.

When the head of the line simulates a transition of the machine, it broadcasts a message (the instruction for one of the counters), which is transmitted by each node of the line until the tail receives it. A classical way of forwarding the message through receptions and broadcasts would not give a phase-bounded protocol. 
Hence, during the transmission, the tail only receives messages and all other nodes only broadcast and do not receive any message. 
The main idea is that we do not use the reception of messages to move into the next state of the execution but to detect errors (and in that case, go to a bad sink state from which the process can not do anything). 
The processes will have to guess the correct message to send, and the correct instant to send it, otherwise some of them will go to the sink state upon the reception of this "wrong" message. Hence, when everyone makes the correct guesses, the only reception that occurs in the transmission is done by the tail process, whereas when someone makes an incorrect guess, a process goes to a bad state with a reception. In the reduction, if the halting state of the Minsky Machine is not reachable, there will be no way to make a correct guess that allows to cover the final state.
In the next subsection, we explain how this is achieved. To do so, we explain the mechanism by abstracting away the actual instruction, and just show how
to transmit a message.


\subsection{Propagating a message using only broadcasts in a line}\label{subsec:undec:transmission}
%


\begin{figure}
	\begin{minipage}[c]{0.5\linewidth}
		\resizebox*{0.8\linewidth}{!}{
			\tikzset{box/.style={draw, minimum width=4em, text width=4.5em, text centered, minimum height=17em}}

\begin{tikzpicture}[-, >=stealth', shorten >=1pt,node distance=2cm,on grid,auto, initial text = {}] 
	\node[state, initial] (q0) {$\textsf{s}_0$};
	\node[state] (q1) [right of = q0] {$\textsf{s}_1$};
	\node[state] (q2) [right of = q1] {$\textsf{s}_2$};
	
	\node[inner sep = 0] (f0) [below of =q0, yshift =10] {$\Huge\frownie$};
	\node[inner sep = 0] (f1) [below of =q1, yshift =10] {$\Huge\frownie$};
	\node[inner sep = 0] (f2) [below of = q2, yshift =10] {$\Huge\frownie$};
	
	\path[->, thick] 
	(q0) edge node {$!!{\text{td}}_0$} (q1)
	(q1) edge node {$!!\overline{\text{td}}_0$} (q2)
	(q0) edge node [xshift = -27]{
		\begin{tabular}{c c}
				$?m, \ m\in \Sigma$
		\end{tabular}}(f0)
	
	(q1) edge node [xshift = -30]{
		\begin{tabular}{c c}
			$?m, \ m\nin $\\
			$\set{{\text{td}}_1, \ {\text{d}}_1}$
	\end{tabular}}(f1)

(q2) edge node [xshift = -30]{
	\begin{tabular}{c c}
		$?m, \ m\nin $\\
		$\set{\overline{\text{td}}_1, \ \overline{\text{d}}_1}$
	\end{tabular}}(f2)

%
	;
\end{tikzpicture}
		}
		\caption{Protocol $P_h$ executed by $v_0$.}\label{fig:undec:phead}
	\end{minipage}
	\hfill
	\begin{minipage}[c]{0.5\linewidth}
		\resizebox*{1.0\linewidth}{!}{
			\tikzset{box/.style={draw, minimum width=4em, text width=4.5em, text centered, minimum height=17em}}

\begin{tikzpicture}[-, >=stealth', shorten >=1pt,node distance=2cm,on grid,auto, initial text = {}] 
	\node[state, inner sep = 1, initial above] (p) {$\idleState$};
	\node[state, inner sep = 1] (p1) [right of = p, xshift = 0] {$\checkState$};
	
	\node[inner sep = -3] (f0) [left of =p, xshift = -30] {$\Huge\frownie$};
	\node[inner sep = -3] (f1) [right of =p1, xshift = 15] {$\Huge\frownie$};

	\path[->, thick] 
	(p) edge [bend left =15] node {$?{\text{d}}_1$} (p1)
	(p1) edge [bend left =15] node {$?\overline{\text{d}}_1$} (p)
	(p) edge node [above] {
		$?m, m \neq {\text{d}}_1$} (f0)
	(p1) edge node {$?m, m\neq \overline{\text{d}}_1$} (f1)
	;
\end{tikzpicture}
		}
		\caption{Protocol $P_t$ executed by $v_n$.}\label{fig:undec:ptail}
	\end{minipage}
\end{figure} 

\begin{figure}
	\begin{minipage}[c]{0.32\linewidth}
	\resizebox*{1.2\linewidth}{!}{
		\tikzset{box/.style={draw, minimum width=4em, text width=4.5em, text centered, minimum height=17em}}

\begin{tikzpicture}[-, >=stealth', shorten >=1pt,node distance=2cm,on grid,auto, initial text = {}] 
	\node[state, inner sep = 1, initial left] (q0) {$\idleState^0$};
	\node[state, inner sep = 1] (q1) [below of = q0, yshift = -15] {$\execState^0$};
	\node[state, inner sep = 1] (q2) [below of = q1, yshift = 10] {$\restState^0$};
	
	\node[inner sep =-3] (f0) [right of = q0, xshift = 35] {$\Huge\frownie$};
	\node[inner sep = -3] (f1) [right of = q1, xshift = 35] {$\Huge\frownie$};
	\node[inner sep = -3] (f2) [right of = q2, xshift = 35] {$\Huge\frownie$};
	
	\node[state, inner sep = 1] (qa) [right = of q0, yshift = 35, xshift = -10] {$\repState_{\text{td}}^0$};
	\node[state, inner sep = 1] (qb) [right = of q0, yshift = -35, xshift = -10] {$\repState_{\text{d}}^0$};
	
	

	
	\path[->, thick] 
	(q0) 
	edge [in = 170, out = 90] node [yshift = -1, xshift = 5]{$!!{\text{td}}_0$} (qa)
	edge [in = 190, out = -70] node [below] {$!!{\text{d}}_0$} (qb)
	edge node [left, yshift = -5, xshift = 2]{$!!{\text{td}}_0$} (q1)
	
	(qa) edge [out = 190, in =70] node [yshift =4]{$!!\overline{\text{td}}_0$} (q0)
	(qb) edge [out = 170, in = -50] node [above, xshift = 7, yshift = -4]{$!!\overline{\text{d}}_0$} (q0)
	
	(q1) edge node  [left, yshift = 0, xshift = 2] {$!!\overline{\text{d}}_0$} (q2)
	
	%
	
	(q0) edge [bend left =20] node [above, yshift = -25, xshift =-0] {\footnotesize\begin{tabular}{c c}
			$?m,$\\
			$m \nin \set{ {\text{td}}_2,$\\
				$ {\text{d}}_2, \overline{\text{td}}_1, \overline{\text{d}}_1}$
	\end{tabular}} (f0)
	
	(q1) edge node [above, yshift = -25,xshift = -5] {\footnotesize\begin{tabular}{c c c}
			$?m,$\\
			$m\nin \set{\overline{\text{td}}_2, \overline{\text{d}}_2,$\\
				${\text{td}}_1, {\text{d}}_1}$
	\end{tabular}} (f1)
	
	(q2) edge [bend right = 0] node [above, yshift = -26, xshift = -5] {\footnotesize\begin{tabular}{c  c}
			$?m,$\\
			$m \nin \set{{\text{td}}_2, {\text{d}}_2,$\\
				$\overline{\text{td}}_1, \overline{\text{d}}_1}$
	\end{tabular}} (f2)
	
	(qa) edge [out =0, in = 140] node [right, yshift = 8, xshift = -10]{\begin{tabular}{l c}
			$?m,m\nin$\\
			$ \set{\overline{\text{td}}_2, {\text{td}}_1}$
	\end{tabular}}(f0)
	
	(qb) edge [out =0, in = 220]node [right, yshift = -10, xshift =-10]  {\begin{tabular}{l c}
			$?m,m\nin$\\
			$ \set{\overline{\text{d}}_2, {\text{d}}_1}$
	\end{tabular}}(f0)
	%
	
	
	;
	
\end{tikzpicture}
	}
	\caption{$\PP_0$.}\label{fig:undec:p0}
\end{minipage}
	\hfill
	\begin{minipage}[c]{0.32\linewidth}
		\resizebox*{1.2\linewidth}{!}{
			\tikzset{box/.style={draw, minimum width=4em, text width=4.5em, text centered, minimum height=17em}}

\begin{tikzpicture}[-, >=stealth', shorten >=1pt,node distance=2cm,on grid,auto, initial text = {}] 
	\node[state, inner sep = 1, initial left] (q0) {$\idleState^1$};
	\node[state, inner sep = 1] (q1) [below of = q0, yshift = -15] {$\execState^1$};
	\node[state, inner sep = 1] (q2) [below of = q1, yshift = 10] {$\restState^1$};

	\node[inner sep =-3] (f0) [right of = q0, xshift = 35] {$\Huge\frownie$};
	\node[inner sep = -3] (f1) [right of = q1, xshift = 35] {$\Huge\frownie$};
	\node[inner sep = -3] (f2) [right of = q2, xshift = 35] {$\Huge\frownie$};
	
	\node[state, inner sep = 1] (qa) [right = of q0, yshift = 35, xshift = -10] {$\repState_{\text{td}}^1$};
	\node[state, inner sep = 1] (qb) [right = of q0, yshift = -35, xshift = -10] {$\repState_{\text{d}}^1$};
	
	

	
	\path[->, thick] 
	(q0) 
	edge [in = 170, out = 90] node [yshift = -1, xshift = 5]{$!!{\text{td}}_1$} (qa)
	edge [in = 190, out = -70] node [below] {$!!{\text{d}}_1$} (qb)
	edge node [left, yshift = -5, xshift = 2]{$!!{\text{td}}_1$} (q1)
	
	(qa) edge [out = 190, in =70] node [yshift =4]{$!!\overline{\text{td}}_1$} (q0)
	(qb) edge [out = 170, in = -50] node [above, xshift = 7, yshift = -4]{$!!\overline{\text{d}}_1$} (q0)
	
	(q1) edge node  [left, yshift = 0, xshift = 2] {$!!\overline{\text{d}}_1$} (q2)
	
%
	
	(q0) edge [bend left =20] node [above, yshift = -25, xshift =-0] {\footnotesize\begin{tabular}{c c}
			$?m,$\\
			$m \nin \set{ {\text{td}}_0,$\\
				$ {\text{d}}_0, \overline{\text{td}}_2, \overline{\text{d}}_2}$
		\end{tabular}} (f0)
	
	(q1) edge node [above, yshift = -25,xshift = -5] {\footnotesize\begin{tabular}{c c c}
			$?m,$\\
			$m\nin \set{\overline{\text{td}}_0, \overline{\text{d}}_0,$\\
				${\text{td}}_2, {\text{d}}_2}$
	\end{tabular}} (f1)

	(q2) edge [bend right = 0] node [above, yshift = -26, xshift = -5] {\footnotesize\begin{tabular}{c  c}
			$?m,$\\
			$m \nin \set{{\text{td}}_0, {\text{d}}_0,$\\
				$\overline{\text{td}}_2, \overline{\text{d}}_2}$
	\end{tabular}} (f2)

	(qa) edge [out =0, in = 140] node [right, yshift = 8, xshift = -10]{\begin{tabular}{l c}
			$?m,m\nin$\\
			$ \set{\overline{\text{td}}_0, {\text{td}}_2}$
		\end{tabular}}(f0)
	
	(qb) edge [out =0, in = 220]node [right, yshift = -10, xshift =-10]  {\begin{tabular}{l c}
			$?m,m\nin$\\
			$ \set{\overline{\text{d}}_0, {\text{d}}_2}$
	\end{tabular}}(f0)
%

	
	;
	
\end{tikzpicture}
		}
		\caption{$\PP_1$.}\label{fig:undec:p1}
	\end{minipage}
	\hfill
\begin{minipage}[c]{0.32\linewidth}
	\resizebox*{1.2\linewidth}{!}{
		\tikzset{box/.style={draw, minimum width=4em, text width=4.5em, text centered, minimum height=17em}}

\begin{tikzpicture}[-, >=stealth', shorten >=1pt,node distance=2cm,on grid,auto, initial text = {}] 
	\node[state, inner sep = 1, initial left] (q0) {$\idleState^2$};
	\node[state, inner sep = 1] (q1) [below of = q0, yshift = -15] {$\execState^2$};
	\node[state, inner sep = 1] (q2) [below of = q1, yshift = 10] {$\restState^2$};
	
	\node[inner sep =-3] (f0) [right of = q0, xshift = 35] {$\Huge\frownie$};
	\node[inner sep = -3] (f1) [right of = q1, xshift = 35] {$\Huge\frownie$};
	\node[inner sep = -3] (f2) [right of = q2, xshift = 35] {$\Huge\frownie$};
	
	\node[state, inner sep = 1] (qa) [right = of q0, yshift = 35, xshift = -10] {$\repState_{\text{td}}^2$};
	\node[state, inner sep = 1] (qb) [right = of q0, yshift = -35, xshift = -10] {$\repState_{\text{d}}^2$};
	
	

	
	\path[->, thick] 
	(q0) 
	edge [in = 170, out = 90] node [yshift = -1, xshift = 5]{$!!{\text{td}}_2$} (qa)
	edge [in = 190, out = -70] node [below] {$!!{\text{d}}_2$} (qb)
	edge node [left, yshift = -5, xshift = 2]{$!!{\text{td}}_2$} (q1)
	
	(qa) edge [out = 190, in =70] node [yshift =4]{$!!\overline{\text{td}}_2$} (q0)
	(qb) edge [out = 170, in = -50] node [above, xshift = 7, yshift = -4]{$!!\overline{\text{d}}_2$} (q0)
	
	(q1) edge node  [left, yshift = 0, xshift = 2] {$!!\overline{\text{d}}_2$} (q2)
	
	%
	
	(q0) edge [bend left =20] node [above, yshift = -25, xshift =-0] {\footnotesize\begin{tabular}{c c}
			$?m,$\\
			$m \nin \set{ {\text{td}}_1,$\\
				$ {\text{d}}_1, \overline{\text{td}}_0, \overline{\text{d}}_0}$
	\end{tabular}} (f0)
	
	(q1) edge node [above, yshift = -25,xshift = -5] {\footnotesize\begin{tabular}{c c c}
			$?m,$\\
			$m\nin \set{\overline{\text{td}}_1, \overline{\text{d}}_1,$\\
				${\text{td}}_0, {\text{d}}_0}$
	\end{tabular}} (f1)
	
	(q2) edge [bend right = 0] node [above, yshift = -26, xshift = -5] {\footnotesize\begin{tabular}{c  c}
			$?m,$\\
			$m \nin \set{{\text{td}}_1, {\text{d}}_1,$\\
				$\overline{\text{td}}_0, \overline{\text{d}}_0}$
	\end{tabular}} (f2)
	
	(qa) edge [out =0, in = 140] node [right, yshift = 8, xshift = -10]{\begin{tabular}{l c}
			$?m,m\nin$\\
			$ \set{\overline{\text{td}}_1, {\text{td}}_0}$
	\end{tabular}}(f0)
	
	(qb) edge [out =0, in = 220]node [right, yshift = -10, xshift =-10]  {\begin{tabular}{l c}
			$?m,m\nin$\\
			$ \set{\overline{\text{d}}_1, {\text{d}}_0}$
	\end{tabular}}(f0)
	%
	
	
	;
	
\end{tikzpicture}
	}
	\caption{$\PP_2$.}\label{fig:undec:p2}
\end{minipage}
\end{figure}

In a line, a node has at most two neighbors, but cannot necessarily distinguish between the two (its left and its right one). 
To do so, nodes broadcast messages with subscript 0, 1 or 2, and we ensure that:
if a node broadcasts with subscript 1, its right [resp. left] neighbor broadcasts with subscript 0 [resp. subscript 2]. Similarly, if a node broadcasts with subscript 0 [resp. 2], its right neighbor broadcasts with subscript 2 [resp. 1] and its left one with subscript 1 [resp. 0].

Consider the five protocols displayed in \cref{fig:undec:p1,fig:undec:p2,fig:undec:p0,fig:undec:phead,fig:undec:ptail}. The states marked as initial are the ones
from which a process enters the protocol. 
Protocol $\PP_h$ is executed by the head of the line, $\PP_t$ by the tail of the line and other nodes execute either $\PP_0$, $\PP_1$ or $\PP_2$. 
Observe that messages 
go by pairs: $\text{td}_i, \overline{\text{td}}_i$ and $\text{d}_i$, $\overline{\text{d}}_i$ for all $i \in \set{0, 1, 2}$. 

The head broadcasts a request to be done with the pair of messages $\text{td}_0$, $\overline{\text{td}}_0$.
Each process in one of the $\PP_i$ starts in $\idleState^i$ and has a choice: either it transmits a message without executing it, or it ``executes'' it and tells it to the others. When it transmits a message not yet executed, it broadcasts the messages $\text{td}_i$ and $\overline{\text{td}}_i$ and visits states $\repState_{\text{td}}^i$ and $\idleState^i$. When it executes the request, it broadcasts the messages $\text{td}_i$ and $\overline{\text{d}}_i$ and visits states $\execState^i$ and $\restState^i$. Finally, when it transmits a request already done, it broadcasts the messages $\text{d}_i$ and $\overline{\text{d}}_i$ and visits states 
$\repState_{\text{d}}^i$ and $\idleState^i$. Once a process has executed the request (i.e. broadcast a pair $\text{td}_j$, $\overline{\text{d}}_j$ for some $j \in \set{0,1,2}$), only pairs $\text{d}_j$, $\overline{\text{d}}_j$, with $j \in \set{0,1,2}$, are transmitted in the rest of the line.

\begin{figure}
	\begin{center}	
		\resizebox*{0.9\linewidth}{!}{
			\tikzset{box/.style={draw, minimum width=4em, text width=4.5em, text centered, minimum height=17em}}

\begin{tikzpicture}[-, >=stealth', shorten >=1pt,node distance=2cm,on grid,auto, initial text = {}] 
	\node[rounded rectangle, draw, inner sep =2] (v0) {$v_0: \textsf{s}_0$};
	\node[rounded rectangle, draw, inner sep =2] (v1) [right of =v0] {$v_1: \idleState^1$};
	\node[rounded rectangle, draw, inner sep =2] (v2) [right of =v1] {$v_2: \idleState^2$};
	\node[rounded rectangle, draw, inner sep =2] (v0b) [right of =v2] {$v_3: \idleState^0$};
	\node[rounded rectangle, draw, inner sep =2] (v1b) [right of = v0b] {$v_4:\idleState^1$};
	\node[rounded rectangle, draw, inner sep =2] (v1bb) [right of = v1b, xshift = 10] {$v_{n-1}: \idleState^1$};
	\node[rounded rectangle, draw, inner sep =2] (v2t) [right = of v1bb] {$v_n: \idleState$};
	
	\path (v1b) -- node[auto=false]{\ldots} (v1bb);
	
	\path[-] 
	(v0) edge node {} (v1)
	(v1) edge node {} (v2)
	(v2) edge node {} (v0b)
	(v0b) edge node {} (v1b) 
	(v1bb) edge node {} (v2t) 
	;
\end{tikzpicture}
		}
		\caption{A configuration from which the transmission can happen: a node in state $\idleState^i$ can only broadcast messages with subscript $i$.
		}\label{fig:undec:line-topo}
	\end{center}
\end{figure}
\begin{figure}
	\begin{minipage}[c]{0.45\linewidth}
		\begin{center}	
		\resizebox*{!}{0.15\paperheight}{
			\tikzset{box/.style={draw, minimum width=4em, text width=4.5em, text centered, minimum height=17em}}

\begin{tikzpicture}[-, >=stealth', shorten >=1pt,node distance=1.9cm,on grid,auto, initial text = {}] 
	\LineTikzZoomOne{0cm}{\textsf{s}_0}{\idleStateAbrev^1}{\idleStateAbrev^2}{\idleStateAbrev^0}{q_1^2}{p}
	\draw[->] (0,-0.12) -- (0, -0.7);
	\node at (0.5,-0.4) () {$!!{\text{td}}_0$};
	\LineTikzZoomOne{-0.85cm}{\textsf{s}_1}{\idleStateAbrev^1}{\idleStateAbrev^2}{q_0^0}{q_0^1}{p}
	\draw[->] (1.8,-1.10) -- (1.8, -1.65);
	\node at (2.3,-1.35) () {$!!{\text{td}}_1$};
	\LineTikzZoomOne{-1.85cm}{\textsf{s}_1}{\repStateAbrev_{\text{td}}^1}{\idleStateAbrev^2}{q_0^0}{q_0^1}{p}
	\draw[->] (0,-2.0) -- (0, -2.6);
	\node at (0.5,-2.3) () {$!!\overline{\text{td}}_0$};
	\LineTikzZoomOne{-2.7cm}{\textsf{s}_2}{\repStateAbrev_{\text{td}}^1}{\idleStateAbrev^2}{q_0^0}{q_0^1}{p}
	\draw[->] (3.5,-2.95) -- (3.5, -3.4);
	\node at (4.,-3.15) () {$!!{\text{td}}_2$};
	\LineTikzZoomOne{-3.6cm}{\textsf{s}_2}{\repStateAbrev_{\text{td}}^1}{\repStateAbrev_{\text{td}}^2}{q_0^0}{q_0^1}{p}
	\draw[->] (1.8,-3.82) -- (1.8, -4.4);
	\node at (2.3,-4.1) () {$!!\overline{\text{td}}_1$};
	\LineTikzZoomOne{-4.6cm}{\textsf{s}_2}{\idleStateAbrev^1}{\repStateAbrev_{\text{td}}^2}{q_0^0}{q_0^1}{p}
	
\end{tikzpicture}
		}
		\subcaption{$C_0 \trans C_1 \trans \cdots \trans C_5$.}\label{fig:undec:start-transmission}
	\end{center}
	\end{minipage}
	\hfill
	\begin{minipage}[c]{0.45\linewidth}
		\begin{center}	
		\resizebox*{!}{0.15\paperheight}{
			\tikzset{box/.style={draw, minimum width=4em, text width=4.5em, text centered, minimum height=17em}}

\begin{tikzpicture}[-, >=stealth', shorten >=1pt,node distance=1.9cm,on grid,auto, initial text = {}] 
	\LineTikzZoomTwo{0cm}{q_0}{\idleStateAbrev^1}{\repStateAbrev_{\text{td}}^2}{\idleStateAbrev^0}{\idleStateAbrev^1}{\idleStateAbrev}
	\draw[->] (2.9,-0.25) -- (2.9, -0.75);
	\node at (2.5,-0.45) () {$!!{\text{td}}_0$};
	\LineTikzZoomTwo{-0.95cm}{q_1}{\idleStateAbrev^1}{\repStateAbrev_{\text{td}}^2}{\execStateAbrev^0}{\idleStateAbrev^1}{\idleStateAbrev}
	\draw[->] (1.2,-1.2) -- (1.2, -1.8);
	\node at (0.8,-1.45) () {$!!\overline{\text{td}}_2$};
	\LineTikzZoomTwo{-2cm}{q_1}{q_{\text{td}}^1}{\idleStateAbrev^2}{\execStateAbrev^0}{\idleStateAbrev^1}{\idleStateAbrev}
	\draw[->] (4.7,-2.25) -- (4.7, -2.8);
	\node at (4.3,-2.6) () {$!!{\text{d}}_1$};
	\LineTikzZoomTwo{-3.05cm}{q_2}{q_{\text{td}}^1}{\idleStateAbrev^2}{\execState^0}{\repStateAbrev_{\text{d}}^1}{\checkStateAbrev}
	\draw[->] (2.9,-3.3) -- (2.9, -3.8);
	\node at (2.5,-3.55) () {$!!\overline{\text{d}}_0$};
	\LineTikzZoomTwo{-4.cm}{q_1}{q_{\text{td}}^1}{\idleStateAbrev^2}{\restStateAbrev^0}{\repStateAbrev_{\text{d}}^1}{\checkStateAbrev}
	\draw[->] (4.7,-4.25) -- (4.7, -4.85);
	\node at (4.3,-4.5) () {$!!\overline{\text{d}}_1$};
	\LineTikzZoomTwo{-5.05cm}{q_1}{\idleStateAbrev^1}{\idleStateAbrev^2}{\restStateAbrev^0}{\idleStateAbrev^1}{\idleStateAbrev}
	\fill [fill=violet, fill opacity = 0.15] (2.1,0.3) rectangle (3.7,-4.35);
\end{tikzpicture}
		}
		\subcaption{$C_6 \trans C_7 \trans \cdots \trans C_{11}$.}\label{fig:undec:end-transmission}
		\end{center}
	\end{minipage}
\caption{Example of correct transmission.}
\end{figure}
	

\emph{Correct transmission of a request.}
Take for instance the configuration $C_0$ depicted in \cref{fig:undec:line-topo}\ for $n =5$ (i.e. there are six vertices).
We say that a configuration is \emph{stable} if the head is in $\textsf{s}_0$ or $\textsf{s}_2$, the tail is in $\idleState$ and other nodes are in $\idleState^i$ or $\restState^i$ for $i \in \set{0,1,2}$. Note that $C_0$ is stable.
We depict a transmission in \cref{fig:undec:start-transmission,fig:undec:end-transmission}, starting from $C_0$. 
We denote the successive depicted configurations $C_0, C_1, \dots C_{11}$. 
Note that $C_{11}$ is stable.
Between $C_0$ and $C_{11}$, the following happens: 
\Ifshort{Between $C_0$ and $C_3$, $v_0$ broadcasts the request with messages $\text{td}_0$ and $\overline{\text{td}}_0$. Between $C_1$ and $C_8$, $v_1$ and $v_2$ successively repeat the request to be done with messages $\text{td}_1$ and $\overline{\text{td}}_1$ for $v_1$ and $\text{td}_2$ $ \overline{\text{td}}_2$ for $v_2$. Between $C_6$ and $C_{10}$, $v_3$ executes the request by broadcasting messages $\text{td}_0$ and $\overline{\text{d}}_0$. Between $C_7$ and $C_{11}$, $v_4$ transmits the done request with messages $\text{d}_1$ and $\overline{\text{d}}_1$.}
\Iflong{
\begin{itemize}
	\item between $C_0$ and $C_3$, $v_0$ broadcasts the request with messages $\text{td}_0$ and $\overline{\text{td}}_0$;
	\item between $C_1$ and $C_8$, $v_1$ and $v_2$ successively repeat the request to be done with messages $\text{td}_1$ and $\overline{\text{td}}_1$ for $v_1$ and $\text{td}_2$ $ \overline{\text{td}}_2$ for $v_2$;
	\item between $C_6$ and $C_{10}$, $v_3$ executes the request by broadcasting messages $\text{td}_0$ and $\overline{\text{d}}_0$;
	\item between $C_7$ and $C_{11}$, $v_4$ transmits the done request with messages $\text{d}_1$ and $\overline{\text{d}}_1$.
\end{itemize}
}
Hence, the request is executed by exactly one vertex 
(namely $v_3$), as highlighted in \cref{fig:undec:end-transmission}. Observe that the processes sort of spontaneously emit broadcast to avoid to receive a
message. A correct guess of when to broadcast yields the interleaving of broadcasts that we have presented in this example. 

\emph{How to prevent wrong behaviors?}
Observe that, when a node is in state $\idleState^1$, if one of its neighbor broadcasts a message which is not $\text{td}_0, \text{d}_0$ or $\overline{\text{td}}_2, \overline{\text{d}}_2$, then the node in $\idleState^1$ reaches $\frownie$. We say that a process \emph{fails} whenever it reaches $\frownie$.
We have the following lemma:
\begin{lemma}
	Let $C \in \CC$ be a \emph{stable} configuration such that $C_0 \trans^+ C$. Then in $C$, it holds that
	 $v_0$ is in $\textsf{s}_2$, and 
	there is exactly one vertex $v \in \set{v_1, v_2, v_3, v_4}$ on a state $\restState^j$ for some $j \in \set{0,1,2}$.
\end{lemma}
Indeed, let $C$ be a \emph{stable} configuration such that $C_0 \trans^+ C$. It holds that:
\Iflong{\begin{enumerate}
	\item From $C_0$, the first broadcast is from $v_0$ and is message $\text{td}_0$.\label{item:undec:1}
\end{enumerate}
Indeed, if another vertex than $v_0$ broadcasts a message $m$ with subscript $i$ from $C_0$, its left neighbor would fail with transition $(\idleState^j, ?m, \frownie)$ as $j = (i-1) \mod3$ and $m = \text{td}_i$ or $ m= \text{d}_i$.
\begin{enumerate}\addtocounter{enumi}{1}
	\item Each vertex (except the tail) broadcasts one pair of messages between $C_0$ and $C$.\label{item:undec:2}
\end{enumerate}
Assume for instance that $v_1$ does not broadcast anything, recall that from \cref{item:undec:1}, $v_0$ broadcasts $\text{td}_0$. Then at some point
it will also broadcasts $\overline{\text{td}}_0$ otherwise it would not be in  $\textsf{s}_0$ or $\textsf{s}_2$ in $C$.
Hence $v_1$ fails as depicted in \cref{fig:undec:transmission-wrong-1}. Actually, each vertex (except the tail) broadcasts exactly one pair, as if it broadcasts more, its left neighbor would fail as well.
\begin{enumerate}\addtocounter{enumi}{2}
	\item When a node
	broadcasts a pair ($\text{td}_j$, $\overline{\text{td}}_j$), its right neighbor broadcasts either a pair ($\text{td}_i$, $\overline{\text{td}}_i$) or 
	($\text{td}_i$, $\overline{\text{d}}_i$), for $j, i \in \set{0,1,2}$.\label{item:undec:3}
\end{enumerate}
Assume its right neighbor broadcasts some $\text{d}_i$, it must be that $i = (j+1) \mod 3$. Such an example is depicted in \cref{fig:undec:transmission-wrong-2}: $v_1$ fails with $(\repState_{\text{td}}^1, ?\text{d}_2, \frownie)$.
With similar arguments, we claim:
\begin{enumerate}\addtocounter{enumi}{3}
	\item When a node
	broadcasts a pair ($\text{td}_j$, $\overline{\text{d}}_j$) or a pair ($\text{d}_j$, $\overline{\text{d}}_j$), its right neighbor broadcasts a pair ($\text{d}_i$,
	 $\overline{\text{d}}_i$), for $j, i \in \set{0,1,2}$;\label{item:undec:4}
\end{enumerate}}
\Ifshort{\begin{enumerate}
		\item \emph{From $C_0$, the first broadcast is from $v_0$ and it broadcasts $\text{td}_0$.}\\
		Indeed, if another vertex than $v_0$ broadcasts a message $m$ with subscript $i$ from $C_0$, its left neighbor would fail with transition $(\idleState^j, ?m, \frownie)$ as $j = (i-1) \mod3$ and $m \in \set{\text{td}_i,\text{d}_i}$.
		Let us consider an example depicted in \cref{fig:undec:transmission-wrong-2}: Assume $v_1$ is in state $\idleState^1$ and $v_2$ broadcasts $\text{td}_2$ or $\text{d}_2$ (it issues a request whereas $v_1$ is not broadcasting any request), then $v_1$ receives the message with transition that goes from $\idleState^1$ to \frownie, as depicted in Figure 7.  Hence, we can not reach a stable configuration from there. 
		\label{item:undec:1}
		\item \emph{Each vertex (except the tail) broadcasts one pair of messages between $C_0$ and $C$.}\\
		Assume for instance that $v_1$ does not broadcast anything. From \cref{item:undec:1}, $v_0$ broadcasts $\text{td}_0$, and so at some point
		it will also broadcasts $\overline{\text{td}}_0$ otherwise it would not be in  $\textsf{s}_0$ or $\textsf{s}_2$ in $C$.
		Hence $v_1$ fails as depicted in \cref{fig:undec:transmission-wrong-1}. Actually, each vertex (except the tail) broadcasts exactly one pair: if it broadcasts more, its left neighbor would fail as well.\label{item:undec:2}
		\item \emph{When a node
			broadcasts a pair ($\text{td}_j$, $\overline{\text{td}}_j$), its right neighbor broadcasts either a pair ($\text{td}_i$, $\overline{\text{td}}_i$) or 
			($\text{td}_i$, $\overline{\text{d}}_i$), for $j, i \in \set{0,1,2}$.}\\
		Assume its right neighbor broadcasts $\text{d}_i$, it must be that $i = (j+1) \mod 3$. Such an example is depicted in \cref{fig:undec:transmission-wrong-2}: $v_1$ fails with $(\repState_{\text{td}}^1, ?\text{d}_2, \frownie)$.
		Similarly, we have:\label{item:undec:3}
		\item \emph{When a node
		broadcasts a pair ($\text{td}_j$, $\overline{\text{d}}_j$) or a pair ($\text{d}_j$, $\overline{\text{d}}_j$), its right neighbor broadcasts a pair ($\text{d}_i$,
		$\overline{\text{d}}_i$), for $j, i \in \set{0,1,2}$.}\label{item:undec:4}
\end{enumerate}}


\begin{figure}
	\begin{minipage}[c]{0.45\linewidth}
		\begin{center}
		\resizebox*{!}{0.065\paperheight}{
			\tikzset{box/.style={draw, minimum width=4em, text width=4.5em, text centered, minimum height=17em}}

\begin{tikzpicture}[-, >=stealth', shorten >=1pt,node distance=1.8cm,on grid,auto, initial text = {}] 
	\LineTikzZoomFour{0cm}{\textsf{s}_0}{\idleStateAbrev^1}{\idleStateAbrev^2}{\textsf{id}^0}{q_0^1}{p}
	\draw[->] (0,-0.14) -- (0, -0.7);
	\node at (0.5,-0.4) () {$!!{\text{td}}_0$};
	\LineTikzZoomFour{-0.8cm}{\textsf{s}_1}{\idleStateAbrev^1}{\idleStateAbrev^2}{\idleStateAbrev^0}{q_0^1}{p}
	\draw[->] (0,-0.95) -- (0, -1.55);
	\node at (0.5,-1.2) () {$!!\overline{\text{td}}_0$};
	\LineTikzZoomFourFrOne{-1.65cm}{\textsf{s}_2}{q_1^1}{\idleStateAbrev^2}{\idleStateAbrev^0}{\idleStateAbrev^1}{p}
\end{tikzpicture}
		}
		\subcaption{$v_1$ does not transmit the request.}\label{fig:undec:transmission-wrong-1}
	\end{center}
		\end{minipage}
		\hfill
	\begin{minipage}[c]{0.45\linewidth}
\begin{center}
	\resizebox*{!}{0.065\paperheight}{
		\tikzset{box/.style={draw, minimum width=4em, text width=4.5em, text centered, minimum height=17em}}

\begin{tikzpicture}[-, >=stealth', shorten >=1pt,node distance=1.8cm,on grid,auto, initial text = {}] 
	\LineTikzZoomFive{0cm}{\textsf{s}_1}{\idleStateAbrev^1}{\idleStateAbrev^2}{q_0^0}{q_0^1}{p}
	\draw[->] (0.9,-0.22) -- (0.9, -0.65);
	\node at (1.5,-0.45) () {$!!{\text{td}}_1$};
	\LineTikzZoomFive{-0.88cm}{\textsf{s}_1}{\repStateAbrev_{\text{td}}^1}{\idleStateAbrev^2}{q_{\text{td}}^0}{q_0^1}{p}
	\draw[->] (2.7,-1.1) -- (2.7, -1.6);
	\node at (3.1,-1.35) () {$!!{\text{d}}_2$};
	\LineTikzZoomFiveFrOne{-1.8cm}{\textsf{s}_1}{\idleStateAbrev^1}{\repStateAbrev_{\text{d}}^2}{q_{\text{td}}^0}{q_{\text{td}}^1}{p}
\end{tikzpicture}

	}
	\subcaption{$v_2$ broadcasts the wrong pair of messages.}\label{fig:undec:transmission-wrong-2}
\end{center}
\end{minipage}
\caption{Example of wrong behaviors during the transmission.}
\end{figure}

\subsection{Putting everything together}

We adapt the construction of \cref{subsec:undec:transmission}\ to propagate operations on counters of the machine issued by the head of the line. Counters processes will evolve in three different protocols as in \cref{subsec:undec:transmission}.
They can be either in a $\textsf{zero}$ state, from which all the types of instructions can be transmitted, or in a state $1_{\counter}$ for $\counter$ one of the 
two counters, from which all the types of operations can be transmitted, except 0-tests of $\counter$.
Increments and decrements of a counter \counter\  are done in a similar fashion as in \cref{subsec:undec:transmission}\ (exactly one node changes its state).
0-tests are somewhat easier: no node changes state nor executes anything, and the tail accepts the same pair as the one broadcast by the head. However, if a node is in a $1_\counter$ when $\counter$ is the counter compared to 0, it fails when its left neighbor broadcasts the request.

We ensure that we can select a line with a similar structure as the one depicted in \cref{fig:undec:line-topo}\ thanks to a first part of the protocol where each node: $(i)$ receives an announcement message from its predecessor with a subscript $j$ (except the head which broadcasts first), $(ii)$ broadcasts an announcement message with the subscript $(j+1)\mod 3$ (head broadcasts with subscript 0) and $(iii)$ waits for the announcement of its successor with subscript $(j+2) \mod 3$ (except for the tail). If it receives any new announcement at any point of its execution, it fails. When considering only line topologies, as each node has at most two neighbors, this part can be achieved with fewer alternations.
We get the two following theorems.
\begin{theorem}
	\Cover\ and \CoverTree\ are undecidable for $k$-phase-bounded protocols with $k \geq 6$.
\end{theorem}
%
\begin{theorem}
	\CoverLine\ is undecidable for $k$-phase-bounded protocols with $k \geq 4$.
\end{theorem}

%
%
%
%
%

\section{\Cover\ in 1-Phase-Bounded Protocols}\label{sec:about-1pb}

We show that $\Cover[\Topo]$ restricted to 1-phase-bounded
protocols is $\textsc{ExpSpace}$-complete.

We begin by proving that
for such protocols $\Cover[\Topo]$ and $\CoverStar$ are equivalent (where
$\Stars$ correspond to the tree topologies of height one)
. To get
this property, we first rely on Theorem
\ref{thm:Cover-CoverTree-equivalent}\ 
(stating that $\Cover$ and $\CoverTree$ are equivalent)
and without loss of generality we can assume that if a
control state can be covered with a tree topology, it can be
covered by the root of the tree. We then observe that when dealing
with  1-phase-bounded protocols, the behaviour of the processes of a tree
which  are located at a height 
strictly greater than  1 have no incidence on the
root node. Indeed if a process at depth 2 performs a broadcast received by a node at
depth 1, then this latter node will not be able to influence  the state of the
root because in 1-phase-bounded protocols, once a process has performed
a reception, it cannot broadcast anymore.
In the sequel we fix a 1-phase-bounded protocol $\PP = (Q, \Sigma,
\qinit, \Delta)$ and a state $q_f \in Q$. We then have:

\begin{lemma}\label{lemma:1-bounded-cover-star}
There exist $\Gamma \in \Topo$, $C = (\Gamma, L)\in \II_\PP$ and $D
=(\Gamma, L')\in \CC_\PP$ and $v\in \Vert{\Gamma}$ such that $C
\trans^\ast D$ and $L'(v) = q_f$ iff  there exists  $\Gamma' \in
\Stars$, $C'= (\Gamma', L'')\in \II$ and $D'
=(\Gamma', L''') \in \CC_\PP$ such that $C'
\trans^\ast_\PP D' $ and $L'''(\epsilon) = q_f$.
\end{lemma}

To solve $\CoverStar$ in \textsc{ExpSpace}, we
proceed as follows (1) we first propose an abstract representation
for the configurations reachable by executions where
the root node does not perform any reception, and that only keeps track of
states in $Q_0$ and $Q^b_1$ (2) we show that we can decide in
polynomial space whether  a configuration corresponding to a given abstract representation can be reached from an initial
configuration (3) relying on reduction to the control state reachability
problem in VASS (Vector Addition System with States), we show how to
decide whether there exists a configuration corresponding to a given abstract
representation from which $q_f$ can be covered in an execution where
the root node does not perform any broadcast. This reasoning relies on
the fact that  a process executing a 1-phase-bounded protocol first
performs only broadcast (or internal actions) and then performs only
receptions (or internal actions).

We use $Q^b$ to represent the set $Q_0\cup Q^b_1$ and we say  that a
configuration $C=(\Gamma,L)$ in $\CC_\PP$ is a \emph{star-configuration}
whenever $\Gamma \in \Stars$. For a star-configuration  $C=(\Gamma,L)$ in $\CC_\PP$ such that 
$L(\epsilon) \in Q^b$, the \emph{broadcast-print} of $C$, denoted 
by $\bprint{C}$, is the pair $(L(\epsilon),\set{L(v) \in
  Q^b \mid v \in \Vert{\Gamma} \setminus \set{\epsilon}})$ in $ Q^b\times
2^{Q^b}$. We call such a configuration $C$ a b-configuration. Note
that any initial star-configuration $C_{in}=(\Gamma_{in},L_{in}) \in
\II$ is a b-configuration verifying $\bprint{C_{in}} \in
\set{(\qinit,\emptyset),(\qinit,\set{\qinit})}$ (the first case
corresponding to $\Vert{\Gamma}=\set{\epsilon}$). We now
define a transition relation $\Rightarrow$ between
broadcast-prints. Given $(q,\Lambda)$ and $(q',\Lambda')$ in $ Q^b\times
2^{Q^b}$, we write $(q,\Lambda) \Rightarrow (q',\Lambda')$ if there
exists two $b$-configurations $C$ and $C'$ such that
$\bprint{C}=(q,\Lambda)$ and $\bprint{C'}=(q',\Lambda')$ and  $C
\trans C'$. We denote by $\Rightarrow^\ast$ the reflexive and
transitive closure of $\Rightarrow$.

One interesting point of this abstract representation
is that we can compute in polynomial time the $\Rightarrow$-successor
of a given broadcast-print. The intuition is simple:
either the root performs a broadcast of $m \in \Sigma$, and in that case we have to
remove from the set $\Lambda$ all the states from which a reception of
$m$ can be done (as the associated processes in $C'$ will not be in a
state in $Q^b$ anymore) or one process in a state of $\Lambda$
performs a broadcast and in that case it should not be received by the
root node (otherwise the reached configuration will not be a
b-configuration anymore).
\begin{lemma}\label{lemma:successor-bprint-ptime}
 Given $(q,\Lambda) \in Q^b\times
2^{Q^b}$, we can compute in polynomial time the set $\set{(q',\Lambda')
\mid (q,\Lambda) \Rightarrow (q',\Lambda') }$.
\end{lemma}
In order to show that our abstract representation can be used to solve
$\CoverStar$, we need to rely on some further formal definitions. Given
  two star-configurations  $C=(\Gamma,L)$ and $C'=(\Gamma',L')$,
   we write $C \preceq C'$ iff the two  following conditions hold $(i)$ $L(\epsilon)=L'(\epsilon)$, and, $(ii)$
  $|\set{v \in \Vert{\Gamma} \setminus\set{\epsilon} \mid L(v)=q}| \leq |\set{v \in \Vert{\Gamma'}\setminus\set{\epsilon} \mid L'(v)=q}|$ for
  all $q \in Q^b$. We then have the following lemma where the two
  first points show that when dealing with star-configurations, the
  network generated by $1$-phase-bounded protocol enjoys some
  monotonicity properties. Indeed, if
  the root node performs a broadcast received by other nodes, then
  if we put more nodes in the same state, they will also
  receive the message. On the other hand if it is another node that
  performs a broadcast, only the root node is able to receive it. The last point of the lemma shows that we
  can have as many processes as we want in reachable states in $Q^b$
  (as soon as the root node does not perform any reception) by
  duplicating nodes and  mimicking behaviors.

 \begin{lemma}\label{lemma:monotonicity}
  The  following properties hold:
   \begin{enumerate}
   \item[(i)] If $C_1$, $C'_1$ and $C_2$ are star-configurations such that $C_1 \trans C'_1$ and $C_1 \preceq C_2$ then there
     exists a star-configuration $C'_2$ such that $C_1' \preceq C'_2$ and $C_2
     \trans^\ast C'_2$.
 \item[(ii)] If $C_1$, $C'_1$ and $C_2$ are
     b-configurations such that $C_1 \trans  C'_1$ and
     $\bprint{C_1}=\bprint{C_2}$ and $C_1 \preceq C_2$ then there
     exists a b-configuration $C'_2$ such that $C_1' \preceq C'_2$ and
     $\bprint{C'_1}=\bprint{C'_2}$ and $C_2
     \trans^\ast C'_2$ .
   \item[(iii)] If $C$ is a b-configuration such that $C_{in}
     \trans^\ast C$ for some initial configuration $C_{in}$ then for all $N \in \nat$,
     there exists an initial configuration $C'_{in}$ and a b-configuration $C'=(\Gamma',L')$ such that $C'_{in} \trans^\ast
     C'$ and  $\bprint{C}=\bprint{C'}=(q,\Lambda)$ and
     $|\set{v \in \Vert{\Gamma'}\setminus\set{\epsilon} \mid L'(v)=q'}| \geq N$ for all $q' \in \Lambda$.
   \end{enumerate}
 \end{lemma}

 We can now prove that we
 can reason  in a sound and complete way with broadcast prints to
 characterise the b-configurations reachable from initial star-configurations. To prove this next lemma, we rely on the two last
 points of the previous lemma and reason by induction on the length
 of the $\Rightarrow$-path leading from $(\qinit,\Lambda_{in})$ to $(q,\Lambda)$.

\begin{lemma}\label{lemma:completeness-bprint}
Given  $(q,\Lambda) \in  Q^b\times
2^{Q^b}$, we have $(\qinit,\Lambda_{in}) \Rightarrow^\ast (q,\Lambda)$
with $\Lambda_{in} \in \set{\emptyset,\set{\qinit}}$ iff there exist
  two $b$-configurations $C_{in} \in \II$ and $C\in \CC$  such
  that $C_{in}\trans^\ast C$ and $\bprint{C}=(q,\Lambda)$. 
\end{lemma}

Finally, we show that we can verify in exponential space
 whether there exists a configuration with a given broadcast-print $(q,\Lambda)$
 from which we can reach  a configuration covering $q_f$ thanks to an
 execution where the root node does not perform any broadcast. This result is obtained
 by a reduction to the control state
 reachability problem in (unary) VASS 
 which is known to be \textsc{ExpSpace}-complete \cite{lipton76reachability,rackoff78covering}. VASS are
 finite state machines equipped with  variables (called counters) taking their values in $\nat$, 
 and where each  transition of the machine
 can either change the value of a counter, by incrementing or
 decrementing it, or do nothing. In our reduction, we encode the state
 of the root in the control state of the VASS and we associate a
 counter to each state of $Q^b$ to represent the number of processes in
 this state. In a first phase,
 the VASS generates a configuration with $(q,\Lambda)$ as
 broadcast-print and in a second phase it simulates the network. For
 instance, if a process performs a broadcast received by the root node,
 then we decrement the counter associated to the source state of the
 broadcast, we increment the one associated to the target state
 and we change the control state of the VASS representing the state of the
 root node accordingly. We need a last definition to characterise  executions
 where the root node does not perform any broadcast: given two star-configurations
 $C=(\Gamma,L)$ and $C'=(\Gamma,L')$, we write $C \transup{}_r C'$
 whenever there exist $v \in \Vert{\Gamma}$ and $\delta \in \Delta$
 such that $C \transup{v,\delta} C'$ and either $v \neq \epsilon$ or
 $\delta=(q,\tau,q')$ for some $q,q' \in Q$. We denote by
 $\trans_r^\ast$ the reflexive and transitive closure of
 $\trans_r$.

\begin{lemma}\label{lemma:expspace-cover-bprint}
Given $(q,\Lambda) \in  Q^b\times
2^{Q^b}$, we can decide in \textsc{ExpSpace} whether
there exist a b-configuration $C=(\Gamma_f,L)$ and a star-configuration
$C_f=(\Gamma_f,L_f)$  such that $\bprint{C} =(q,\Lambda)$ and  $L_f(\epsilon)=q_f$
and  $C \trans_r^\ast C_f$.
 \end{lemma}

Combining the results of the previous lemmas leads to an
\textsc{ExpSpace}-algorithm to solve $\CoverStar$. We first guess a
broadcast-print $(q,\Lambda)$ and check in polynomial space whether it is
$\Rightarrow$-reachable from an initial broadcast-print in
$\set{(\qinit,\emptyset),(\qinit,\set{\qinit})}$ thanks to Lemma
\ref{lemma:successor-bprint-ptime}  (relying on a non-deterministic
polynomial space algorithm for reachability). Then we use
\cref{lemma:expspace-cover-bprint} to check the existence of a b-configuration $C$ with $\bprint{C}=(q,\Lambda)$ from which we can cover
$q_f$. By Savitch's theorem~\cite{DBLP:journals/jcss/Savitch70}, we conclude that the problem is in \textsc{ExpSpace}. 
The completeness of this method is direct. For the soundess, we  reason as follows: using Lemma
\ref{lemma:completeness-bprint}, there exists a configuration $C$
reachable from an initial star-configuration such that
$\bprint{C}=(q,\Lambda)$, and by~\cref{lemma:expspace-cover-bprint}, there is a configuration $C'$ such that $\bprint{C'}=(q,\Lambda)$
from which we cover $q_f$.  Thanks to
\cref{lemma:monotonicity}.$(iii)$, there is a configuration $C''$ reachable
from an initial configuration such that $C\preceq C''$ and $C' \preceq C''$  and
$\bprint{C''}=(q,\Lambda)$. Thanks to
\cref{lemma:monotonicity}.$(i)$ applied to each transition,  we can build an execution from $C''$ that covers
$q_f$. The lower bound is obtained by a reduction  from the
control state reachability in VASS.
 
 \begin{theorem}\label{theorem:1pb-expspace}
	$\Cover[\Topo]$ and $\CoverTree$ are $\textsc{ExpSpace}$-complete for 1-phase-bounded protocols.
\end{theorem}

\section{Decidability Results for 2-Phase-Bounded Protocols}\label{sec:about-2pb}


\subsection{\Cover\ and \CoverTree\ are Decidable on 2-PB Protocols}\label{subsec:proof-dec-2phases-covertree}
%
A \emph{simple path between $u$ and $u'$} in a topology $\Gamma = (V,E)$ is a sequence of distinct vertices $v_0, \dots, v_k$ such that $u = v_0$, $u' = v_k$, and for all $0 \leq i <k$, $(v_i, v_{i+1}) \in E$. Its length is denoted $d(v_0,\dots, v_k)$ and is equal to $k$. \iflong{In a tree topology $\Gamma'$, for two vertices $u, u'$, there exists a unique simple path between $u$ and $u'$, hence we denote $d(u, u')$ to denote the length of the unique path between $u$ and $u'$. Furthermore, for all vertex $u$, $d(\epsilon, u) = |u|$.}\fi
%
%
%
%
Given an integer $K$, we say that a topology $\Gamma$ is $K$-bounded path (and we write $\Gamma\in K-\textsf{BP}$) if there is no simple path $v_0, \dots, v_k$ such that $d(v_0, \dots, v_k) > K$
%
The result of this subsection relies on the following theorem.

\begin{theorem}[\cite{DelzannoSZ10},Theorem 5]\label{th:concur-2010}
	For $K\geq 1$, \textsc{Cover[$K$-\textsf{BP}]}\ is decidable.
\end{theorem}

Hence, we show that if a state $q_f$ of a protocol $\PP$ is coverable with a tree topology, then $q_f$ is actually coverable with a tree topology that is
also $2(|Q|+1)-\textsf{BP}$. 
\Iflong{
Given $(\PP, q_f)$ a positive instance of $\CoverTree$, we let $f(\PP, q_f)$ the minimal number of processes needed to cover $q_f$ with a tree topology and we fix 
$\Gamma = (V,E)$ a tree topology such that $|V| = f(\PP,q_f)$ and that covers $q_f$. Let $v \in \Vert{\Gamma}$ and $(C_i = (\Gamma, L_i))_{0\leq i\leq n}$ configurations such that $C_0 \in \II$, $C_0 \trans C_1\trans\dots \trans C_n$ and $L_n(v) = q_f$. We assume wlog that $v$ is the root of the tree, i.e. $v = \epsilon$.

For all $u \in \Vert{\Gamma}$, we define $b(u)$ as the first index $0 \leq i<n$ from which $u$ takes a broadcast transition, and $\infty$ if it never broadcasts anything.
%
%
%
The next lemma relies on the fact that a node (different from the root that will reach $q_f$)
that never broadcasts anything is useless for covering $q_f$ and could then be removed. Since $\Gamma$ is the smallest topology allowing to cover $q_f$, it
cannot happen in $\Gamma$. Also, since the protocol is 2-phase bounded, a node that would first broadcast after the first broadcast of its father would also be useless
for the covering of $q_f$: this broadcast will only be received by its father in the last phase of reception, hence it will have no influence on the behavior of $\epsilon$.
The formal proof can be found in \cref{appendix:dec-2pb-covertree}.
\begin{lemma}\label{obs:decidability-2phase:everybody-broadcasts}
	For all $u \in \Vert{\Gamma} \setminus \set{\epsilon}$, $b(u) \neq \infty$. Moreover, for all $u_1, u_2 \in \Vert{\Gamma} \setminus \set{\epsilon}$ such that $u_2 = u_1 \cdot x$ for some $x \in \nat$,
	 it holds that $b(u_1) > b(u_2)$.

\end{lemma}
\cref{obs:decidability-2phase:everybody-broadcasts} allows to establish the following lemma. 
}
\Ifshort{To establish this result, consider a coverable state $q_f$ of a protocol $\PP$ with a tree topology $\Gamma$, such that $\Gamma$ is minimal in the number of nodes needed to cover $q_f$. We can suppose wlog that $q_f$ is covered by the root of the tree. We argue that all nodes (except maybe the root) in the execution covering $q_f$ broadcast something, as otherwise they are useless and could then be removed. We also argue that, since $\PP$ is 2-phase-bounded, a node that would first broadcast after the first broadcast of its father would also be useless for the covering of $q_f$: this broadcast will only be received by its father in its \emph{last phase of reception}, hence it will have no influence on the behavior of the root.
These two properties are the key elements needed to establish the following lemma.
}
%
%

\begin{lemma}\label{obs:decidability-2phase:bounded-path-to-v}
	Let $\PP = (Q, \Sigma, \qinit, \Delta)$ be a 2-phase-bounded protocol and $q_f \in Q$. If $q_f$ can be covered with a tree topology, then it can be covered with a topology $\Gamma\in\Trees$ such that, for all $u \in \Vert{\Gamma}$, $|u| \leq |Q|+1$.
\end{lemma}

Indeed, a counting argument implies that if this is not the case, there exist two nodes $u_1$ and $u_2$ on the same branch, different from the root, with $u_1$ 
a prefix of $u_2$, that both execute their first broadcast from the same state $q$.
In this case, we could replace the subtree rooted in $u_1$ by the subtree rooted in $u_2$, and still obtain an execution covering $q_f$. Once $u_1$ has 
reached $q$ (possibly by receiving broadcasts from the children of $u_2$), it will behave as in the initial execution. Behaviors of the children of $u_1$ might differ in this second part, but it can only influence $u_1$ in its reception phase, which will be the last phase, and hence will not influence the behavior of the root. 
\Iflong{We are able to build this execution thanks to \cref{obs:decidability-2phase:everybody-broadcasts}.}
Thanks to~\cref{{thm:Cover-CoverTree-equivalent},th:concur-2010}, we can then conclude. 
\begin{theorem}\label{thm:Cover-decidable-2pb}
	\Cover\ and \CoverTree\ are decidable for 2-phase-bounded protocols.
\end{theorem}
\Iflong{
\begin{proof}
We prove the decidability by reducing both \Cover\ and \CoverTree\ to $\textsc{Cover}[K-\textsf{BP}]$. Let $\PP$ a protocol and $q_f$ a state of $\PP$.
Assume that there is a topology $\Gamma$ with which 
$q_f$ is coverable. Then, by~\cref{{thm:Cover-CoverTree-equivalent}}, there exist a tree topology $\Gamma'$ with wich $q_f$ is coverable. Assume that 
$|V|= f(P,q_f)$. 
%
	Then, each simple path in $\Gamma'$ has a length bounded by $2(|Q| +1)$.
	Indeed, let $u_1, u_2 \in \Vert{\Gamma'} $, as $\Gamma'$ is a tree topology, the unique simple path between $u_1$ and $u_2$ either contains $\epsilon$ and $d(u_1, u_2) = |u_1| + |u_2|$, or there exists $u \in \nat^+$ such that $u_1 = u\cdot u_1'$ and $u_2 = u \cdot u'_2$, and $u_1'$ and $u_2'$ have no common prefix. In the latter case, the length of the simple path between $u_1$ and $u_2$ is $|u_1'| + |u_2'| < |u_1|+ |u_2|$. In both cases, $d(u_1, u_2) \leq |u_1| + |u_2|$, and hence from  \cref{obs:decidability-2phase:bounded-path-to-v}, $d(u_1, u_2) \leq 2(|Q| +1)$. Then $q_f$ is coverable in a $2(|Q|+1)$-bounded path topology. 
	
	Conversely, assume that there exists a $2(|Q|+1)$-bounded path topology $\Gamma$ with which $q_f$ is coverable. Then, immediately, there is a topology  with which 
$q_f$ is coverable. By~\cref{thm:Cover-CoverTree-equivalent}, there also exists a tree topology with which $q_f$ is coverable. \cref{th:concur-2010}
allows to conclude to decidability of both \Cover\ and \CoverTree. 
\end{proof}
}
%

\subsection{Polynomial Time Algorithm for \CoverLine\ on 2-PB Protocols} \label{subsec:coverLine-inP}

%
In the rest of this section, we fix a 2-phase-bounded protocol $\PP = (Q, \Sigma, \qinit, \Delta)$ and a state $q_f\in Q$ to cover.
For an execution $\rho = C_0 \transup{} C_1 \transup{} \cdots \transup{}C_n$ with $C_n = (\Gamma ,L_n)$, for all $v \in \Vert{\Gamma}$, we denote by \firstBroadcast{v}{\rho} the \emph{smallest} index $0 \leq i < n$ such that $C_i\transup{v, t} C_{i+1}$ with $t=(q,!!m,q')\in \Delta$.
If $v$ never broadcasts anything, $\firstBroadcast{v}{\rho} = -1$. We also denote by \lastBroadcast{v}{\rho} the \emph{largest} index $0\leq i < n$, 
such that $C_i \transup{v,t} C_{i+1}$ for some transition $t \in \Delta_{}$. If $v$ never issues any transition, we let  $\lastBroadcast{v}{\rho} = -1$.

The polynomial time algorithm relies on the fact that to cover a state, one can consider only executions that have a specific shape, described in the following lemma.
\begin{lemma}\label{lemma:CoverLine-2pb-inP:execution-shape-2}
	If $q_f$ is coverable with a line topology $\Gamma$ such that $\Vert{\Gamma} = \set{v_1, \dots, v_\ell}$ 
	then there exists an execution $\rho = C_0 \transup{} C_1 \transup{} \cdots \transup{}  C_n$ such that $C_n = (\Gamma, L_n)$,
	and $3\leq N\leq \ell-2$  with $L_n(v_N)= q_f$, 
	and
	\begin{enumerate}
		\item there exist $0 \leq j_1 < j_2 <n$ such that for all $0 \leq j <n$, if we let $C_j \transup{v^j, t^j} C_{j+1}$:		 \label{item:lemma:exec-shape-steps} \Ifshort{\\
			\emph{(a)} if $0\leq  j < j_1$, then $v^j \in \set{v_1, \dots, v_{N-2}}$ and if $v^j = v_{N-2}$, then $t^j = (q, \tau, q')$ for some $q,q'\in Q$; and \label{item:lemma:exec-shape-step-1} \\
			\emph{(b)}  if $j_1\leq j < j_2$, then $v^j \in \set{v_{N+2}, \dots, v_{\ell}}$ and if $v^j = v_{N+2}$, then $t^j = (q, \tau, q')$ for some $q,q'\in Q$;\label{item:lemma:exec-shape-step-2} and \\
			\emph{(c)} if $j_2\leq j < n$, then $v^j \in \set{v_{N-2}, \dots, v_{N+2}}$.\label{item:lemma:exec-shape-step-3}}
	\Iflong{
		\begin{enumerate}
		\item if $ j < j_1$, then $v^j \in \set{v_1, \dots, v_{N-2}}$ and if $v^j = v_{N-2}$, then $t^j$ is an internal transition;\label{item:lemma:exec-shape-step-1}
		\item if $j_1\leq j < j_2$, then $v^j \in \set{v_{N+2}, \dots, v_{n}}$ and if $v^j = v_{N+2}$, then $t^j$ is an internal transition;\label{item:lemma:exec-shape-step-2}
		\item if $j_2\leq j < n$, then $v^j \in \set{v_{N-2}, \dots, v_{N+2}}$.\label{item:lemma:exec-shape-step-3}
		\end{enumerate}}
		\item\label{item:lemma:exec-shape-order}
		\Ifshort{\emph{(a)} for all $1 \leq i \leq N-2$, $\lastBroadcast{v_i}{\rho} \leq \firstBroadcast{v_{i+1}}{\rho}$,\label{item:lemma:exec-shape-order-left} and \\
			\emph{(b)} for all $N+2 \leq i \leq \ell$, $\lastBroadcast{v_i}{\rho} \leq \firstBroadcast{v_{i-1}}{\rho}$\label{item:lemma:exec-shape-order-right}.}
		\Iflong{
		\begin{enumerate}
			\item for all $1 \leq i \leq N-2$, $\lastBroadcast{v_i}{\rho} \leq \firstBroadcast{v_{i+1}}{\rho}$;\label{item:lemma:exec-shape-order-left}
			\item for all $N+2 \leq i \leq \ell$, $\lastBroadcast{v_i}{\rho} \leq \firstBroadcast{v_{i-1}}{\rho}$\label{item:lemma:exec-shape-order-right}.
		\end{enumerate} }
	\end{enumerate}
\end{lemma}
\begin{figure}
	\begin{center}
		\tikzset{box/.style={draw, minimum width=4em, text width=4.5em, text centered, minimum height=17em}}
\usetikzlibrary{fit,calc,positioning,decorations.pathreplacing,matrix}

\begin{tikzpicture}[-, >=stealth', shorten >=1pt,node distance=1.8cm,on grid,auto, initial text = {},decoration = {brace}] 
	\LineTikzLong{0cm}{\qinit}{\qinit}{\qinit}{\qinit}{\qinit}{\qinit}{\qinit}{\qinit}
	\draw[->,thick] (1.2,-0.3) -- (1.2, -0.7);
	\node at (1.35,-0.7) () {$\ast$};
	\node at (-1.1,0) () {$C_0$};
	\LineTikzLong{-1.0cm}{\_}{q_1}{\qinit}{\qinit}{\qinit}{\qinit}{\qinit}{\qinit}
	\draw[->,thick] (11.3,-1.3) -- (11.3, -1.7);
	\node at (11.45,-1.7) () {$\ast$};
	\node at (-1.1,-1) () {$C_{j_1}$};
	\LineTikzLong{-2.0cm}{\_}{q_1}{\qinit}{\qinit}{\qinit}{q_2}{\_}{}
	\draw[->,thick] (6.2,-2.35) -- (6.2, -2.67);
	\node at (6.35,-2.65) () {$\ast$};
	\node at (-1.1,-2) () {$C_{j_2}$};
	\LineTikzLong{-3.0cm}{\_}{\_}{\_}{q_f}{\_}{\_}{\_}{}
	\node at (-1.1,-3) () {$C_n$};
	\fill [fill=violet, fill opacity = 0.2] (-0.6,0.3) rectangle (3.5,-1.3);
	\fill [fill=green, fill opacity = 0.2] (9,-0.7) rectangle (13.2,-2.3);
	\fill [fill=orange, fill opacity = 0.2] (1.9,-1.7) rectangle (10.6,-3.3);
	
	\draw[->, color=blue] (1.5,0.5) -- (2.5, -0.6);
	\node[rectangle,draw,color=blue, inner sep =2] at (1.4,0.6) {\scriptsize no broadcast from $v_{N-2}$};
	
	\draw[->, color=red] (10.9,-0.7) -- (9.9, -1.5);
	\node[rectangle,draw,color=red, inner sep=2] at (11,-0.5) {\scriptsize no broadcast from $v_{N+2}$};
	
	
\end{tikzpicture}
	\end{center}
	\caption{Illustration of execution $\rho$ obtained from \cref{lemma:CoverLine-2pb-inP:execution-shape-2}.}\label{fig:inP-proof-illustration}
\end{figure}
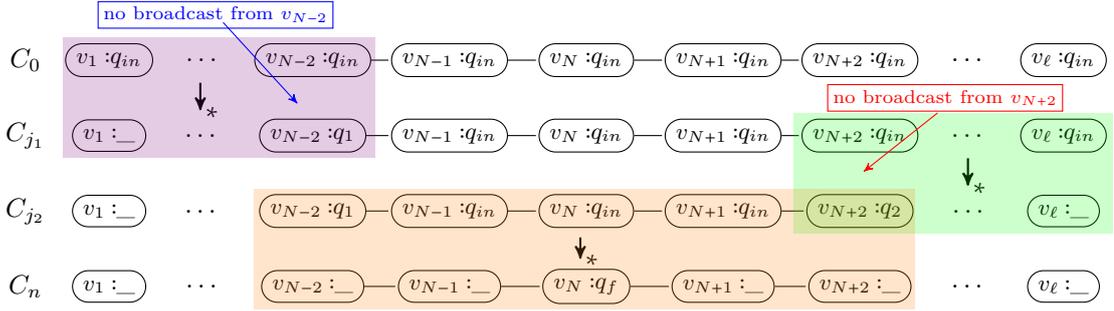
\cref{fig:inP-proof-illustration} illustrates the specific form of the execution described in \cref{item:lemma:exec-shape-steps} of \cref{lemma:CoverLine-2pb-inP:execution-shape-2}: the first nodes to take actions are the ones in the purple part (on the left), then, only nodes in the green part (on the right) issue transitions), and finally the nodes in the orange central part take actions in order to reach $q_f$. The fact that $\PP$ is 2-phase bounded allows us to 
establish \cref{item:lemma:exec-shape-order} of \cref{lemma:CoverLine-2pb-inP:execution-shape-2}: when $v_{i+1}$ starts broadcasting, no further broadcasts from $v_i$ will influence $v_{i+1}$'s broadcasts (it can only receive them in its last reception phase). 

\cref{fig:inP-proof-illustration} highlights why we get a polynomial time algorithm: when we reach the orange part of the execution, the nodes $v_{N-1}$, $v_N$
and $v_{N+1}$ are still in the initial state of the protocol. Moreover, in the orange part (which is the one that witnesses the covering of $q_f$), only five nodes take 
actions. Once one has computed in which set of states the nodes $v_{N-2}$ and $v_{N+2}$ can be at the beginning of the orange part, it only remains to compute
the set of reachable configurations from a finite set of configurations. Let $H$ be the set of possible states in which $v_{N-2}$ and $v_{N+2}$ can be at the
beginning of the last part of the execution, and for $q_1,q_2\in H$, let $C_{q_1,q_2}=(\Gamma_5, L_{q_1,q_2})$ where $\Gamma_5$ is the line topology 
with five vertices $\set{v_1, v_2, v_3, v_4, v_5}$ and $L_{q_1, q_2} (v_1) =  q_1$, $L_{q_1, q_2}(v_5)= q_2$ and for all other vertex $v$, $L_{q_1, q_2}(v) = \qinit$. 

Our algorithm is then: \Ifshort{(1) Compute $H$; (2) For all $q_1, q_2 \in H$, explore reachable configurations from $C_{q_1, q_2}$; (3) Answer yes if we reach a configuration covering $q_f$, answer no otherwise.}\Iflong{
\begin{enumerate}
\item Compute $H$.
\item For all $q_1, q_2 \in H$, explore reachable configurations from $C_{q_1, q_2}$. 
\item Answer yes if we reach a configuration covering $q_f$, answer no otherwise.
\end{enumerate}
}
It remains to explain how to compute $H$. This computation relies on \cref{item:lemma:exec-shape-order} of \cref{lemma:CoverLine-2pb-inP:execution-shape-2}:
locally, each node $v_i$ at the left of $v_{N-1}$ (resp. at the right of $v_{N+1}$) stops issuing transitions once its right neighbor $v_{i+1}$ (resp. its left neighbor $v_{i-1}$) starts broadcasting.

Hence we compute iteratively set of coverable pairs of states $S\subseteq Q\times Q$ by relying on a family $(S_i)_{i\in \nat}$ of subsets of $Q \times Q$ formally defined as follows:
\begin{align*}
	&S_0 =\set{(\qinit,\qinit)} \\
	&S_{i+1} = S_i \ \cup 
	\set{(q_1,q_2) \mid  \textrm{there exist } (p_1, p_2) \in S_i, j\in \set{1,2}  \textrm{ s.t. } (p_j, \tau ,q_j) \in \Delta\text{ and } p_{3-j} = q_{3-j}} \\
	&\cup \set{(q_1,q_2) \mid  \textrm{there exists } (p_1, p_2) \in S_i,  \textrm{ s.t. } (p_2, !!m ,q_2) \in \Delta, (p_1, ?m ,q_1) \in \Delta,
		m \in \Sigma} \\
	&\cup \set{(q_1,q_2) \mid  \textrm{there exists } p_2 \in Q \textrm{ s.t. }(q_1, p_2) \in S_i,  \textrm{ and } (p_2, !!m ,q_2) \in \Delta \textrm{ and }m \nin R(q_1)} \\
	 &\cup 
	\set{(\qinit,q) \mid  \textrm{there exists } (q, q') \in S_i \textrm{ for some }q' \in Q}.
\end{align*}
We then define $S = \bigcup_{n \in \nat} S_n$, and $H=\set{q\in Q\mid \textrm{ there exists }q'\textrm{ and } (q,q')\in S}$. 
Observe that $(S_i)_{i \in \nat}$ is an increasing sequence bounded by $|Q|^2$. The computation reaches then a fixpoint and $S$ can be computed in polynomial time.
We define $H = \set{q\mid \exists q' \in Q, (q,q') \in S}$. Note that $H \subseteq Q_0 \cup Q_1^r$, as expected by~\cref{item:lemma:exec-shape-order}\ of \cref{lemma:CoverLine-2pb-inP:execution-shape-2}.
\Iflong{Going back to the execution $\rho$ of \cref{lemma:CoverLine-2pb-inP:execution-shape-2}\ and \cref{fig:inP-proof-illustration}, we prove in \cref{appendix:subsec:CoverLine-inP:lemma:complete}, that $q_1 \in H$ and $q_2 \in H$. Finally, we get the following lemma.

\lug{oublier ce qu'il y a au dessus et dire en dessous le schéma de raisonnement: le lemme repose sur l'execution du 7.6 et etats atteint par rho sont dans H}
Correctness and completeness of the algorithm rely on the following lemmas, proved in \cref{appendix:subsec:CoverLine-inP:lemma:complete} and
 \cref{appendix:subsec:CoverLine-inP:lemma:correct}.

\begin{lemma}\label{lemma:CoverLine-2pb-inP:complete}
	If $q_f$ is coverable with a line topology, then there exist $q_1, q_2 \in H$ such that $C_{q_1, q_2} \trans^\ast C$ and $C = (\Gamma_5, L)$ with $L(v_3)  = q_f$.
\end{lemma}
}
We also state that our construction is complete and correct, leading to the following theorem. 
\Iflong{
Conversely, any state $q\in H$ is coverable with a line topology in an execution where $q$ is reached by the node at the extremity of the line, which never broadcasts. 
Hence,
\begin{lemma}\label{lemma:CoverLine-2pb-inP:correct}
	If there exist $q_1, q_2 \in H$, $v \in \Vert{\Gamma_5}$, such that $C_{q_1, q_2} \trans^\ast C$ and $C = (\Gamma_5, L)$ with $L(v) = q_f$, then there exist $\Gamma \in \Lines$, $C_0 \in \II$, $C' = (\Gamma, L')\in \CC$ such that $C_0 \trans^\ast C'$ and $L'(v) = q_f$.
\end{lemma}

We obtain then the following result. 
}
\begin{theorem}
	\CoverLine\ is in P for $k$-phase-bounded protocols with $k \in \set{1,2}$.
\end{theorem}

\begin{proof}
We explain why the algorithm takes a polynomial time: step 1 (computing $H$) is done in polynomial time as explained above. For step 2, there are at most $|H| \times |H| \leq |Q|^2$ pairs, and for each pair, we explore a graph of at most $|Q|^5$ nodes in which each vertex represents a configuration $C = (\Gamma_5, L)$. Accessibility in a graph can be done non-deterministically in logarithmic space, and so in polynomial time.
Observe that all the lemmas of this section hold true when considering 1-phase-bounded protocols, hence the theorem.
\end{proof}

\bibliographystyle{plainurl}
\bibliography{pb-main}

\clearpage
\appendix
\Iflong{\section{Proofs of \cref{subsec:Cover-CoverTree-equivalent}}\label{appendix:covertree-cover-equivalent}}
\Ifshort{\section{\Cover\ and \CoverTree\ are equivalent}\label{subsec:Cover-CoverTree-equivalent}}

\Ifshort{

Let $\PP = (Q, \Sigma, \qinit, \Delta)$ be a broadcast protocol, and $q_f\in Q$. Let $\rho=C_0 \trans \cdots \trans C_n$ with $C_i = (\Gamma, L_i)\in\CC_\PP$ for all $0 \leq i \leq n$, and a vertex $v_f \in \Vert{\Gamma}$ such that $L_n(v_f) = q_f$.
We will build an execution covering $q_f$ with a tree topology $\Gamma'$, rooted in $v_f$, the node that reaches $q_f$. 
$\Gamma'$ is actually an \emph{unfolding} of the topology $\Gamma$.

%
%

We first define inductively the set of nodes $V'\subseteq \nat^\ast$, along with a labelling function $\lambda$ which associates to each node $v' \in V'$ a node $v \in \Vert{\Gamma}$. 
\begin{enumerate}
	\item $\epsilon\in V'$ and $\lambda(\epsilon) = v_f$;\label{enumerate:def-tree:epsilon}
	\item Let $\NeighG{\Gamma}{v_f} = \set{v_1, \dots , v_k}$. For all $1 \leq i \leq k$, $i \in V'$ and $\lambda(i) = v_i$;\label{enumerate:def-tree:depth-1}
	\item Let $w \cdot x \in V'$, with $w \in \nat^\ast$ such that $|w| < n-1$, and $x \in \nat$. Let $\NeighG{\Gamma}{\lambda(w\cdot x)} \setminus \set{\lambda(w)} = \set{v_1, \dots, v_k}$. Then, for all $1 \leq i \leq k$, $w\cdot x \cdot i \in V'$ and $\lambda'(w\cdot x \cdot i) = v_i$.\label{enumerate:def-tree:induction} 
\end{enumerate}
Finally, define $E' = \set{\langle w, w\cdot x\rangle  \mid w \in V', w\cdot x \in V'}$.
Note that $\epsilon \in V'$ and $V'\subseteq \nat^\ast$ and  is prefix closed.
Furthermore, by construction, for all $w \in V'$, $|w| \leq n$, and for all $w \in V'$, in fact $w\in \set{1,\dots, d}^\ast$ 
where $d$ is the maximal degree of $\Gamma$. 
Hence, $V'$ is a finite set and  $\Gamma'=(V',E')$ is a tree topology. 

The way we built $\Gamma'$ ensures that each node $v\in V'$ (except the leaves) enjoys the same set of neighbors than $\lambda(v)\in V$. This is formalised in the
following lemma.
\begin{lemma}\label{lemma:Covertree-Cover:well-formed-tree}
	For all $u\in V'$, for all $u'\in \NeighG{\Gamma'}{u}$, $\lambda(u')\in \NeighG{\Gamma}{\lambda(u)}$. Then, we let $f_u:\NeighG{\Gamma'}{u} \rightarrow \NeighG{\Gamma}{\lambda(u)}$ defined by $f_u(u')=\lambda(u')$. If $|u|< n$, $f_u$
	is a bijection.
\end{lemma}
}
\Iflong{\begin{proofof}{\cref{lemma:Covertree-Cover:well-formed-tree}}}
\Ifshort{\begin{proof}}
	We prove the lemma inductively on the structure of the nodes of $V'$. 
	\begin{itemize}
		\item Let $u=\epsilon$, then $\lambda(u)=v_f$ and let $\NeighG{\Gamma}{v_f}=\set{v_1,\dots, v_k}$. Then by definition, for all $1\leq i\leq k$, $i\in V'$
		and $\lambda(i)=v_i$. By definition of $E'$, $\set{1,\dots, k}\subseteq \NeighG{\Gamma'}{\epsilon}$. Now let $u\in\NeighG{\Gamma'}{\epsilon}$, again by definition of 
		$E'$, $|u|=1$ and $u\in \set{1,\dots,k}$ because all other nodes added in $V'$ are of length greater of equal than 2. So $\NeighG{\Gamma'}{\epsilon}=\set{1,\dots,k}$, and the function $f_\epsilon$ such that $f_\epsilon(i)=\lambda(i)=v_i$ is obviously a bijection.
		
		\item Let $w\cdot x\in V'$ such that $|w\cdot x|< n$ and $w\in\nat^\ast$. By induction hypothesis, $f_w:\NeighG{\Gamma'}{w}\rightarrow \NeighG{\Gamma}{\lambda(w)}$ is a bijection.
		Hence, since by definition of $E'$, $w\cdot x\in\NeighG{\Gamma'}{w}$, $f_w(w\cdot x) = \lambda(w\cdot x)\in\NeighG{\Gamma}{\lambda(w)}$ and $(\lambda(w), \lambda(w\cdot x))\in E$. 
		Let then $\NeighG{\Gamma}{\lambda(w\cdot x)}=\set{v_1,\dots, v_\ell,\lambda(w)}$. By definition of $V'$ and $E'$, $\NeighG{\Gamma'}{w\cdot x}=\set{w, w\cdot x\cdot 1,\dots, w\cdot x\cdot \ell}$ and $\lambda(w\cdot x\cdot i) = v_i$ for all $1\leq i \leq \ell$. Hence, $f_{w\cdot x}$ is a bijection. 
	\end{itemize}
	
	Moreover, if $w\cdot x\in V'$ with $|w\cdot x| = n$, then $w\cdot x$ is a leaf and $\NeighG{\Gamma'}{w\cdot x} = \set{w}$. By construction, $\lambda(w\cdot x)=v\in\NeighG{\Gamma}{\lambda(w)}$, which implies that $\lambda(w)\in\NeighG{\Gamma}{\lambda(w\cdot x)}$. 
	\Ifshort{\end{proof}}
\Iflong{\end{proofof}}

\Ifshort{

From the execution $\rho = C_0\rightarrow \dots C_n$ we will build a similar execution on $\Gamma'$. The idea is that for each step of the execution $\rho$, 
for each node $v\in V$, all the nodes in $\Gamma'$ that are labelled by $v$ will behave in the same way. This is possible thanks to~\cref{lemma:Covertree-Cover:well-formed-tree}.
Observe though that the leaves might no be able to behave as expected because they might not have the same set of neighbors than the node they are labelled by, so they might not be able to receive some
broadcast message. However, we can ensure some weaker version of correctness, defined as follows.
Let $0 \leq h \leq n$ and $C= (\Gamma, L)$ be a configuration. We say that a configuration $C' = (\Gamma', L')$ is \emph{$h$-correct} for $C$ if for all $u\in\Vert{\Gamma'}$, if $|u|\leq h$ then $L'(u)=L(\lambda(u))$.

The following lemma gives the main ingredient that allows to mimick the execution $\rho$ on $\Gamma'$.

\begin{lemma}\label{lemma:Covertree-Cover:induction-step}
	Let $C_1 = (\Gamma, L_1)$, $C_2 = (\Gamma, L_2)\in \CC_\PP$ such that $C_1 \trans C_2$ and let $C'_1 = (\Gamma', L'_1)\in\CC_\PP$ and $0< h \leq n$ such that $C'_1$ is $h$-correct for $C_1$.
	There exists $C'_{2}\in \CC_\PP$ such that $C'_1 \trans^\ast C'_{2}$, and $C'_2$ is $(h-1)$-correct for $C_2$.	
\end{lemma}
}
\Iflong{\begin{proofof}{\cref{lemma:Covertree-Cover:induction-step}}}
	\Ifshort{\begin{proof}}
	Denote $C_1 \transup{v, t} C_2$ with $t = (q, \alpha, q')$ and $\alpha \in !! \Sigma \cup \set{\tau}$. 
	Let $\set{u_1, \dots, u_K} $ be the set of vertices in $V'$ such that for all $1 \leq k \leq K$, $|u_k| \leq h$ and $\lambda(u_k) = v$. We will build inductively an execution $C'_{1,0}\transup{u_1, t} C'_{1,1}\transup{u_2, t}\dots\transup{u_K,t}C'_{1,K}$ of
	configurations over $\Gamma'$
	such that for all $0\leq k\leq K$, $C'_{1,k} = (\Gamma', L'_{1, k})$ and for all node $u\in V'$ such that $|u|\leq h$, 
	\begin{itemize}
		\item if $u\in \bigcup_{1\leq j\leq k}(\set{u_j}\cup\NeighG{\Gamma'}{u_j})$, then $L'_{1,k}(u)=L_2(\lambda(u))$, 
		\item otherwise, $L'_{1,k}(u)=L'_1(u)$.
	\end{itemize}
	
	We let $C'_{1,0}=C'_1$, since $\bigcup_{1\leq j\leq k}(\set{u_j}\cup\NeighG{\Gamma'}{u_j})=\emptyset$, $C'_1$ trivially meets the requirements. 
	
	Let $0\leq k< K$ and assume now that we have built $C'_{1,0}, \dots, C'_{1,k}$ as required. First, observe that 
	\begin{equation}\label{eq:u-k+1-u-j}
		(u_{k+1}\cup\NeighG{\Gamma'}{u_{k+1}})\cap (\bigcup_{1\leq j\leq k} (\set{u_j}\cup\NeighG{\Gamma'}{u_j}))=\emptyset. 
	\end{equation}
	Indeed assume otherwise and let $u$ an element of this intersection. If $u=u_{k+1}$, then $u\notin \bigcup_{1\leq j\leq k}\set{u_j}$ by definition. Then it must
	be in $\bigcup_{1\leq j\leq k}\NeighG{\Gamma'}{u_j}$ and let $j$ such that $u_{k+1}\in \NeighG{\Gamma'}{u_j}$. By~\cref{lemma:Covertree-Cover:well-formed-tree},
	$f_{u_j}(u_{k+1})=\lambda(u_{k+1})=v\in \NeighG{\Gamma}{\lambda(u_j)}=\NeighG{\Gamma}{v}$. Hence it implies that $(v,v)\in E$ which is impossible. 
	If $u\in\NeighG{\Gamma'}{u_{k+1}}$ then by a similar argument, $u\notin\bigcup_{1\leq j\leq k} \set{u_j}$. Then let $1\leq j\leq k$ such that $u\in\NeighG{\Gamma'}{u_j}$. Then $u_j\in\NeighG{\Gamma'}{u}$ and $u_{k+1}\in\NeighG{\Gamma'}{u}$. If $|u| <n$ ($u$ is not a leaf), then $f_u(u_j)=\lambda(u_j)=v$ and 
	$f_u(u_{k+1})=\lambda(u_{k+1})=v$ which contradicts the fact that $f_u$ is bijective from~\cref{lemma:Covertree-Cover:well-formed-tree}. If $|u|=n$, then 
	it is a leaf and it has only one neighbor: its father in the tree. So it is not possible that $\set{u_j, u_{k+1}}\in\NeighG{\Gamma'}{u}$. 
	
	To define $C'_{1,k+1}$ we need to differentiate between 
	the possible values of $\alpha$ (recall that $t = (q, \alpha, q')$).
	
	\begin{itemize}
		\item if $\alpha=\tau$, let $L'_{1,k+1}(u_{k+1})=q'$ and $L'_{1,k+1}(u)=L'_{1,k}(u)$ for all $u\neq u_{k+1}$. First, by~\cref{eq:u-k+1-u-j}\ and induction hypothesis,
		$u_{k+1}\notin \bigcup(\set{u_j}\cup\NeighG{\Gamma'}{u_j})$, then $L'_{1,k}(u_{k+1})=L'_1(u_{k+1})$. Moreover, since $C'_1$ is $h$-correct for
		$C_1$, $L'_1(u_{k+1})=L_1(v)=q$. Then $C'_{1,k}\transup{u_{k+1}, t} C'_{1,k+1}$. Let $u\in V'$ such that $|u|\leq h$ and $u\in\bigcup_{1\leq j\leq k+1}
		(\set{u_j}\cup\NeighG{\Gamma'}{u_j})$. If $u\in\bigcup_{1\leq j\leq k}
		(\set{u_j}\cup\NeighG{\Gamma'}{u_j})$, by induction hypothesis, $L'_{1,k}(u)=L_2(\lambda(u))$. Moreover, by~\cref{eq:u-k+1-u-j}, $u\neq u_{k+1}$. Then,
		$L'_{1,k+1}(u)=L'_{1,k}(u)=L_2(\lambda(u))$. By construction, $L'_{1,k+1}(u_{k+1})=q'=L_2(\lambda(u_{k+1}))$. Now, if $u\in\NeighG{\Gamma'}{u_{k+1}}$,
		or $u\notin \bigcup_{1\leq j\leq k+1}
		(\set{u_j}\cup\NeighG{\Gamma'}{u_j})$. 
		Then, $L'_{1,k+1}(u)= L'_{1,k}(u)$. By induction hypothesis, $L'_{1,k}(u)=L'_1(u)=L_1(\lambda(u))$ because $C'_1$ is $h$-correct for $C_1$. Moreover, 
		$L_1(\lambda(u))=L_2(\lambda(u))$. Hence, $C'_{1,k+1}$ meets the requirements. 
		
		\item If $\alpha = !!m$ for some $m\in\Sigma$, we define $L'_{1,k+1}$ as follows.
		\begin{itemize}
			\item $L'_{1,k+1}(u_{k+1})=q'$ (the node $u_{k+1}$ performs the broadcast).
			\item for all $u\in\NeighG{\Gamma'}{u_{k+1}}$ such that $|u|\leq h$, $L'_{1,k+1}(u)=L_2(\lambda(u))$. 
			\item for all $u\in\NeighG{\Gamma'}{u_{k+1}}$ such that $|u| > h$, if there exists $p\in Q$ and $(L'_{1,k}(u), ?m, p)\in \Delta$, $L'_{1,k+1}(u)=p$, otherwise, $L'_{1,k+1}(u)=L'_{1,k}(u)$ (its neighbors receive the broadcast).
			\item For all other $u\in V'$, $L'_{1,k+1}(u)=L_{1,k}(u)$. 
		\end{itemize}
		
		We first show that $C'_{1,k}\transup{u_{k+1}, t} C'_{1,k+1}$. By~\cref{eq:u-k+1-u-j}, $u_{k+1}\notin \bigcup_{1\leq j\leq k}(\set{u_j}\cup\NeighG{\Gamma'}{u_j})$. Hence, the induction hypothesis ensures that $L'_{1,k}(u_{k+1})=L_1(\lambda(u_{k+1})=L_1(v)$. Let $u\in \NeighG{\Gamma'}{u_{k+1}}$ such that $|u|\leq h$. By~\cref{eq:u-k+1-u-j},
		$u\notin\bigcup_{1\leq j\leq k}(\set{u_j}\cup\NeighG{\Gamma'}{u_j})$. Then, by induction hypothesis, $L'_{1,k}(u)=L'_1(u)=L_1(\lambda(u))$ since
		$C'_1$ is $h$-correct for $C_1$. Moreover, since $u\in\NeighG{\Gamma'}{u_{k+1}}$, \cref{lemma:Covertree-Cover:well-formed-tree} implies that 
		$\lambda(u)\in\NeighG{\Gamma}{v}$. Since $C_1\transup{v,t} C_2$, we know that either $(L_1(\lambda(u)), ?m, L_2(\lambda(u)))\in \Delta$, hence
		$(L'_{1,k}(u), ?m, L'_{1,k+1}(u))\in \Delta$,
		or $L'_{1,k+1}(u)=L_2(\lambda(u))=L_1(\lambda(u))=L'_{1,k}(u)$. In both cases, we can conclude that $C'_{1,k}\transup{u_{k+1}, t} C'_{1,k+1}$. 	
		
		Let now $u\in V'$ such that $|u|\leq h$ and $u\in \bigcup_{1\leq j\leq k+1}(\set{u_j}\cup\NeighG{\Gamma'}{u_j})$. If $u\in \bigcup_{1\leq j\leq k}(\set{u_j}\cup\NeighG{\Gamma'}{u_j})$, again by~\cref{eq:u-k+1-u-j}, $u\notin \set{u_{k+1}}\cup\NeighG{\Gamma'}{u_{k+1}}$, and $L'_{1,k+1}(u)=L'_{1,k}(u)=L_2(\lambda(u))$ by induction hypothesis. If $u=u_{k+1}$, $L'_{1,k+1}(u)=q'=L_2(v)=L_2(\lambda(u_{k+1})$. If $u\in\NeighG{\Gamma'}{u_{k+1}}$, by definition we have that
		$L'_{1,k+1}(u)=L_2(\lambda(u))$. If now $|u|\leq h$ and $u\nin \bigcup_{1\leq j\leq k+1}(\set{u_j}\cup\NeighG{\Gamma'}{u_j})$, $L'_{1,k+1}(u)=L'_{1,k}(u)=L'_1(u)$ by induction hypothesis. 
	\end{itemize}

	We let $C'_2=C'_{1,K}$ and we prove that $C'_{2}$ is $(h-1)$-correct for $C_2$. 
	Let $u \in V'$ such that $|u| \leq h-1 < n$. As $C'_1$ is $h$-correct for $C_1$, $L'_1(u) = L_1(\lambda(u))$.
	
	If $u$ is such that $\lambda(u) \in \NeighG{\Gamma}{v}$, then by \cref{lemma:Covertree-Cover:well-formed-tree}, there exists $u' \in \NeighG{\Gamma'}{u}$ such that $\lambda(u') = v$. Furthermore, since $\Gamma'$ is a tree topology, and $|u|\leq h-1$, either $|u'| \leq h-2$ or $|u'| \leq h-1+1=h$. In both cases, $|u'| \leq h$. As a consequence, $u' \in \set{u_1, \dots, u_K}$, and so, $L'_{2}(u) = L'_{1,K}(u)=L_{2}(\lambda(u))$.
	If $u$ is such that $\lambda(u) = v$, then $u \in \set{u_1, \dots, u_K}$, and again, $L'_{2}(u) =L'_{1,K}(u)= L_{2}(\lambda(u))$.
	In other cases, $\lambda(u) \nin \NeighG{\Gamma}{v} \cup \set{v}$, and $L_2(\lambda(u))=L_1(\lambda(u))$. 
	By definition of $\set{u_1,\dots, u_K}$ and by~\cref{lemma:Covertree-Cover:well-formed-tree}, $u\notin \bigcup_{1\leq j\leq k}(\set{u_j}\cup\NeighG{\Gamma'}{u_j})$. As a consequence, $L'_{2}(u) = L'_{1,K}(u)=L'_1(u)=L_1(\lambda(u))=L_2(\lambda(u))$. 
	
	Hence, $C'_{2}$ is $(h-1)$-correct for $C_2$.
	\Ifshort{\end{proof}}
\Iflong{\end{proofof}}

We build now an execution covering $q_f$ with $\Gamma'$: let $C'_0=(\Gamma', L'_0)$ defined by $L'_0(v)=\qinit$ for all $v\in\Vert{\Gamma'}$. Obviously,
$C'_0$ is $n$-correct. By~\cref{lemma:Covertree-Cover:induction-step}, there exists a sequence of configurations $(C'_i)_{1\leq i \leq n}$ such that
for all $1\leq i\leq n$, $C'_{i-1}\trans^\ast C'_{i}$ and $C'_i$ is $(n-i)$-correct for $C_i$. Hence, $C'_0\trans^\ast C'_n$ and $C'_n$ is $0$-correct for $C_n$ and
$L'_n(\epsilon)=L_n(\lambda(\epsilon))=L_n(v_f)=q_f$. 

This allows to prove the following result. 
\begin{theorem}
	$\Cover$ and $\CoverTree$ are equivalent. 
\end{theorem}


\Iflong{\section{Proof of \cref{sec:cover-cover-phase-bounded}} \label{app:cover-cover-phase-bounded}}
\Ifshort{
\section{Phase-bounded protocols as an under-approximation}\label{sec:cover-cover-phase-bounded}

Phase-bounded protocols can be seen as a semantic restriction of general protocols in which each process can only switch a bounded number of times between phases where it receives messages and phases where it can send messages. When, usually, restricting the behavior of processes immediately yields an
underapproximation of the reachable states, we highlight here the fact that preventing messages from being received can in fact lead to new reachable states. 
This motivates our definition of phase-bounded protocols, in which a process always ends in a reception phase. 

\begin{figure}
	\begin{minipage}[c]{0.5\columnwidth}
		\resizebox*{1.0\columnwidth}{!}{
			\tikzset{box/.style={draw, minimum width=4em, text width=4.5em, text centered, minimum height=17em}}

\begin{tikzpicture}[-, >=stealth', shorten >=1pt,node distance=2cm,on grid,auto, initial text = {}] 
	\node[state,initial] (q0) [] {$\qinit$};
	\node[state] (q1) [right of = q0, yshift = 15, xshift =0] {$q_1$};
	\node[state] (q2) [right  of = q1, xshift =0] {$q_2$};
	\node[state] (q3) [right of = q2, xshift =0] {$q_3$};
	\node[state] (q4) [right of = q0, yshift = -15, xshift = 0] {$q_4$};
	\node[state] (q5) [right  of = q4, xshift =0] {$q_5$};
	\node[state] (q6) [right of = q5, xshift =0] {$q_6$};
	
	\node[state] (p) [below of = q4, yshift= 20] {$p$};

	\path[->] 
	(q0) edge node {$!!c$} (q1)
	(q0) edge node {$?c$} (q4)
	(q0) edge [bend right] node {$?m$} (p)
	
	(q1) edge node {$!!m$} (q2)
	
	(q4) edge node {$!!b$} (q5)
	(q4) edge node {$?m$} (p)
	(q2) edge node {$?a$} (q3)
	(q5) edge [bend left] node {$?m$} (p)
	(q5) edge node {$!!a$} (q6)
	;
\end{tikzpicture}
		}
		
		\subcaption{An example of a broadcast protocol $\overline{\PP}$\label{fig:bp:example-2}}
	\end{minipage}
	\hfill
	\begin{minipage}[c]{0.5\columnwidth}
		\resizebox*{1.0\columnwidth}{!}
		{\tikzset{box/.style={draw, minimum width=4em, text width=4.5em, text centered, minimum height=17em}}

\begin{tikzpicture}[-, >=stealth', shorten >=1pt,node distance=2cm,on grid,auto, initial text = {}] 
	\node[state,initial] (q0) [] {$\qinit$};
	\node[state, ellipse, inner sep= 0.01] (q1) [right of = q0, yshift = 15, xshift =10] {\footnotesize$(q_1,b,1)$};
	\node[state, ellipse, inner sep= 0.01] (q2) [right  of = q1, xshift =10] {\footnotesize$(q_2,b,1)$};
	\node[state, ellipse,inner sep= 0.01] (q3) [right of = q2, xshift =10] {\footnotesize$(q_3,r,2)$};
	\node[state, ellipse,inner sep= 0.01] (q4) [right of = q0, yshift = -15, xshift = 10] {\footnotesize$(q_4,r,1)$};
	\node[state, ellipse,inner sep= 0.01] (q5) [right  of = q4, xshift =10] {\footnotesize$(q_5,b,2)$};
	\node[state,ellipse,inner sep= 0.01] (q6) [right of = q5, xshift =10] {\footnotesize$(q_6,b,2)$};
	
	\node[state, ellipse,inner sep= 0.01] (p) [below of = q4, yshift = 20] {$(p,r,1)$};

	\path[->] 
	(q0) edge node {$!!c$} (q1)
	(q0) edge node {$?c$} (q4)
	(q0) edge [bend right] node {$?m$} (p)
	
	(q1) edge node {$!!m$} (q2)
	
	(q4) edge node {$!!b$} (q5)
	(q4) edge node {$?m$} (p)
	(q2) edge node {$?a$} (q3)
	(q5) edge node {$!!a$} (q6)
	;
\end{tikzpicture}}
		\subcaption{Unfolding of $\overline{\PP}$ limited to 2 phases: $(q,b,i)$ means that the state $q$ is reached in a broadcast phase when $(q,r,i)$ means that it is reached in a reception phase. The number of the phase is given by $i$.}\label{fig:bp:example-2-unfold}
	\end{minipage}
	
	\caption{A protocol and its bounded unfolding showing that it may not be an underapproximation.}
\end{figure}
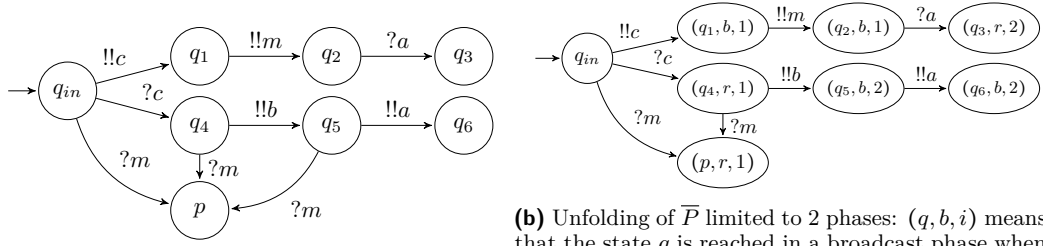

Indeed, consider the protocol pictured on~\cref{fig:bp:example-2}. The state $q_3$ is not coverable: to cover $q_3$, a node $v_1$ needs to receive message $a$ when it is on state $q_2$ from one of its neighbor $v_2$. By construction, $v_1$ broadcasts message $m$ before reaching $q_2$.
Vertex $v_2$ can only broadcast message $a$ when it is on state $q_5$. To reach $q_5$, vertex $v_2$ visited exactly states $\qinit, q_4$ and $q_5$. From each of those states, there is an outgoing reception transition labelled with $m$ going to $p$. Hence, at any moment of the execution, the broadcast of message $m$ by vertex $v_1$ would have brought $v_2$ in $p$, preventing the broadcast of $a$.
However, a naive 2-phase bounded unfolding of this protocol, in which we limit the number of phases of sending and reception to 2, is illustrated in~\cref{fig:bp:example-2-unfold}. In 
this protocol, $(q_3,r,2)$ (hence $q_3$) is coverable:
consider $\Gamma = (\set{v_1, v_2}, \set{\langle v_1, v_2\rangle})$ and the following execution:
\begin{align*}
	&(\Gamma, \set{v_1 \mapsto \qinit, v_2\mapsto \qinit}) \trans (\Gamma, \set{v_1 \mapsto (q_1,b,1), v_2\mapsto (q_4,r,1)}) \trans (\Gamma, \set{v_1 \mapsto (q_1,b,1),\\
		& v_2\mapsto (q_5,b,2)}) \trans 
	(\Gamma, \set{v_1 \mapsto (q_2,b,1), v_2\mapsto (q_5,b,2)}) \trans(\Gamma, \set{v_1 \mapsto (q_3,r,2), v_2\mapsto (q_6,b,2)}).
\end{align*}

In fact, in state $(q_5,b,2)$ a process is not allowed to switch anymore, hence the transition allowing to receive message $m$ has been removed. Doing so, 
we have made state $q_3$ coverable. This shows that this type of bounded semantics does not give an underapproximation of the coverable states, in spite of what was expected. 

Actually, the reception of a message is something that is not under the control of a process. If another process broadcasts a message, a faithful behavior of the system is that all the processes that can receive it indeed do so, no matter in which phase they are in their own execution. Hence, in a restriction that attempts to limit the number of switches 
between sending and receiving phases, one should not prevent a reception to happen. This motivates our definition of a phase-bounded protocol.

Let $\PP = (Q, \Sigma, \qinit, \Delta)$ be a broadcast protocol, and $k \in \nat$. 

We define the $k$-unfolding of $\PP$ denoted by $\PP_k$ as the following protocol: 
$\PP_k =(Q_k, \Sigma, \qinit, \Delta_k)$ with 
$Q_k= \{q^0\mid q\in Q\}\cup \{q^{b,j}, q^{r,j}\mid q\in Q, 1\leq j\leq k\}$. To ease the notations we let $q^{r,0}=q^0$, $q^{b,0}=q^0 = $ for all $q^0 \in Q_0$.
\begin{align*}
	\Delta_k & =  \set{(q^0, \tau, p^0) \mid (q, \tau, p) \in \Delta}\\
	&\cup		
	\set{(q^{r,j}, \alpha, p^{r,j}) \mid 1 \leq j \leq k \text{ and }(q, \alpha, p) \in \Delta \text{ and } \alpha \in \set{\tau} \cup ?\Sigma}\\
	& \cup 
	\set{(q^{b,j}, \alpha, p^{b,j}) \mid 1 \leq j \leq k \text{ and }(q, \alpha, p) \in \Delta \text{ and } \alpha \in \set{\tau} \cup !!\Sigma}\\
	& \cup 
	\set{(q^{r,j}, !!m, p^{b,j+1}) \mid 0 \leq j < k \text{ and }(q, !!m, p) \in \Delta}\\
	& \cup 
	\set{(q^{b,j}, ?m, p^{r,j+1}) \mid 0 \leq j < k \text{ and }(q, ?m, p) \in \Delta}\\
	& \cup 
	\set{(q^{b,k}, ?m, p^{r,k}) \mid (q, ?m, p) \in \Delta}.
\end{align*}

The last case of the definition of $\Delta_k$ implements the fact that we never prevent a reception from occurring, even in the last phase.

The following lemma establishes that this definition of unfolding can be used as an underapproximation for $\Cover$. 
\begin{lemma}\label{lemma:cover-general-protocol-phase-bounded-protocol}
	$q_f$ can be covered in $\PP$ if and only if there exist $k\in\nat$, $y\in\set{r,b}$ and $0\leq j\leq k$ such that $(q_f,y,j)$ can be covered in $\PP_k$.
\end{lemma}
}
\Iflong{\begin{proofof}{\cref{lemma:cover-general-protocol-phase-bounded-protocol}}}
	\Ifshort{\begin{proof}}
		\textbf{Left-to-right direction:}
	Assume that $q_f$ can be covered in $\PP$, and let $C_0 \trans C_1 \trans \cdots \trans C_n$ be an execution of $\PP$ such that $C_i = (\Gamma, L_i)$ for all $0 \leq i \leq n$ and there exists $v_f \in \Vert{\Gamma}$ such that $L_n(v_f) = q_f$. 
	For all $v \in V$, we define $\kappa_v$ a function associating to each $0 \leq i \leq n$ the number and the type of phase in which vertex $v$ is. More formally, $\kappa_v(0) =0$ and for $0 \leq i < n$, 
	$$
	\kappa_v(i+1) = \left\{
	\begin{array}{ll}
		(b,1) & \mbox{if }\kappa_v(i) = 0 \text{ and } C_i \transup{v, t} C_{i+1} \text{ for some } t  =(q, !!m, q') \in \Delta,\\ &  q,q' \in Q, m \in \Sigma\\
		(r,1) & \mbox{if } \kappa_v(i) =0 \text{ and }  C_i \transup{v', t} C_{i+1} \text{ with } v \in \Neigh{v'} \text{ and } \\ &  t  =(q, !!m, q') \in \Delta, q,q' \in Q, m \in \Sigma \text{ and } (L_i(v), ?m, L_{i+1}(v)) \in \Delta\\
		(b,j) & \mbox{if }\kappa_v(i) = (r,j-1) \text{ and }  C_i \transup{v, t} C_{i+1} \text{ for some } t  =(q, !!m, q') \in \Delta, \\& q,q' \in Q, m \in \Sigma\\
		(r,j) & \mbox{if } \kappa_v(i) = (b,j-1) \text{ and }  C_i \transup{v', t} C_{i+1} \text{ with } v \in \Neigh{v'} \text{ and } \\ & t  =(q, !!m, q') \in \Delta, q,q' \in Q, m \in \Sigma \text{ and } (L_i(v), ?m, L_{i+1}(v)) \in \Delta\\
		\kappa_v(i) &\text{otherwise.}
	\end{array}
	\right.
	$$
	We define $k = \max_{v\in \Vert{\Gamma}} \set{x_v \mid \kappa_v(n) = (x_v, y), y \in \set{r,b}}$. Since, for all $0\leq i<n$, for all $v\in\Vert{\Gamma}$, 
	$\kappa_v(i)\leq \kappa_v(i+1)$, it holds that $k\geq  \set{x_v \mid \kappa_v(i) = (x_v, y), y \in \set{r,b}, 0\leq i\leq n}$. We now consider $\PP_k = (Q_k, \Sigma, \qinit^0, \Delta_k)$ the $k$-unfolding of $\PP$.
	For each $0\leq i\leq n$, for all $v\in\Vert{\Gamma}$, if $L_i(v)=q$ we let $L'_i(v)=q^{\kappa_v(i)}$. We now show $C'_0 \trans C'_1 \trans \cdots \trans C'_n$ is an execution of $\PP_k$ where for all $0 \leq i \leq n$, $C'_i = (\Gamma, L'_i)$.
	Consider $C'_0 = (\Gamma, L'_0)$ with $L'_0(v)= \qinit^0$ for all $v \in \Vert{\Gamma}$. As $\kappa_v(0) = 0$ for all $v\in \Vert{\Gamma}$, the induction property holds and $C'_0$ is initial.
	Assume we proved that $C'_0 \trans^\ast C'_i$.  
	We have to prove that $C'_i \trans C'_{i+1}$.
	Denote $v_!$ the vertex such that $C_i \transup{v_!, t} C_{i+1}$.
	
	If $t = (q_1, \tau, q_2)$ for some $q_1, q_2 \in Q$, for all $v \in \Vert{\Gamma}$, $\kappa_v(i+1) = \kappa_v(i)$ and if $v \neq v_!$, $L_{i+1}(v) = L_i(v)$ hence, $L'_{i+1}(v) = L'_i(v)$. By definition of $\PP_k$, 
	$(q_1^{\kappa_{v_!}(i)}, \tau, q_2^{\kappa_{v_!}(i+1)}) \in \Delta_k$. Hence, if $\alpha = \tau$, $C'_i \trans C'_{i+1}$.

	Now let $t = (q_1, !!m, q_2)$ for some $m \in \Sigma$ and $q_1, q_2 \in Q$. We start by observing that for all $v \in \Vert{\Gamma} \setminus  (\set{v_!} \cup \Neigh{v_!})$, $\kappa_v(i+1) = \kappa_v(i)$ and $L_{i+1}(v) = L_i(v)$ hence, $L'_{i+1}(v) = L'_i(v)$.
	
	Observe that, either (a) $\kappa_{v_!}(i+1) = (b,j)$ and $\kappa_{v_!}(i) = (r, j-1)$ (or $\kappa_v(i) = 0$) for some $j \leq k$, or (b) $\kappa_{v_!}(i+1) = \kappa_{v_!}(i)$.
	In case (a), $L'_i(v_!)=q_1^{(r,j-1)}$ or $L'_i(v_!)=q_1^{0}$, and $L'_{i+1}(v_!)= L_{i+1}(v_!)^{\kappa_{v_!}(i+1)}=L_{i+1}(v_!)^{(b,j)} = q_2^{(b,j)}$. By definition, $(q_1^{(r,j-1)}, !!m, q_2^{(b,j)})\in \Delta_k$ (or $(q_1^0, !!m, q_2^{b,1}) \in \Delta_k$).

	In case (b), $\kappa_v(i)=\kappa_v(i+1)= (b,j)$ for some $1\leq j\leq k$. Then, $L'_i(v_!)=q_1^{(b,j)}$ and $L'_{i+1}(v_!)= L_{i+1}(v_!)^{\kappa_{v_!}(i+1)}=L_{i+1}(v_!)^{b,j} = q_2^{(b,j)}$. By definition, $(q_1^{(b,j)}, !!m, q_2^{(b,j)})\in \Delta_k$. 
	
	Let now $v \in \Neigh{v_!}$. 	
	If there is no $q' \in Q$ such that $(L_i(v), ?m, L_{i+1}(v))\in \Delta$, then $L_i(v) = L_{i+1}(v)$ and $\kappa_v(i+1) = \kappa_v(i)$. Hence $L'_{i+1}(v) = L_{i+1}(v)^{\kappa_v(i+1)} =  L_{i}(v)^{\kappa_v(i)} = L'_i(v)$. Furthermore, by construction of $\PP_k$, there is no transition $(L'_{i}(v), ?m, p^\kappa) \in \Delta_k$ for some $p^\kappa \in Q_k$ as otherwise, there would be a transition $(L_{i}(v), ?m ,p) \in \Delta$, which contradicts the fact that $L_i(v) = L_{i+1}(v)$.
	
	Let now $v \in \Neigh{v_!}$ such that  $(L_i(v), ?m, L_{i+1}(v))\in \Delta$.
	By definition of $\kappa_v$, again, either (a) $\kappa_v(i+1) = (r,j)$ and $\kappa_v(i) = (b, j-1)$ (or $\kappa_v(i) = 0$) for some $1\leq j$, or (b) $\kappa_v(i+1) = \kappa_v(i)$. In that case, $\kappa_v(i)=(r,j)$ for some $j\geq 1$. In both cases, by construction of $k$, it holds that $j \leq k$. 
	Hence, by definition of $\Delta_k$,
	in case (a), $L'_i(v)=L_i(v)^{(b,j-1)}$ and $L'_{i+1}(v)=L_{i+1}(v)^{(r,j)}$ and $(L'_i(v), ?m, L'_{i+1}(v))\in \Delta_k$, and in case (b), 
	$L'_i(v)=L_i(v)^{(r,j)}$ and $L_{i+1}(v)=L_{i+1}(v)^{(r,j)}$. Hence, $(L'_i(v), ?m, L'_{i+1}(v))\in \Delta_k$. 
	
	
		We conclude that $C'_i \trans C'_{i+1}$.
		
	Hence, $L'_n(v_f) = q_f^{\kappa_{v_f}(n)}$ and so there exists $0 \leq j \leq k$ and $y \in \set{r,b}$ such that $q_f^{y,j}$ is coverable in $\PP_k$.

	\textbf{Right-to-left direction:} Conversely, assume that there exist $k \in \nat$, $0 \leq j \leq k$ and $y \in \set{r,b}$ such that $q_f^{y,j}$ is coverable in $\PP_k$.
	Let $C_0 \trans C_1 \trans \cdots \trans C_n$ an execution such that $C_n = (\Gamma, L_n)$ and there exists $v_f \in \Vert{\Gamma}$ with $L_n(v_f) = q_f^{y,j}$. We show that there exists an execution $C'_0 \trans C'_1 \trans \cdots \trans C'_n$ such that for all $1 \leq i\leq n$, for $C'_i = (\Gamma, L'_i)$ and any vertex $v \in \Vert{\Gamma}$, $L'_i(v) = q$ if and only if 
	there exists $\kappa_{v,i} \in \set{0}\cup (\set{r,b} \times [0,n])$ such that $L_i(v) = q^{\kappa_{v,i}}$. 
	
	Let $C'_0 = (\Gamma, L'_0)$ with $L_0(v) = \qinit$ for all $v\in \Vert{\Gamma}$. As $L'_0(v) = \qinit^0$ for all $v \in \Vert{\Gamma}$, the induction hypothesis holds. 
	Assume we proved that $C'_0 \trans^\ast C'_i$ where for all $v \in \Vert{\Gamma}$, $L'_i(v) = q$ if and only if 
	there exists $\kappa_{v,i} \in \set{0}\cup (\set{r,b} \times [0,n])$ such that $L_i(v) = q^{\kappa_{v,i}}$. 
	
	Let $C'_{i+1} = (\Gamma, L'_{i+1})$ be the configuration such that for all $v \in \Vert{\Gamma}$, $L'_{i+1}(v) = q$ if and only if there exists $\kappa_{v,i+1} \in \set{0}\cup (\set{r,b} \times [0,n])$ such that $L_{i+1}(v) = q^{\kappa_{v,i+1}}$. We prove now that $C'_i \trans C'_{i+1}$.
	
	Denote $C_i \transup{v_!, t} C_{i+1}$ and $t = (q_1^{\kappa_{v_!,i}}, \alpha, q_2^{\kappa_{v_!,i+1}}) \in \Delta_{k}$. From the definition of $\Delta_k$, it holds that $(q_1, \alpha, q_2) \in \Delta$. 
	Hence, if $\alpha = \tau$, $C'_i \trans C'_{i+1}$, as for all other nodes, $L'_{i+1}(v) = q$ if and only if there exists $\kappa_{v,i+1} \in \set{0}\cup (\set{r,b} \times [0,n])$ such that $L_{i+1}(v) = q^{\kappa_{v,i+1}} $. Furthermore, $L_{i+1}(v) = L_i(v) = q^{\kappa_{v,i+1}}$ hence $L'_i(v) = q = L'_{i+1}(v)$. 
	
	Assume now that $\alpha = !!m$ for some $m \in \Sigma$ and let $v\in \Neigh{v_!}$. Either 
	(a) $(L_i(v), ?m, L_{i+1}(v)) \in \Delta_{k}$, or (b) there is no $p\in Q_k$ such that $(L_i(v), ?m, p) \in \Delta_k$. In case (a), by construction of $\Delta_k$, there exists $(p_1, ?m, p_2) \in \Delta$ such that $L_i(v) = p_1^{\kappa_1}$ and $L_{i+1}(v) = p_2^{\kappa_2}$ for some $\kappa_1, \kappa_2 \in \set{0} \cup (\set{r,b} \times [1,n])$ Hence, $(L'_i(v), ?m, L'_{i+1}(v)) \in \Delta$.
	
	In case (b), we show that there is no $p\in Q$ such that $(L'_i(v), ?m, p) \in \Delta$. Assume by contradiction that this is the case, and denote $L_i(v) = p_1^\kappa$ and $(p_1, ?m, p_2) \in \Delta$. 
	If $\kappa = 0$, then by definition of $\Delta_k$, $(p_1^\kappa, ?m, p_2^{r, 1}) \in \Delta_k$ and we reach a contradiction. 
	If $\kappa = (r, j)$ for some $1 \leq j \leq k$, by definition of $\Delta_k$, $(p_1^{(r,j)}, ?m, p_2^{(r,j)})\in\Delta_k$ and again, we reach a contradiction. 
	If $\kappa = (b,j)$, 
	by definition of $\Delta_k$: $(p_1^{(b,j)}, ?m, p_2^{(r, j+1)}) \in \Delta_k$ and we reach a contradiction.
	Finally, if $\kappa = (b,k)$ then, 
	by definition of $\Delta_k$: $(p_1^{(b,k)}, ?m, p_2^{(r, k)}) \in \Delta_k$ and we reached a contradiction.
	
	Hence, there is no transition $(p_1, ?m, p_2) \in \Delta$ for some $p_2 \in Q$. 
	
	We conclude that $C'_i \trans C'_{i+1}$. Hence, there is an execution $C'_0 \trans^\ast C'_n$ of $\PP$ with $L_n(v_f) = q_f$ and $q_f$ is coverable in $\PP$.
	\Ifshort{\end{proof}}
\Iflong{\end{proofof}}

\section{Undecidability proof of \Cover\ (\cref{sec:undecidable})}\label{app:undecidability}

We reduce the coverability problem of a Minsky machine to the \Cover\ problem in broadcast networks with 6-phase-bounded protocols.

\paragraph*{Minsky Machines.}
A Minsky Machine $M$ is a tuple $M=(\Loc, \TransM, \ell_0, \ell_f, \counter_1, \counter_2)$ such that $\Loc$ is a finite state of locations, $\ellinit \in \Loc$, $\ell_f \in \Loc$, $\counter_1, \counter_2$ are two counters 
and $\TransM$ is a finite set of transitions such that $\TransM \subseteq \Loc \times
\set{\inc{\counter_i}, \dec{\counter_i}, \test{\counter_i} \mid i = 1, 2} \times \Loc$.
We denote a configuration of $M$ by a tuple $(\ell, x_1, x_2)$ where $\ell \in \Loc$ and $x_1, x_2 \in \nat$ are values of the two counters (respectively $\counter_1$ and $\counter_2$).
Let $(\ell, x_1, x_2)$ and $(\ell', x_1', x_2')$ be two configurations, and $t \in \TransM$. We note $(\ell, x_1, x_2) \transRelM{t} (\ell', x_1, x_2 )$ with $t = (\ell, \text{op}, \ell')$ and one of the following conditions holds:
\begin{itemize}
	\item $\text{op} = \inc{\counter_i}$ for some $i \in \set{1, 2}$ and $x'_i = x_i +1$ and $x'_{3-i} = x_{3-i}$;
	\item $\text{op} = \dec{\counter_i}$ for some $i \in \set{1, 2}$ and $x'_i = x_i -1$ and $x'_{3-i} = x_{3-i}$;
	\item $\text{op} = \test{\counter_i}$ for some $i \in \set{1, 2}$ and $x'_i = x_i = 0$ and $x'_{3-i} = x_{3-i}$.
\end{itemize}
The halting problem asks if there is a sequence of transitions $t_1, \dots, t_n$ such that $(\ellinit, 0,0) \transRelM{t_1} (\ell_1, x_1^1, x_2^1) \transRelM{t_2} \cdots \transRelM{t_n} (\ell_f, x_1^n, x_2^n)$ for $x_1^i, x_2^i \in \nat$ for all $1 \leq i \leq n$ and $x_1^n = x_2^n = 0$. It is known to be undecidable. We denote $X$ for the set of counters $\set{\counter_1, \counter_2}$.

\subsection{Definition of $\PP$.}

Let $M = (\text{Loc}, \text{Trans}, \ellinit, \ell_f, \counter_1, \counter_2)$ be a Minsky Machine. 
The protocol $\PP$ of the reduction is described in \cref{fig:undec:prot-line} and is defined on a finite alphabet $\Sigma$ that we define now. We start by defining $\text{OP} = \set{\testmess{\counter}{},  \todoinc{\counter}{}, \tododec{\counter}{},\ovtodoinc{\counter}{}, \ovtododec{\counter}{} \mid \counter \in X}$. Next, for $0 \leq i \leq 2$, we define $\text{OP}^i = \set{\AOp^i \mid \AOp \in \text{OP}}$ and 
$\text{OK}^i = \set{\doneop{\AOp}{i} \mid \AOp \in \text{OP}}$.
We can now define the alphabet $\Sigma = \set{\text{done, 0, 1, 2}} \cup \bigcup_{0 \leq i \leq 2} \text{OP}^i  \cup \text{OK}^i$. We also define for $0 \leq i\leq 2$ two operations: $i\ominus 1 = (i-1) \mod 3$ and $i\oplus 1 = (i+1) \mod3$.

A transition from a sub protocol (an orange box) to a state means that from any states of the sub protocol, there is an outgoing transition with the same label and to the same state. Hence, any state $q$ of $\PP_M$ has outgoing transitions $(q, ?1, \frownie)$, 
$(q, ?2, \frownie)$ and $(q, ?0, \frownie)$.

The main idea is to build a protocol that will ensure, during an initialization phase, a specific organization among processes depicted in \cref{fig:undec:line}. 
In particular, after the initialization phase, processes will be organized as line $v_0, v_1, \dots, v_n$ with the first process $v_0$ in a state $\ellinit$ corresponding to the initial location of the machine, the last process $v_n$ in a state $q^{\textsf{tail}}$, and for $1 \leq i <n$, process $v_i$ is on state $\zerostate^{b_i}$ where $b_i = i \mod 3$. This way, a process $v_i$ (for some $0\leq i <n$) can \emph{distinguish} between its two neighbors (intuitively between its right and left). For instance, if $n <5$, process $v_4$ is on a state $\zerostate^1$ and has its left neighbor $v_3$ on $\zerostate^0$ and its right neighbor $v_5$ on a state $\zerostate^2$.

$\PP_M$, $\PP_\textsf{tail}$ and $\PP_i$ for $0 \leq i \leq 2$ will be defined on the same alphabet of messages $\Sigma$.


\begin{itemize}
	\item We describe $\PP_M= (Q^M, \Delta^M)$ in \cref{fig:undec:prot-M}, this is where the machine is simulated: the process $v_0$ evolves in the sets $\Loc$, we will see that if it takes a transition of the machine which should not be taken (a false test to 0, a decrement not possible, or an increment not feasible because of lack of processes), then other processes will fail in their execution, and $q_f$ will not be covered.
	\item We describe $\PP_1= (Q^1, \Delta^1)$ in \cref{fig:undec:prot}. The processes evolving on $\PP_1$ will simulate counters' values, each of them representing one unit of one of the counter: when the process is on $\counter_1^1$ [resp. $\counter_2^1$] it represents a unit of counter $\counter_1$ [resp. $\counter_2$]. 

	\item $\PP_2$ and $\PP_3$ are obtained from $\PP_1$ by replacing each $j \in \set{0,1,2}$ appearing in the protocol (states and transitions) by $j\oplus 1$ for $\PP_2$ and $j\ominus 1$ for $\PP_0$. As in $\PP_1$, processes evolving on $\PP_2$ or $\PP_0$ will simulate counters' values: for $\counter \in X$, a process on $\counter^2$ or $\counter^0$ represents one unit of $\counter$.
	
	\item We describe $\PP_{\textsf{tail}} = (Q^t, \Delta^t)$ in \cref{fig:undec:prot-tail}. The process $v_n$ will check that all the operations have been correctly executed. When it will receive message "$\text{done}$" (first sent by the process $\PP_M$), it will reach state $q_f$.
	
\end{itemize}
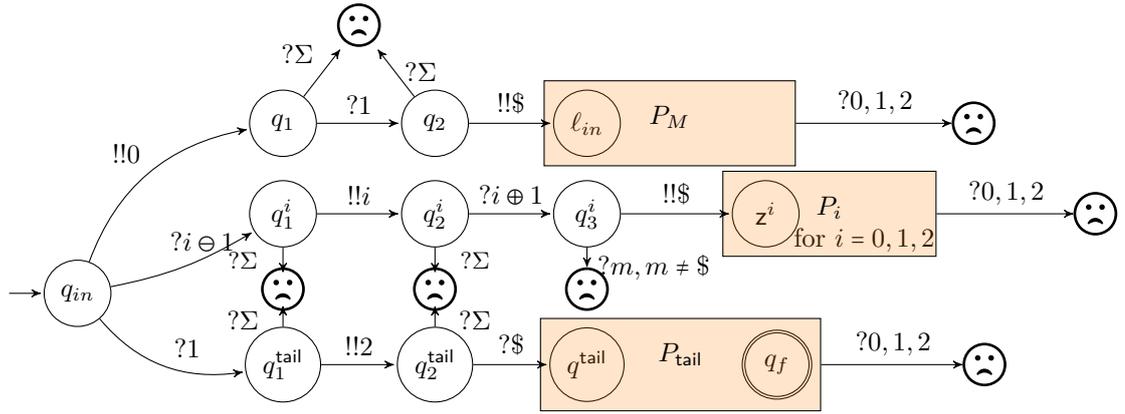
\begin{figure}
	\begin{center}
		\tikzset{box/.style={draw, minimum width=4em, text width=4.5em, text centered, minimum height=17em}}

\begin{tikzpicture}[-, >=stealth', shorten >=1pt,node distance=2cm,on grid,auto, initial text = {}] 
	\node[state,initial] (q0) {$\qinit$};
	
	\node[state] (q2) [right  of = q0, xshift = 20, yshift = 30] {$q_1^i$};
	\node[state] (q3) [right  of = q2] {$q_2^i$};
	\node[state] (q3b) [right  of = q3] {$q_3^i$};
	\node[state] (q4) [right  of = q3b, xshift = 10] {$\textsf{z}^i$};
	
	\node[state] (q1) [above = 1.2 of  q2] {$q_1$};
	\node[state] (q1b) [ right = of  q1] {$q_2$};
	\node[state] (l0) [right  of = q1b] {$\ellinit$};
	
	\node[state] (q5) [below = 2 of q2] {$q_1^{\textsf{tail}}$};
	\node[state] (q5b) [right = of q5] {$q_2^{\textsf{tail}}$};
	\node[state] (q6) [right = of q5b] {$q^{\textsf{tail}}$};
	
	\node [] (l0b) [right = 2.5 of l0] {};
	\node [] (q4b) [right = 2 of q4] {};
	\node [state,accepting] (q6b) [right = 2.5 of q6] {$q_f$};
	\node [] () [right =  0 of q4b, yshift = -10, xshift = -20] {for $i = 0,1, 2$};
	
	\node[draw, fill = orange, fill opacity = 0.2, text opacity = 1, fit=(l0) (l0b), text height=0.04 \columnwidth, label ={[shift={(0ex,-5ex)}]:$\PP_M$}] (AM) {};
	\node[draw, fill = orange, fill opacity = 0.2, text opacity = 1, fit=(q4) (q4b), text height=0.06 \columnwidth, label ={[shift={(0ex,-5ex)}]:$\PP_i$}] (A1) {};
	\node[draw, fill = orange, fill opacity = 0.2, text opacity = 1, fit=(q6) (q6b), text height=0.06 \columnwidth, label ={[shift={(0ex,-5ex)}]:$\PP_{\textsf{tail}}$}] (AT) {};
	
	\node[inner sep = -0] (b1) [above = 1.3 of q1, xshift = 1cm] {$\Huge\frownie$};
	\node[inner sep = -3] (b2) [below = 1of q2, xshift = 0] {$\Huge\frownie$};
	\node[inner sep = -1] (b2b) [below = 1 of q3b, xshift = 0] {$\Huge\frownie$};
	\node[inner sep = -3] (b2bb) [below = 1 of q3, xshift = 0] {$\Huge\frownie$};
	\node[inner sep = -3] (b3) [right = 4 of AM] {$\Huge\frownie$};
	\node[inner sep = -3] (b4) [right =3.5 of A1] {$\Huge\frownie$};
	\node[inner sep = -3] (b5) [right = 4 of AT] {$\Huge\frownie$};
	
	\path[->] 
	(q0) edge [bend left]node {$!!0$} (q1) 
	(q1) edge  node {$?1$} (q1b)
	(q1b) edge node {$!!\$$} (l0)
	
	(q0) edge [bend right = 10]node [xshift = 22, yshift =2]{$?i\ominus1$} (q2)
	(q2) edge node {$!!i$} (q3)
	(q3) edge node {$?i\oplus 1$} (q3b)
	(q3b) edge node {$!!\$$} (q4)
	
	(q0) edge [bend right]node {$?1$} (q5)
	(q5) edge node {$!!2$} (q5b)
	(q5b) edge node {$?\$$} (q6)
	
	(q2) edge node [below, yshift = 10, xshift = -15]{
		\begin{tabular}{c}
			$?\Sigma$
		\end{tabular}
		} (b2)
	(q3) edge node [below, yshift = 10, xshift = 15]{\begin{tabular}{c c}
			$?\Sigma$
	\end{tabular}} (b2bb)
	(q5) edge node [below, yshift = 5, xshift = -15]{$?\Sigma$} (b2)
	(q5b) edge node [below, yshift = 5, xshift =15]{$?\Sigma$} (b2bb)
	(q3b) edge node [below, yshift = 5, xshift =25]{$?m, m\neq \$$} (b2b)
	(q1) edge node {$?\Sigma$} (b1)
	(q1b) edge node [right] {$?\Sigma$} (b1)
	(AM) edge node {$?0,1,2$} (b3)
	(A1) edge node {$?0,1,2$} (b4)
	(AT) edge node {$?0,1,2$} (b5)
	;
\end{tikzpicture}
	\end{center}
	\caption{Description of \PP, the protocol of the reduction. }\label{fig:undec:prot-line}
\end{figure}
\begin{figure}
	\begin{center}
		\tikzset{box/.style={draw, minimum width=4em, text width=4.5em, text centered, minimum height=17em}}

\begin{tikzpicture}[-, >=stealth', shorten >=1pt,node distance=2cm,on grid,auto, initial text = {}] 
	\node[rounded rectangle, draw, inner sep =2] (v0) {$v_0: \ellinit$};
	\node[rounded rectangle, draw, inner sep =2] (v1) [right of =v0] {$v_1: \zerostate^1$};
	\node[rounded rectangle, draw, inner sep =2] (v2) [right of =v1] {$v_2: \zerostate^2$};
	\node[rounded rectangle, draw, inner sep =2] (v0b) [right of =v2] {$v_3: \zerostate^0$};
	\node[rounded rectangle, draw, inner sep =2] (v1b) [right of = v0b] {$v_4:\zerostate^1$};
	\node[rounded rectangle, draw, inner sep =2] (v1bb) [right of = v1b, xshift = 10] {$v_{n-1}: \zerostate^1$};
	\node[rounded rectangle, draw, inner sep =2] (v2t) [right = of v1bb] {$v_n: q^{\textsf{tail}}$};
	
	\path (v1b) -- node[auto=false]{\ldots} (v1bb);
	
	\path[-] 
	(v0) edge node {} (v1)
	(v1) edge node {} (v2)
	(v2) edge node {} (v0b)
	(v0b) edge node {} (v1b) 
	(v1bb) edge node {} (v2t) 
	;
\end{tikzpicture}
	\end{center}
	\caption{Particular shape of configurations from which we can simulate the Minsky machine.}\label{fig:undec:line}
\end{figure}
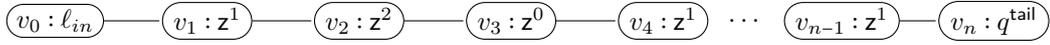 
\begin{figure}
\begin{minipage}[c]{0.45\linewidth}
	\tikzset{box/.style={draw, minimum width=4em, text width=4.5em, text centered, minimum height=17em}}

\begin{tikzpicture}[-, >=stealth', shorten >=1pt,node distance=2cm,on grid,auto, initial text = {}] 
	\node[state] (l) {$\ell$};
	\node[state] (q) [right = of l, yshift = 0] {$q_t$};
	\node[state] (l2) [right = 4.5 of l] {$\ell'$};
	\node[inner sep = 0] (d1) [below = of l,yshift = 10] {$\Huge\frownie$};
	\node[inner sep = 0] (d2) [below = of q,yshift = 10] {$\Huge\frownie$};
	\node[inner sep = 0] (d3) [below = of l2,yshift = 10] {$\Huge\frownie$};
	
	\path[->] 
	(l) edge node [xshift = 0]{$!!\incmess{\counter}{0}$} (q)
	(q) edge node [xshift = -0, yshift = 0]{$!!\okincmess{\counter}{0}$} (l2)
	(q) edge node [right] {\begin{tabular}{l l l}
			$?m, $\\
			$ m \nin \set{\incmess{\counter}{1}, $\\
			$\overincmess{\counter}{1}
			}$
		\end{tabular}
	} (d2)
	(l) edge node [right] {\begin{tabular}{l l }
			$?m, $\\
			$m\nin \text{OK}^1$
		\end{tabular}
	} (d1)
	(l2) edge node [right] {\begin{tabular}{l l }
			$?m,  $\\
			$m \nin  \text{OK}^1$
		\end{tabular}
	} (d3)
	;
\end{tikzpicture}
	\subcaption{Translation of a decrement transition $t = (\ell, \inc{\counter}, \ell')$.}
\end{minipage}
\hfill
\begin{minipage}[c]{0.45\linewidth}
	\tikzset{box/.style={draw, minimum width=4em, text width=4.5em, text centered, minimum height=17em}}

\begin{tikzpicture}[-, >=stealth', shorten >=1pt,node distance=2cm,on grid,auto, initial text = {}] 
	\node[state] (l) {$\ell$};
	\node[state] (q) [right = of l, yshift = 0] {$q_t$};
	\node[state] (l2) [right = 4.5 of l] {$\ell'$};
	\node[inner sep = 0] (d1) [below = of l,yshift = 10] {$\Huge\frownie$};
	\node[inner sep = 0] (d2) [below = of q,yshift = 10] {$\Huge\frownie$};
	\node[inner sep = 0] (d3) [below = of l2,yshift = 10] {$\Huge\frownie$};
	
	\path[->] 
	(l) edge node [xshift = 0]{$!!\decmess{\counter}{0}$} (q)
	(q) edge node [xshift = -0, yshift = 0]{$!!\okdecmess{\counter}{0}$} (l2)
	(q) edge node [right] {\begin{tabular}{l l l}
			$?m,$\\
			$ m \nin  \set{\decmess{\counter}{1},$\\
				$ \overdecmess{\counter}{1}}$
		\end{tabular}
	} (d2)
	(l) edge node [right] {\begin{tabular}{l l}
			$?m, $\\
			$m \nin	\text{OK}^1$
		\end{tabular}
	} (d1)
	(l2) edge node [right] {\begin{tabular}{l l}
			$?m,$\\
			$	m \nin  \text{OK}^1$
		\end{tabular}
	} (d3)
	;
\end{tikzpicture}

%
%
	\subcaption{Translation of a decrement transition $t = (\ell, \dec{\counter}, \ell')$.}
\end{minipage}
\\
\begin{minipage}[c]{0.45\linewidth}
	\vspace*{5mm}
	\tikzset{box/.style={draw, minimum width=4em, text width=4.5em, text centered, minimum height=17em}}

\begin{tikzpicture}[-, >=stealth', shorten >=1pt,node distance=1.9cm,on grid,auto, initial text = {}] 
	\node[state] (l) {$\ell$};
	\node[state] (q) [right = of l, yshift = 0] {$q_t$};
	\node[state] (l2) [right = of q] {$\ell'$};
	\node[inner sep =0] (d1) [below = of l,yshift = 10] {$\Huge\frownie$};
	\node[inner sep = 0] (d2) [below = of q,yshift = 10] {$\Huge\frownie$};
	\node[inner sep = 0] (d3) [below = of l2,yshift = 10] {$\Huge\frownie$};
	
	\path[->] 
	(l) edge node [xshift = 0]{$!!\text{test}_\counter^0$} (q)
	(q) edge node [xshift = -0, yshift = 0]{$!!\ovtest{\counter}{0}$} (l2)
	(q) edge node [right] {\begin{tabular}{l l }
			$?m,$\\
			$ m \neq {\text{test}^1_\counter}$
		\end{tabular}
	} (d2)
	(l) edge node [right] {\begin{tabular}{ll}
			$?m,$\\
			$m \nin \text{OK}^1$
		\end{tabular}
	} (d1)
	(l2) edge node [right] {\begin{tabular}{ll}
			$?m,$\\
			$m \nin \text{OK}^1$
		\end{tabular}
	} (d3)
	;
\end{tikzpicture}

%
%
%
	\subcaption{Translation of a test transition $t = (\ell, \test{\counter}, \ell')$.}
\end{minipage}
\hfill 
\begin{minipage}[c]{0.4\linewidth}
	\tikzset{box/.style={draw, minimum width=4em, text width=4.5em, text centered, minimum height=17em}}

\begin{tikzpicture}[-, >=stealth', shorten >=1pt,node distance=2cm,on grid,auto, initial text = {}] 
	\node[state] (l) {$\ell_f$};
	\node[state] (q) [right = of l] {$q_f^M$};

	\path[->] 
	(l) edge node [xshift = 0]{$!!\text{done}$} (q)
	;
\end{tikzpicture}
	\subcaption{End of the simulation.}
\end{minipage}
\caption{Description of $\PP_M$.}\label{fig:undec:prot-M}
\end{figure}
\begin{figure}
	\begin{center}
		\tikzset{box/.style={draw, minimum width=4em, text width=4.5em, text centered, minimum height=17em}}

\begin{tikzpicture}[-, >=stealth', shorten >=1pt,node distance=2cm,on grid,auto, initial text = {}] 
	\node[state] (qt) {$q^{\textsf{tail}}$};
	\node[state] (rep) [right = 4 of qt] {$q^t_{\textsf{op}}$};
	\node[inner sep = -3] (bad1) [left = 4 of qt] {$\Huge\frownie$};
	\node[inner sep = -3] (bad2) [right = 3 of rep] {$\Huge\frownie$};
	\node[state, accepting] (qf) [below = of qt, yshift = 10] {$q_f$};
	
	\node[] () [below = 1.2 of rep] {\begin{tabular}{c c}
			$\forall \textsf{op} \in \set{\overincmess{\counter}{}, \overdecmess{\counter}{},$\\
			$\testmess{\counter}{} \mid \counter \in X}$
		\end{tabular}};
	
	\path[->] 
	(qt) edge node [xshift = 0]{$?\text{done}$} (qf)
	(qt) edge node [] {\begin{tabular}{c c c}
			$?\AOp^1$\\
			$\forall \AOp \in \set{\incmess{\counter}{}, \decmess{\counter}{} \mid$\\
			$  \counter \in X}$
		\end{tabular}
		 } (bad1)
	(qt) edge [bend left = 10] node {$?\AOp^1$} (rep)
	(rep) edge [bend left = 10] node {$?\ovanOp{\AOp}{1}$} (qt)
	(rep) edge node {$?m, \forall m \in \Sigma$} (bad2)
	;
\end{tikzpicture}
	\end{center}
	\caption{Description of $\PP_{\textsf{tail}}$.}\label{fig:undec:prot-tail}
\end{figure}
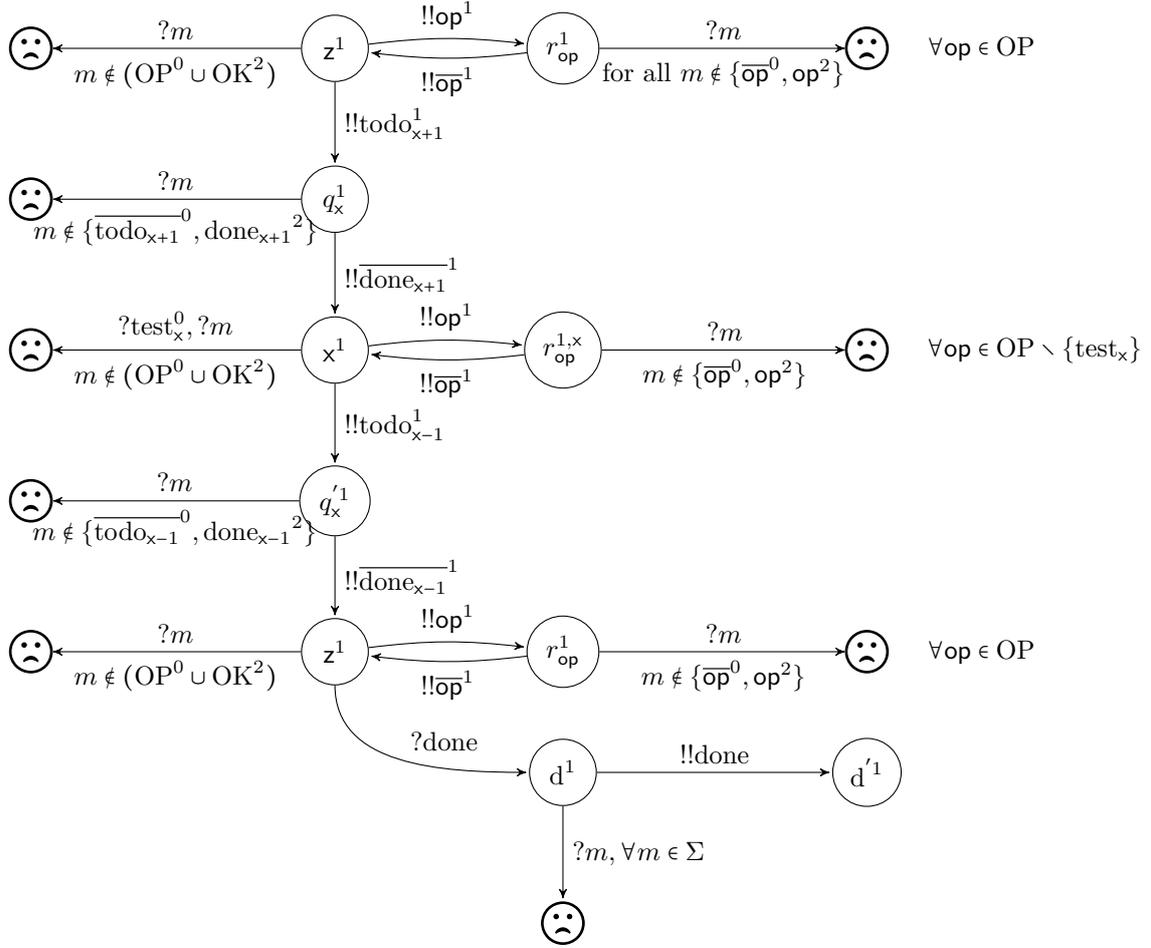
\begin{figure}
	\begin{center}
		\tikzset{box/.style={draw, minimum width=4em, text width=4.5em, text centered, minimum height=17em}}

\begin{tikzpicture}[-, >=stealth', shorten >=1pt,node distance=2cm,on grid,auto, initial text = {}] 
	\node[state] (0) {$\zerostate^1$};
	\node[state] (0x) [below =of 0] {$q_\counter^1$};
	\node[state] (x) [below = 2 of 0x] {$\counter^1$};
	\node[state] (x0) [below = 2 of x] {$q_\counter^{'1}$};
	\node[state] (0b) [below =  of x0] {$\zerostate^1$};
	
	\node[state] (rep0) [right = 3 of 0]{$r_{\textsf{op}}^{1}$};
	\node[state] (rep1) [right = 3 of x]{$r_{\textsf{op}}^{1,\counter}$};
	\node[state] (rep2) [right = 3 of 0b]{$r_{\textsf{op}}^{1}$};

	\node[inner sep= -3] (bad1) [left = 4 of 0] {$\Huge\frownie$};
	\node[inner sep= -3] (bad2) [left = 4 of 0x] {$\Huge\frownie$};
	\node[inner sep= -3] (bad3) [left = 4 of x] {$\Huge\frownie$};
	\node[inner sep= -3] (bad4) [left = 4 of x0] {$\Huge\frownie$};
	\node[inner sep= -3] (bad5) [left = 4 of 0b] {$\Huge\frownie$};
	
	\node[inner sep= -3] (bad6) [right = 4 of rep0] {$\Huge\frownie$};
	\node[inner sep= -3] (bad7) [right = 4 of rep1] {$\Huge\frownie$};
	\node[inner sep= -3] (bad8) [right = 4 of rep2] {$\Huge\frownie$};
	
	\node[] () [right  = 1.5 of bad6] {$\forall \textsf{op} \in \text{OP}$};
	\node[] () [right  = 1.5 of bad7, xshift = 20] {$\forall \textsf{op} \in \text{OP} \setminus \set{\text{test}_\counter}$};
	\node[] () [right = 1.5 of bad8] {$\forall \textsf{op} \in \text{OP}$};
	
	\node[state] (d1) [below = 1.6 of rep2] {$\text{d}^1$};
	\node[state] (d2) [right = 4 of d1] {$\text{d}^{'1}$};
	\node[inner sep = 0] (bad9) [below = of d1] {$\Huge\frownie$};

	\path[->]
	(0) edge node {$!!\incmess{\counter}{1}$} (0x)
	(0x) edge node {$!!\okoverincmess{\counter}{1}$} (x)
	(x) edge node {$!!\decmess{\counter}{1}$} (x0)
	(x0) edge node {$!!\okoverdecmess{\counter}{1}$} (0b)
	


	(0) edge [bend left = 7] node []{$!!\textsf{op}^1$} (rep0)
	(rep0) edge [bend left = 7] node []{$!!\overline{\textsf{op}}^1$} (0)
	
	(x) edge [bend left = 7] node []{$!!\textsf{op}^1$} (rep1)
	(rep1) edge [bend left = 7] node []{$!!\overline{\textsf{op}}^1$} (x)
	
	(0b) edge [bend left = 7] node []{$!!\textsf{op}^1$} (rep2)
	(rep2) edge [bend left = 7] node []{$!!\overline{\textsf{op}}^1$} (0b)
	
%
	(0) edge node [above, yshift = -0] {$?m$} node[below, yshift=-0]{
	$m \notin (\text{OP}^0 \cup \text{OK}^2)$}
		 (bad1)
	
	(0x) edge node [above, yshift = 0] {$?m$}
	node[below, yshift=0]{$m \notin 
			 \set{\okincmess{\counter}{0}, \overincmess{\counter}{2}
				}$
	} (bad2)

	(x) edge [bend left = 0]node [above] {$?\testmess{\counter}{0}, ?m$}
	node[below, yshift=0]{
		 $m\notin (\text{OP}^0 \cup \text{OK}^2)$}
	 (bad3)

	
	(x0) edge node [above, yshift = -0] {$?m$} node[below]{$m \notin\set{\okdecmess{\counter}{0} ,\overdecmess{\counter}{2}}$
	} (bad4)

	(0b) edge node [above, yshift = 0] {$?m$} node[below]{
		 $m\notin (\text{OP}^0\cup \text{OK}^2)$
	} (bad5)


	(rep0) edge node [above, yshift = 0] 
			{$?m$} node [below, yshift=0] 
			{for all $m\notin\set{\overline{\textsf{op}}^{0}, {\textsf{op}}^{2}}$}
	 (bad6)

	(rep1) edge node [above, yshift = 0] {
			$?m$} node[below]{
			$m\notin \set{\overline{\textsf{op}}^{0}, {\textsf{op}}^{2}}$
	} (bad7)

	(rep2) edge node [above, yshift = 0] {$?m$}  node[below]{$m
			\notin \set{\overline{\textsf{op}}^{0},{\textsf{op}}^{2} }$
		} (bad8)

	(0b) edge [out = -90, in = 180] node {$?\text{done}{}$} (d1)
	(d1) edge node {$!!\text{done}$} (d2)
	(d1) edge node {$?m, \forall m \in \Sigma$} (bad9)
	; 

\end{tikzpicture}
	\end{center}
	\caption{Description of $\PP_1$. We draw two states $\zerostate^1$ for readability's sake. $\PP_{2}$ is obtained by replacing each $j \in \set{0, 1, 2}$ by $j \oplus 1$ and $\PP_0$ by replacing each $j  \in \set{0, 1, 2}$ by $j \ominus 1$.
	}\label{fig:undec:prot}
\end{figure}

State $\frownie$ is a deadlock state (there is no outgoing transitions from $\frownie$). In order for $\PP$ to be phase-bounded, we actually duplicate $\frownie$ state as much as needed so there is no contradictions (one state per reception phase). For instance, if there are two receptions transitions $(q, ?m, \frownie)$, $(q', ?m', \frownie)$ with $q$ and $q'$ not in the same phase, we duplicate $\frownie$ into two states $\frownie, \frownie'$ such that $(q, ?m, \frownie)$, and $(q', ?m', \frownie')$.
 For ease of notation, when referring to any of those deadlock states, we use $\frownie$.

As a result, we get the following lemma.

\begin{lemma}\label{lemma:undecidability-bounded-phases-prot}
$\PP$ is 6-phase-bounded.
\end{lemma}
\begin{proof}
We give the partition of states satisfying the 6-phase-bounded protocols condition.
\begin{itemize}
	\item $Q_0 = \set{\qinit}$;
	\item $Q_1^b= \set{q_1}$ and $Q_1^r= \set{q_1^i\mid i \in \set{0,1,2}}\cup \set{q_1^{\textsf{tail}}}\cup \set{\frownie}$;
	\item $Q_2^b = \set{q_2^i\mid i \in \set{0,1,2}} \cup \set{q_2^{\textsf{tail}}}$ and $Q_2^r = \set{q_2} \cup \set{\frownie}$;
	\item $Q_3^b = \set{\ell, q_t, q_f^M \mid \ell \in \Loc, t \in \TransM}$ and $Q_3^r = \set{q_3^i \mid i \in \set{0,1,2}} \cup Q^t \cup \set{\frownie}$;
	\item $Q_4^b = \set{q \in Q^i \setminus{\frownie, \text{d}^i, \text{d}'^i} i \in \set{0,1,2}}$ and $Q_4^r = \set{\frownie}$;
	\item $Q_5^b = \emptyset$ and $Q_5^r = \set{\text{d}^i \mid i \in \set{0,1,2}} \cup \set{\frownie}$;
	\item $Q_6^b = \set{\text{d}'^i \mid i \in \set{0,1,2}}$ and $Q_6^r= \set{\frownie}$.
\end{itemize}
%
%
%
%
%
%
%
\end{proof}

Note that $\PP$ is also $k$-phase-bounded for any $k > 6$ by taking $Q_7^r= \set{\frownie}$, $Q_7^b = \emptyset$ and $Q_j^y = \emptyset$ for any $j \geq 8$ and $y \in \set{r,b}$.

In the rest of this section, we denote $b_i = i \mod 3$ for all $i \in \nat$.

\subsection{Completeness of the Reduction.}

\begin{lemma}
	If there is a transition sequence $t_1, \dots t_n$ such that $(\ellinit, 0,0) \transRelM{t_1} (\ell_1, x_1^1, x_2^1) \transRelM{t_2} \cdots \transRelM{t_n} (\ell_f, 0, 0)$, then there exist $C \in \II, C' =(\Gamma' ,L') \in \CC$, such that $C \trans^\ast C'$ and there exists $v \in \Vert{\Gamma'}$ such that $L'(v) = q_f$.
\end{lemma}
\begin{proof}
	We denote by $m'$ the maximal value reached by counters during the machine execution, i.e. $m' = \max_{1 \leq i \leq n}(x_1^i + x_2^i)$.
	We define $m$ as $m'+1$.
	We suppose wlog that $m \mod 3 =1$, if it is not the case, we redefine $m$ as the first integer greater than $m$ equal to 1 modulo 3.
	
	Define $\Gamma = (V, E)$ such that: $V = \set{v_0, v_1, \dots, v_m, v_{m+1}}$ and $E = \set{(v_i, v_{i+1}) \mid 0 \leq i < m+1}$, and take $C_0 = (\Gamma, L_0)$ where for all $v \in V$, $L_0(v) = \qinit$.
	
	We define the following sequence of configurations: 
	$C_0 \transup{v_0, t_0} C_1 \transup{v_1, t_1} \dots \transup{v_{m+1}, t_{m+1}} C_{m+2} \transup{v_0, t'_0} C_{m+2} \transup{v_1, t'_1} \cdots \transup{v_m, t'_m} C_{2m+3}$ where:
	$$
	t_i = \left\{
	\begin{array}{lll}
		(\qinit, !!0, q_1)& \mbox{for } i = 0 \\
		(q_1^{b_i}, !!b_i, q_2^{b_i}) &\mbox{for } 1 \leq i \leq m\\
		(q_1^{\textsf{tail}}, !!2, q^{\textsf{tail}}) & \mbox{for } i = m+1.
	\end{array}
	\right.
	$$
	and
	$$
	t'_i = \left\{
	\begin{array}{lll}
		(q_2, !!\$, \ellinit)& \mbox{for } i = 0 \\
		(q_3^{b_i}, !!\$, \zerostate^{b_i}) &\mbox{for } 1 \leq i \leq m.
	\end{array}
	\right.
	$$
	\Iflong{
	Denote $C_i = (\Gamma, L_i)$ for all $1 \leq i \leq m+2$. 
	It holds that, $L_1(v_0) = q_1$, 
	$L_1(v_1) = q_1^1$ 
	and for all $v \in V \setminus\set{v_0, v_1}$, $L_1(v) = \qinit$, and for all $1< i < m+2$:
	$$
	L_i(v_j) = \left\{
	\begin{array}{lllll}
		\ellinit& \mbox{for } j=0\\
		\textsf{z}^{b_j}& \mbox{and }0< j \leq i-2\\
		q_2^{b_j} & \mbox{and } 0<j, j = i -1, i\leq m+1\\
		q_1^{b_j} & \mbox{for } 0<j, j = i\mbox{ and } i < m+1\\
		\qinit & \mbox{for } j > i\\
		q_1^{\textsf{tail}} & \mbox{for } j = i = m+1\\
		q^{\textsf{tail}} & \mbox{for } j =m+1,  i = m+2\\
	\end{array}
	\right.
	$$
	and $L_{m+2}(v_0) = \ellinit$, $L_{m+2}(v_{m+1}) = q^{\textsf{tail}}$ and for all $1 \leq i \leq m$, $L_{m+2}(v_i) = \mathsf{z}^{b_i}$.}

	\Ifshort{
		Denote $C_{m+2} = (\Gamma, L_{m+2})$ and $C_{2m+3} = (\Gamma, L_{2m+3})$, it holds that: $L_{m+2}(v_0) = q_2$, and $L_{m+2}(v_i) = q_3^{b_i}$ for all $1\leq i \leq m$ and $L_{m+2}(v_{m+1}) = q_2^{\textsf{tail}}$.
		And:  $L_{2m+3}(v_0) = \ellinit$, $L_{2m+3}(v_{m+1}) = q^{\textsf{tail}}$ and for all $1\leq i \leq m$, $L_{2m+3}(v_i) = \textsf{z}^{b_i}$.
	}
	
	We will now build $C'_0 \trans^+ C'_1 \trans^+ \dots \trans^+ C'_n$ such that for all $0 \leq i \leq n$, $C'_i = (\Gamma, L'_i)$, and: $L'_i(v_0) = \ell_i$, $L'_i(v_{m+1}) = q^{\textsf{tail}}$, and $x_1^i = |\set{v \mid L'_i(v) = \counter_1^j, 0 \leq j \leq 2}|$, $x_2^i = |\set{v \mid L'_i(v) = \counter_2^j, 0 \leq j \leq 2}|$, and for all $1 \leq j \leq m$, $L'_i(v_j) \in \set{\mathsf{z}^{b_j}, \counter_1^{b_j}, \counter_2^{b_j}}$.
	
	For $i = 0$, we take $C'_0 = C_{2m+3}$. By construction, it respects the conditions.
	Assume we built $C'_i$ for some $0 \leq i < n$, and consider $(\ell_i, x_1^i, x_2^i) \transup{t_{i+1}} (\ell_{i+1}, x_1^{i+1}, x_2^{i+1})$. For readability's sake, we note $t = (\ell, \text{op}, \ell')$ and we rename $x_1^i, x_2^i, x_1^{i+1}, x_2^{i+1}$ by $x_1, x_2, x_1', x_2'$, finally we rename $C'_i$ by $C$.
	
	By induction hypothesis, $C = (\Gamma, L)$ is such that $L(v_0) = \ell$, $L(v_{m+1}) = q^{\textsf{tail}}$ and $x_1= |\set{v \mid L(v) = \counter_1^j, 0 \leq j \leq 2}|$, $x_2 = |\set{v \mid L(v) = \counter_2^j, 0 \leq j \leq 2}|$, and for all $1 \leq j \leq m$, $L(v_j) \in \set{\mathsf{z}^k, \counter_1^k, \counter_2^k \mid k = j \mod 3}$.
	We distinguish the following cases:
	\begin{itemize}
		\item $\text{op} = \inc{\counter}$ with $\counter\in X$, then by definition of $m$, $x_1 +x_2 < m-1$. 
		Hence, $|\set{v \mid L(v) = \counter_1^j, 0 \leq j \leq 2}|+ |\set{v \mid L(v) = \counter_2^j, 0 \leq j \leq 2}| < m-1$, and so there exists $v_{f} \in \set{v_1, \dots, v_{m-1}}$ such that $L(v_f) \nin \set{\counter_2^j, \counter_1^j \mid 0 \leq j \leq 2}$ with $1 \leq f \leq m$. Let us denote $f$ the first such index. By hypothesis, $L(v_f) \in \set{\mathsf{z}^k, \counter_1^k, \counter_2^k }$ for $k = f \mod 3$, hence $L(v_f) = \mathsf{z}^k$. 
		
		We denote $\othercounter$ the counter such that $\set{\counter, \othercounter} = \set{\counter_1, \counter_2}$. 
		We also denote $p_i:=L(v_i)$, and $p'_i$ the state such that:\\
		$$
		p'_i = \left\{
		\begin{array}{ll}
			r^{b_i}_{\text{inc}_\counter}& \mbox{if } p_i = \mathsf{z}^{b_i}, i \neq f\\
			r^{b_i, \counter}_{\text{inc}_\counter}& \mbox{if } p_i = \counter^{b_i}, i \neq f\\
			r^{b_i \counter}_{\text{inc}_\counter}& \mbox{if } p_i = \othercounter^{b_i}, i \neq f\\
			q_\counter^{b_i}& \mbox{if }i=f
		\end{array}
		\right.
		$$
		
		Consider the following sequence:
		\begin{align*}
		 &C^0 \transup{v_0, (\ell, !!\todoinc{\counter}{0}, q_t)}
		C^1 \transup{v_1, (p_1, !!\todoinc{\counter}{b_1}, p'_1)} 
		C^2 \transup{v_0, (q_t, !!\okincmess{\counter}{0}, \ell')}
		C^3 \transup{v_2, (p_2, !!\todoinc{\counter}{b_2}, p'_2)} \\
		&C^4 \transup{v_1, (p'_1, !!\okincmess{\counter}{b_1}, p_1)}
		C^5 \transup{v_3, (p_3, !!\todoinc{\counter}{b_3}, p'_3)} 
		C^6\transup{v_2, (p'_2, !!\okincmess{\counter}{b_2}, p_2)} \\
		&C^7 \transup{v_4, (p_4 , !!\todoinc{\counter}{b_4}, p'_4)}
		\cdots \transup{v_{f-1}, (p_{f-1}, !!\todoinc{\counter}{b_{f-1}}, p'_{f-1})} 
		C^{2(f-1)} \transup{v_{f-2}, (p'_{f-2}, !!\okincmess{\counter}{b_{f-2}}, p_{f-2})}\\
		&C^{2f-1} \transup{v_{f}, (p_{f}, !!\todoinc{\counter}{b_f}, p'_{f})} 
		C^{2f}  \transup{v_{f-1}, (p'_{f-1}, !!\okincmess{\counter}{b_{f-1}}, p_{f-1})}\\
	&	C^{2f+1}  \transup{v_{f+1}, (p_{f+1}, !!\doneinc{\counter}{b_{f+1}}, p'_{f+1})}
		 C^{2f + 2} \transup{v_{f}, (p'_{f}, !!\ovdoneinc{\counter}{b_{f}}, \counter^{b_f})} \\
	&	C^{2f + 3}\transup{v_{f+2}, (p_{f+2}, !!\doneinc{\counter}{b_{f+2}}, p'_{b_{f+2}})}
		 \cdots \transup{v_{m}, (p_{m}, !!\ovdoneinc{\counter}{b_{m}}, p'_{m})}\\
		 & C^{2m}  \transup{v_{m-1}, (p'_{m-1}, !!\ovdoneinc{\counter}{b_{m-1}}, p_{m-1})} C^{2m+1}  \transup{v_{m}, (p'_{m}, !!\ovdoneinc{\counter}{b_{m}}, p_{b_m})} C^{2m+2}
	\end{align*}
		Note that if $f=1$, $v_1$ broadcasts ${\text{inc}}_\counter^{b_1}$ and $\ovdoneinc{\counter}{b_{1}}$, and every vertices except $v_0$ broadcast overlined messages as well. Furthermore, as $ f\leq m-1$, $v_m$ always broadcasts $\ovtodoinc{\counter}{b_m} \cdot \ovdoneinc{\counter}{b_{m}}$, hence $L^{2m}(v_{m+1}) = q_{\ovtodoinc{\counter}{}}^t$ and $L^{2m+2}(v_{m+1}) = q^{\textsf{tail}}$.
		At the end of the sequence, $C^{2m+2} =(\Gamma, L^{2m+2})$ is such that $L^{2m+2}(v_0) = \ell'$, $L^{2m+2}(v_{m+1}) = q^{\textsf{tail}}$, and, $L^{2m+2}(v_f) = \counter^{b_f}$ and for all $1 \leq i \leq m$, $i\neq f$, $L^{2m+2}(v_i) = L(v_i)$.
		
		\item $\text{op} = \dec{\counter}$, the sequence is analogous, this time we justify the existence of one process on a state $\counter^{j}$ for some $0 \leq j \leq 2$ by induction hypothesis. 
		
		\item $\text{op} = \test{\counter}$, the only difference is that the propagated message does not change, all processes can propagate the message as no processes is on a state $\counter^i$ for some $0 \leq i \leq 2$ (this is true by induction hypothesis).
	\end{itemize} 
	Hence, we built a sequence of configurations $C'_0 \trans^+ C'_1 \trans^+ \cdots \trans^+ C'_n$ such that $C'_n =(\Gamma, L_n)$ with $L_n(v_0) = \ell_f$, $L_n(v_{m+1}) = q^{\textsf{tail}}$ and for all $0 \leq i \leq m$, $L_n(v_i) = \mathsf{z}^{b_i}$. Finally, we build the sequence leading to the final configuration, from $C'_n$: 
	$C'_n \transup{v_0, (\ell_f, !!\text{done}, q_f^M)} C'_{n+1} \transup{v_1, (d^{b_1}, !!\text{done}, d^{'b_1})} C'_{n+2} \transup{v_2, (d^{b_2}, !!\text{done}, d^{'b_2})} \cdots  \transup{v_m, (d^{b_m}, !!\text{done}, d^{'b_m})} C'_{n+m+1}$. Denote $C'_{n+m+1} = (\Gamma, L_{n+m+1})$, it holds that $L_{n+m+1}(v_{m+1}) = q_f$.
\end{proof}

\subsection{Correctness of the Reduction.}

This subsection is devoted to prove the following lemma.
\begin{lemma}\label{lemma:undecidability:correctness}
	If there exist $C \in \II$, $C' =(\Gamma', L') \in \CC$ such that $C \trans^\ast C'$ and there exists $v \in \Vert{\Gamma'}$ such that $L'(v) = q_f$, then there exists a transition sequence $t_1, \dots, t_n$ such that $(\ellinit, 0,0) \transup{t_1} (\ell_1, x_1^1, x_2^1) \transup{t_2} \cdots \transup{t_n} (\ell_f, 0 ,0)$.
\end{lemma}


We say that a process $v \in \Vert{\Gamma}$ is active if it broadcasts (at least) one message during the execution.
We denote $C_0 \trans C_1 \trans \cdots \trans C_n$ where for each $0 \leq i \leq n$, $C_i = ( \Gamma, L_i)$, $C_0 = C$ and $C_n = C'$. Furthermore, we denote $v_f$ the process covering $q_f$. 
Lastly, for sake of readability, for $C, C' \in \CC$, $m \in \Sigma$ and $v\in \Vert{\Gamma}$, we denote $C \transup{v, a} C'$ when there exists a transition $(q, !!a, q') \in \Delta$ for some $q, q' \in Q$ such that $C \transup{v, (q,!!a, q')} C'$.

\subsubsection{First step: Structure.}

In a first step we prove that from the execution between $C$ and $C'$, one can extract a set of vertices $V = \set{v_0, v_1, \dots, v_m, v_{m+1}}$ for some $m \in \nat$ forming a \emph{line} (i.e. for all $0 \leq i \leq m$, $\langle v_i, v_{i+1}\rangle \in \Edges{\Gamma}$). And such that $v_{m+1}$ covers $q_f$ and neighbors of vertices in $V$ which are not in $V$ are not active (i.e. they do not broadcast anything).
This subsubsection is devoted to prove the following lemma.

\begin{lemma}\label{lemma:undecidability:correctness:structure}
	There exists $ V = \set{v_0, v_1, \dots, v_m, v_{m+1}} \subseteq \Vert{\Gamma}$, such that for all $0 \leq i\leq m$, $\langle v_i, v_{i+1} \rangle \in \Edges{\Gamma}$ and $\exists j_0, j_1 , \dots, j_m $ such that $L_{j_0}(v_0) = \ellinit$, $L_{j_{m}}(v_{m+1}) = q^{\textsf{tail}}$ and for all $1 \leq i \leq m$, $L_{j_i}(v_i) = \zerostate^{b_i}$ and for all other vertex $v \in \Vert{\Gamma} \setminus V$ which is a neighbor of a vertex $v' \in V$, $v$ is not active. Furthermore, $m\mod 3 = 1$.
\end{lemma}	
By construction of the protocol, $v_f$ receives message "\text{done}" from a process $v_1$ before reaching $q_f$. This process again either sends "\text{done}" from state $\ell_f$ or receives it before broadcasting it from a new process $v_2$ (note that $v_2 \neq v_f$ as "\text{done}" can be sent only once and $v_f$ broadcasts it after the broadcast of $v_1$). In the second case, we can reiterate the reasoning. We do so until finding a process (let us denote it $v_{m+1}$) sending "\text{done}" from $\ell_f$ and we note $v_1, \dots, v_m$ the intermediate processes. The process $v_{m+1}$ exists as there is a finite number of processes. Furthermore, $\langle v_f, v_1\rangle \in \Edges{\Gamma}$ and for all $0 \leq i < m$, $\langle v_i, v_{i+1}\rangle  \in \Edges{\Gamma}$. We rename $v_i$ by $v_{m+1-i}$ for all $1 \leq i \leq m+1$, and $v_f$ by $v_{m+1}$. Denote $V = \set{v_0, \dots, v_{m+1}}$. 
Note that there exists $j_0 < j_{m+1}$ such that $L_{j_0}(v_0) = \ellinit$ as $v_0$ broadcasts "done" from $\ell_f$, and  $L_{j_{m+1}}(v_{m+1}) = q^{\textsf{tail}}$ as $v_{m+1}$ reaches state $q_f$ from $q^{\textsf{tail}}$ after $v_0$ reached $\ell_f$.
Furthermore, observe that for all $v \in V \setminus \set{v_0, v_{m+1}}$, there exists $i \in \set{0, 1, 2}$ such that there exists $j_1 < j$ with $L_{j}(v) = \zerostate^i$.

We get the two following lemmas.
\begin{lemma}\label{lemma:undecidability:correctness:structure:v0-neighbor}
	$v_0$ has only one active neighbor before broadcasting "done".
\end{lemma}
\Iflong{
\begin{proof}
	There is at least one as $v_0$ receives message "1" to go from $q_1$ to $\ellinit$.
	If there is more than one, then there exists two messages $i_1, i_2 \in \set{0, 1, 2}$ which are sent from neighbors of $v_0$. Indeed, any active vertex starts by broadcasting a message $i \in \set{0,1,2}$. If $i_1$ or $i_2$ is sent before $v_0$ broadcasts 0, then it can not reach $\ell_f$ as it reaches $q_1^{i_1}$ or $q_1^{i_2}$, hence $i_1, i_2$ are both sent after the first broadcasts of $v_0$. Note that $v_0$ can receive at most once a message $i \in \set{0,1,2}$ (in fact $i = 1$) as otherwise it reaches $\frownie$ and cannot broadcast "done".
\end{proof}
}
\begin{lemma}\label{lemma:undecidability:correctness:structure:vm-neighbor}
	$v_{m+1}$ has only one active neighbor before receiving "done".
\end{lemma}
\Iflong{
\begin{proof}
	Same arguments as for \cref{lemma:undecidability:correctness:structure:v0-neighbor}.
\end{proof}
}
Indeed, otherwise, $v_0$ (or $v_m$) has two neighbors sending two messages $i_1, i_2 \in \set{0,1,2}$, by construction of \PP, one of the message will bring $v_0$ in $\frownie$.\lug{preuve plus en détail en Iflong }
With a similar argument, we get the following lemma.

\begin{lemma}\label{lemma:undecidability:correctness:structure:vk-neighbor}
	Let $v_k \in V$ for some $1 \leq k \leq m$, $v_k$ has only two active neighbors before broadcasting "done".
\end{lemma}
\Iflong{
\begin{proof}
	Denote $i_k \in \set{0,1,2}$ such that $v_k$ receives "done" from $\zerostate^{i_k}$. $v_k$ has at least two active neighbors as it must receive $(i_k-1)\mod3$ and $(i_k+1) \mod 3$ and each process can send at most one message $i \in \set{0,1,2}$. 
	If there is more than two, three messages $i_1, i_2,i_3\in \set{0, 1, 2}$ are sent from neighbors of $v_k$. Remark that the first message received by $v_k$ is $(i_k-1)\mod3$, assume wlog that it is $i_1$. If the second message is received before $v_k$ broadcasts $i_k$, then $v_k$ reaches $\frownie$, hence, the second message is received after the broadcasts of $i_k$. If the second message received is different from $(i_k+1)\mod3$, then $v_k$ reaches $\frownie$, hence the second message received is $(i_k+1)\mod3$ (we suppose wlog that it is $i_2$). Let us denote $j$ the index just after the reception of $i_2$, for all $j' \geq j$, $L_{j'}(v_k) \in Q_i $ by construction of the protocol. For all $q \in Q_i$, the only outgoing reception transition is $(q, ?i_3, \frownie )$, hence, the reception of $i_3$ brings $v_3$ in $\frownie$, and $v_3$ can not broadcast "done" from $\frownie$.
\end{proof}
}\lug{preuve plus en détail en Iflong }
We can now prove the following lemma.
\begin{lemma}\label{lemma:undecidability:correctness:structure-indixes-mod}
	There exist $ j_0 < j_1 <\cdots < j_m $,
	such that for all $0 \leq k \leq m+1$, $C_{j_k} \transup{v_k, b_k} C_{j_k+1}$.
	Furthermore $m+1 = 2 \mod 3$.
\end{lemma}
\begin{proof}
	We prove the lemma by induction on $0 \leq k \leq m+1$.
	For $k = 0$, it follows directly from the construction of $\PP$ and the fact that there exists $j$ such that $L_j(v_0) = \ell_f$ ($v_0$ is the first process to broadcast "done"). Denote $j_0$ such that $C_{j_0}\transup{v_0, 0} C_{j_0 +1}$.
	
	Assume the property to hold for some $1 \leq k \leq m+1$ and denote  $j_0 < j_1 < \cdots < j_{k}$ such that for all $0 \leq k' \leq k$, $C_{j_{k'}} \transup{v_{k'}, b_{k'}} C_{j_{k'}+1}$. \\
	
	\textbf{Case $k + 1< m +1$.}
	As $k+1 <m+1$, by construction of $V$, there exists $i \in \set{0,1,2}$ such that there exists $j$ with $L_j(v_{k+1}) = \zerostate^i$.
	By \cref{lemma:undecidability:correctness:structure:vk-neighbor}, $v_{k+1}$ has exactly two active neighbors, which are by construction of $V$, $v_k$ and $v_{k+2}$. 
	Denote $i_{k+1} \in \set{0,1,2}$ the integer such that $v_{k+1}$ receives "done" from $\zerostate^{i_{k+1}}$.
	By induction hypothesis, $C_{j_{k}} \transup{v_{k}, b_{k}} C_{j_{k}+1}$. 
	We distinguish two cases: (a) $k = 0$ and (b) $k > 0$.
	
	\textbf{Subcase (a) $k = 0$.} 
	Assume $v_{k+1}$ broadcasts $i_{k+1}$ from an index $j < j_0$. Observe that $L_j(v_0) = \qinit$ and so $L_{j+1}(v_0) = q_1^{i_{k+1}}$ which contradicts the fact that $v_0$ broadcasts "done" from $\ell_f$, hence, $j_0 < j$, and so $L_{j_0}(v_{k+1}) = \qinit$, and $L_{j_0 +1}(v_1) = q_1^1$. Hence $i_{k+1}= 1$ and $k+1 = 1$.
	
	\textbf{Subcase $k >0$.} Consider $v_{k-1}$, by induction hypothesis $j_{k-1} < j_k$ with $C_{j_{k-1}} \transup{v_{k-1}, b_{k-1}} C_{j_{k-1}+1}$ and $C_{k} \transup{v_{k}, b_k} C_{k +1}$. Hence, as by \cref{lemma:undecidability:correctness:structure:vk-neighbor}, $v_k$ has only two active neighbors, only $v_{k+1}$ can broadcast $(b_k +1) \mod3$ once $v_k$ broadcast $b_k$. Hence, $i_{k+1} = (b_k +1) \mod3 = b_{k+1}$. Hence there exists $j_{k+1} > j_k$ such that $C_{j_{k+1}} \transup{v_{k+1}, b_{k+1}} C_{j_{k+1}+1}$.

	\textbf{Case $k+1 = m+1$.} By construction of $V$, there exists $j'_{m+1}$ such  that $L_{j'_{m+1}}(v_{m+1}) = q^{\textsf{tail}}$. By \cref{lemma:undecidability:correctness:structure:vm-neighbor}, $v_{m+1}$ has exactly one active neighbor and this neighbor broadcasts "1". As $m \geq 1$, consider $v_{m}$. By induction hypothesis, and construction of $\PP$, there exists a unique $j_m$ such that $C_{j_m} \transup{v_m, b_m} C_{j_{m} + 1}$. Hence, $b_m = m \mod 3 = 1$, finally there exists $j_{m+1} > j_m$ such that $C_{j_{m+1}} \transup{v_{m+1}, 2} C_{j_{m+1} + 1}$.%
\end{proof}

We are now ready to prove \cref{lemma:undecidability:correctness:structure}.

\begin{proofof}{\cref{lemma:undecidability:correctness:structure}}	
	Let $V = \set{v_0, \dots, v_{m+1}}$ such that, there exist $ j_0, j_1,\dots, j_{m+1}$ with $L_{j_0}(v_0) = \ell_f$,  $L_{j_m}(v_{m+1}) = q^\textsf{tail}$ and for all $1 \leq k \leq m$, there exists $i \in \set{0, 1, 2}$ with $L_{j_k}(v_k) = \zerostate^i$.
	
	From \cref{lemma:undecidability:correctness:structure-indixes-mod}\ there exists $j'_0 < \dots < j'_{m+1}$ with $C_{j'_k} \transup{b_k} C_{j'_{k+1}}$ for all $0\leq k \leq m+1$, and $(m+1) \mod3 = 2$.
	Furthermore, $L_{j'_1}(v_0) = q_1$ and $L_{j'_1+1}(v_0) = q_2$. In the same way, for all $1 \leq k \leq m$, $L_{j'_{k+1}}(v_k) = q_2^{b_k}$ and $L_{j'_{k+1} + 1}(v_k) = q_3^{b_k}$. Lastly, $L_{j'_{m+1} + 1}(v_{m+1}) = q_2^\textsf{tail}$. Hence, by construction of $\PP$, there exist $j''_0 \dots j''_{m},j''_{m+1}$ such that $L_{j''_i}(v_i) = \zerostate^{b_i}$ for all $1 \leq i \leq m$, and $L_{j''_0}(v_0) = \ellinit$ and $L_{j''_{m+1}}(v_{m+1}) = q^{\textsf{tail}}$.
\end{proofof}

\subsubsection{Second step: the simulation.} 
We are now ready to extract the execution of the machine from the execution between $C$ and $C'$.
Fix $V = \set{v_0, \dots, v_{m+1}}$ the set of vertices from \cref{lemma:undecidability:correctness:structure}. We forget about nodes not in $V$, as they are useless for $v_f$ to cover $q_f$: either they are 
neighbors of some node in $V$ and so they are inactive, either they are not connected to $v_f$ through an inactive neighbor, or they are not connected to $v_f$.

Let $C, C' \in \CC$ and $v \in \Vert{\Gamma}$, we denote $C \transup{\mid v}^\ast C'$ whenever $C = C'$ or there exist $(u_1, t_1), \dots (u_N,t_N)$, with $N \in \nat$ and for all $1 \leq i \leq N$, $u_i \in \Vert{\Gamma} \setminus \set{v}$, $t_i \in \Delta$, and $C_1 \transup{u_1, t_1} \transup{u_2,t_2} \cdots \transup{u_N,t_N} C'$.

Let $\AOp \in \text{OP}$, we denote the set $\text{OK}(\AOp)$ as follows:
\begin{itemize}
	\item $\text{OK}(\todoinc{\counter}{})= \set{\ovtodoinc{\counter}{}, \ovdoneinc{\counter}{}}$ for some $\counter \in X$;
	\item $\text{OK}(\tododec{\counter}{})= \set{\ovtododec{\counter}{}, \ovdonedec{\counter}{}}$ for some $\counter \in X$;
	\item $\text{OK}(\overincmess{\counter}{})= \set{\ovdoneinc{\counter}{}}$ for some $\counter \in X$;
	\item $\text{OK}(\overdecmess{\counter}{})= \set{\ovdonedec{\counter}{}}$ for some $\counter \in X$;
	\item $\text{OK}({\testmess{\counter}{}})= \set{\ovtest{\counter}{}}$ for some $\counter \in X$.
\end{itemize}

We denote $\text{Act} = \set{\counter +1, \counter - 1, \counter = 0 \mid \counter \in \mathsf{X}}$.

We start by some observations.
\begin{lemma}\label{lemma:undecidability:correctness:language-broadcasts}
	For all $0 < i < m+1$, the word broadcast by $v_i$ belongs to the language $E_i = b_i \cdot \$\cdot  (\sum_{\AOp \in \text{OP}, \AOk \in \text{OK}(\AOp)} \AOp^{b_i} \cdot \AOk^{b_i} )^\ast \cdot \text{done}$. Furthermore, the word broadcast by $v_0$ belongs to the language $E_0 = 0 \cdot \$\cdot (\sum_{\AOp \in \set{\todoinc{\counter}{}, \tododec{\counter}{}, \testmess{\counter}{} \mid \counter \in \mathsf{X}} }(\AOp^0\cdot \ovanOp{\AOp}{0}))^\ast \cdot \text{done}$.
\end{lemma}
\begin{proof}
	Follows from the construction of $V$ and $\PP$.
\end{proof}
We denote $n_i$ the index from which $v_i$ broadcasts "done" for all $0 \leq i \leq m$. It holds that $n_0 < \cdots < n_m$ by construction of $V$.
\begin{lemma}\label{lemma:undecidability:correctness:interleaving-1}
	For all $0 \leq i < m$, $0 < j_1 < j_2 \leq n_i$, $\AOp \in \text{OP}$,  $\AOk \in \text{OK}(\AOp)$, such that $C_{j_1} \transup{v_i, \AOp^{b_i}} C_{j_1 +1} \transup{\mid v_i}^\ast C_{j_2} \transup{v_i, \AOk^{b_i}} C_{j_2 + 1}$ there exists a unique $j_1 < j_3 < j_2$ such that $C_{j_3} \transup{v_{i+1} ,w} C_{j_3 +1}$ for some $w \in \Sigma$. Furthermore, if $\AOk = \ovanOp{\AOp_2}{}$, then $w = \AOp_2^{b_{i+1}}$.
\end{lemma}
\begin{proof}
	We denote $b := b_{i+1}$ for readability's sake, and denote $\AOk = \ovanOp{\AOp_2}{}$.
	Recall that $n_i < n_{i+1}$ and $b= b_i\oplus 1$. Hence:
	\begin{itemize}
		\item if $\AOp_2 = \todoinc{\counter}{}$ for some $\counter \in \mathsf{X}$, then $L_{j_2}(v_{i+1}) \in \set{r^b_{\AOp_2}, r^{b,\counter_1}_{\AOp_2}, r^{b,\counter_2}_{\AOp_2}, q^b_\counter}$;
		\item if $\AOp_2 = \tododec{\counter}{}$ for some $\counter \in \mathsf{X}$, then $L_{j_2}(v_{i+1}) \in \set{r^b_{\AOp_2}, r^{b,\counter_1}_{\AOp_2}, r^{b,\counter_2}_{\AOp_2}, q'^b_\counter}$;
		\item otherwise, $L_{j_2}(v_{i+1}) \in \set{r^b_{\AOp_2}, r^{b,\counter_1}_{\AOp_2}, r^{b,\counter_2}_{\AOp_2}}$.
	\end{itemize}
	Indeed, otherwise, $L_{j_2+1}(v_{i+1}) = \frownie$ as $(L_{j_2}(v_{i+1}), ?\ovanOp{\AOp_2}{b_i}, \frownie) \in \Delta_b$.
	Furthermore:
	\begin{itemize}
	\item if $\AOp = \testmess{\counter}{}$ for some $\counter \in \mathsf{X}$, then $L_{j_1}(v_{i+1}) \in \set{\textsf{z}^b, \othercounter^b \mid \othercounter \in \textsf{X}\setminus \set{\counter}}$;
	\item otherwise, $L_{j_1}(v_{i+1}) \in \set{\textsf{z}^b, \othercounter^b \mid \othercounter \in \textsf{X}}$.
	\end{itemize}
	Indeed, otherwise, $L_{j_2+1}(v_{i+1}) = \frownie$ as $(L_{j_1}(v_{i+1}), ?\AOp^{b_i}, \frownie) \in \Delta_b$.
	Note that any path between $L_{j_1 + 1}(v_{i+1})$ and $L_{j_2}(v_{i+1})$ requires at least one broadcast from $v_{i+1}$, and in particular a broadcast of $\AOp_2^b$.
	Hence, there exists $j_1 < j_3 < j_2$ such that $C_{j_3} \transup{v_{i+1}, \AOp_2^b} C_{j_3+1}$, consider the smallest such index. We now show that this is the unique index between $j_1$ and $j_2$ from which $v_{i+1}$ broadcasts something.

	If there exists $j_1 < j_4 <j_3$ such that $C_{j_4} \transup{v_{i+1}, w} C_{j_4+1}$ for some $w \in \Sigma$, and take $j_4$ the largest such index. By construction of $j_4$ and $j_3$,  $L_{j_4 + 1}(v_{i+1}) \in \set{\textsf{z}^b, \othercounter^b \mid \othercounter \in \textsf{X}}$, and so $w \in \text{OK}^b$. Hence, it holds that $L_{j_4}(v_i) \nin \set{r^{b_i}_{\AOp_2}, r^{b_i,\counter_1}_{\AOp_2}, r^{b_i,\counter_2}_{\AOp_2}, q^{b_i}_\counter, q'^{b_i}_\counter}$, as otherwise, $(L_{j_4}(v_i) , ?w, \frownie)$. This contradicts the fact that $C_{j_1 +1}\transup{\mid v_i}^\ast C_{j_4}$ and $C_{j_1} \transup{v_i, \AOp^{b_i}} C_{j_1+1}$.
	With a similar argument, there is no $j_3 < j_5 \leq j_2$ such that $C_{j_5} \transup{v_{i+1}, w} C_{j_5+1}$ for some $w \in \Sigma$.

\end{proof}

\begin{lemma}\label{lemma:undecidability:correctness:interleaving-2}
	For all $0 \leq i < m$, $0 < j_1 < j_2 \leq n_i$, $\AOp \in \text{OP}$, $w \in \Sigma$, such that $C_{j_1} \transup{v_i, 
		\ovanOp{\AOp}{b_i}} C_{j_1 +1} \transup{\mid v_i}^\ast C_{j_2} \transup{v_i, w} C_{j_2 + 1}$, there exists a unique $j_1 < j_3 < j_2$ such that $C_{j_3} \transup{v_{i+1} ,w'} C_{j_3 +1}$ for some $w' \in \Sigma$. Furthermore, $w' = \AOk^{'b}$ with $\AOk' \in \text{OK}(\AOp)$.
\end{lemma}
\begin{proof}
	We denote $b := b_{i+1}$ for readability's sake.
	Note that from \cref{lemma:undecidability:correctness:language-broadcasts}, $w \in \text{OP}^{b_i} \cup \set{\text{done}}$.
	Recall that $n_i < n_{i+1}$ and $b = b_i \oplus 1$. Hence, $L_{j_2}(v_{i+1}) \in \set{\zerostate^b, \counter_1^b, \counter_2^b}$, as otherwise $L_{j_2+1}(v_{i+1}) = \frownie$, hence $(L_{j_2}(v_{i+1}) ,?w, \frownie) \in \Delta_b$.
	Furthermore, 
	\begin{itemize}
		\item if $\AOp = \todoinc{\counter}{}$ for some $\counter \in \mathsf{X}$, then $L_{j_1}(v_{i+1})\in \set{r_{\AOp}^b, r_{\AOp}^{b,\counter_1}, r_{\AOp}^{b,\counter_2}, q_\counter^b}$;
		\item if $\AOp = \tododec{\counter}{}$ for some $\counter \in \mathsf{X}$, then $L_{j_1}(v_{i+1})\in \set{r_{\AOp}^b, r_{\AOp}^{b,\counter_1}, r_{\AOp}^{b,\counter_2}, q'^b_\counter}$;
		\item otherwise, $L_{j_1}(v_{i+1})\in \set{r_{\AOp}^b, r_{\AOp}^{b,\counter_1}, r_{\AOp}^{b,\counter_2}}$.
	\end{itemize}
	Note that any path between $L_{j_1}(v_{i+1})$ and $L_{j_2}(v_{i+1})$ requires at least one broadcast from $v_{i+1}$ of a message $w'$, and if $w'$ is the last broadcast from $v_{i+1}$ before reaching $L_{j_2}(v_{i+1})$, it holds that:
	\begin{itemize}
		\item if $\AOp = \todoinc{\counter}{}$, then $w' \in \set{\ovtodoinc{\counter}{b}, \incmess{\counter}{b}}$;
		\item if $\AOp = \tododec{\counter}{}$, then $w' \in \set{\ovtododec{\counter}{b}, \decmess{\counter}{b}}$;
		\item otherwise, $w'= \ovanOp{\AOp}{b}$.
	\end{itemize}
	Hence $w' = \AOk'^b$ with $\AOk' \in \text{OK}(\AOp)$.
	We now show that $j_3$ is the unique index between $j_1$ and $j_2$ from which $v_{i+1}$ broadcasts something. Remember that it has been chosen as the greatest such index.

	If there exists $j_1 < j_4 <j_3$ such that $C_{j_4} \transup{v_{i+1}, w''} C_{j_4+1}$ for some $w'' \in \Sigma$, and take $j_4$ the largest such index. By construction of $j_3$ and $j_4$, $L_{j_4 + 1}(v_{i+1}) \in \set{r_{\AOp}^b, r_{\AOp}^{b,\counter},  q_\counter^b, q'^b_\counter \mid \counter \in \mathsf{X}}$, and so $w'' \in \text{OP}^b$. Hence, it holds that $L_{j_4}(v_i) \nin \set{r^{b_i}_\AOp, r^{b_i, \counter}_\AOp \mid \AOp \in \text{OP}, \counter \in \mathsf{X}}$, as otherwise, $(L_{j_4}(v_i) , ?w, \frownie)$. This contradicts the fact that $C_{j_1 +1}\transup{\mid v_i}^\ast C_{j_4}$ and $C_{j_1} \transup{v_i, \ovanOp{\AOp}{b_i}} C_{j_1+1}$.

\end{proof}

\begin{lemma}\label{lemma:undecidability:correctness:seq-of-indexes}
	Let $0 <j_1^0 < j_2^0< n_0$ and $\AOp \in \set{\todoinc{\counter}{}, \tododec{\counter}{}, \testmess{\counter}{} \mid \counter \in X}$ such that $C_{j_1^0} \transup{v_0,w_0} C_{j_1^0 + 1 } \transup{\mid v_0}^\ast C_{j_2^0 } \transup{v_0, w'_0} C_{j_2^0 +1}$ and $w_0 = \AOp^0$ and $w'_0 = \ovanOp{\AOp}{0}$. Then there exists a unique \emph{sequence of indices} $j_1^1, j_2^1, \dots, j_1^m, j_2^m$ such that for all $1 \leq i \leq m$:
	\begin{itemize}
		\item $j_1^{i-1} < j_1^i  < j_2^{i-1} <j_2^i < n_i$, and
		\item $C_{j_1^i} \transup{v_i, w_i} C_{j_1^i +1} \trans^\ast C_{j_2^i} \transup{v_i, w'_i} C_{j_2^i +1}$ for some $w_i,w'_i \in \Sigma$, and
		\item if $w_{i-1} = \ovanOp{\AOp_i}{b_{i-1}}$ for some $\AOp_i \in \text{OP}$, then $w_i = \AOp_{i}^{b_i}$ and $w'_i = \AOk_i^{b_i}$ with $\AOk_i\in \text{OK}(\AOp_i)$.
	\end{itemize}
	\end{lemma}
\begin{proof}

	Let $0 <j_1^0 < j_2^0< n_0$ and $\AOp \in \set{\todoinc{\counter}{}, \tododec{\counter}{}, \testmess{\counter}{} \mid \counter \in X}$ such that $C_{j_1^0} \transup{v_0, \AOp^0} C_{j_1^0 + 1 } \transup{\mid v_0}^\ast C_{j_2^0 } \transup{v_0, \ovanOp{\AOp}{0}} C_{j_2^0 +1}$. 
	
	We build the sequence of indexes $j_1^1, j_2^1, \dots, j_1^m, j_2^m$ and the sequence of operations $\AOp_1, \dots, \AOp_m$ such that $w_i = \AOp_i^{b_i}$ and $w'_i = \AOk_i^{b_i}$ with $\AOk_i \in \text{OK}(\AOp_i)$ inductively. 
	Using \cref{lemma:undecidability:correctness:interleaving-1}, there exists a unique $j_1^0 < j_1^1 < j_2^0$ such that $C_{j_1^1} \transup{v_1, w_1} C_{j_1^1 +1}$ for some $w_1 \in \Sigma$, and it holds that $w_1= \AOp^1$. As $j_2^0 <n_0$, there exists $j_2^0 <n\leq n_0$ such that $C_{n} \transup{v_0, u} C_{n+1}$ for some $u \in \Sigma$. Using  \cref{lemma:undecidability:correctness:interleaving-2}, there exists a unique $j_2^0 < j_2^1 < n $ such that $C_{j_2^1} \transup{v_1, w'_1} C_{j_2^1 +1}$ for some $w'_1 \in \Sigma$ with $w'_1 = \AOk_1^{1}$ and $\AOk_1 \in \text{OK}(\AOp)$.
	
	Define $\AOp_1 = \AOp$.
	We found a unique pair of indexes $j_1^1, j_2^1$ such that $j_1^0 < j_1^1 < j_2^0 < j_2^1 \leq n_0 < n_1$ and $C_{j_1^1} \transup{v_1, w_1} C_{j_1^1 +1} \trans^\ast C_{j_2^1} \transup{v_1, w'_1} C_{j_2^1 +1}$ for some $w_1, w'_1 \in \Sigma$. Furthermore, $w'_0= \ovanOp{\AOp}{0}$ and $w_1= \AOp_1^1$ and  $w'_1 = \AOk_1^{b}$ with $\AOk_1 \in \text{OK}(\AOp_1)$.

	Assume we have built $j_1^1, j_2^1, \dots, j_1^k, j_2^k$ for some $k < m$.
	And denote $C_{j_1^k} \transup{v_k, \AOp_k^{b_k}} C_{j_1^k +1}$ and $C_{j_2^k} \transup{v_k, \AOk_k^{b_k}} C_{j_2^k +1}$ with $\AOk_k = \ovanOp{\AOp'_k}{}$.
	Using  \cref{lemma:undecidability:correctness:interleaving-1}, there exists a unique $j_1^k < j_1^{k+1} < j_2^k$ such that $C_{j_1^{k+1}} \transup{v_{k+1}, w_{k+1}} C_{j_1^1 +1}$ for some $w_{k+1}\in \Sigma$. Furthermore $w_{k+1}= \AOp_k^{'b_{k+1}}$. 
	
	As $j_2^k <n_k$, there exists $j_2^k <n\leq n_k$ such that $C_{n} \transup{v_k, u} C_{n+1}$ for some $u \in \Sigma$. Using  \cref{lemma:undecidability:correctness:interleaving-2}, there exists a unique $j_2^k < j_2^{k+1} < n \leq n_k < n_{k+1}$ such that $C_{j_2^{k+1}} \transup{v_{k+1}, w'_{k+1}} C_{j_2^{k+1} +1}$ for some $w'_{k+1} \in \Sigma$. Furthermore, $w'_{k+1} = \AOk^{'b}$ and $\AOk' \in \text{OK}(\AOp'_k)$.

\end{proof}

Let $0 <j_1^0 < j_2^0< n_0$ and $\AOp \in \set{\todoinc{\counter}{}, \tododec{\counter}{}, \testmess{\counter}{} \mid \counter \in X}$ such that $C_{j_1^0} \transup{v_0, \AOp^0} C_{j_1^0 + 1 } \transup{\mid v_0}^\ast C_{j_2^0 } \transup{v_0, \ovanOp{\AOp}{0}} C_{j_2^0 +1}$. 
We denote $\textsf{Seq}(j_1^0, j_2^0)$ the unique sequence of indexes $j_1^1, j_2^1, \dots,\linebreak j_1^m, j_2^m$ defined in the previous lemma.

\begin{lemma}\label{lemma:undecidability:correctness:seq-test}
	If $\AOp = \testmess{\counter}{}$ for some $\counter \in \mathsf{X}$, then for all $1 \leq i \leq m$, $L_{j_1^i}(v_i) = L_{j_2^i + 1}(v_i) \in \set{\zerostate^{b_i}, \othercounter^{b_i} \mid \othercounter\neq \counter}$.
\end{lemma}
\begin{proof}
	From \cref{lemma:undecidability:correctness:seq-of-indexes}, as $\text{OK}(\testmess{\counter}{}) = \set{\ovtest{\counter}{}}$, it holds that 
	for all $1 \leq i \leq m$, $C_{j_1^i} \transup{v_i,\testmess{\counter}{b_i}} C_{j_1^i +1}$, and $C_{j_2^i} \transup{v_i,\ovtest{\counter}{b_i}} C_{j_2^i +1}$. As $C_{j_1^i + 1} \transup{\mid v_i}^\ast C_{j_2^i}$, by construction of the protocol, we have that that $L_{j_1^i}(v_i) = L_{j_2^i + 1}(v_i) \in \set{\zerostate^{b_i}, \othercounter^{b_i} \mid \othercounter\neq \counter}$.
\end{proof}

\begin{lemma}\label{lemma:undecidability:correctness:seq-inc}
	If $\AOp = \todoinc{\counter}{}$ for some $\counter \in \mathsf{X}$, then, there exists a unique $1 \leq p \leq m$ such that $L_{j_1^p}(v_p) = \zerostate^{b_p}$, $L_{j_2^p + 1}(v_p) =\counter^{b_p}$, and for all $1 \leq i \leq m$, $i\neq p$, it holds that
	$L_{j_1^i}(v_i) = L_{j_2^i + 1}(v_i) \in \set{\zerostate^{b_i}, \counter_1^{b_i}, \counter_2^{b_i}}$.
\end{lemma}
\begin{proof}
	Denote $\AOp_0, \dots, \AOp_m \in \text{OP}$, $\AOk_0, \dots, \AOk_m \in \text{OK}$ as defined in \cref{lemma:undecidability:correctness:seq-of-indexes}.
	
	We start by oberving the following: 
	Let $0 \leq i < m$ such that $\AOp_i = \doneinc{\counter}{}$, then $\AOp_{i+1} = \doneinc{\counter}{}$ and $\AOk_{i+1} = \ovdoneinc{\counter}{}$. Indeed, it holds that $\AOk_i = \overincmess{\counter}{}$ by definition of $\text{OK}(\overincmess{\counter}{})$, hence .
	by \cref{lemma:undecidability:correctness:seq-of-indexes}, $\AOp_{i+1} = \doneinc{\counter}{}$, and so $\AOk_{i+1} = \ovdoneinc{\counter}{}$ by definition of $\text{OK}(\overincmess{\counter}{})$.
	
	As a consequence, $L_{j_1^{i+1}}(v_{i+1}) = L_{j_2^{i+1}+ 1}(v_{i+1}) \in \set{\zerostate^{b_{i+1}}, \counter_1^{b_{i+1}}, \counter_2^{b_{i+1}}}$.
	
	Now let $0  < i \leq m$, such that $\AOp_i = \todoinc{\counter}{}$, then $\AOp_{i-1} = \todoinc{\counter}{}$ and $\AOk_{i-1}= \ovtodoinc{\counter}{}$. 
	Indeed, by \cref{lemma:undecidability:correctness:seq-of-indexes}, as $\AOp_i = \todoinc{\counter}{}$, it must be that $\AOk_{i-1} = \ovtodoinc{\counter}{}$. And from \ref{lemma:undecidability:correctness:language-broadcasts}, $\AOk_{i-1} \in \text{OK}(\AOp_{i-1})$.
	Note that the only $\AOp \in \text{OP}$ such that $\ovtodoinc{\counter}{} \in \text{OP}(\AOp)$ is $\todoinc{\counter}{}$, hence $\AOp_{i-1} = \todoinc{\counter}{}$. 
	
	As a consequence, $L_{j_1^{i-1}}(v_{i+1}) = L_{j_2^{i-1}+ 1}(v_{i-1}) \in \set{\zerostate^{b_{i-1}}, \counter_1^{b_{i-1}}, \counter_2^{b_{i-1}}}$.
	
	Recall that $\AOp_0 = \todoinc{\counter}{}$ and observe that $\textsf{op}_m = \overincmess{\counter}{}$, as otherwise, necessarily $L_{j_1^m + 1}(v_{m+1})= \frownie$. As a consequence, there exists a unique $0 < p < m$ such that $\AOp_p = \todoinc{\counter}{}$ and $\AOp_{p+1} = \overincmess{\counter}{}$ and it holds that for all $i < p$, $\AOp_i = \todoinc{\counter}{}$ and for all $i > p$, $\AOp_i = \overincmess{\counter}{}$.
	
	Observe that, from \cref{lemma:undecidability:correctness:seq-of-indexes}, as $\AOp_{p+1} = \overincmess{\counter}{}$, $\AOk_p = \ovdoneinc{\counter}{}$, and so by construction of $\PP$, it holds that $L_{j_1^p}(v_p) = \zerostate^{b_p}$, $L_{j_2^p + 1}(v_p) =\counter^{b_p}$.
	
\end{proof}

\begin{lemma}\label{lemma:undecidability:correctness:seq-dec}
		If $\AOp = \tododec{\counter}{}$ for some $\counter \in \mathsf{X}$, then, there exists a unique $1 \leq p \leq m$ such that $L_{j_1^p}(v_p) = \counter^{b_p}$, $L_{j_2^p + 1}(v_p) =\zerostate^{b_p}$, and for all $1 \leq i \leq m$, $i\neq p$, it holds that
	$L_{j_1^i}(v_i) = L_{j_2^i + 1}(v_i) \in \set{\zerostate^{b_i}, \counter_1^{b_i}, \counter_2^{b_i}}$.
\end{lemma}
\begin{proof}
	Same proof as \cref{lemma:undecidability:correctness:seq-inc}.
\end{proof}

 We now associate to $j_1^0, j_2^0$ two configurations $\textsf{prec}(j_1^0, j_2^0) = (\ell, x_1, x_2)$, $\textsf{succ}(j_1^0, j_2^0) = (\ell', x_1', x_2')$ of the machine as follows: 
\begin{itemize}
	\item $\ell = L_{j_1^0}(v_0)$, $x_1 = |\set{i \mid L_{j_1^i}(v_i) = \counter_1^{b_i}, 1 \leq i \leq m}|$, $x_2 = |\set{i \mid L_{j_1^i}(v_i) = \counter_2^{b_i}, 1 \leq i \leq m}|$;
	\item $\ell' = L_{j_2^0 + 1}(v_0)$, $x'_1 = |\set{i \mid L_{j_2^i + 1}(v_i) = \counter_1^{b_i}, 1 \leq i \leq m}|$, $x'_2 = |\set{i \mid L_{j_2^i + 1}(v_i) = \counter_2^{b_i}, 1 \leq i \leq m}|$.
\end{itemize}

\begin{lemma}\label{lemma:undecidability:correctness:machine-conf-step-trans}
	$\textsf{prec}(j_1^0, j_2^0) \transRelM{} \textsf{succ}(j_1^0, j_2^0)$
\end{lemma}
\begin{proof}
	Denote $\AOp_0, \dots, \AOp_m \in \text{OP}$, $\AOk_0, \dots, \AOk_m \in \text{OK}$ as defined in \cref{lemma:undecidability:correctness:seq-of-indexes}.
	
	Observe that by construction, $L_{j_1^0}(v_0) = \ell$, $L_{j_2^0 +1}(v_0) = \ell'$ for some $\ell, \ell' \in \Loc$. 
	
	\begin{itemize}
		\item If $\AOp_0 = \todoinc{\counter}{}$, then by construction, there exists $t = (\ell, \inc{\counter}, \ell') \in \TransM$. Furthermore, by \cref{lemma:undecidability:correctness:seq-inc}, there exists a unique $1 \leq p < m$ such that $L_{j_1^p}(v_p) \neq L_{j_2^p +1 }(v_p)$ and $L_{j_2^p +1 }(v_p) = \counter^{b_p}$ and $L_{j_1^p}(v_p)=\zerostate^{b_p}$. Denote $i \in \set{1, 2}$ such that $\counter =\counter_i$ and $\textsf{prec}(j_1^0, j_2^0) = (\ell, x_1, x_2)$.
		Hence, $\textsf{succ}(j_1^0, j_2^0) = (\ell', x'_1, x'_2)$ with $x'_i = x_i + 1$, and $x'_{3-i} = x_{3-i}$.
		
		\item If $\AOp_0 = \testmess{\counter}{}$, then by construction, there exists $t = (\ell, \test{\counter}, \ell') \in \TransM$. Denote $i \in \set{1, 2}$ such that $\counter =\counter_i$
		By \cref{lemma:undecidability:correctness:seq-test}, for all $1 \leq k \leq m$, $L_{j_1^k}(v_k) = L_{j_2^k+1 }(v_k) \in \set{\zerostate^{b_k}, \counter_{3-i}^{b_k}}$. Hence, if $\textsf{prec}(j_1^0, j_2^0) = (\ell, x_1, x_2)$, then $x_i = 0$. Denote $\textsf{succ}(j_1^0, j_2^0) = (\ell', x'_1, x'_2)$, it holds that $x'_1 = x_1$, $x_2 = x'_2$ and $x'_i = x_i = 0$.
		
		\item If $\AOp_0 = \tododec{\counter}{}$, then by construction, there exists $t = (\ell, \dec{\counter}, \ell') \in \TransM$. Furthermore, by \cref{lemma:undecidability:correctness:seq-dec}, there exists a unique $1 \leq p < m$ such that $L_{j_1^p}(v_p) \neq L_{j_2^p +1 }(v_p)$ and furthermore, $L_{j_2^p +1 }(v_p) = \zerostate^{b_p}$ and $L_{j_1^p}(v_p)=\counter^{b_p}$. Denote $i \in \set{1, 2}$ such that $\counter =\counter_i$ and $\textsf{prec}(j_1^0, j_2^0) = (\ell, x_1, x_2)$.
		Hence, $\textsf{succ}(j_1^0, j_2^0) = (\ell', x'_1, x'_2)$ with $x'_i = x_i -1$, and $x'_{3-i} = x_{3-i}$.
	\end{itemize}
	
\end{proof}

We now denote $j < j_{1, 0} < j_{2,0} < \cdots < j_{1,k} < j_{2,k} < n_0$ for some $k \in \nat$ the indices such that: 
$C_0 \transup{\mid v_0}^\ast C_j \transup{v_0, 0} C_{j+1} \transup{\mid v_0}^\ast C_{j_{1,0}} \transup{v_0, {\AOp^0}^0} C_{j_{1,0} + 1}\transup{\mid v_0}^\ast C_{j_{2,0}} \transup{v_0, \ovanOp{\AOp^0}{0}} C_{j_{2,0} + 1} \transup{\mid v_0}^\ast C_{j_{1,1}}\cdots C_{j_{2,k}} \transup{v_0, \ovanOp{\AOp^k}{0}} C_{j_{2,k} + 1} \transup{\mid v_0}^\ast C_{n_0}$ with $\AOp^0, \dots, \AOp^k \in \set{\todoinc{\counter}{}, \tododec{\counter}{}, \testmess{\counter}{} \mid \counter \in \set{\counter_1, \counter_2}}$.

\begin{lemma}\label{lemma:undecidability:correctness:machine-conf-init}
	$\textsf{prec}(j_{1,0}, j_{2,0}) = (\ellinit, 0, 0)$.
\end{lemma}
\begin{proof}
	Denote $\textsf{Seq}(j_{1,0}, j_{2,0}) = j_1^1, j_2^1, \dots, j_1^m, j_2^m$.
	
	First note that by construction, $L_{j_1^0}(v_0) = \ellinit$.
	
	In fact: for all $1 \leq i \leq m$, it holds that $C_0 \transup{\mid v_i}^\ast C_{x_i} \transup{v_i, b_i} C_{x_i+1} \transup{\mid v_i}^\ast C_{y_i} \transup{v_i, \$} C_{y_i+1} \transup{\mid v_i}^\ast C_{j_{1}^i}$ for some indices $x_i, y_i$. Assume this is not the case and take the smallest index $p$ such that this is not the case: there exists at least two indices from which it broadcasts $b_p$ and $\$$ as $L_{j_1^p}(v_p) \in \set{\zerostate^p, \counter_1^p, \counter_2^p}$. Assume now, there are more than two indices from which $v_p$ broadcasts, i.e. there exist $x_p < y_p < z_p$ such that 
	$C_0 \transup{\mid v_p}^\ast C_{x_p} \transup{v_p, b_p} C_{x_p+1} \transup{\mid v_p}^\ast C_{y_p} \transup{v_p, \$} C_{y_p+1} \transup{\mid v_p}^\ast C_{z_p} \transup{v_p, w} C_{z_p+1} \trans^\ast C_{j_{1}^p}$. By construction of the protocol, $w\in \text{OP}^{b_p}$.
	If $p = 1$, then $L_{z_p}(v_0) = \ellinit$, and so $L_{z_p +1}(v_0) = \frownie$ which contradicts the existence of $j_{1,0}$. 
	Assume $p > 1$, then $L_{z_p}(v_{p-1}) = \zerostate^{b_{p-1}}$ and so $L_{z_p+1}(v_{p-1}) =  \frownie$ which contradicts the fact that $v^{p-1}$ broadcasts "done".
	
	Hence, for all $1 \leq i \leq m$, it holds that $C_0 \transup{\mid v_i}^\ast C_{x_i} \transup{v_i, b_i} C_{x_i+1} \transup{\mid v_i}^\ast C_{y_i} \transup{v_i, \$} C_{y_i+1} \transup{\mid v_i}^\ast C_{j_{1}^i}$ for some indices $x_i, y_i$. As a consequence, $L_{j_1^i}(v_i) = \zerostate^{b_i}$, and so $\textsf{prec}(j_{1,0}, j_{2,0}) = (\ellinit, 0, 0)$.
\end{proof}

\begin{lemma}\label{lemma:undecidability:correctness:machine-conf-step-eq}
	For all $0 \leq i < k$, $\textsf{succ}(j_{1,i}, j_{2,i}) = \textsf{prec}(j_{1,i+1}, j_{2,i+1})$.
\end{lemma}
\begin{proof}
	Let $0 \leq i < k$ such that $\textsf{succ}(j_{1,i}, j_{2,i}) \neq \textsf{prec}(j_{1,i+1}, j_{2,i+1})$.
	Denote $\textsf{succ}(j_{1,i}, j_{2,i}) = (\ell, x_1, x_2)$ and $ \textsf{prec}(j_{1,i+1}, j_{2,i+1}) = (\ell', x'_1, x'_2)$.
	Denote $\textsf{Seq}(j_{1,i}, j_{2,i}) = j_1^1, j_2^1, \dots, j_1^m, j_2^m$ and $\textsf{Seq}(j_{1,i+1}, j_{2,i+1}) = i_1^1, i_2^1, \dots,i_1^m, i_2^m$.
	By construction, $C_{j_{2,i}  + 1} \transup{\mid v_0}^\ast C_{j_{1,i+1}}$ hence, $\ell = \ell'$.
	Hence, there exists $1 \leq p \leq m$ such that $L_{j_2^p +1}(v_p) \neq L_{i_1^p}(v_p)$. Consider the first such index $p$. 
	By \ref{lemma:undecidability:correctness:seq-of-indexes}, $L_{j_2^p +1}(v_p) \in \set{\zerostate^{b_p}, \counter_1^{b_p}, \counter_2^{b_p}}$. Again, by  \ref{lemma:undecidability:correctness:seq-of-indexes}, $L_{i_1^p}(v_p) \neq \frownie$ and so there exists $j_2^p +1 \leq j < i_1^p$ such that $C_j \transup{v_p, w} C_{j+1}$. Take $j$ the first such index, by construction of $\PP$, $w \in \text{OP}^{b_p}$.
	If $p = 1$, then $L_{j+1}(v_0) = \frownie$, as $w \in \text{OP}^{b_p}$ and $L_{j+1}(v_0)  = L_{j_{1,i}}(v_0) = \ell$.
	If $p > 1$, then $L_{j}(v_{p-1}) = L_{j_2^{p-1} +1} (v_{p-1})\in \set{\zerostate^{b_{p-1}}, \counter_1^{b_{p-1}}, \counter_2^{b_{p-1}}}$. Hence, $L_{j+1}(v_{p-1}) = \frownie$ which contradicts the existence of $i_{1}^{p-1}$.
\end{proof}

We are now ready to prove \cref{lemma:undecidability:correctness}:

\begin{proofof}{\cref{lemma:undecidability:correctness}}
	We found a sequence of indices $j < j_{1, 0} < j_{2,0} < \cdots < j_{1,k} < j_{2,k} < n_0$  such that:
	\begin{itemize}
		\item $\textsf{prec}(j_{1,0}, j_{2,0}) = (\ellinit, 0, 0)$ (\cref{lemma:undecidability:correctness:machine-conf-init})
		\item for all $0 \leq j \leq k$, $\textsf{prec}(j_{1,j}, j_{2,j}) \transRelM{} \textsf{succ}(j_{1,j}, j_{2,j})$ (\cref{lemma:undecidability:correctness:machine-conf-step-trans})
		\item for all $0 \leq i < k$, $\textsf{succ}(j_{1,i}, j_{2,i}) = \textsf{prec}(j_{1,i+1}, j_{2,i+1})$ (\cref{lemma:undecidability:correctness:machine-conf-step-eq})
		\item $\textsf{succ}(j_{1,k}, j_{2,k}) = (\ell_f, x_1, x_2 )$ for some $x_1, x_2 \in \nat$, as $L_{j_{2,k}}(v_0) = L_{n_0}(v_0) = \ell_f$.
	\end{itemize}
Putting everything together, we get a sequence $(\ellinit, 0, 0)\transRelM{} \textsf{succ}(j_{1,0}, j_{2,0})\transRelM{} \textsf{succ}(j_{1,1}, j_{2,1}) \transRelM{} \cdots \transRelM{}\textsf{succ}(j_{1,k}, j_{2,k}) = (\ell_f, x_1, x_2 )$.
\end{proofof}

\section{Proofs of  \cref{sec:about-1pb}}
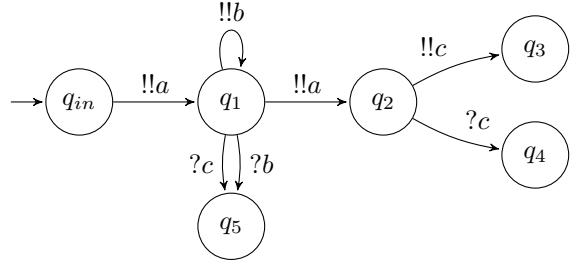
\begin{figure}
	\begin{center}
		\tikzset{box/.style={draw, minimum width=4em, text width=4.5em, text centered, minimum height=17em}}

%
%
%
%
%

\begin{tikzpicture}[-, >=stealth', shorten >=1pt,node distance=2cm,on grid,auto, initial text = {}] 

	\node[state,initial] (q0) [] {$\qinit$};
	\node[state] (q1) [right of = q0, yshift = 0, xshift =0cm] {$q_1$};
	\node[state] (q2) [right  of = q1, yshift = 0, xshift = 00cm] {$q_2$};
	\node[state] (q3) [right  of = q2, yshift = 20, xshift =0] {$q_3$};
	\node[state] (q4) [right of = q2, yshift = -20, xshift = 0] {$q_4$};
	\node[state] (q5) [below  of = q1, yshift = 10]{$q_5$};
	
	\node [] (p1) [left of = q4, yshift = -18] {};
	\node[] (p2) [left of = q1, yshift = 0, xshift =-2cm] {};

	\path[->] 
	(q0) edge [] node {$!!a$} (q1) 
	(q1) edge [bend right = 10] node [left]{$?c$} (q5) 
	(q1) edge [bend left = 10] node {$?b$} (q5) 
	(q1) edge [loop above] node {$!!b$} ()
	(q1) edge [] node {$!!a$} (q2)
	(q2) edge [bend left = 10] node {$!!c$} (q3)
	(q2) edge [bend right = 10] node {$?c$} (q4)

	;
\end{tikzpicture}
	\end{center}
	\caption{Example of a 1-phase-bounded broadcast protocol denoted $\PP$}
\label{fig:bp:example-1pb}
\end{figure}

We start by an example. 
\begin{example}
Consider the protocol $\PP$ depicted in \cref{fig:bp:example-1pb}. $\PP$ is 1-phase-bounded. We depict two star-configurations $C_1$ and $C_2$ in \cref{fig:star-conf}. Both are b-configurations as the root node $v_0$ is in $q_1$ in both configurations and $q_1 \in Q_1^b$. The broadcast-print of $C_1$ is equal to the one of $C_2$ and $\bprint{C_1} = \bprint{C_2} = (q_1, \set{\qinit, q_1, q_2})$.
\end{example}

\begin{figure}
	\begin{minipage}[c]{0.45\linewidth}
		\begin{center}
			\resizebox*{!}{0.065\paperheight}{
				\tikzset{box/.style={draw, minimum width=4em, text width=4.5em, text centered, minimum height=17em}}

\begin{tikzpicture}[-, >=stealth', shorten >=1pt,node distance=1.5cm,on grid,auto, initial text = {}] 
	\node[rounded rectangle, draw, inner sep = 2] (v0) [] {$v_0: q_1$} ;
	\node[rounded rectangle, draw, inner sep = 2] (v1) [left of =v0, yshift = 25] {$v_1: \qinit$};
	\node[rounded rectangle, draw, inner sep = 2] (v2) [left of =v0, yshift = -25] {$v_2: q_5$};
	\node[rounded rectangle, draw, inner sep = 2] (v3) [right of =v0, yshift = 25] {$v_3: q_1$};
	\node[rounded rectangle, draw, inner sep = 2] (v4) [right of =v0, yshift = -25] {$v_4: q_2$};
	\path[-] 
	(v0) edge node {} (v1)
	(v0) edge node {}  (v2)
	(v0) edge node {} (v3)
	(v0) edge node {} (v4);
\end{tikzpicture}
			}
		\end{center}
	\end{minipage}
	\hfill
	\begin{minipage}[c]{0.45\linewidth}
		\begin{center}
			\resizebox*{!}{0.065\paperheight}{
				\tikzset{box/.style={draw, minimum width=4em, text width=4.5em, text centered, minimum height=17em}}

\begin{tikzpicture}[-, >=stealth', shorten >=1pt,node distance=1.5cm,on grid,auto, initial text = {}] 
	\node[rounded rectangle, draw, inner sep = 2] (v0) [] {$v_0: q_1$} ;
	\node[rounded rectangle, draw, inner sep = 2] (v1) [left of =v0, yshift = 25] {$v_1: \qinit$};
	\node[rounded rectangle, draw, inner sep = 2] (v2) [left of =v0, yshift = -25] {$v_2: q_1$};
	\node[rounded rectangle, draw, inner sep = 2] (v3) [right of =v0, yshift = 25] {$v_3: q_1$};
	\node[rounded rectangle, draw, inner sep = 2] (v4) [right of =v0, yshift = -25] {$v_4: q_2$};
	\path[-] 
	(v0) edge node {} (v1)
	(v0) edge node {}  (v2)
	(v0) edge node {} (v3)
	(v0) edge node {} (v4);
\end{tikzpicture}
			}
		\end{center}
	\end{minipage}
	\caption{Example of two star-configurations of $\PP$ denoted $C_1$ to the left and $C_2$ to the right.}\label{fig:star-conf}
\end{figure}
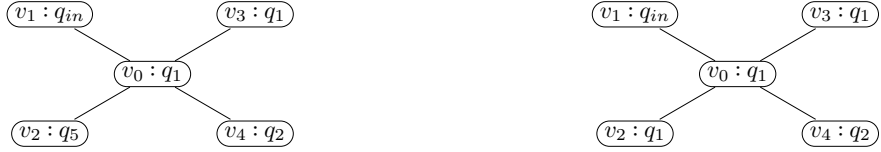

\subsection{Proof of \cref{lemma:1-bounded-cover-star}}

\begin{proofof}{\cref{lemma:1-bounded-cover-star}}
  Assume there exists $\Gamma=(V,E)$ in $\Topo$, $C = (\Gamma, L)\in \II$ and $D
  =(\Gamma, L')\in \CC_\PP$ and $v\in \Vert{\Gamma}$ such that $C
  \trans^\ast D$ and $L'(v) = q_f$. Thanks to \cref{thm:Cover-CoverTree-equivalent}, we can assume that $\Gamma \in \Trees$ and without loss of generality that $v=\epsilon$ (i.e. it the root of the tree). We consider the execution:
  $$
  C_0 \transup{v_1,\delta_1}  C_1\transup{v_2,\delta_2} C_2 \transup{v_3,\delta_3} \cdots \transup{v_n,\delta_n} C_n
  $$
  such that $C_0=C$ and $C_n=D$ and $C_i=(\Gamma,L_i)$ for all $i \in [0,n]$. Note in particular that $L_n(\epsilon)=q_f$. Without loss of generality we can as assume that $|v_n| \leq 1$ otherwise it means that the last step of this execution does not affect the root node and hence we could have stop it at $C_{n-1}$. We denote by $m$ the number of steps in this executions labelled by the root node or the node at height $1$ whose label in the previous step belongs to $Q_0 \cup Q^b_1$, i.e. $m=|\set{i \in [1,n] \mid v_i=\epsilon \mbox{ or  } (|v_i| =1 \mbox{ and } L_{i-1}(v_i) \in Q_0 \cup Q^b_1)}|$. Intuitively $m$ is the number of steps in the execution, which can affect the root node or which an internal move of a node at height $1$ which stays in the $Q_0 \cup Q^b_1$ component. Note that since we are dealing with 1-phase-bounded protocol if for some $i \in [1,n]$ and some nodes $v \in V$, we have $L_i(v) \notin Q_0 \cup Q^b_1$ then $L_j(v) \notin Q_0 \cup Q^b_1$ for all $j \in [i,n]$.
  We consider then the configuration $\Gamma' \in \Stars$ which is  a restriction of the tree $\Gamma$ to nodes of height at most one, i.e. $\Gamma'=(V',E')$ with $V'=\set{v \in V \mid |v| \leq 1}$ and $E'=\set{(v,v') \in E \mid v,v' \in V'}$. 

  We will show how to build an execution of length $m$ from $C'=(\Gamma',L'_0)\in \II$ to a configuration $D'=(\Gamma',L''')$ with $L'''(\epsilon)=q_f$. For this matter we consider an injective function $f : [1,m] \mapsto [1,n]$ which select $m$ indices of the  step of the previously introduced execution keeping only the one labelled by nodes in $V'$. Formally $f$ is the unique injective function respecting the following conditions:
  \begin{itemize}
  \item for all $i \in [1,m]$,$v_{f(i)}=\epsilon$ or $(|v_{f(i)}| =1$ and  $L_{f(i)-1}(v_{f(i)}) \in Q_0 \cup Q^b_1$) ;
  \item for all $i,j \in [1,m]$, if $i <j$ then $f(i)<f(j)$.
  \end{itemize}
  Since $v_n \in V'$, we have $f(m)=n$.  We shall now show that there exists an execution of the form:
  $$
  C'_0 \transup{v_{f(1)},\delta_{f(1)}}  C'_1\transup{v_{f(2)},\delta_{f(2)}} C'_2 \transup{v_{f(3)},\delta_{f(3)}} \cdots \transup{v_{f(m)},\delta_{f(m)}} C'_m
  $$
  where $C'_0=(\Gamma',L'_0)$ and for all $i \in [1,m]$ we  have $C'_i=(\Gamma',L'_i)$ with $L'_i$ satisfying the following condition: $L_i'(\epsilon)=L_{f(i)}(\epsilon)$  and for all $v \in V' \setminus \set{\epsilon}$, if $L_{f(i)}(v) \in Q^0 \cup Q^b_1$ then $L'_i(v)
  =L_{f(i)}(v)$. We shall now see how to build this execution by induction. First note that if we extend $f$ to $0$ by setting $f(0)=0$ for $i=0$, we have that $L_0'(\epsilon)=L_{f(0)}(\epsilon)=\qinit$ and for all $v \in V' \setminus \set{\epsilon}$, $L'_0(v)=L_{f(0)}(v)=\qinit$ and $\qinit \in Q^0$. Now assume that we have build 
  $
  C'_0 \transup{v_{f(1)},\delta_{f(1)}}  C'_1\transup{v_{f(2)},\delta_{f(2)}} C'_2 \transup{v_{f(3)},\delta_{f(3)}} \cdots \transup{v_{f(i)},\delta_{f(i)}} C'_i
  $ with $C'_i=(\Gamma',L'_i)$ with $L'_i$ satisfying the following condition: $L_i'(\epsilon)=L_{f(i)}(\epsilon)$  and for all $v \in V' \setminus \set{\epsilon}$, if $L_{f(i)}(v) \in Q^0 \cup Q^b_1$ then $L'_i(v)
  =L_{f(i)}(v)$. Consider the node $v_{f(i+1)}$,  we have two cases:
  \begin{enumerate}
  \item $v_{f(i+1)}=\epsilon$. Since between $C_{f(i)}$ and $C_{f(i+1)-1}$ the root node did not perform an internal action, nor any broadcast and none of the node at height 1 performed a broadcast, we deduce that $L_{f(i+1)-1}(v_{f(i+1)})=L_{f(i)}(v_{f(i+1)})=L'_i(v_{f(i+1)})$, hence from $C'_i$ the transition $\delta_{f(i+1)}$ can be taken from the root node. And we have $C'_i \transup{v_{f(i+1)},\delta_{f(i+1)}}  C'_{i+1}$ with $C'_{i+1}=(\Gamma',L'_{i+1})$. Note that we have necessarily  $L_{i+1}'(\epsilon)=L_{f(i+1)}(\epsilon)$. And for all $v \in V' \setminus \set{\epsilon}$, if $L_{f(i+1)}(v) \in Q^0 \cup Q^b_1$, it means that the action of the root node did not affect $v$ then $L_{f(i+1)}(v)=L_{f(i+1)-1}(v)$ and by the same reasoning as we did for the root node, since $L_{f(i+1)-1}(v) \in Q^0 \cup Q^b_1 $ , we have $L_{f(i+1)-1}(v)=L_{f(i)}(v)=L'_i(v)$. But in the transition from $C'_{i}$ to $C'_{i+1}$ the same action still not affect $v$, this allows us to deduce that $L'_{i+1}(v)=L'_{i}(v)$. Hence $L'_{i+1}(v)=L_{f(i+1)}(v)$.
  \item $v_{f(i+1)} \neq \epsilon$. In that case $|v_{f(i+1)}|= 1$ and  $L_{f(i+1)-1}(v_{f(i+1)}) \in Q_0 \cup Q^b_1$. Here again we can deduce that $L_{f(i+1)-1}(v_{f(i+1)})=L_{f(i)}(v_{f(i+1)})=L'_i(v_{f(i+1)})$, hence from $C'_i$ the transition $\delta_{f(i+1)}$ can be taken from the  node $v_{f(i+1)}$. Hence we have  $C'_i \transup{v_{f(i+1)},\delta_{f(i+1)}}  C'_{i+1}$ with $C'_{i+1}=(\Gamma',L'_{i+1})$. Consider $v \in V' \setminus \set{\epsilon}$, if $v=v_{f(i+1)}$, it is clear that $L'_{i+1}(v)=L_{f(i+1)}$. If $v \neq v_{f(i+1)}$ and $L_{f(i+1)}(v) \in Q^0 \cup Q^b_1$, we can deduce that $L'_{i+1}(v)=L'_i(v)$ because the action of node at height 1 does not affect the other node at height 1, and $L_{f(i+1)}(v)=L_{f(i+1)-1}(v)$ for the same reason, and $L_{f(i+1)-1}(v)=L_{f(i)}(v)$ by definition of $f$ and $L_{f(i)}(v)=L'_i(v)$. Hence $L'_{i+1}(v)=L_{f(i+1)}(v)$. If $v=\epsilon$, we can show that $L'_i(v)=L_{f(i)}(v)=L_{f(i+1)-1}(v)$ and this allows us to deduce that $L'_{i+1}(v)=L_{f(i+1)}(v)$.
  \end{enumerate}

  We have hence proven that we can build the execution $
  C'_0 \transup{v_{f(1)},\delta_{f(1)}}  C'_1\transup{v_{f(2)},\delta_{f(2)}} C'_2 \transup{v_{f(3)},\delta_{f(3)}} \cdots \transup{v_{f(m)},\delta_{f(m)}} C'_m
  $ where $C'_m=(\Gamma',L'_m)$ such that $L'_m(\epsilon)=L_{f(m)}(\epsilon)=L_n(\epsilon)=q_f$.
\end{proofof}

\subsection{Proof of \cref{lemma:successor-bprint-ptime}}

\begin{proofof}{\cref{lemma:successor-bprint-ptime}}
  Let $(q,\Lambda) \in Q^b \times 2^{Q^b}$. We consider $F$ subset of $Q^b \times 2^{Q^b}$, built according to the following rules:
  \begin{itemize}
  \item Initially, $F=\emptyset$;
  \item for all transitions $(q,\tau,q') \in \Delta$, add $(q',\Lambda)$ to $F$;
  \item for all transitions $(q,!!m,q') \in \Delta$, add $(q',\Lambda')$ to $F$ where $\Lambda'=\Lambda \setminus \set{q'' \in \Lambda \mid \exists q''' \in Q. (q'',?m,q''') \in \Delta}$;
  \item for all transitions $(q',\tau,q'') \in \Delta$, with $q' \in \Lambda$ add $(q,\Lambda \cup \set{q''})$ and  $(q,(\Lambda \setminus \set{q'}) \cup \set{q''})$ to $F$ (the latter case being here to deal with the case where in the configuration with broadcast-print $(q,\Lambda)$ there is a single node labelled with $q'$);
  \item for all transitions $(q',!!m,q'') \in \Delta$ such that $q' \in \Lambda$ and there does not exist $q''' \in Q$ with $(q,?m,q''') \in \Delta$, add $(q,\Lambda \cup \set{q''})$ and  $(q,(\Lambda \setminus \set{q'}) \cup \set{q''})$ to $F$ .
  \end{itemize}

  It is clear that $F$ can be built in polynomial time and a case analysis allows to show $F=\set{(q',\Lambda')
    \mid (q,\Lambda) \Rightarrow (q',\Lambda') }$.
\end{proofof}

\subsection{Proof of \cref{lemma:monotonicity}}

\begin{proofof}{\cref{lemma:monotonicity}}
  \begin{enumerate}
  \item[(i)] Let $C_1=((V_1,E_1),L_1)$, $C'_1=((V_1,E_1),L_1)$ and $C_2=((V_2,E_2),L_2)$ be star-configurations such that $C_1 \transup{v,\delta}  C'_1$ and $C_1 \preceq C_2$. First note that if $v \neq \epsilon$ and $L_1(v) \notin Q^b$, then $C'_1 \preceq C_2$ and $C_2
    \trans^\ast C_2$. If $v =\epsilon$, then $L_2(v)=L_1(v)$ by definition of  $\preceq$, hence we have $C_2
    \transup{v,\delta} C'_2$ and one can show that $C'_1 \preceq C'_2$ (the root node have obviously the same label and if some nodes at height one are affected because of a broadcast from $C_1$ to $C'_1$, their label will go out of $Q^b$ in $C'_1$ and the nodes with the same label will be affected similarly going to $C'_2$). If $v \neq \epsilon$ and $L_1(v) \notin Q^b$, then there exists a node $v' \in V_2 \setminus\set{\epsilon}$ such that $L_2(v')=L_1(v)$ and we have $C_2
    \transup{v',\delta} C'_2$ with $C'_1 \preceq C'_2$ (here if $\delta$ is a broadcast it should affect the root node from $C_1$ and $C_2$ similarly and does not affect the other nodes).

  \item[(ii)] Assume $C_1=((V_1,E_1),L_1)$, $C'_1=((V_1,E_1),L_1)$ and $C_2=((V_2,E_2),L_2)$ are
    b-configurations such that $C_1 \transup{v,\delta}  C'_1$ and
    $\bprint{C_1}=\bprint{C_2}$ and $C_1 \preceq C_2$. From the previous point, we easily deduce that  there
    exists a b-configuration $C'_2$ such that $C'_1 \preceq C'_2$ and $C_2
    \trans^\ast C'_2$ . However it might be the case that $\bprint{C'_1} \neq \bprint{C'_2}$. In that case, we have necessarily  $\bprint{C'_1}=(q,\Lambda_1)$ and $\bprint{C'_2}=(q,\Lambda_2)$ with $\Lambda_1\subset \Lambda_2$.  This happens if $v\neq \epsilon$ is the only node labelled by $L_1(v) \in Q^b$  in $C_1$ and in $C_2$ there are strictly more than $1$ node labelled by $L_1(v)$. Assume $\set{v' \in V_2 \mid L_2(v)=L_1(v)}=\set{v_1,\ldots,v_k}$ then the execution $C_2 \transup{v_1,\delta} C_{2,1}\transup{v_2,\delta} C_{2,2} \transup{v_3,\delta} C_{2,3} \cdots \transup{v_k,\delta} C_{2,k}$ is a valid execution. Note that along this execution the state of the root node does not change, because $C'_1$ is a b-configuration hence $C_1 \transup{v,\delta}  C'_1$ with $v \neq \epsilon$ does not affect the root node (otherwise it would perform a reception and its state would not be in $Q^b$ anymore). Furthermore we have $C'_1 \preceq C'_{2,k}$ and $C_2
    \trans^\ast C'_{2,k}$ and $\bprint{C'_1} =\bprint{C'_{2,k}}$.
  \item[(iii)] Let $C=((V,E),L),$ be a  b-configuration such that $C_{in}
    \trans^\ast C$ for some initial configuration $C_{in}$ and such that $\bprint{C}=(q,\Lambda)$. Let $N \in \nat$.  We assume that $\Lambda=\set{q_1,\ldots,q_k}$. There exist $k$ nodes $u_1,\ldots,u_k$ in $V \setminus \set{\epsilon}$ such that $L(u_i)=q_i$ for all $i \in [1,k]$. Assume furthermore that we have an execution of the form:

    $$
    C_0 \transup{v_1,\delta_1}  C_1\transup{v_2,\delta_2} C_2 \transup{v_3,\delta_3} \cdots \transup{v_n,\delta_n} C_n
    $$
    with $C_0=C_{in}$ and $C_n=C$. Here again since $C$ is a b-configuration and $q_i \in Q^b$ for all $i \in [1,k]$ , none of the transitions $\delta_i$ performed by a node $v_i$ different than $\epsilon$ changes the state of $\epsilon$ and none of the transitions performed by $\epsilon$ changes the states of nodes $u_1, \dots ,u_k$. We build an initial star configuration $C'_{in}=((V',E'),L'_{in})$ where $V'=V \cup \set{w_{1,1},w_{1,2},\ldots,w_{1,N},w_{2,1},w_{2,2},\ldots,w_{2,N},\ldots,w_{k,1},w_{k,2},\ldots,w_{k,N}}$ with $w_{i,j}$ is a node of depth one such that $w_{i,j}=\max(v \in V \setminus\set{\epsilon})+(i-1)*N+j$ and $E'$ is defined such that $(V',E')$ is a star with root $\epsilon$. Now  we build an execution from $C'_{in}$ to $C'=((V',E'),L')$ following the same steps as the execution from $C_{in}$ to $C$ and replacing each step $\transup{v_\ell,\delta_\ell}$ with $v_\ell=u_i$ for some $i \in [1,k]$, by the sequence 
    $\transup{v_\ell,\delta_\ell} \transup{w_{i,1},\delta_\ell} \cdots \transup{w_{i,N},\delta_\ell}$. It is easy to see that at the end   $\bprint{C}=\bprint{C'}=(q,\Lambda)$ and
    $|\set{v \in V'\setminus\set{\epsilon} \mid L'(v)=q'}|\geq N$ for all $q' \in \Lambda$.

  \end{enumerate}
\end{proofof}

\subsection{Proof of \cref{lemma:completeness-bprint}}

\begin{proofof}{\cref{lemma:completeness-bprint}} 
  Let $(q,\Lambda) \in  Q^b\times
  2^{Q^b}$. First assume that there exist
  two $b$-configurations $C_{in} \in \II$ and $C$  such
  that $C_{in}\trans^\ast C$ and $\bprint{C}=(q,\Lambda)$. Hence there exists an execution of the form $C_0 \trans C_1 \trans  \cdots \trans C_n$ with $C_0=C_{in}$ and $C_n=C$. Let $(q_i,\Lambda_i)=\bprint{C_i}$ for all $i \in [0,n]$.  By definition of $\Rightarrow$, we have $(q_0,\Lambda_0) \Rightarrow (q_1,\Lambda_1) \Rightarrow \cdots \Rightarrow (q_n,\Lambda_n)$, hence $(q_0,\Lambda_0) \Rightarrow^\ast (q_n,\Lambda_n)$. Since $C_0$ is an initial configuration and $(q_0,\Lambda_0)=\bprint{C_0}$, we have  $q_0=\qinit$ and $\Lambda_0 \in  \set{\emptyset,\set{\qinit}}$ and by definition we have $(q_n,\Lambda_n)=\bprint{C_n}=\bprint{C}=(q,\Lambda)$.\\

  Now assume that  $(\qinit,\Lambda_{in}) \Rightarrow^\ast (q,\Lambda)$
  with $\Lambda_{in} \in \set{\emptyset,\set{\qinit}}$. Hence we have $(q_0,\Lambda_0) \Rightarrow (q_1,\Lambda_1) \Rightarrow \cdots \Rightarrow (q_n,\Lambda_n)$ with $(q_0,\Lambda_0)=(\qinit,\Lambda_{in}) $ and $(q_n,\Lambda_n)=(q,\Lambda)$. We reason by induction on $n$ to show that for all $i \in [0,n]$, there exist
  two $b$-configurations $C_{in} \in \II$ and $C_i$  such
  that $C_{in}\trans^\ast C_i$ and $\bprint{C_i}=(q_i,\Lambda_i)$. If $n=0$, then it is clear that there exists an initial $b$ configuration $C_{in} \in \II$ such that $\bprint{C_{in}}=(q_0,\Lambda_0)$ and $C_{in}\trans^\ast C_{in}$.

  Suppose the property holds for $i \in [0,n-1]$ and let us prove it still holds for $i+1$. By induction hypothesis, there there exist
  two $b$-configurations $C_{in} \in \II$ and $C_i$  such
  that $C_{in}\trans^\ast C_i$ and $\bprint{C_i}=(q_i,\Lambda_i)$.

  Since $(q_i,\Lambda_i) \Rightarrow (q_{i+1},\Lambda_{i+1})$, there exists two b-configurations $C'_i=((V'_i,E'_i),L'_i)$ and $C_{i+1}$ such that $C'_i \rightarrow C_{i+1}$ and $\bprint{C'_i}=(q_i,\Lambda_i)$ and $\bprint{C_{i+1}}=(q_{i+1},\Lambda_{i+1})$. We let $N= \max_{q' \in Q^b}(|\set{v \in V'_i \mid L_i(v) = q'}|)$.
  Since $C_{in}\trans^\ast C_i$ and $\bprint{C_i}=(q_i,\Lambda_i)$, using Lemma \ref{lemma:monotonicity}.(iii), there exists an initial configuration $C'_{in}$ and a b-configuration $C''_i=((V''_i,E''_i),L''_i)$ such that $C'_{in} \trans^\ast
  C''_i$ and  $\bprint{C''_i}=(q_i,\Lambda_i)$ and
  $|\set{v \in V''_i \mid L''_i(v)=q'}| \geq N$ for all $q' \in \Lambda_i$. By definition we have $C'_i \preceq C''_i$. By Lemma \ref{lemma:monotonicity}.(ii), there
  exists a b-configuration $C'_{i+1}$ such that $C_{i+1} \preceq C'_{i+1}$ and
  $\bprint{C_{i+1}}=\bprint{C'_{i+1}}=(q_{i+1},\Lambda_{i+1})$ and $C''_i
  \trans^\ast C'_{i+1}$. We deduce that we have $C'_{in} \trans^\ast C'_{i+1}$ and  $\bprint{C'_{i+1}}=(q_{i+1},\Lambda_{i+1})$.
\end{proofof}

\subsection{Unary VASS and the control state reachability problem}

We present here the syntax and semantics of (unary) Vector Addition System with States (VASS). In our context a VASS $V$ is a tuple $(S,X,T)$ where: $S$ is a finite set of control states, $X$ is a finite set of variables taking their value in the natural and $T$ is a finite set of transitions of the form $(s,a,s')$ with $s,s' \in S$ and $a \in \set{x++,x-- \mid x \in X} \cup \set{\mathtt{skip}}$. A configuration of such a VASS is a pair $(s,\nu)$ with $s \in S$ and $\nu:X \mapsto \nat$. We define the transition relation $\vartriangleright$ between VASS configurations as follows: $(s,\nu) \vartriangleright (s',\nu')$ iff there exists a transition $(s,a,s')$ in $T$ one of the following condition holds:
\begin{itemize}
\item $a=\mathtt{skip}$  and $\nu=\nu'$, or,
\item $a=x++$ and $\nu'(x)=\nu(x)+1$ and $\nu'(x')=\nu(x')$ for all $x' \in X \setminus \set{x}$, or,
\item $a=x--$ and $\nu'(x)=\nu(x)-1$ and $\nu'(x')=\nu(x')$ for all $x' \in X \setminus \set{x}$.
\end{itemize}
Remark that if $\nu(x)=0$ then it is not possible to take the transition $(s,x--,s')$ from $(s,\nu)$. We denote by $\vartriangleright^\ast$ the reflexive and transitive closure of $\vartriangleright$. The control state reachability problem for VASS can be defined as follows:

\begin{decproblem}
  \problemtitle{$\VASSCover$~}
  \probleminput{A VASS $V=(S,X,T)$, an initial configuration $(s_{in},\nu_{in})$ and 
    a  state $s_f \in S$;} 
  \problemquestion{Does there exist $\nu' : X \mapsto \nat$ such that $(s_{in},\nu_{in}) \vartriangleright^\ast (s_f,\nu')$ ?}
\end{decproblem}

From \cite{lipton76reachability,rackoff78covering}, we know that $\VASSCover$ is \textsc{ExpSpace}-complete.

\subsection{Proof of \cref{lemma:expspace-cover-bprint}}

\begin{proofof}{\cref{lemma:expspace-cover-bprint}}
  Let $(q,\Lambda) \in  Q^b\times
  2^{Q^b}$.  From the 1-phase-bounded protocol $\PP = (Q, \Sigma,
  \qinit, \Delta)$, we  build a VASS $V=(S,X,T)$ as follows:
  \begin{itemize}
  \item $S=Q  \cup (Q \times \delta) \cup \set{s_{in}}$;
  \item $X=Q^b$;
  \item $T$ is the smallest set verifying the following conditions:
    \begin{itemize}
    \item $(s_{in},q'++,s_{in}) \in T$ for all $q' \in \Lambda$ (initialisation phase adding processes in $\Lambda$); 
    \item  $(s_{in},\mathtt{skip},q) \in T$;
    \item for all $(q_1,\tau,q_2) \in \Delta$, we have $(q_1,\mathtt{skip},q_2) \in T$ (central node does an internal action);
    \item for all $\delta=(q_1,\tau,q_2) \in \Delta$ with $q_1 \in Q^b$ and all $q' \in Q$, we have $(q',q_1--,(q',\delta)),((q',\delta),q_2++,q') \in T$ (node at height one does an internal action);
    \item for all $\delta=(q_1,!!m,q_2) \in \Delta$ and $(q',?m,q'')\in \Delta$, we have $(q',q_1--,(q'',\delta)),((q'',\delta),q_2++,q'') \in T$ (node at height one broadcasts a message received by the root);
    \item for all $\delta=(q_1,!!m,q_2)\in \Delta$ and all $q' \in Q$ such that there deos not exist $q''$ in $Q$ verifying $(q',?m,q'')\in \Delta$, we have $(q',q_1--,(q',\delta)),((q',\delta),q_2++,q') \in T$ (node at height one broadcasts a message not received by the root).
      
    \end{itemize}

  \end{itemize}

  Intuitively, the control state of the VASS tracks the state of the root whereas the counters count how many processes are in states $Q^b$. One can verify that there exist a b-configuration $C$ and a star-configuration
  $C_f=(\Gamma_f,L_f)$  such that $\bprint{C} =(q,\Lambda)$ and  $L_f(\epsilon)=q_f$
  and  $C \trans_r^\ast C_f$ iff there exists $\nu' : X \mapsto \nat$ such that $(s_{in},\nu_{in}) \vartriangleright^\ast (q_f,\nu')$ where $\nu_{in}(q)=1$ for all $q \in \Lambda$ and $\nu_{in}(q)=0$ for all $q \in Q^b \setminus \Lambda$. We have hence shown that given $(q,\Lambda) \in  Q^b\times
  2^{Q^b}$, proving whether
  there exist a b-configuration $C=(\Gamma_f,L)$ and a star-configuration
  $C_f=(\Gamma_f,L_f)$  such that $\bprint{C} =(q,\Lambda)$ and  $L_f(\epsilon)=q_f$
  and  $C \trans_r^\ast C_f$ reduces to the control state reachability problem for VASS, $\VASSCover$, which is \textsc{ExpSpace}-complete.

\end{proofof}
\subsection{Proof of \cref{theorem:1pb-expspace}}

Thanks to Theorems \ref{thm:Cover-CoverTree-equivalent} and \ref{lemma:1-bounded-cover-star}, we have to prove that given 1-phase-bounded-protocol $\PP = (Q, \Sigma,
\qinit, \Delta)$ and $q_f \in Q$, deciding whether there exists  $\Gamma \in
\Stars$, $C= (\Gamma, L)\in \II$ and $D
=(\Gamma, L') \in \CC_\PP$ such that $C
\trans^\ast D$ and $L'(\epsilon) = q_f$ is an \textsc{ExpSpace}-complete problem.

We begin with the upper bound providing an \textsc{NExpSpace} algorithm. Let $\PP = (Q, \Sigma,
\qinit, \Delta)$ be a 1-phase-bounded protocol and $q_f \in Q$. We first guess a broadcast-print $(q,\Lambda) \in Q^b \times 2^{Q^b}$ and show we have  $(\qinit,\Lambda_{in}) \Rightarrow^\ast (q,\Lambda) $ with $\Lambda_{in} \in \set{\emptyset,\set{\qinit}}$. This boils down to a reachability query in the graph $(Q^b \times 2^{Q^b},\Rightarrow)$ which can be achieved in $\textsc{NPSpace}$ thanks to \cref{lemma:successor-bprint-ptime} and because the number of vertices in this graph is smaller than $|Q| *2^{|Q|}$. Thanks to Savitch's theorem, we can do this reachability query in \textsc{PSpace}. Then we look for a b-configuration $C'=(\Gamma_f,L)$ and a star-configuration
$C_f=(\Gamma_f,L_f)$  such that $\bprint{C} =(q,\Lambda)$ and  $L_f(\epsilon)=q_f$
and  $C \trans_r^\ast C_f$. Thanks to \cref{lemma:expspace-cover-bprint}, this can be done in \textsc{ExpSpace}.  The overall procedure gives rise to an \textsc{NExpSpace} algorithm and using again Savitch's theorem, we obtain an \textsc{ExpSpace}-algorithm.

Let us show that this algorithm is complete. Assume there exist $\Gamma \in
\Stars$, $C= (\Gamma, L)\in \II$ and $D
=(\Gamma, L') \in \CC_\PP$ such that $C
\trans^\ast D$ and $L'(\epsilon) = q_f$. Since the protocol is $1$-phase-bounded, there exists a b-configuration $C'$ such that $C \trans^\ast C' \trans^\ast_r D$. If $q_f \notin Q^b$, take for $C'=(\Gamma,L'')$ the last configuration in the execution $C \trans^\ast D$ such that $L''(\epsilon) \in Q^b$ and otherwise take $C'=D$. Thanks to Lemma \ref{lemma:completeness-bprint}, if $\bprint{C'}=(q,\Lambda)$, we have $(\qinit,\Lambda_{in}) \Rightarrow^\ast (q,\Lambda)$
with $\Lambda_{in} \in \set{\emptyset,\set{\qinit}}$.-
As $L'(\epsilon)=q_f$ and $C' \trans_r^\ast D$, we can conclude that our algorithm is complete.

We shall now show it is sound. Assume there exists $(q,\Lambda) \in  Q^b\times
2^{Q^b}$ such that $(\qinit,\Lambda_{in}) \Rightarrow^\ast (q,\Lambda)$
with $\Lambda_{in} \in \set{\emptyset,\set{\qinit}}$ and such that there exist a b-configuration $C'=(\Gamma_f,L)$ and a star-configuration
$C_f=(\Gamma_f,L_f)$  verifying $\bprint{C'} =(q,\Lambda)$ and  $L_f(\epsilon)=q_f$
and  $C' \trans_r^\ast C_f$. Thanks to \cref{lemma:completeness-bprint}, there exist
two $b$-configurations $C_{in} \in \II$ and $C =((V,E),L)$ in  $\CC$  such
that $C_{in}\trans^\ast C$ and $\bprint{C}=(q,\Lambda)$. We denote by $N=\max_{q' \in Q^b}(|\set{v \in V \mid L(v) = q'}|)$. Using \cref{lemma:monotonicity}.(iii), there exists an initial configuration $C'_{in}$ and a b-configuration $C''=(\Gamma'',L'')$ such that $C'_{in} \trans^\ast
C''$ and  $\bprint{C}=\bprint{C''}=(q,\Lambda)$ and
$|\set{v \in \Vert{\Gamma'}\setminus\set{\epsilon} \mid L'(v)=q'}| \geq N$ for all $q' \in \Lambda$. But we have then that $C' \preceq C''$. Thanks to  \cref{lemma:monotonicity}.(i) applied to each transition of the execution $C' \trans_r^\ast C_f$, we deduce that there exists $C'_f=(\Gamma'_f,L'_f)$ such that
$C_f \preceq C'_f$ and $C'' \trans^\ast C'_f$, Since $C_f \preceq C'_f$ and $L_f(\epsilon)=q_f$ we deduce that $L'_f(\epsilon)=q_f$. Since we have $C'_{in} \trans^\ast C'' \trans^\ast C'_f$, our algorithm is sound.

It remains to prove the lower bound. For this matter, we provide a reduction from $\VASSCover$,  the control state reachability problem for VASS. The intuition being that the root node keeps track of the states in $S$ and the other nodes  represent the value of the counters, the value of $X$ at a certain time being the number of processes in state $x_1$. Then the nodes encoding the counters will only perform broadcast saying whether they increment or decrement a counter and if at some point the root node receives an action on a counter, it is not suppose to do according to its control state, it will go in an error state $err$ from which it will not be able to reach the final state anymore. Let $V=(S,X,T)$ be a VASS, $(s_{in},\nu_{in})$ an initial configuration  and $s_f \in S$. Without loss of generality we assume that $\nu_{in}(x)=0$ for all $x \in X$. We build the following broadcast protocol $\PP = (Q, \Sigma, \qinit, \Delta)$ with:
\begin{itemize}
\item $Q=\set{\qinit,err} \cup S \cup \set{x_0,x_1 \mid x \in X}$;
\item $\Sigma=\set{x++,x-- \mid x \in X}$;
\item $\Delta$ is the smallest set satisfying the following conditions:
  \begin{itemize}
  \item $(\qinit,\tau,s_{in})$ belong to $\Delta$;
  \item for all $m \in \Sigma$, $(\qinit, ?m, err)$ belong to $\Delta$;
  \item for all $x \in X$, $(\qinit,\tau,x_0)$ belong to $\Delta$;
   \item $(x_0,!!x++,x_1)$ and $(x_1,!!x--,x_0)$ are in $\Delta$ for all $x \in X$;
   \item for all $(s,x++,s') \in T$, we have $(s,?x++,s') \in \Delta$;
   \item for all $(s,x--,s') \in T$, we have $(s,?x--,s') \in \Delta$;
   \item for all $(s,\mathtt{skip},s') \in T$, we have $(s,\tau,s'') \in \Delta$;
   \item for all $s \in S$ and all $x \in X$, if there does not exist $s' \in S$ such that $(s,x++,s') \in T$ then $(s,?x++,err) \in \Delta$;
    \item for all $s \in S$ and all $x \in X$, if there does not exist $s' \in S$ such that $(s,x--,s') \in T$ then $(s,?x--,err) \in \Delta$.
   \end{itemize}
 \end{itemize}
 Note that $\PP$ is 1-phase bounded (the transitions leaving a state in $\set{x_0,x_1 \mid x \in X}$ only perform broadcasts and the the transitions leaving a state in $\set{s\mid s \in S}$ only perform receptions or internal actions). We can then show that there exists  $\Gamma \in
\Stars$, $C= (\Gamma, L)\in \II$ and $D
=(\Gamma, L') \in \CC_\PP$ such that $C
\trans^\ast D$ and $L'(\epsilon) = s_f$ iff there exists $\nu' : X \mapsto \nat$ such that $(s_{in},\nu_{in}) \vartriangleright^\ast (s_f,\nu')$.

\section{Proofs of \cref{sec:about-2pb}}

\subsection{Proofs of \cref{subsec:proof-dec-2phases-covertree}}\label{appendix:dec-2pb-covertree}
\Ifshort{
We recall the definition of simple paths. 
A \emph{simple path between $u$ and $u'$} in a topology $\Gamma = (V,E)$ is a sequence of distinct vertices $v_0, \dots, v_k$ such that $u = v_0$, $u' = v_k$, and for all $0 \leq i <k$, $(v_i, v_{i+1}) \in E$. Its length is denoted $d(v_0,\dots, v_k)$ and is equal to $k$. In a tree topology $\Gamma'$, for two vertices $u, u'$, there exists a unique simple path between $u$ and $u'$, hence we denote $d(u, u')$ to denote the length of the unique path between $u$ and $u'$. Furthermore, for all vertex $u$, $d(\epsilon, u) = |u|$.

	\subsubsection{Proof of \cref{obs:decidability-2phase:bounded-path-to-v}.}
	Given $(\PP, q_f)$ a positive instance of $\CoverTree$, we let $f(\PP, q_f)$ the minimal number of processes needed to cover $q_f$ with a tree topology and we fix 
	$\Gamma = (V,E)$ a tree topology such that $|V| = f(\PP,q_f)$ and that covers $q_f$. Let $v \in \Vert{\Gamma}$ and $(C_i = (\Gamma, L_i))_{0\leq i\leq n}$ configurations such that $C_0 \in \II$, $C_0 \trans C_1\trans\dots \trans C_n$ and $L_n(v) = q_f$. We assume wlog that $v$ is the root of the tree, i.e. $v = \epsilon$.
	
	For all $u \in \Vert{\Gamma}$, we define $b(u)$ as the first index $0 \leq i<n$ from which $u$ takes a broadcast transition, and $\infty$ if it never broadcasts anything.
	%
	%
	%
	
}

We define $\Gamma[u]$ as the tree topology obtained from $\Gamma$ by removing $u$ and all $u' \in \Vert{\Gamma}$ which admits $u$ as a prefix. More formally, $\Gamma[u] = (V', E')$ with $V' = V \setminus \set{w \in V \mid w = u \cdot w', w' \in \nat^\ast}$ and $E' = E \cap (V'\times V') $.

\Ifshort{
	We establish the following lemma before proving \cref{obs:decidability-2phase:bounded-path-to-v}.
\begin{lemma}\label{obs:decidability-2phase:everybody-broadcasts}
	For all $u \in \Vert{\Gamma} \setminus \set{\epsilon}$, $b(u) \neq \infty$. Moreover, for all $u_1, u_2 \in \Vert{\Gamma} \setminus \set{\epsilon}$ such that $u_2 = u_1 \cdot x$ for some $x \in \nat$,
	it holds that $b(u_1) > b(u_2)$.
	
\end{lemma}
}
\Iflong{\begin{proofof}{\cref{obs:decidability-2phase:everybody-broadcasts}}}
\Ifshort{\begin{proof}}
We first prove the first part of the lemma.
	Assume there exists $u \in \Vert{\Gamma}$ such that $u\neq v$ and $u$ does not broadcast anything during the execution. Take $\Gamma' = \Gamma[u]= (V',E')$.	
	From  $C'_0 = (\Gamma', L'_0) \in \II$, we make the processes in $\Vert{\Gamma'}$ perform the same sequence of transitions as in the original execution.  The remaining parent of $u$ can perform the same sequence of transitions as it never receives anything from $u$.

Now we prove the second part of the lemma. 
	Assume for the sake of contradiction that there exist $u_1, u_2 \in \Vert{\Gamma} \setminus \set{\epsilon}$ such that $u_2 = u_1 \cdot x$ for some $x \in \nat$,
	and $b(u_1) < b(u_2)$.
	From the first part of the lemma, $b(u_1) \leq b(u_2) < n$. 
	We show now that $q_f$ is covered with $\Gamma' = \Gamma[u_2]$. For that, we will define by induction configurations $C'_0 = (\Gamma', L'_0)$, \dots, $C'_n=(\Gamma', L'_n)$ such that $C'_0\in \II$, $C'_1, \dots, C'_n \in \CC$ and $C'_0 \trans^\ast C'_1 \trans^\ast \cdots \trans^\ast C'_n$. We will prove that,
	for all $0\leq k\leq n$, 
	\begin{align*}
	P(k): &  \textrm{ for all $v' \in \Vert{\Gamma'}\setminus\set{u_1}$, $L'_k(v') = L_k(v')$, and if $L_k(u_1) \nin Q_2^r$, $L'_k(u_1) = L_k(u_1)$}.
	\end{align*}

	For $k=0$, $C'_0$ is uniquely defined, and obviously, for all $v'\in \Vert{\Gamma'}$, $L'_0(v')=L_0(v')$. Let $0\leq k <n$, and assume we have built $C'_0 ,\dots, C'_k \in \CC$ and that $P(k)$ holds. Let $u\in\Vert{\Gamma}$, $t\in \Delta$ such that $C_k\transup{u,t} C_{k+1}$. 
	
	\begin{itemize}
	\item If $t= (q,\tau,q')$, we define $L'_{k+1}$ as follows. 
	\begin{equation*}
	L'_{k+1}(u')=\begin{cases} q' & \textrm{ if $u'=u$ and $L'_k(u)=L_k(u)$}\\
						L'_k(u') & \textrm{ otherwise.}
	\end{cases}
	\end{equation*}
	\textbf{First case: $L'_k(u)=L_k(u)$}. Then $L'_k(u)=q$, and it follows immediately that $C'_k\transup{u,t} C'_{k+1}$. Also, let $v'\in\Vert{\Gamma'}\setminus\set{u_1}$. By induction hypothesis,
	$L'_k(v')=L_k(v')$ hence if $v'\neq u$, $L'_{k+1}(v')=L'_k(v')=L_k(v')=L_{k+1}(v')$ by definition of $C_k \transup{u,t}C_{k+1}$. Also, $L'_{k+1}(u)=L_{k+1}(u)$ 
	by definition
	of $C_k\transup{u,t} C_{k+1}$. If $L_{k+1}(u_1)\notin Q_2^r$ and $u\neq u_1$, then $L_k(u_1)=L_{k+1}(u_1)\notin Q_2^r$, and 
	$L'_{k+1}(u_1)= L'_k(u_1)=L_k(u_1)=L_{k+1}(u_1)$.  If $u=u_1$, then by construction $L'_{k+1}(u_1)=q'=L_{k+1}(u_1)$. Then $P(k+1)$ holds. 
	
		\textbf{Second case: $L'_k(u)\neq L_k(u)$}.  Then $C'_{k+1}=C'_k$ and $C'_k\trans^* C'_{k+1}$. Moreover, since $P(k)$ holds, it implies that $u=u_1$ and $L_k(u_1)\in Q_2^r$.
	Then $L_{k+1}(u_1)\in Q_2^r$ and $P(k+1)$ holds by induction hypothesis. 
	
	\item If $t=(q,!!m,q')$, we differentiate three cases.
	\begin{enumerate}
		\item if $u \in \Vert{\Gamma'}$, and $u\neq u_1$, we let $C'_{k+1}=(\Gamma', L'_{k+1})$ defined as follows: 
		
		\begin{align*}
		L'_{k+1}(u') &= L_{k+1}(u') \textrm{ for all $u'\in\Vert{\Gamma'}\setminus\set{u_1}$}\\
		L'_{k+1}(u_1) &=\begin{cases}
					L'_k(u_1) & \textrm{ if $u_1\nin\NeighG{\Gamma'}{u}$}\\
					L_{k+1}(u_1) & \textrm{ if $u_1\in\NeighG{\Gamma'}{u}$ and $L'_k(u_1)=L_k(u_1)$}\\ 		
					p & \textrm{ if $u_1\in\NeighG{\Gamma'}{u}$, $L'_k(u_1)\neq L_k(u_1)$ and there exists $(L'_k(u_1), ?m, p)\in \Delta$}\\
					L'_k(u_1) & \textrm{ otherwise}.			
					\end{cases}
		\end{align*}
		
		First, we show that $C'_k\transup{u,t} C'_{k+1}$. By induction hypothesis, $L_k(u')=L'_k(u')$ for all $u'\in\Vert{\Gamma'}\setminus\set{u_1}$. Then
		$L'_{k}(u)=L_k(u)=q$ and $L'_{k+1}(u)=q'$. For all $u'\neq u_1$, $L'_{k+1}(u')=L_{k+1}(u')$ and $L_k(u')=L'_k(u')$. 
		If $u_1\notin\NeighG{\Gamma'}{u}$, then $L'_{k+1}(u_1)=L'_k(u_1)$, and $C'_k\transup{u,t}C'_{k+1}$.  Otherwise, if $L'_k(u_1)=L_k(u_1)$,
		then $L'_{k+1}(u_1)=L_{k+1}(u_1)$ and $u_1$ has behaved correctly, hence $C'_k\transup{u,t}C'_{k+1}$. If $u_1\in\NeighG{\Gamma'}{u}$ and $L'_k(u_1)
		\neq L_k(u_1)$, then either there exists a transition $(L'_k(u_1), ?m, p)\in \Delta$ and $L'_{k+1}(u_1)=p$ or $L'_{k+1}(u_1)=L'_k(u_1)$. 
		In any case, $C'_k\transup{u,t} C'_{k+1}$. 
		
		Let $v'\in\Vert{\Gamma'}\setminus\set{u_1}$. By definition, $L'_{k+1}(v')=L_{k+1}(v')$. 
		If $L_{k+1}(u_1)\notin Q_2^r$ 
		then as $\PP$ is 2-phase bounded, $L_k(u_1) \nin Q_2^r$ and so by induction hypothesis $L'_k(u_1) = L_k(u_1)$.
		Then, by construction of $L'_{k+1}(u_1)$: either $u_1 \nin \NeighG{\Gamma'}{u}$ and so $L'_{k+1}(u_1) = L'_k(u_1) = L_k(u_1) = L_{k+1}(u_1)$, or $u_1 \in \NeighG{\Gamma'}{u}$ and $L'_{k+1}(u_1) =L_{k+1}(u_1)$.\lug{j'ai changé ici car je ne voyais pas pourquoi ce qui était commenté était vrai}
		
		\item If $u=u_1$, note that, as $(L_k(u_1), !!m, L_{k+1}(u_1)) \in \Delta$, by definition of 2-phase-bounded protocols, $L_k(u_1) \nin Q_2^r$. Hence, by induction hypothesis, $L'_{k}(u_1) = L_k(u_1) = q$. 
		We then simply let $L'_{k+1}(u') = L_{k+1}(u')$ for all $u'\in\Vert{\Gamma'}$. As $C_k \trans C_{k+1}$, and by the induction hypothesis,
		$C'_k \trans C'_{k+1}$. Moreover, by construction, for all $u' \in \Vert{\Gamma'}$, $L'_k(u') = L'_{k+1}(u')$ hence $P(k+1)$ holds. 
		
		\item If $u\nin \Vert{\Gamma'}$, we let $C'_{k+1} = C'_k$ and then $C'_{k}\trans^\ast C'_{k+1}$. If $u \neq u_2$, then $\NeighG{\Gamma}{u} \cap \Vert{\Gamma'} = \emptyset$. Indeed, $u=u_2\cdot w$ for $w\in \nat ^+$, hence $u[-1]=u_2\cdot w'\notin\Vert{\Gamma'}$ and for every $x \in \nat$, $u\cdot x \nin \Vert{\Gamma'}$. Hence, for $u'\in\Vert{\Gamma'}$, $L_{k+1}(u')=L_k(u')$. Then, $P(k+1)$ holds.

		If $u = u_2$, its only neighbor in $\Vert{\Gamma'}$ is $u_1$. Furthermore, as $b(u_1) < b(u_2) \leq k $, $u_1$ already broadcast some messages, and as $\PP$ is 2-phase-bounded, $L_k(u_1) \in Q_1^b\cup Q_2^b \cup Q_2^r$.
		Either $L_{k+1}(u_1)  = L_k(u_1)$, and since $L'_{k+1}(u_1) = L'_k(u_1)$, the induction property still holds, or $(L_k(u_1), ?m, L_{k+1}(u_1)) \in \Delta$. Then, as $\PP$ is 2-phase-bounded, $L_{k+1}(u_1) \in Q_2^r$, and the induction property holds too. 
	\end{enumerate}

	\end{itemize}
	
	Then $L'_n(\epsilon)=L_n(\epsilon)$ (since $u_1\neq\epsilon)$, and $L'_n(\epsilon)=q_f$. 
	Hence we found a tree topology $\Gamma' = (V' ,E')$ with $|V'| < |V|$ with which $q_f$ is coverable. This is not possible as $|V| = f(\PP, q_f)$, hence $b(u_1) > b(u_2)$.
	\Ifshort{\end{proof}}
\Iflong{\end{proofof}}

We are now ready to prove \cref{obs:decidability-2phase:bounded-path-to-v}.

\begin{proofof}{\cref{obs:decidability-2phase:bounded-path-to-v}}
	Assume there exists $u \in \Vert{\Gamma}$ such that $|u| > |Q|+1$. For all $1 \leq i \leq |u|$, we denote $v_i$ the prefix of $u$ of length $i$. By definition of a tree topology, for al $1 \leq i \leq |u|$, $v_i \in \Vert{\Gamma}$.
	Let $1 \leq i \leq n$, from \cref{obs:decidability-2phase:everybody-broadcasts}\ we know that $0\leq b(v_i)< n$. Hence, we can
	denote $q_i$ the state from which $v_i$ performs its first broadcast, i.e. $C_{b(v_i)} \transup{v_i, (q_i, !!m_i, q'_i)} C_{b(v_i) +1}$ for some $m_i \in \Sigma$ and $q'_i \in Q$. 
	
	As $|u| > |Q|+1$, there exists $1 <i_1 <i_2 \leq |u|$ such that $q_{i_1} = q_{i_2}$. Note that we take two such nodes in $v_2,\dots, v_{|u|}$ instead of in $v_1, \dots, v_{|u|}$. We will see later that we need $v_{i_1}$ to \emph{not have} $\epsilon$ as neighbor in order to apply \cref{obs:decidability-2phase:everybody-broadcasts}\ on $v_{i_1}$ and its neighbor $v_{i_1 - 1}$.
	
	From \cref{obs:decidability-2phase:everybody-broadcasts}, $b(v_{i_1}) > b(v_{i_2})$. Take $\Gamma' = (V', E')$ 	
	defined as
	the tree topology obtained from $\Gamma$ by only keeping processes $u$ such that (1) the unique path between path from $u$ to $v$ does not contain $v_{i_1}$, i.e. $v_{i_1}$ is not a prefix of $u$, or (2) the unique path between $u$ and $v$ contains $v_{i_2}$, i.e. $v_{i_2}$ is a prefix of $u$. We also add an edge between $v_{i_2}$ and the remaining parent of $v_{i_1}$. To stay consistent with the definition of our tree topology, we also rename nodes.
	More formally:
	$V' = V \setminus \set{u \in V \mid v_{i_1}\text{ is a prefix of } u} \cup \set{v_{i_1}\cdot w \mid w \in \nat^\ast \text{ and }v_{i_2} \cdot w \in V}$, and $E' = \set{(w, w\cdot x) \mid x \in \nat, w \in V', w\cdot x \in V'}$.
	Let $w \in V'$, we let 
	\begin{align*}
	\textsf{tr}: V'&\rightarrow V\\
	\textsf{tr}:w&\mapsto \begin{cases} w & \textrm{ if $v_{i_1}$ is not a prefix of $w$ in $V'$}\\
						v_{i_2}\cdot w' & \textrm{ if $w = v_{i_1}\cdot w'$ in $V'$ for some $w' \in \nat^\ast$}.
						\end{cases}
	\end{align*}
	
	Note that for all $u_1, u_2 \in \Vert{\Gamma'}$, such that $(u_1,u_2) \neq (v_{i_1 - 1}, v_{i_1})$, it holds that: $(u_1, u_2) \in \Edges{\Gamma'}$ if and only if $(\textsf{tr}(u_1), \textsf{tr}(u_2)) \in \Edges{\Gamma}$.
%
	
	We prove the existence of an execution covering $q_f$ with $\Gamma'$ in three steps.
	\paragraph*{First step.} We start by building $C'_0 = (\Gamma', L'_0) \in \II$, $C'_1, \dots, C'_{b(v_{i_2})} \in \CC$ such that for all $0 \leq i \leq b(v_{i_2})$, and $C'_i =(\Gamma', L'_i)$, for all $u \in \Vert{\Gamma'}$, $L'_i(u) = L_i(\textsf{tr}(u))$.
	We now prove by induction that 
	$C'_0 \trans^\ast C'_1 \trans^\ast \cdots \trans^\ast C'_{b(v_{i_2})}$.
	Let $k\geq 0$, and assume that we have proved that $C'_0 \trans^\ast C'_k$. 

	Assume that $C_k\transup{u,t} C_{k+1}$.
	
	\begin{itemize}
		\item if $t = (q, \tau, q')$ and there exists $v \in \Vert{\Gamma'}$ such that $\textsf{tr}(v) = u$, then by construction, $L'_k(v) = L_k(v) = q$ and $L'_{k+1}(v) = L_{k+1}(u) =q'$. For all other nodes $v' \in \Vert{\Gamma'}$, $L'_{k+1}(v') = L_{k+1}(\textsf{tr}(v')) = L_{k}(\textsf{tr}(v')) = L'_k(v')$ as $\textsf{tr}$ is injective. Hence, $C'_k\transup{v,t} C'_{k+1}$.
		\item if $t = (q, \tau, q')$ and there is no $v \in \Vert{\Gamma'}$ such that $\textsf{tr}(v) = u$, then for all $v \in \Vert{\Gamma'}$, $L'_{k+1}(v) = L_{k+1}(\textsf{tr}(v)) = L_{k}(\textsf{tr}(v')) = L'_k(v)$ and so $C'_{k+1} = C'_k$.
		\item if $t = (q, !!m, q')$ and there $v \in \Vert{\Gamma'}$ such that $\textsf{tr}(v) = u$, then by construction, $L'_k(v) = L_k(v) = q$ and $L'_{k+1}(v) = L_{k+1}(u) =q'$. 
		Furthermore, for all nodes $v' $ such that $(v', v) \in \Edges{\Gamma'}$, then either $(\textsf{tr}(v'), \textsf{tr}(v)) \in \Edges{\Gamma'}$ or $\set{v, v'} = \set{v_{i_1 -1}, v_{i_1}}$.
		In the first case, as $C_k \transup{\textsf{tr}(v), t} C_{k+1}$, either $(L'_k(v'), ?m, L'_{k+1}(v')) \in \Delta$ or $L'_k(v') = L'_{k+1}(v')$ and there is no $p \in Q$ such that $(L'_k(v'), ?m,p) \in \Delta$. In the latter case, recall that $k < b(v_{i_2})$ and, from  \cref{obs:decidability-2phase:everybody-broadcasts}, $b(v_{i_2}) < \cdots < b(v_{i_1}) <b(v_{i_1 - 1})$. Hence, $u \nin \set{v_{i_2}, v_{i_1 - 1}}$. Hence, $v \nin \set{v_{i_1 -1}, v_{i_1}}$, and we reach a contradiction. 
		
		All other nodes $u'$ not neighbors of $v$ are such that $L'_{k+1}(u') = L_{k+1}(\textsf{tr}(u'))$ and $(\textsf{tr}(u'), \textsf{tr}(v)) \nin \Edges{\Gamma'}$, as $v \nin v_{i_1}$. Hence, $L'_{k+1}(u') = L_{k+1}(\textsf{tr}(u')) = L_k(\textsf{tr}(u')) = L'_k(u')$.
		
		Hence $C'_k \transup{v,t} C'_{k+1}$.
		
		\item  if $t = (q, !!m, q')$ and there is no $v \in \Vert{\Gamma'}$ such that $\textsf{tr}(v) = u$, hence $u = v_{i_1}\cdot x$ for some $x \in \nat^\ast$ and $v_{i_2}$ is not a prefix of $u$. 
		We prove that for all $v \in \Vert{\Gamma}$, it holds that $L_{k+1}(\textsf{tr}(v)) = L_k(\textsf{tr}(v))$. Assume this is not the case, then $(\textsf{tr}(v), u) \in \Edges{\Gamma}$. 
		If $\textsf{tr}(v)$ is a parent of $u$ and so $\textsf{tr}(v) = v_{i_1 - 1}$ and $u = v_{i_1}$. With \cref{obs:decidability-2phase:everybody-broadcasts}\ it implies that $k< b(v_{i_2})< b(u)$, which is absurd as $C_k \transup{u, t} C_{k+1}$. 
		If $u$ is a parent of $\textsf{tr}(v)$, then $\textsf{tr}(v)= v_{i_2}$ and $u = v_{i_2 - 1}$. For the same reason, we get  $k< b(v_{i_2})< b(u)$, which is absurd as $C_k \transup{u, t} C_{k+1}$. Hence $(\textsf{tr}(v), u) \nin \Edges{\Gamma}$ and $L_{k+1}(\textsf{tr}(v)) = L_k(\textsf{tr}(v))$. Hence, $L'_{k+1}(v) = L_{k+1}(\textsf{tr}(v)) = L_k(\textsf{tr}(v)) = L'_k(v)$ for all $v \in \Vert{\Gamma'}$ and so $C'_k = C'_{k+1}$.
	\end{itemize}
Hence $C'_0 \trans^\ast C'_{b(v_{i_2})}$ where for all $u \in \Vert{\Gamma'}$, $L'_{b(v_{i_2})}(u) = L_{b(v_{i_2})}(\textsf{tr}(u))$.
	\paragraph*{Second step.} 
	Next, we build $C'_{b(v_{i_2})  +1} ,C'_{b(v_{i_2})  +2}, \dots, C'_{b(v_{i_1})} \in \CC$ such that $C'_{b(v_{i_2}) } \trans^\ast  \cdots \trans^\ast C'_{b(v_{i_1})}$. Moreover, we ensure that for all $b(v_{i_2}) \leq k \leq b(v_{i_1})$, if we let $C'_k =(\Gamma', L'_k)$, 
	\begin{align*}
	\textrm{for all }u \in \Vert{\Gamma'} \textrm{ such that }u = \textsf{tr}(u), L'_k(u) = L_k(u)\\
	\textrm{ and }L'_k(v_{i_1})=q_{i_2}. 
	\end{align*}
	
	Observe that, from the first step, $C'_{b(u_{i_2})}$ is such that for all $u\in\Vert{\Gamma'}$, $L'_{b(u_{i_2})}(u)=L_{b(u_{i_2})}(\textsf{tr}(u))$, and then
	$L'_{b(v_{i_2})}(v_{i_1})=L_{b(v_{i_2})}(\textsf{tr}(v_{i_1}))=L_{b(v_{i_2})}(v_{i_2})=q_{i_2}$. So the induction hypothesis holds for $k=b(v_{i_2})$. 
	
	Now, let $b(u_{i_2})\leq k < b(u_{i_1})$ and assume that we have built $C'_{b(u_{i_2})}$, \dots, $C'_k$. Let $v\in\Vert{\Gamma}$
	and $t=(q,\alpha,q')$ such that $C_{k}\transup{v,t} C_{k+1}$. If there is no $u\in\Vert{\Gamma'}$ such that $\textsf{tr}(u)=v$, or if the only $u\in\Vert{\Gamma'}$
	such that $\textsf{tr}(u)=v$ is such that $u\neq v$, then we let $C'_{k+1}=C'_k$. 
	\begin{itemize}
	\item If $\alpha=\tau$, we know that for all $u'\in\Vert{\Gamma}\setminus\set{v}$,
	$L_{k+1}(u')=L_k(u')$, hence, by induction hypothesis, if $u'=\textsf{tr}(u')$, we have that $L'_{k+1}(u')=L_{k+1}(u')$, and $L'_{k+1}(v_{i_1})=q_{i_2}$. 
	
	\item If $\alpha=!!m$ for some $m\in\Sigma$, then the neighbors of $v$ may have changed state. 
	Since $k<b(v_{i_1})$, we know that $u\neq v_{i_1}$. Hence, $u=v_{i_1}\cdot w$ for some $w\in\nat^+$ and then, for all $u'\in\Vert{\Gamma'}$, if
	$u'=\textsf{tr}(u')$, then $u'\notin \NeighG{\Gamma}{v}$. Thus, for all $u'\in\Vert{\Gamma'}$ such that $u'=\textsf{tr}(u')$, $L_{k+1}(u')=L_k(u')$ and
	then $L'_{k+1}(u')=L_{k+1}(u')$. Moreover, $L'_{k+1}(v_{i_1})=L'_k(v_{i_1})=q_{i_2}$. 
	\end{itemize}
	
	Assume now that $v=\textsf{tr}(v)$. By definition, $v_{i_1}$ is not a prefix of $v$ in $V'$. 
	\begin{itemize}
	\item If $\alpha = \tau$, we let $L'_{k+1}(v)=q'$ and, for all $u\in\Vert{\Gamma'}\setminus\set{v}$, $L'_{k+1}(u)=L_{k+1}(u)=L_k(u)$. By induction hypothesis, $L'_k(v)=L_k(v)$, hence $(L'_k(v), \tau, L'_{k+1}(v))\in\Delta$ and $C'_k\transup{v,t} C'_{k+1}$.
	Moreover, for all $u\in\Vert{\Gamma'}$ such that $u=\textsf{tr}(u)$, $L'_{k+1}(u)=L_{k+1}(u)$. Also, $v\neq v_{i_1}$ hence $L'_{k+1}(v_{i_1})=L'_k(v_{i_1})=q_{i_2}$. 
	\item If $\alpha = !!m$, let $L'_{k+1}(v)=q'$ and, for all $u\in\NeighG{\Gamma'}{v}$, if $\textsf{tr}(u)=u$, let $L'_{k+1}(u)=L_{k+1}(u)$. For all other $u\in\Vert{\Gamma'}$, let $L'_{k+1}(u)=L'_k(u)$. First, we show that $C'_k\transup{v,t} C'_{k+1}$. By induction hypothesis, $L'_k(v)=L_k(v)$, hence $(L'_k(v), !!m, L'_{k+1}(v))\in\Delta$. Let $u\in\NeighG{\Gamma'}{v}$ and assume that $\textsf{tr}(u)\neq u$. It means that $u=v_{i_1}\cdot w$ for some $w\in\nat^*$. Since 
	$v=\textsf{tr}(v)$, we also know that $v_{i_1}$ is not a prefix of $v$.
	If $v=u\cdot x$ then $v=v_{i_1}\cdot w\cdot x$, which is impossible since $u_{i_1}$ is not a prefix of $v$. If $u=v\cdot x$, then either $v_{i_1}$ is also
	a prefix of $v$ which is impossible, or $u=v_{i_1}$ and then $v_{i_1}=v\cdot x$. This means that $v=v_{i_1-1}$, the father of $v_{i_1}$. By~\cref{obs:decidability-2phase:everybody-broadcasts}, it would mean that $b(v) = b(v_{i_1})< b(v)$ but $k<b(v_{i_1})$ and $v$ broadcasts a message from $C_k$. 
	This leads to a contradiction. Hence, for all $u\in\NeighG{\Gamma'}{v}$, $\textsf{tr}(u)=u$, and then, for all $u\in\NeighG{\Gamma'}{v}$, $L'_{k+1}(u)=L_{k+1}(u)$. By induction hypothesis, $L_k(u)=L'_k(u)$, then either $L_k(u)=L_{k+1}(u)$ and there is no $(L'_k(u), ?m, p)\in\Delta$, or $(L'_k(u), ?m, L'_{k+1}(u))$. Hence
	$C'_k\transup{v,t} C'_{k+1}$. Moreover, by induction hypothesis, for all $u\in\Vert{\Gamma'}$ such that $\textsf{tr}(u)=u$, $L'_k(u)=L_k(u)$. Then, if $u\in\NeighG{\Gamma'}{v}$, $L'_{k+1}(u)=L_{k+1}(u)$, and otherwise, $L'_{k+1}(u)=L'_k(u)=L_k(u)$. Let again $u$ such that $u = \textsf{tr}(u)$ and $(u,v)\notin \Edges{\Gamma'}$, then $(\textsf{tr}(u), \textsf{tr}(v))\notin\Edges{\Gamma}$, hence $(u,v)\notin\Edges{\Gamma}$. Then $u\notin\NeighG{\Gamma}{v}$ and $L_{k+1}(u)=L_k(u)$. Hence, for all $u\in\Vert{\Gamma'}$
	such that $u=\textsf{tr}(u)$, $L'_{k+1}(u)=L_{k+1}(u)$. Since $v_{i_1}\neq\textsf{tr}(v_{i_1})$, and since $v\neq v_{i_1}$, $L'_{k+1}(v_{i_1})=L'_k(v_{i_1})=q_{i_2}$ by induction hypothesis. 
		\end{itemize}

		In this part of the execution, we have forgotten the transitions issued by nodes $u=v_{i_1}\cdot w$ for $w\in\nat^+$ in $\Vert{\Gamma}$. The configuration
		we reach is such that $L'_{b(v_{i_1})}(v_{i_1})=q_{i_2}$ and correspond to the states reached in $C_{b(v_{i_1})}$ in the rest of the tree. This will allow us
		to reach $q_f$. The subtree rooted in $v_{i_1}$ might visit different states than in the original execution, but this will not influence the states visited
		by the root, thanks to the fact that the protocol is 2-phase bounded. 

	\paragraph*{Third step.}
	

	Finally, we build $C'_{b(v_{i_1})  } ,C'_{b(v_{i_1})  +1}, \dots, C'_{n} \in \CC$ such that $C'_{b(v_{i_1}) } \trans^\ast C'_{b(v_{i_1})  +1} \trans^\ast \cdots \trans^\ast C'_{n}$. Moreover, we ensure that for all $b(v_{i_1}) \leq k \leq n$, if we let $C'_k =(\Gamma', L'_k)$,  
	
	\begin{align*}
	\textrm{for all }u \in \Vert{\Gamma'}\textrm{ such that }u = \textsf{tr}(u), L'_k(u) = L_k(u) \\
	\textrm{ and } L'_k(v_{i_1}) = L_k(v_{i_1})\textrm{ or }L_k(v_{i_1}) \in Q_2^r.
	\end{align*}
	
	Processes $u \in \Vert{\Gamma'}$ such that $u = \textsf{tr}(u)$, perform the same sequence of transitions than between $C_{b(v_{i_1}) } $ and $C_{n}$ and process $v_{i_1}$ broadcasts the sequence of messages broadcast by the same node in $\Gamma$. 
	Formally we build the sequence by induction on $b(v_{i_1}) \leq k \leq n$. $C'_{b(v_{i_1})}$ is already defined and satisfies the induction property.
	Assume we have built $C'_{b(v_{i_1})  } , \dots ,C'_{k}$ for some $b(v_{i_2}) \leq k <n $ such that:  
	for all $u \in \Vert{\Gamma'}$ with $u = \textsf{tr}(u)$, it holds that $L'_k(u) = L_k(u)$,
	and, furthermore either $L'_k(v_{i_1})=L_k(v_{i_1})$ or $L_k(v_{i_1}) \in Q_2^r$.
	
	Denote $C_k \transup{u, t} C_{k+1}$ with $t = (q, \alpha, q')$ and $\alpha \in \set{\tau} \cup !! \Sigma$.
	\begin{itemize}
	\item First, if $\alpha = \tau$ and $u = \textsf{tr}(u)$ or $u = v_{i_1}$, then define $L'_{k+1}$ as follows: if $L'_k(u) = L_k(u)$, define $L'_{k+1}(u) = q'$ and $L'_{k+1}(u') = L'_{k}(u')$ for all other nodes, otherwise $L'_{k+1} = L'_k$. 
	\begin{itemize}
	\item If $L'_k(u)=L_k(u)$, $L_k(u) = q$ and so we immediately get that $C'_k \trans C'_{k+1}$. Furthermore, for all $u' \in \Vert{\Gamma'} \setminus \set{u}$, it holds that $L'_{k+1}(u') = L'_k(u')$. If $u' \neq v_{i_1}$ and $u'=\textsf{tr}(u')$, by induction hypothesis, $L'_{k+1}(u') = L'_k(u') = L_k(u') = L_{k+1}(u')$.
	Furthermore, either $u=v_{i_1}$ and $L'_{k+1}(v_{i_1})=L_{k+1}(v_{i_1})$, or $u \neq v_{i_1}$ and as $L_{k+1}(v_{i_1}) = L_k(v_{i_1})$, it holds that $L'_{k+1}(v_{i_1}) = L_{k+1}(v_{i_1})$ or $L_{k+1}(v_{i_1}) = L_k(v_{i_1}) \in Q_2^r$.
	\item If $L'_k(u)\neq L_k(u)$, by induction hypothesis $u = v_{i_1}$
	and $L_k(v_{i_1}) \in Q_2^r$. As $\PP$ is 2-phase-bounded, for all $p$ such that $(L_k(v_{i_1}), \tau, p)\in \Delta$, it holds that $p \in Q_2^r$. Hence, $L_{k+1}(v_{i_1}) \in Q_2^r$ and the induction property holds.
		\end{itemize}
		Note that if  $u \neq \textsf{tr}(u)$ and $u \neq  v_{i_1}$, for all $u'$ such that $u' = \textsf{tr}(u')$ or $u '= v_{i_1}$, it holds that $L_{k+1}(u') = L_k(u')$, hence by taking  $C'_{k+1} = C'_k$, the induction property holds.

	\item In the following, we assume that $\alpha = !!m$ for some $m \in \Sigma$.
	We start by defining $C'_{k+1}$, then we prove that $C'_k \trans^\ast C'_{k+1}$ and finally we prove that $C'_{k+1}$ respects the induction property. We build $C'_{k+1}$ depending on $u$. 
	\begin{enumerate}
		\item if $\textsf{tr}(u) = u$ and $u\neq v_{i_1 - 1}$, then let $L'_{k+1}(u')=L_{k+1}(u')$ for all $u'$ such that $\textsf{tr}(u') = u'$, and $L'_{k+1}(u') = L'_{k}(u')$ for all other nodes.
		
		\item if $u = v_{i_1- 1}$, then let $L'_{k+1}(u')=L_{k+1}(u')$ for all $u'$ such that $\textsf{tr}(u') = u'$, and $L'_{k+1}(u') = L'_{k}(u')$ for all other nodes $u' \neq v_{i_1}$. It is left to define $L'_{k+1}(v_{i_1})$. If $L'_{k}(v_{i_1}) = L_k(v_{i_1})$, then define $L'_{k+1}(v_{i_1})=L_{k+1}(v_{i_1})$, otherwise if there exists $(L'_k(v_{i_1}), ?m, p) \in \Delta$, define $L'_{k+1}(v_{i_1}) = p$, and else define $L'_{k+1}(v_{i_1}) = L'_k(v_{i_1})$.
		
		\item if $u = v_{i_1}$, then let $L'_{k+1}(v_{i_1}) = L_{k+1}(v_{i_1}) = q'$, and let $L'_{k+1}(u')=L_{k+1}(u')$ for all $u'$ such that $\textsf{tr}(u') = u'$. Furthermore, for all other nodes $u'$, if $(v_{i_1}, u')\in \Edges{\Gamma'}$, then if there exists $(L'_k(u'), ?m, p) \in \Delta$, let $L'_{k+1}(u')=p$, otherwise $L'_{k+1}(u')=L'_{k}(u')$. If $(v_{i_1}, u')\nin \Edges{\Gamma'}$, then $L'_{k+1}(u')=L'_{k}(u')$.
		
		\item otherwise, $C'_{k+1} = C'_k$.
		
	\end{enumerate}
	We prove now that in all cases 1, 2 and 3, it holds that $C'_k \trans C'_{k+1}$.

	In case 1, observe that as $u \neq v_{i_1-1}$ and $\textsf{tr}(u) = u$, then all neighbors $u'\in\NeighG{\Gamma'}{u}$ are such that $\textsf{tr}(u') = u'$. Indeed, if $v_{i_1}$ is not a prefix of $u$, it is not a prefix of its parent (if any), and for $x \in \nat$, $v_{i_1}$ is a prefix of $u\cdot x$ if and only if $u\cdot x = v_{i_1}$, i.e. $u = v_{i_1 -1}$. Hence, for all $(u',u) \in \Edges{\Gamma'}$, it holds that $(\textsf{tr}(u'), \textsf{tr}(u')) = (u, u')\in \Edges{\Gamma}$, and by induction hypothesis $L'_k(u') = L_k(u')$. Therefore, as $C_k \trans C_{k+1}$, $(L'_k(u'), ?m, L'_{k+1}(u')) \in \Delta$ or there is no transition labeled by $?m$ from $L'_k(u')$ and $L'_{k+1}(u') = L'_k(u')$.
	Furthermore, $(L'_k(u), !!m, L'_{k+1}(u)) \in \Delta$. Let $u'\in\Vert{\Gamma'}$ which is not a neighbor of $u$, and such that $\textsf{tr}(u') = u'$. Then, $u'\notin\NeighG{\Gamma}{u}$, and by induction hypothesis $L'_k(u') = L_k(u')$. As $C_k \trans C_{k+1}$, $L_k(u') = L_{k+1}(u')$. And so by definition of $L'_{k+1}$, $L'_{k+1}(u') = L_{k+1}(u') = L'_k(u')$. Let $u'\in\Vert{\Gamma'}$ which is not a neighbor of $u$ and such that $\textsf{tr}(u')\neq u'$. Then, by definition of $L'_{k+1}$, $L'_{k+1}(u') =  L'_k(u')$.
	Hence $C'_k \trans C'_{k+1}$.
	
	In case 2, the only difference is that $v_{i_1}$ is now a neighbor of the broadcasting node $u = v_{i_1 - 1}$. If $L'_{k}(v_{i_1}) = L_k(v_{i_1})$, then we defined $L'_{k+1}(v_{i_1}) = L_{k+1}(v_{i_1})$. As $(v_{i_1-1}, v_{i_1}) \in \Edges{\Gamma}$ and $C_k \trans C_{k+1}$, then $(L'_k(v_{i_1}), ?m, L'_{k+1}(v_{i_1})) \in \Delta$ or there is no transition labeled by $?m$ from $L'_k(v_{i_1})$ and $L'_{k+1}(v_{i_1}) = L'_k(v_{i_1})$. Hence, $C'_k \trans C'_{k+1}$.
	Otherwise, $L'_k(v_{i_1}) \neq L_k(v_{i_1})$, and by definition of $L'_{k+1}$, either $(L'_{k}(v_{i_1}), ?m, L'_{k+1}(v_{i_1})) \in \Delta$, or $L'_{k+1}(v_{i_1}) = L'_k(v_{i_1})$ and there is no transitions $(L'_{k}(v_{i_1}), ?m, p) \in \Delta$ for some $p \in Q$. Hence, $C'_k \trans C'_{k+1}$.
	
%
%
%
	
	In case 3, note that by definition of a 2-phase-bounded protocol, $L_k(v_{i_1}) \nin Q_2^r$ as $(L_k(v_{i_1}), !!m,q') \in \Delta$. Hence,	by induction hypothesis, $L'_k(v_{i_1}) = L_k(v_{i_1})$ and so $v_{i_1}$ can perform the broadcast of $m$ from $C'_k$.
	Furthermore, by induction hypothesis, its parent $v_{i_1 -1}$ is such that $L'_k(v_{i_1 - 1}) = L_k(v_{i_1-1})$ and, by definition, $L'_{k+1}(v_{i_1 - 1}) = L_{k+1}(v_{i_1-1})$.  As $(v_{i_1-1}, v_{i_1}) \in \Edges{\Gamma}$, and $C_k \trans C_{k+1}$, either $(L_k(v_{i_1-1}), ?m,  L_{k+1}(v_{i_1-1})) \in \Delta$ or there is no transitions $(L_k(v_{i_1-1}), ?m ,p) \in \Delta$ for some $p \in Q$. 
	For other neighbors $u'$ of $v_{i_1}$, $u'\neq \textsf{tr}(u')$ and by construction of $L'_{k+1}$, $(L'_k(u'), ?m, L'_{k+1}(u')) \in \Delta$, or there is no reception transitions labeled by $?m$ from $L'_k(u')$ and $L'_{k+1}(u')=L'_{k}(u')$. Other nodes $u'$ which are not neighbors of $v_{i_1}$ are either such that $u' = \textsf{tr}(u')$, and it holds that $L'_{k+1}(u') = L_{k+1}(u') = L_k(u') = L'_k(u')$, as $(u', v_{i_1}) \nin \Edges{\Gamma}$. Or $u' \neq \textsf{tr}(u')$ and by construction, $L'_{k+1}(u') = L'_k(u')$.
	Hence it holds that $C'_k \trans C'_{k+1}$.\\

	We prove now that in all cases 1, 2, 3 and 4, $C'_{k+1}$ satisfies the induction property, i.e. for all $u \in \Vert{\Gamma'}$ such that $\textsf{tr}(u) = u$, it holds that $L'_{k+1}(u) = L_{k+1}(u)$, and $L'_{k+1}(v_{i_1}) = L_{k+1}(v_{i_1})$ or $ L_{k+1}(v_{i_1}) \in Q_2^r$. In case 1, the first part is trivial by construction. As $(u, v_{i_1}) \nin \Edges{\Gamma}$, $L_{k+1}(v_{i_1}) = L_k(v_{i_1})$, hence either $L'_{k+1}(v_{i_1}) = L'_k(v_{i_1}) = L_k(v_{i_1}) = L_{k+1}(v_{i_1})$ by induction hypothesis, or $ L_{k+1}(v_{i_1}) = L_k(v_{i_1})\in Q_2^r$.
	
	In case 2, either $L'_{k+1}(v_{i_1}) = L_{k+1}(v_{i_1})$ or $L'_{k}(v_{i_1}) \neq L_{k}(v_{i_1})$. In the latter case, by induction hypothesis, $L_k(v_{i_1}) \in Q_2^r$, and so, as $\PP$ is 2-phase-bounded, $L_{k+1}(v_{i_1}) \in Q_2^r$.
	Case 3 is direct from the definition of $L'_{k+1}$.
	
	In case 4, for all $u' \in \Vert{\Gamma'}$ such that $u' = \textsf{tr}(u')$, $L_{k+1}(u') = L_k(u')$ as $(u' ,u) \nin \Edges{\Gamma}$. Hence, $L'_{k+1}(u') = L'_k(u') = L_k(u') = L_{k+1}(u')$. However, if $u = v_{i_1+1}$, $v_{i_1}$ might change its state between $C_k$ and $C_{k+1}$ upon reception of $m$. If it does so, $(L_k(v_{i_1} ), ?m, L_{k+1}(v_{i_1})) \in \Delta$. As $k > b(v_{i_1})$ and $\PP$ is 2-phase-bounded, $v_{i_1}$ already went through its broadcasting phase, 
	and $L_{k+1}(v_{i_1}) \in Q_2^r$, and so the induction property holds.
\end{itemize}	
	
	\paragraph*{Conclusion.} Hence we found an execution covering $q_f$ with a tree topology $\Gamma' = (V', E')$ and $|V'| < |V|$. It contradicts the fact that $|V| = f(\PP,q_f)$.
	
\end{proofof}

\Ifshort{
\subsubsection{Proof of \cref{thm:Cover-decidable-2pb}}
\begin{proofof}{\cref{thm:Cover-decidable-2pb}}
	We prove the decidability by reducing both \Cover\ and \CoverTree\ to $\textsc{Cover}[K-\textsf{BP}]$. Let $\PP$ a protocol and $q_f$ a state of $\PP$.
	Assume that there is a topology $\Gamma$ with which 
	$q_f$ is coverable. Then, by~\cref{{thm:Cover-CoverTree-equivalent}}, there exist a tree topology $\Gamma'$ with wich $q_f$ is coverable. Assume that 
	$|V|= f(P,q_f)$. 
	%
	Then, each simple path in $\Gamma'$ has a length bounded by $2(|Q| +1)$.
	Indeed, let $u_1, u_2 \in \Vert{\Gamma'} $, as $\Gamma'$ is a tree topology, the unique simple path between $u_1$ and $u_2$ either contains $\epsilon$ and $d(u_1, u_2) = |u_1| + |u_2|$, or there exists $u \in \nat^+$ such that $u_1 = u\cdot u_1'$ and $u_2 = u \cdot u'_2$, and $u_1'$ and $u_2'$ have no common prefix. In the latter case, the length of the simple path between $u_1$ and $u_2$ is $|u_1'| + |u_2'| < |u_1|+ |u_2|$. In both cases, $d(u_1, u_2) \leq |u_1| + |u_2|$, and hence from  \cref{obs:decidability-2phase:bounded-path-to-v}, $d(u_1, u_2) \leq 2(|Q| +1)$. Then $q_f$ is coverable in a $2(|Q|+1)$-bounded path topology. 
	
	Conversely, assume that there exists a $2(|Q|+1)$-bounded path topology $\Gamma$ with which $q_f$ is coverable. Then, immediately, there is a topology  with which 
	$q_f$ is coverable. By~\cref{thm:Cover-CoverTree-equivalent}, there also exists a tree topology with which $q_f$ is coverable. \cref{th:concur-2010}
	allows to conclude to decidability of both \Cover\ and \CoverTree. 
\end{proofof}
}
\subsection{Proofs of \cref{subsec:coverLine-inP}}\label{appendix:proof-inP}
\subsubsection{Proof of \cref{lemma:CoverLine-2pb-inP:execution-shape-2}}\label{appendix:subsec:CoverLine-inP:lemma:shape}
Let $q_f \in Q$ be a coverable state.
We fix $\Gamma$ with $\Vert{\Gamma} = \set{v_1, \dots, v_\ell}$ the line topology such that there exist $\rho = C_0 \trans \cdots \trans C_n$ with $C_n = (\Gamma, L_n)$ and $L_n(v_N) = q_f$ for some $v_N \in \Vert{\Gamma}$. We suppose wlog that $N \geq 3$ and $N \leq \ell- 2$, otherwise, we can just add artificial nodes not issuing any transition in the execution.

We start by proving the following lemma.
\begin{lemma}\label{lemma:CoverLine-2pb-inP:execution-shape-1}
There exists $\rho' = C'_0 \trans^\ast C'_n$ such that $C'_n = (\Gamma, L'_n)$ with $L'_n(v_N)= q_f$, and
	\begin{itemize}
		\item for all $1 \leq i \leq N-2$, $\lastBroadcast{v_i}{\rho'} \leq \firstBroadcast{v_{i+1}}{\rho'}$;
		\item for all $N+2 \leq i \leq \ell$, $\lastBroadcast{v_i}{\rho'} \leq \firstBroadcast{v_{i-1}}{\rho'}$.
	\end{itemize} 
\end{lemma}

\paragraph{Proof of \cref{lemma:CoverLine-2pb-inP:execution-shape-1}.}

To prove \cref{lemma:CoverLine-2pb-inP:execution-shape-1}, we give two symmetric lemmas.
\begin{lemma}\label{lemma:CoverLine-2pb-inP:execution-shape-right}
	Let $\rho = C_0\trans^\ast C$ with $C = (\Gamma, L)$ and $1 \leq K < \ell$. 
	There exists $\rho' = C_0  \trans^\ast C'$ such that $C' = (\Gamma, L')$, and 
	\begin{itemize}
	\item for all $v \in \set{v_{K+1}, \dots, v_\ell}$, $L'(v) = L(v)$;
	\item for all $1 \leq i \leq K-1$, $\lastBroadcast{v_i}{\rho'} \leq \firstBroadcast{v_{i+1}}{\rho'}$;
	\item for all $K \leq i \leq \ell$, $\lastBroadcast{v_i}{\rho'} \leq \firstBroadcast{v_{i+1}}{\rho'}$ if and only if $\lastBroadcast{v_i}{\rho} \leq \firstBroadcast{v_{i+1}}{\rho}$.
	\end{itemize}
\end{lemma}
\begin{proof}
We show that we can build $\rho_K$ by induction on $K$:
	for $K = 1$, there is nothing to do and $\rho_1 = \rho$. Let $1 < K < \ell$, and assume we have built $\rho_{K-1} = C_0 \transup{v^0, t^0} C_1 \transup{v^1, t^1} \cdots \transup{v^{n-1}, t^{n-1}} C_n$ such that $C_n = (\Gamma, L_n)$ and:
	\begin{itemize}
		\item for all $v \in \set{v_{K}, \dots, v_\ell}$, $L_n(v) = L(v)$;
		\item for all $1 \leq i \leq K-2$, $\lastBroadcast{v_i}{\rho_{K-1}} \leq \firstBroadcast{v_{i+1}}{\rho_{K-1}}$;
		\item for all $K \leq i \leq \ell$, $\lastBroadcast{v_i}{\rho_{K-1}} \leq \firstBroadcast{v_{i+1}}{\rho_{K-1}}$ if and only if $\lastBroadcast{v_i}{\rho} \leq \firstBroadcast{v_{i+1}}{\rho}$.
	\end{itemize}
	If $\lastBroadcast{v_{K-1}}{\rho_{K-1}} \leq \firstBroadcast{v_K}{\rho_{K-1}}$, there is nothing to do.
	
	Assume $\lastBroadcast{v_{K-1}}{\rho_{K-1}} > \firstBroadcast{v_K}{\rho_{K-1}}$, then $\lastBroadcast{v_{K-1}}{\rho_{K-1}} \geq 0$ and $v_{K-1}$ issues at least one transition in $\rho_{K-1}$ after $\firstBroadcast{v_{K}}{\rho_{K-1}}$. Let
	\begin{equation*}
	j_1 = \begin{cases} \min\set{j > \firstBroadcast{v_K}{\rho_{K-1}}\mid C_j \transup{v_{K-1}, t} C_{j+1} \text{ for some } t \in \Delta} & \textrm{ if }
	\firstBroadcast{v_{K-1}}{\rho_{K-1}}\leq\firstBroadcast{v_K}{\rho_{K-1}}\\
	\min{\set{j \geq 0\mid C_j \transup{v_i, t} C_{j+1} \text{ for some } t \in \Delta\textrm{ and for some }1\leq i\leq K-1}} & \textrm{otherwise.}
	\end{cases}
	\end{equation*}
	
	We now build the following sequence of configurations $C'_{j_1}, \dots C'_n$:
	\begin{itemize}
		\item $C'_{j_1} = C'_{j_1 +1 } =C_{j_1}$;
		\item for $j_1 + 1 \leq j < n$, $C'_{j+1} = (\Gamma, L'_{j+1})$ and 
		\begin{itemize}
			\item $L'_{j+1}(v) = L_{j+1}(v)$ for all $v \in \set{v_{K+1}, \dots, v_\ell}$, 
			\item $L'_{j+1}(v) = L'_{j}(v)$ for all $v \in \set{v_1, \dots, v_{K-2}}$,
			\item $L'_{j+1}(v_{K-1})$ is defined as follows: if $v^j = v_K$ and $t^j $ is a broadcast transition of some message $m \in \Sigma$, and $(L'_j(v_{K-1}) ,?m, p) \in \Delta$ for some $p \in Q$, then $L'_{j+1}(v_{K-1}) = p$, otherwise $L'_{j+1}(v_{K-1}) = L'_j(v_{K-1})$;
			\item $L'_{j+1}(v_{K})$ is defined as follows: 
			\begin{itemize}
				\item if $v^j = v_{K+1}$ and $L'_j(v_K) \neq L_j(v_K)$ and $t^j $ is a broadcast transition of some message $m \in \Sigma$, and $(L'_j(v_{K}), ?m, p) \in \Delta$ for some $p \in Q$, then $L'_{j+1}(v_{K}) = p$;
				\item if $v^j = v_{K+1}$ and $L'_j(v_K) = L_j(v_K)$, then $L'_{j+1}(v_{K}) = L_{j+1}(v_K)$;
				\item if $v^j = v_K$ and $L'_j(v_K) = L_j(v_K)$, then $L'_{j+1}(v_K) = L_{j+1}(v_K)$;
				\item otherwise, $L'_{j+1}(v_K) = L'_j(v_K)$.
			\end{itemize}
		\end{itemize}
	\end{itemize}
	
	At $j_1$, actions of process $v_{K-1}$ (and hence actions of all processes in $\set{v_1, \dots, v_{K-1}}$ will have no influence any more on $\set{v_{K+1},
	\dots, v_\ell}$, either
	because $v_K$ has emitted its first broadcast, and thus will only receive messages from $\set{v_1, \dots, v_{K-1}}$ in its last reception phase (then
	will never feed it to $v_{K+1}$), or because anyway the first broadcast of $v_{K-1}$ happens to late (after the first broadcast of $v_K$). 	
	In order to prove the induction step, we have just built a sequence of configurations that ignore subsequent actions of $\set{v_1,\dots, v_{K-1}}$, starting from
	$j_1$, while
	still being an execution. The state of the protocol in which $v_K$ is during this part of the execution is either faithful to what it was in $\rho_{K-1}$, or 
	$v_{K}$ is in its last reception phase. 
	
	Formally, we now prove that, for all $j_1 +1 \leq j \leq n$,
	\begin{enumerate}
	\item  $P_1(j)$: either
	$C'_{j-1} \transup{v^{j-1}, t^{j-1}} C'_{j}$ and $v^{j-1}\in \set{v_K, \dots, v_\ell}$ or $C'_{j} = C'_{j-1}$ and $v^{j-1} \in \set{v_1, \dots v_{K}}$,\label{item:1}
	\item $P_2(j)$: $L'_j(v_K) = L_j(v_K)$ or $L_j(v_K) \in Q_2^r$.\label{item:2}
	\end{enumerate}
	
	We prove $P_1(j)$ and $P_2(j)$ by induction on $j$.
	For $j = j_1+1$, $P_1(j)$ trivially holds as $C'_{j_1 +1 }= C'_{j_1}$ and $C_{j_1}\transup{v_i, t} C_{j_1 +1}$ for $1\leq i\leq K-1$, by
	construction. Let us prove now that $P_2(j_1 + 1)$ holds: recall first that $L'_{j_1 +1}(v_K) = L'_{j_1}(v_K) = L_{j_1}(v_K)$. 
	Since $C_{j_1}\transup{v_i, t} C_{j_1 +1}$, for some $1\leq i\leq K-1$, either $L_{j_1 + 1}(v_K) = L_{j_1}(v_K)$ and $P_2(j)$ holds, or $(L_{j_1}(v_K), ?m, L_{j_1+1}(v_K)) \in \Delta$ and $t=(p,!!m,q)$.  In that case $C_{j_1}\transup{v_{K-1}, t} C_{j_1+1}$. Since either $j_1 > \firstBroadcast{v_{K}}{\rho_{K-1}}$ if 
	$\firstBroadcast{v_{K-1}}{\rho_{K-1}}\leq \firstBroadcast{v_K}{\rho_{K-1}}$, or $\firstBroadcast{v_{K-1}}{\rho_{K-1}}> \firstBroadcast{v_K}{\rho_{K-1}}$ and $C_{j_1} \transup{v_{K-1}, t} C_{j_1 +1}$ with $t = (p, !!m, q)$,
	then in any case $j_1> \firstBroadcast{v_K}{\rho_{K-1}}$. Since $\PP$ is 2-phase-bounded, $L_{j_1}(v_K) \in Q_1^b\cup Q_2^b \cup Q_2^r$ and so $L_{j_1+1}(v_K) \in Q_2^r$. Hence $P_2(j)$ holds.
	
	Let now $j_1 +1\leq j < n$ such that $P_1(j)$ and $P_2(j)$ hold. 
	We make the following observation.
\begin{observation} \label{obs:conf-suc}
Let $C_1 \transup{v,t} C_1'$ with $C_1 = (\Gamma, L_1)$ and $C'_1 = (\Gamma, L'_1)$. Consider $C_2 = (\Gamma, L_2)$, 
$C_2' = (\Gamma, L'_2)$  two configurations such that, for all $u \in \set{v}\cup \Neigh{v}$, $L_2(u) = L_1(u)$ and $L_2'(u) = L_1'(u)$, and for all $u\nin \set{v}\cup \Neigh{v}$, $L'_2(u) = L_2(u)$. Then $C_2 \transup{v,t} C_2'$.
\end{observation}

	The induction step is proved by a case analysis over $v^j \in \Vert{\Gamma}$. (Recall that $C_j \transup{v^j ,t^j} C_{j+1}$.)
	\begin{itemize}
		\item Assume $v^j \in \set{v_{K+2}, \dots, v_\ell}$. Then $\set{v^j} \cup \Neigh{v^j}\subseteq \set{v_{K+1}, \dots, v_\ell}$, and by construction of $C'_j$ and
		$C'_{j+1}$, for all $v \in \set{v^j} \cup \Neigh{v^j}$, $L'_j(v) = L_j(v)$ and $L'_{j+1}(v) = L_{j+1}(v)$. Then, nodes of $\set{v^j}\cup\Neigh{v^j}$
		behave in the same way than in $\rho_{K-1}$. Let $v \nin \set{v^j} \cup \Neigh{v^j}$ and $v \in \set{v_1, \dots, v_K}$. In that case, either 
		$v\in\set{v_1,\dots, v_{K-2}}$, and by construction, $L'_{j+1}(v) = L'_j(v)$, or $v\in\set{v_{K-1}, v_K}$, and since $v^j\notin\set{v_{K-2}, v_{K-1}, v_K, v_{K+1}}$, by construction, $L'_{j+1}(v)=L'_j(v)$. Let $v \nin \set{v^j} \cup \Neigh{v^j}$ such that $v \in \set{v_{K+1}, \dots, v_\ell}$. In that case, $L'_{j+1}(v) = L_{j+1}(v) = L_j(v) = L'_j(v)$. This allows us to conclude that $C'_j \transup{v^j ,t^j} C'_{j+1}$ thanks to \cref{obs:conf-suc}, and so $P_1(j+1)$ holds. By $P_2(j)$, either $L'_j(v_K)=L_j(v_K)$ and we have just proved that $L'_{j+1}(v_K)=L'_j(v_K)$, and that 
		$L_j(v_K)=L_{j+1}(v_K)$. Then $L'_{j+1}(v_K)=L_{j+1}(v_K)$ and $P_2(j+1)$ holds. Or $L_{j}(v_K)\in Q_2^r$ and since $L_j(v_K)=L_{j+1}(v_K)$, 
		$L_{j+1}(v_K)\in Q_2^r$. Hence, $P_2(j+1)$ holds. 
		\item Assume that $v^j =v_{K+1}$, or that $v^j = v_K$ and $L'_j(v_K) = L_j(v_K)$. By construction, $L_{j+1}(v^j) = L'_{j+1}(v^j)$ and $L_j(v^j) = L'_j(v^j)$.
		Observe that $v^j$ has at most two neighbors $v, v'$ such that:
		\begin{itemize}
			\item $v = v_{K+2}$ or $ v =v_{K+1}$, and so by construction $L_{j+1}(v) = L'_{j+1}(v)$ and $L_j(v) = L'_j(v)$;
			\item $v' = v_{K}$ or $v' = v_{K-1}$.
			If $v' = v_{K-1}$ or $v' = v_K$ and $L'_j(v_K) \neq L_j(v_K)$, then by construction either $t^j$ is a broadcast transition of some message $m \in \Sigma$ and $(L'_j(v'), ?m , L'_{j+1}(v')) \in \Delta$. Or 
			$L'_{j+1}(v') = L'_{j}(v')$.
			If $v' = v_K$ and $L'_j(v_K) = L_j(v_K)$, then $L'_{j+1}(v_K) = L_{j+1}(v_K)$ by definition.
		\end{itemize}
		Furthermore, for all $v \nin \set{v^j} \cup \Neigh{v^j}$, either 
		$v\in\set{v_{K+2}, \dots, v_\ell}$ and $L'_{j+1}(v)=L_{j+1}(v)=L_j(v)=L'_j(v)$, or $v\in\set{v_1,\dots, v_{K-2}}$ and $L'_{j+1}(v)=L'_j(v)$, or $v=v_{K-1}$ 
		(if $v^j=v_{K+1}$), and $L'_{j+1}(v)=L'_j(v)$. 
		This allows us to conclude (with \cref{obs:conf-suc}) that $C'_j \transup{v^j ,t^j} C'_{j+1}$ and so $P_1(j+1)$ holds. For $P_2(j+1)$, if $v^j = v_K$ and $L'_j(v_K) = L_j(v_K)$, then by construction $L'_{j+1}(v_K) = L_{j+1}(v_K)$ and so $P_2(j+1)$ trivially holds. If $v^j = v_{K+1}$, then either $L'_j(v_K) = L_j(v_K)$ and so by construction, $L'_{j+1}(v_K) = L_{j+1}(v_K)$, or $L'_j(v_K) \neq L_j(v_K)$ and so by induction hypothesis $L_j(v_K) \in Q_2^r$, and since $\PP$ is 2-phase-bounded, $L_{j+1}(v_K) \in Q_2^r$.
		
		\item Assume that $v^j = v_K$ and $L'_j(v_K) \neq L_j(v_K)$. Then, by $P_2(j)$, $L_j(v_K) \in Q_2^r$. Hence, $t$ is an internal transition (it
		cannot be a broadcast), and so $L_{j+1}(v) = L_j(v) = L'_j(v) = L'_{j+1}(v)$ for all 
		$v\in\set{ v_{K+1}, \dots, v_\ell}$, and $L'_{j+1}(v) = L'_j(v)$ for all $v \in \set{v_1, \dots v_{K-2}}$, and
		$L'_{j+1}(v_{K-1})=L'_j(v_{K-1})$ and $L'_{j+1}(v_K)=L'_j(v_K)$.  As a consequence,
		 $C'_{j+1} = C'_j$ and $P_1(j+1)$ holds. Furthermore, by induction hypothesis, $L_j(v_K) \in Q_2^r$ and so, as $\PP$ is 2-phase-bounded, it holds that $L_{j+1}(v_K) \in Q_2^r$ and we have proved $P_2(j+1)$. 
		 
		 \item Finally, assume that $v^j \in \set{v_1, \dots, v_{K-1}}$. Then, for all $v \in \set{v_{K+1}, \dots, v_\ell}$, $L_{j+1}(v) = L_j(v)$,
		 and $L'_{j+1}(v)=L_{j+1}(v)$ and $L'_j(v)=L_j(v)$ by construction. For all $v^j\in \set{v_1, \dots, v_{K-1}, v_K}$, $L'_{j+1}(v)=L'_j(v)$ 
		  and so $C'_{j+1} = C'_j$ and $P_1(j)$ holds. Furthermore, either $L_j(v_K) = L_{j+1}(v_K)$ or $(L_j(v_K), ?m, L_{j+1}(v_K))\in \Delta$ and
		  $v^j=v_{K-1}$ and $t^j=(L_j(v_{K-1}), !!m, L_{j+1}(v_{K-1}))$. In the first case, 
		  by $P_2(j)$, either $L_j(v_K)=L'_j(v_K)$ and then $L_{j+1}(v_K)=L_j(v_k)=L'_j(v_K)=L'_{j+1}(v_K)$, or  $L_j(v_K)= L_{j+1}(v_K) \in Q_2^r$ and so 
		  $P_2(j+1)$ holds. In the second case, $j > \firstBroadcast{v_K}{\rho_{K-1}}$. Indeed, either $\firstBroadcast{v_{K-1}}{\rho_{K-1}}\leq 
		  \firstBroadcast{v_K}{\rho_{K-1}}$ and then $j>j_1 > \firstBroadcast{v_K}{\rho_{K-1}}$, or $j\geq \firstBroadcast{v_{K-1}}{\rho_{K-1}}> \firstBroadcast{v_{K}}{\rho_{K-1}}$.  By definition of a 2-phase-bounded protocol, $L_j(v_K) \in Q_1^b \cup Q_2^b\cup Q_2^r$ hence, $L_{j+1}(v_K) \in Q_2^r$ and so $P_2(j+1)$ holds.
	\end{itemize}
	Hence, $P_1(j)$ and $P_2(j)$ holds for all $j_1+1\leq j\leq n$. 
	
	We build a sequence of configurations $C'_0\trans^\ast C'_1\trans^\ast C'_n$ by letting $C'_i=C_i$ for all $0\leq i\leq j_1$. 
	Let us show now that $C'_0, \dots, C'_n$ that we have defined form the expected execution $\rho_K$. 
	\begin{itemize}
	\item Let $v \in \set{v_{K+1}, \dots, v_\ell}$. By construction $L'_n(v)=L_n(v)$. By induction hypothesis (on $K$), $L_n(v)=L(v)$. 
	\item Since $P_1(j)$ holds for all $ j_1+1\leq j\leq n$, $\lastBroadcast{v_{K-1}}{\rho_K} < j_1$ and, by definition of $j_1$, it implies that $\lastBroadcast{v_{K-1}}{\rho_K} \leq
	\firstBroadcast{v_K}{\rho_K}$. Let $1\leq i < K-1$. By induction hypothesis, $\lastBroadcast{v_i}{\rho_{K-1}}\leq \firstBroadcast{v_{i+1}}{\rho_{K-1}}$. 
	If $\firstBroadcast{v_{K-1}}{\rho_{K-1}}\leq \firstBroadcast{v_K}{\rho_{K-1}}$, then $\lastBroadcast{v_{K-2}}{\rho_{K-1}}\leq \firstBroadcast{v_{K-1}}{\rho_{K-1}}
	\leq \firstBroadcast{v_K}{\rho_{K-1}}<j_1$. Thus, $\lastBroadcast{v_{K-2}}{\rho_K}=\lastBroadcast{v_{K-2}}{\rho_{K-2}}\leq \firstBroadcast{v_{K-1}}{\rho_{K-1}}
	= \firstBroadcast{v_{K-1}}{\rho_{K}}$. Since, for all $1\leq i < K-2$, $\firstBroadcast{v_{i+1}}{\rho_{K-1}}\leq \lastBroadcast{v_{i+1}}{\rho_{K-1}}<j_1$, and by
	induction hypothesis, $\lastBroadcast{v_{i}}{\rho_{K-1}} \leq \firstBroadcast{v_{i+1}}{\rho_{K-1}}$, we deduce that $\lastBroadcast{v_{i}}{\rho_{K}}=
	\lastBroadcast{v_{i}}{\rho_{K-1}}\leq \firstBroadcast{v_{i+1}}{\rho_{K-1}}=\firstBroadcast{v_{i+1}}{\rho_{K}}$. 
	If now $\firstBroadcast{v_{K-1}}{\rho_{K-1}}> \firstBroadcast{v_K}{\rho_{K-1}}$, then by $P_1(j)$, $\lastBroadcast{v}{\rho_K}=-1$, for all $v\in\set{v_1,\dots, v_{K-1}}$. Obviously then, for all $1 \leq i \leq K-1$, $\lastBroadcast{v_i}{\rho_K} \leq \firstBroadcast{v_{i+1}}{\rho_K}$.
	
	\item Let $K + 1\leq i \leq \ell$. Then, by $P_1(j)$, $\lastBroadcast{v_i}{\rho_K} \leq \firstBroadcast{v_{i+1}}{\rho_K}$ if and only if $\lastBroadcast{v_i}{\rho_{K-1}} \leq \firstBroadcast{v_{i+1}}{\rho_{K-1}}$ if and only if, by induction hypothesis, $\lastBroadcast{v_i}{\rho} \leq \firstBroadcast{v_{i+1}}{\rho}$.
	\end{itemize}
		
	\color{black}
\end{proof}

By symmetry, we get the following lemma.
\begin{lemma}\label{lemma:CoverLine-2pb-inP:execution-shape-left}
	Let $\rho = C_0\trans^\ast C$ with $C = (\Gamma, L)$ and $1 < K \leq \ell$. 
	There exists $\rho' = C_0  \trans^\ast C'$ such that $C' = (\Gamma, L')$, and 
	\begin{itemize}
		\item for all $v \in \set{v_{1}, \dots, v_{K-1}}$, $L'(v) = L(v)$;
		\item for all $K+1 \leq i \leq \ell$, $\lastBroadcast{v_i}{\rho'} \leq \firstBroadcast{v_{i-1}}{\rho'}$;
		\item for all $1\leq i \leq K$, $\lastBroadcast{v_i}{\rho'} \leq \firstBroadcast{v_{i-1}}{\rho'}$ if and only if $\lastBroadcast{v_i}{\rho} \leq \firstBroadcast{v_{i-1}}{\rho}$.
	\end{itemize}
\end{lemma}

We are now ready to prove \cref{lemma:CoverLine-2pb-inP:execution-shape-1}.

\begin{proofof}{\cref{lemma:CoverLine-2pb-inP:execution-shape-1}}
	From $\rho = C_0 \trans^\ast C$ with $C = (\Gamma, L)$ we apply \cref{lemma:CoverLine-2pb-inP:execution-shape-right}\ for $K = N-1$ and we get an execution $\rho' = C'_0 \trans^\ast C'$ such that $C'= (\Gamma, L')$ and 
	\begin{itemize}
		\item for all $v \in \set{v_{N}, \dots, v_\ell}$, $L'(v) = L(v)$;
		\item for all $1 \leq i \leq N-2$, $\lastBroadcast{v_i}{\rho'} \leq \firstBroadcast{v_{i+1}}{\rho'}$;
		\item for all $N-1\leq i \leq \ell$, $\lastBroadcast{v_i}{\rho'} \leq \firstBroadcast{v_{i+1}}{\rho'}$ if and only if $\lastBroadcast{v_i}{\rho} \leq \firstBroadcast{v_{i+1}}{\rho}$.
	\end{itemize}
	Next, from $\rho'$, we apply \cref{lemma:CoverLine-2pb-inP:execution-shape-left}\ for $K = N+1$ and we get an execution $\rho'' = C''_0 \trans^\ast C''$ such that $C'' = (\Gamma, L'')$ and
	\begin{itemize}
		\item for all $v \in \set{v_{1}, \dots, v_{N}}$, $L''(v) = L'(v)$;
		\item for all $N+2 \leq i \leq \ell$, $\lastBroadcast{v_i}{\rho''} \leq \firstBroadcast{v_{i-1}}{\rho''}$;
		\item for all $1\leq i \leq N+1$, $\lastBroadcast{v_i}{\rho''} \leq \firstBroadcast{v_{i-1}}{\rho''}$ if and only if $\lastBroadcast{v_i}{\rho'} \leq \firstBroadcast{v_{i-1}}{\rho'}$.
	\end{itemize}
Hence, $\rho''$ covers $q_f$ as $L''(v_N) = L'(v_N) = L(v_N) = q_f$ and 
\begin{itemize}
	\item for all $N+2 \leq i \leq \ell$, $\lastBroadcast{v_i}{\rho''} \leq \firstBroadcast{v_{i-1}}{\rho''}$;
	\item for all $1\leq i \leq N-2$, $\lastBroadcast{v_i}{\rho''} \leq \firstBroadcast{v_{i+1}}{\rho''}$.
\end{itemize}
\end{proofof}

\paragraph{Proof of \cref{lemma:CoverLine-2pb-inP:execution-shape-2}}
We are now ready to prove \cref{lemma:CoverLine-2pb-inP:execution-shape-2}.
We begin with three observations that hold regardless of whether $\PP$ is 2-phase-bounded or not. 
In a line $\Gamma = (V, E)$ with $V = \set{v_1, \dots v_\ell}$ and $E = \set{\langle v_i, v_{i+1} \rangle \mid 1 \leq i < \ell}$, we define the distance between two nodes $v_i, v_j$ as $d(v_i, v_j) = |j-i|$.
\begin{observation}\label{obs:CoverLine-2pb-inP:reordening-broadcasts}
	Let $C, C' \in \CC$, and $u,v \in \Vert{\Gamma}$ such that $d(v,u) \geq 3$. If $C \transup{u,t} C_1 \transup{v,t'} C'$ for two transitions $t, t' \in \Delta$ and a configuration $C_1 \in \CC$, then there exists $C_2 \in \CC$ such that $C \transup{v,t'} C_2 \transup{u,t} C'$.
\end{observation}

\begin{observation}\label{obs:CoverLine-2pb-inP:reordening-internal-2}
	Let $C, C' \in \CC$, and $u,v \in \Vert{\Gamma}$ such that $d(u, v) \geq 2$. Then if $C \transup{u,t} C_1 \transup{v,t'} C'$ for two transitions $t, t' \in \Delta$ with one of the two transitions \emph{internal} and a configuration $C_1 \in \CC$, then there exists $C_2 \in \CC$ such that $C \transup{v,t'} C_2 \transup{u,t} C'$.
\end{observation}

\begin{observation}\label{obs:CoverLine-2pb-inP:reordening-internal-1}
	Let $C, C' \in \CC$, and $u,v \in \Vert{\Gamma}$. Then if $C \transup{u,t} C_1 \transup{v,t'} C'$ for two internal transitions $t, t' \in \Delta$ and a configuration $C_1 \in \CC$, then there exists $C_2 \in \CC$ such that $C \transup{v,t'} C_2 \transup{u,t} C'$.
\end{observation}

For the rest of this proof we let $\rho = C_0 \transup{v^0, t^0} C_1 \transup{v^1, t^1} \cdots \transup{v^{n-1}, t^{n-1}} C_n$ be the execution obtained from \cref{lemma:CoverLine-2pb-inP:execution-shape-1}. We define $j_1 := \lastBroadcast{v_{N-3}}{\rho} + 1$ and $j_2 = \lastBroadcast{v_{N+3}}{\rho} + 1$.

\begin{lemma}\label{lemma:CoverLine-2pb-inP:delay-1}
	Let $0 \leq j < j_1$ be the maximal index such that $C_j \transup{v^j, t^j} C_{j+1}$, for some $v^j \in \set{v_{N-1}, v_N, \dots, v_\ell}$
	then,
	$C_0 \trans^\ast C_j \transup{v^{j+1}, t^{j+1}} C'_{j+1} \trans^\ast C'_{j_1 -2} \transup{v_{N-3}, t^{j_1}} C'_{j_1 - 1} \transup{v^j,t^j} C_{j_1}$.
\end{lemma}
\begin{proof}
	Let $0 \leq j < j_1$ be the maximal index such that $C_j\transup{v^j,t^j} C_{j+1}$ for some $v \in \set{v_{N-1}, v_{N}, \dots v_{\ell}}$.
	
	Recall that $j_1 = \lastBroadcast{v_{N-3}}{\rho} + 1$, and recall that $ \lastBroadcast{v_{N-3}}{\rho} \leq \firstBroadcast{v_{N-2}}{\rho}$, hence either $\firstBroadcast{v_{N-2}}{\rho} = -1$ and so it never broadcast anything, either $j_1 \leq \firstBroadcast{v_{N-2}}{\rho}$.
	In both cases, $v_{N-2}$ does not broadcast anything between $C_0$ and $C_{j_1}$.

	As $j$ is maximal, for all $j <k < j_1$, $v^{k} \in \set{v_1, \dots, v_{N-2}}$ and if $v^{k} = v_{N-2}$, $t^{k}$ is internal. 
	We prove by induction that for all $j \leq  k < j_1$: 
	
	$P(k)$: there exists $C'_{j+1} ,\dots ,C'_{k}$ such that $C_{j}\transup{v^{j+1}, t^{j+1}} C'_{j+1} \trans \cdots \transup{v^{k},t^{k}} C'_{k} \transup{v^{j}, t^j} C_{k+1}$. 
	
	For $k = j$, the induction property trivially holds. 
	
	Let $j \leq k <j_1 - 1$ such that $C_{j}\transup{v^{j+1}, t^{j+1}} C'_{j+1} \cdots \transup{v^{k},t^{k}} C'_{k} \transup{v^{j}, t^j} C_{k+1}$. Let $C_{k+1} \transup{v^{k+1}, t^{k+1}} C_{k+2}$.
	
	\textbf{First case: $v^j = v_{N-1}$.} Then, since $j \leq k+1\leq j_1-1= \lastBroadcast{v_{N-3}}{\rho}\leq\firstBroadcast{v_{N-2}}{\rho}\leq \lastBroadcast{v_{N-2}}{\rho}
	\leq\firstBroadcast{v_{N-1}}{\rho}$, by~\cref{lemma:CoverLine-2pb-inP:execution-shape-1}, $t^j$ cannot be a broadcast and is thus an internal transition. 
	If $v^{k+1} \in \set{v_1, \dots, v_{N-3}}$, then $d(v_{N-1}, v^{k+1}) \geq 2$ and we apply \cref{obs:CoverLine-2pb-inP:reordening-internal-2}, to get that there exists $C'_{k+1}$ such that $C'_{k} \transup{v^{k+1}, t^{k+1}} C'_{k+1} \transup{v^j, t^j} C_{k+2}$. If $v^{k+1} = v_{N-2}$, then $t^{k+1}$ is internal and we can hence apply \cref{obs:CoverLine-2pb-inP:reordening-internal-1}\ to get that there exists $C'_{k+1}$ such that $C'_{k} \transup{v^{k+1}, t^{k+1}} C'_{k+1} \transup{v^j, t^j} C_{k+2}$.
	
	\textbf{Second case: $v^j = v_{N}$.} Then, if $v^{k+1} \in \set{v_1, \dots v_{N-3}}$ and so $d(v_N, v^{k+1}) \geq 3$: We can apply \cref{obs:CoverLine-2pb-inP:reordening-broadcasts}\ to get that there exists $C'_{k+1}$ such that $C'_{k} \transup{v^{k+1}, t^{k+1}} C'_{k+1} \transup{v^j, t^j} C_{k+2}$. Otherwise, $v^{k +1} = v_{N-2}$ and $t^{k+1}$ is internal and $d(v_N, v_{N-2}) = 2$ and so we can apply \cref{obs:CoverLine-2pb-inP:reordening-internal-2}\ 
	to get that there exists $C'_{k+1}$ such that $C'_{k} \transup{v^{k+1}, t^{k+1}} C'_{k+1} \transup{v^j, t^j} C_{k+2}$.
	
	\textbf{Third case: $v^j\in \set{v_{N+1}, \dots v_\ell}$.} Then $d(v^j, v^{k+1}) \geq 3$ and so we can apply \cref{obs:CoverLine-2pb-inP:reordening-broadcasts}\ to get that there exists $C'_{k+1}$ such that $C'_{k} \transup{v^{k+1}, t^{k+1}} C'_{k+1} \transup{v^j, t^j} C_{k+2}$.
	
\end{proof}

\begin{lemma}\label{lemma:CoverLine-2pb-inP:delay-2}
	Let $j_1 \leq j < j_2$ the \emph{maximal} index such that $C_j \transup{v^j, t^j} C_{j+1}$, for some $v^j \in \set{v_{N-2}, v_{N-1}, v_{N}, v_{N+1}}$, then,
	$C_{j_1} \trans^\ast C_j \transup{v^{j+1}, t^{j+1}} C'_{j+1} \transup{v^{j+2}, t^{j+2}} \cdots \trans C'_{j_2-2} \transup{v_{N+3}, t^{j_2}} C'_{j_2 - 1} \transup{v^j,t^j} C_{j_2}$.
\end{lemma}
\begin{proof}
	Let $j_1 \leq j < j_2$ be the maximal index such that $C_j\transup{v^j,t^j} C_{j+1}$ for some $v \in \set{v_{N-2}, v_{N-1}, v_N, v_{N+1}}$.
	
Recall that $j_1 = \lastBroadcast{v_{N-3}}{\rho} +1$, and that $\lastBroadcast{v_{N-3}}{\rho}  \geq \firstBroadcast{v_{N-3}}{\rho} \geq \lastBroadcast{v_{N-4}}{\rho}$, and so by a simple inductive reasoning: $j_1 > \lastBroadcast{v}{\rho}$ for all $v \in \set{v_1, \dots, v_{N-3}}$. Hence, for all $j_1 \leq k  < j_2$, $v^k \nin \set{v_1, \dots, v_{N-3}}$.
	
	Hence, as $j$ is maximal, for all $j <k < j_1$, $v^{k} \in \set{v_{N+2}, \dots, v_{\ell}}$. 
	We prove by induction that for all $j \leq k < j_2$, there exist $C'_{j+1}, \dots C'_k$ such that $C_j\transup{v^{j+1}, t^{j+1}} C'_{j+1} \trans \cdots C'_k \transup{v^j, t^j} C_{k+1}$.
	For $k =j$, the induction property trivially holds. Let $j \leq k <j_2-1$ such that there exist $C'_{j+1}, \dots C'_k$ with $C_j\transup{v^{j+1}, t^{j+1}} C'_{j+1} \trans \cdots C'_k \transup{v^j, t^j} C_{k+1}$. Denote $C_{k+1} \transup{v^{k+1}, t^{k+1}} C_{k+2}$.
	
	\textbf{First case:} $v^j = v_{N+1}$. Then, $\lastBroadcast{v_{N+2}}{\rho}\leq \firstBroadcast{v_{N+1}}{\rho}$ by~\cref{lemma:CoverLine-2pb-inP:execution-shape-1}, and $\lastBroadcast{v_{N+3}}{\rho}\leq\firstBroadcast{v_{N+2}}{\rho}\leq \lastBroadcast{v_{N+2}}{\rho}$. Since $k+1\leq j_2-1 = \lastBroadcast{v_{N+3}}{\rho}$, $k+1\leq \lastBroadcast{v_{N+3}}{\rho}\leq\firstBroadcast{v_{N+1}}{\rho}$, and so, $t^j$ cannot be a broadcast, and is an internal transition.
	If $v^{k+1} \in \set{v_{N+3}, \dots v_\ell}$ then $d(v_{N+1}, v^{k+1}) \geq 2$ and we can apply \cref{obs:CoverLine-2pb-inP:reordening-internal-2}\ to get that there exists $C'_{k+1}$ such that $C'_k \transup{v^{k+1}, t^{k+1}} C'_{k+1}\transup{v^j, t^j}  C_{k+2}$. Otherwise, $v^{k+1} = v_{N+2}$. Recall that $j_2 = \lastBroadcast{v_{N+3}}{\rho} +1$ and that $k+1\leq \lastBroadcast{v_{N+3}}{\rho} \leq \firstBroadcast{v_{N+2}}{\rho}$, hence $t^{k+1}$ is an internal transition. Hence we can apply \cref{obs:CoverLine-2pb-inP:reordening-internal-1}\ to get that there exists $C'_{k+1}$ such that $C'_k \transup{v^{k+1}, t^{k+1}} C'_{k+1}\transup{v^j, t^j}  C_{k+2}$. 
	
	\textbf{Second case:} $v^j = v_N$. Then, if $v^{k+1} \in \set{v_{N+3}, \dots v_{\ell}}$ and so $d(v_N, v^{k+1}) \geq 3$, we can apply 
	\cref{obs:CoverLine-2pb-inP:reordening-broadcasts}\ to get that there exists $C'_{k+1}$ such that $C'_{k} \transup{v^{k+1}, t^{k+1}} C'_{k+1} \transup{v^j, t^j} 
	C_{k+2}$. Otherwise, $v^{k +1} = v_{N+2}$ and $t^{k+1}$ is internal and $d(v_N, v_{N+2}) = 2$ and so we can apply 
	\cref{obs:CoverLine-2pb-inP:reordening-internal-2}\ to get that there exists $C'_{k+1}$ such that $C'_{k} \transup{v^{k+1}, t^{k+1}} C'_{k+1} \transup{v^j, t^j} 
	C_{k+2}$.
	
	\textbf{Third case: $v^j\in \set{v_{N-2}, v_{N-1}}$.} Then $d(v^j, v^{k+1}) \geq 3$ and so we can apply \cref{obs:CoverLine-2pb-inP:reordening-broadcasts}\ to get that there exists $C'_{k+1}$ such that $C'_{k} \transup{v^{k+1}, t^{k+1}} C'_{k+1} \transup{v^j, t^j} C_{k+2}$.
	
\end{proof}

We are now ready to prove \cref{lemma:CoverLine-2pb-inP:execution-shape-2}.

\begin{proofof}{\cref{lemma:CoverLine-2pb-inP:execution-shape-2}}
	Let 
	$j_1(\rho) := \lastBroadcast{v_{N-3}}{\rho} + 1$ and $j_2(\rho) = \lastBroadcast{v_{N+3}}{\rho} + 1$. We let $i_1, \dots i_k$  the indices $0 \leq i_1 < \cdots < i_k  < j_1(\rho)$ such that $C_{i_j }\transup{v, t_j} C_{i_j +1}$ for some $v \in \set{v_{N-1}, \dots, v_{\ell}}$ and $t_j \in \Delta$ We denote $\textsf{EarlyActions}[\rho]$ the number of such indices.
	
	We build inductively $\rho_0, \rho_1, \dots \rho_{k}$ such that for all $j$, $\rho_j= C_0 \trans^\ast C^j = (\Gamma, L^j)$ with:
	\begin{itemize}
		\item $P_1(j)$: $L^j(v_N) = q_f$, and 
		\item $P_2(j)$: $\textsf{EarlyActions}[\rho_{j}] = k-j$, and 
		\item $P_3(j)$:  the order of \cref{lemma:CoverLine-2pb-inP:execution-shape-1}\ is preserved, i.e.  for all $1 \leq i \leq N-2$,
		$\lastBroadcast{v_i}{\rho_{j}} \leq \firstBroadcast{v_{i+1}}{\rho_{j}}$ and for all $N+2 \leq i \leq \ell$, $\lastBroadcast{v_i}{\rho_{j}} \leq \firstBroadcast{v_{i-1}}{\rho_{j}}$.
	\end{itemize}
	
	For $j = 0$, define $\rho_0 = \rho$. It trivially satisfies the induction properties.
	
	Let $0 \leq j < k$ and assume we have built $\rho_j$ satisfying $P_1(j)$, $P_2(j)$ and $P_3(j)$ and denote it $\rho_j = C^j_0 \transup{v^0, t^0} C^j_1 \transup{v^1, t^1} \cdots \transup{v^{n-1} t^{n-1}} C^j_n= (\Gamma, L^j_n)$.
	Let $i$ be the maximal index such that $i <j_1(\rho_j)$ and $C_{i }\transup{v^i, t^i} C_{i +1}$ for some $v^i \in \set{v_{N-1}, \dots, v_{\ell}}$
	and $t \in \Delta$. From \cref{lemma:CoverLine-2pb-inP:delay-1}, there exists $C'_{i+1}, \dots C'_{j_1(\rho_j) -1}$ such that $C^{j}_0 \trans^\ast C^j_{i} 
	\transup{v^{i+1}, t^{i+1}} C'_{i +1} \transup{v^{i+2}, t^{i+2}}  \cdots \trans C'_{j_1(\rho_j) - 1} \transup{v^i,t^i} C^j_{j_1(\rho_j)}$, hence we define 
	$\rho_{j+1} = C^{j}_0 \trans^\ast C^j_{i} \transup{v^{i+1}, t^{i+1}} C'_{i +1} \trans^\ast  C'_{j_1(\rho_j)- 1} \transup{v^i,t^i} C^j_{j_1(\rho_j)} \trans^\ast C_n$. 
	By construction $P_1(j+1)$ holds.
	Furthermore, $j_1(\rho_{j+1})=\lastBroadcast{v_{N-3}}{\rho_{j+1}} + 1 = j_1(\rho_j) - 1$ and so $\textsf{EarlyActions}[\rho_{j+1}] = 
	\textsf{EarlyActions}[\rho_j] -1= k-j-1$ and so $P_2(j+1)$ holds.
	
	We denote $v^i = v_r$.
	Recall that $v_r \in \set{v_{N-1}, \dots, v_{\ell}}$. 
	Observe that for all $1 \leq k \leq \ell$, it holds that:
	\begin{itemize}
		\item if $k <\ell$ and $k \nin \set{r-1, r}$, then $\lastBroadcast{v_k}{\rho_{j+1}} \leq \firstBroadcast{v_{k+1}}{\rho_{j+1}}$ if and only if $\lastBroadcast{v_k}{\rho_{j}} \leq \firstBroadcast{v_{k+1}}{\rho_{j}}$;
		\item if $k > 0$ and $k \nin \set{r, r+1}$, then $\lastBroadcast{v_k}{\rho_{j+1}} \leq \firstBroadcast{v_{k-1}}{\rho_{j+1}}$ if and only if $\lastBroadcast{v_k}{\rho_{j}} \leq \firstBroadcast{v_{k-1}}{\rho_{j}}$.
	\end{itemize}
	Hence, by $P_3(j)$, for all $1 \leq k \leq N-2$, if $k \nin \set{r-1, r}$
	$\lastBroadcast{v_k}{\rho_{j+1}} \leq \firstBroadcast{v_{k+1}}{\rho_{j+1}}$ and for all $N+2 \leq k \leq \ell$, if $k \nin \set{r, r+1}$, $\lastBroadcast{v_k}{\rho_{j}} \leq \firstBroadcast{v_{k-1}}{\rho_{j}}$.
	
	Let now $1 \leq k\leq N-2$, such that $k \in \set{r-1, r}$. As $N-1 \leq r \leq \ell$, $k = N-2$ and $r = N-1$. 
	As $j_1(\rho_j) > i \geq 0$, from $P_3(j)$, we get that $j_1(\rho_j) - 1 = \firstBroadcast{v_{N-3}}{\rho_j} \leq \lastBroadcast{v_{N-3}}{\rho_j} < \firstBroadcast{v_{r-1}}{\rho_j} \leq \lastBroadcast{v_{r-1}}{\rho_j}< \firstBroadcast{v_{r}}{\rho_j}$. Hence $\lastBroadcast{v_{r-1}}{\rho_{j+1}}  = \lastBroadcast{v_{r-1}}{\rho_j}$, and $ \firstBroadcast{v_{r}}{\rho_{j+1}} =\firstBroadcast{v_{r}}{\rho_j}$.
	
	Let now $N+2 \leq k \leq \ell$, such that $ k \in \set{r, r+1}$, hence, $N +1 \leq r \leq \ell$.
	 If $k = r$, then notice that either $\lastBroadcast{v_r}{\rho_{j}} > i$ or $\lastBroadcast{v_r}{\rho_{j}}=i$. In the first case, 
	 		as for all $i< i_2 < j_1(\rho_j)$, $v^{i_2} \nin \set{v_{N+1}, \dots, v_\ell}$, it holds that $\lastBroadcast{v_r}{\rho_{j}} > j_1(\rho_{j+1})$, and so $\lastBroadcast{v_r}{\rho_{j+1 }} = \lastBroadcast{v_r}{\rho_{j}}$, and, 
	 			from $P_3(j)$, $j_1(\rho_j) \leq \lastBroadcast{v_r}{\rho_{j+1 }} \leq \firstBroadcast{v_{r-1}}{\rho_{j}} = \firstBroadcast{v_{r-1}}{\rho_{j+1}}$.
	 If $\lastBroadcast{v_r}{\rho_{j}}=i$, then $\lastBroadcast{v_r}{\rho_{j+1}}=j_1(\rho_j)-1$. As for all $i< i_2 < j_1(\rho_j)$, $v^{i_2} \nin \set{v_{N+1}, \dots, v_\ell}$, it holds that $\lastBroadcast{v_r}{\rho_{j}} \leq j_1(\rho_j) - 1 \leq \firstBroadcast{v_{r-1}}{\rho_j} = \firstBroadcast{v_{r-1}}{\rho_{j+1}}$. Hence, $ \lastBroadcast{v_r}{\rho_{j+1 }} \leq \firstBroadcast{v_{r-1}}{\rho_{j+1}}$.
	 
	 If $k = r+1$, then, as for all $i< i_2 < j_1(\rho_j)$, $v^{i_2} \nin \set{v_{N+1}, \dots, v_\ell}$, it holds that $\lastBroadcast{v_k}{\rho_{j+1}} = \lastBroadcast{v_k}{\rho_j}$ and if $t^i$ is internal, $\firstBroadcast{v_{r}}{\rho_{j+1}} = \firstBroadcast{v_{r}}{\rho_j}$.
	Otherwise, $\firstBroadcast{v_{r}}{\rho_j} \leq \firstBroadcast{v_{r}}{\rho_{j+1}}$. From $P_3(j)$, it holds that $\lastBroadcast{v_k}{\rho_{j+1}} =\lastBroadcast{v_k}{\rho_j} \leq  \firstBroadcast{v_{k-1}}{\rho_j} \leq \firstBroadcast{v_{k-1}}{\rho_{j+1}}$.
%
%
 Hence $P_3(j+1)$ holds.

	Hence, we build $\rho_k$ such that $\rho_k = C^k_0 \trans^\ast C^k_{j_1(\rho_k)} \trans^\ast C^k_{n}$ and $P_1(k)$, $P_2(k)$ and $P_3(k)$ hold.

	With the same reasoning between $C_{j_1}$ and $C_{j_2}$, and applying this time \cref{lemma:CoverLine-2pb-inP:delay-2}, we finally get an execution $\rho' = C'_0 \trans^\ast C'_{j_1(\rho')} \trans^\ast C'_{j_2(\rho')} \trans^\ast C'_n $ such that:
	\begin{itemize}
		\item for all $1 \leq i \leq N-2$, $\lastBroadcast{v_i}{\rho'} \leq \firstBroadcast{v_{i+1}}{\rho'}$;
		\item for all $N+2 \leq i \leq \ell$, $\lastBroadcast{v_i}{\rho'} \leq \firstBroadcast{v_{i-1}}{\rho'}$.
		\item nodes issuing transitions between $C'_0$ and $C'_{j'_1}$ belong to $\set{v_1, \dots, v_{N-3}, v_{N-2}}$ and $v_{N-2}$ only performs internal transitions (because $j_1(\rho') = \lastBroadcast{v_{N-3}}{\rho'} + 1 \leq \firstBroadcast{v_{N-2}}{\rho'}$;
		\item nodes issuing transitions between $C'_{j'_1}$ and $C'_{j'_2}$ belong to $\set{v_{N+2}, v_{N+3}, \dots, v_{\ell}}$ and $v_{N+2}$ only performs internal transitions;
	\end{itemize}

It is left to prove that in $C'_{j_2(\rho')} \trans^\ast C'_n $, only nodes in $\set{v_{N-2}, v_{N-1} , v_N, v_{N+1}, v_{N+2}}$ issue transitions.
	Assume a node $v \nin \set{v_{N-2}, v_{N-1} , v_N, v_{N+1}, v_{N+2}}$ issues a transition between $C'_{j_2(\rho')}$ and $C'_n$, hence by definition, $\lastBroadcast{v}{\rho'} \geq j_2(\rho')$.
	Either $v \in \set{v_1, \dots v_{N-3}}$ and so $\lastBroadcast{v}{\rho'} \leq \lastBroadcast{v_{N-3}}{\rho'} < j_1(\rho') <j_2(\rho')$ and we reach a contradiction. Or $v \in \set{v_{N+3}, \dots v_{\ell}}$ and so $\lastBroadcast{v}{\rho'} \leq \lastBroadcast{v_{N+3}}{\rho'} < j_2(\rho')$ and we reach a contradiction.
\end{proofof}

\subsubsection{Completeness of the algorithm}\label{appendix:subsec:CoverLine-inP:lemma:complete}
This part is dedicated to prove the following lemma.
\begin{lemma}\label{lemma:CoverLine-2pb-inP:complete}
	If $q_f$ is coverable with a line topology, then there exist $q_1, q_2 \in H$ such that $C_{q_1, q_2} \trans^\ast C$ and $C = (\Gamma_5, L)$ with $L(v_3)  = q_f$.
\end{lemma}

If $q_f$ is coverable, then let $\rho = C_0 \trans^\ast C_{j_1} \trans^\ast C_{j_2} \trans^\ast C_n$ be the execution obtained from \cref{lemma:CoverLine-2pb-inP:execution-shape-2}\ with $j_1 = \firstBroadcast{v_{N-3}}{\rho} + 1$ and $j_2 = \firstBroadcast{v_{N+3}}{\rho} + 1$. We denote $C_i = (\Gamma, L_i)$ for all $0\leq i \leq n$.

We start by proving that $L_{j_1}(v_{N-2}) \in H$ and $L_{j_2}(v_{N+2}) \in H$.
We need three preliminaries lemmas.
\begin{lemma}\label{lemma:proof:lemma:CoverLine2pb-complete-case0}
	For all $0 \leq j \leq \firstBroadcast{v_1}{\rho}$, $(L_{j}(v_1), \qinit) \in S$.
\end{lemma}
\begin{proof}
	We prove the lemma by induction on $j$.
	For $j = 0$, $L_0(v_1) = \qinit$ and by construction $(\qinit, \qinit) \in S$.
	Let $0 \leq j < \firstBroadcast{v_1}{\rho}$ such that $(L_{j}(v_1), \qinit) \in S$. 
		Denote $C_j \transup{v^j, t^j} C_{j+1}$. Observe that either $L_j(v_1) = L_{j+1}(v_1)$ (and there is nothing to do), or $v^j \in \set{v_1, v_2}$. 
		As $j < \firstBroadcast{v_1}{\rho}$, from \cref{item:lemma:exec-shape-steps} of \cref{lemma:CoverLine-2pb-inP:execution-shape-2}, it holds that $j  < \firstBroadcast{v_1}{\rho} \leq \firstBroadcast{v_2}{\rho}$. Hence $t^j$ is internal and if $v^j = v_2$, $L_j(v_1) = L_{j+1}(v_1)$. Otherwise, $v^j = v_1$ and $(L_j(v_1), \qinit) \in S$. Hence, $(L_{j+1}(v_1), \qinit) \in S$ by definition of $S$.
	
\end{proof}

We adopt the convention that $L_{-1}(v) = \qinit$ for all $v$.
We prove the following lemma by induction on $k$.

\begin{lemma}\label{lemma:proof:lemma:CoverLine2pb-complete-case-ind}
	For all $1 < k \leq N-2$, 
	if $L_{\firstBroadcast{v_{k-1}}{\rho}}(v_{k-1}) \in H$, then $L_{j}(v_k)\in H$ for all $\firstBroadcast{v_{k-1}}{\rho} \leq j \leq \firstBroadcast{v_{k}}{\rho}$.
\end{lemma}
\begin{proof}
	Let $1 < k \leq N-2$ such that $L_{\firstBroadcast{v_{k-1}}{\rho}}(v_{k-1}) \in H$, we prove the lemma by induction. 
	For ease of readability, denote $j^{k-1} = \firstBroadcast{v_{k-1}}{\rho}$ and $j^k =  \firstBroadcast{v_{k}}{\rho}$.
	By definition of $H$, there exists $q \in Q$ such that $(L_{j^{k-1}}(v_{k-1}),q)\in S$. 

	First let us show that for all $0 \leq j \leq j^{k-1}$, $(L_j(v_k), L_{j^{k-1}}(v_{k-1})) \in S$: for $j = 0$, as $(L_{j^{k-1}}(v_{k-1}),q)\in S$, by definition of $S$, $(\qinit, L_{j^{k-1}}(v_{k-1}))\in S$.
	Let now $0 \leq j < j^{k-1}$ and denote $C_j \transup{v^j, t^j} C_{j+1}$. 
	Recall that in $\rho$, as $j < j^{k-1} \leq j_1$, it holds that $v^j \in \set{v_1, \dots v_{N-2}}$ (\cref{item:lemma:exec-shape-steps} of \cref{lemma:CoverLine-2pb-inP:execution-shape-2}). Furthermore, as $j < j^{k-1}$, then $j < \firstBroadcast{v}{\rho}$ for all $v \in \set{v_{k-1}, \dots, v_{N-2}}$ (\cref{item:lemma:exec-shape-order} of \cref{lemma:CoverLine-2pb-inP:execution-shape-2}). Hence, either $t^j$ is an internal transition, or $v^j \in \set{v_1, \dots v_{k-2}}$. In the latter case, $L_{j+1}(v_k) = L_j(v_k)$.
	Otherwise, either $v^j \neq v_k$ and so $L_{j+1}(v_k) = L_j(v_k)$ and there is nothing to do, or $v^j = v_k$ and so $(L_j(v_k), \tau, L_{j+1}(v_k))\in \Delta$ and so, as $(L_j(v_k), L_{j^{k-1}}(v_{k-1})) \in S$, $(L_{j+1}(v_k), L_{j^{k-1}}(v_{k-1})) \in S$ by definition of $S$.
	
	Hence $(L_{j^{k-1}}(v_k), L_{j^{k-1}}(v_{k-1})) \in S$. We now prove that for all $j^{k-1} \leq j \leq j^k$, $(L_j(v_k), L_j(v_{k+1})) \in S$. As we just proved it for $j = j^{k-1}$, let $j^{k-1} \leq j < j^k$, and denote $C_j \transup{v^j, t^j} C_{j+1}$. Note that from \cref{lemma:CoverLine-2pb-inP:execution-shape-2}, we get that
	(a) $v^j \in \set{v_1, \dots, v_{N-2}}$ as $j < \firstBroadcast{v_k}{\rho} \leq j_1$; (b) as $j \geq j^{k-1}$, it holds that $j \geq \lastBroadcast{v}{\rho}$ for all $v \in \set{v_1, \dots, v_{k-2}}$, hence $v^j \nin  \set{v_1, \dots, v_{k-2}}$; and (c) if $v^j \in \set{v_{k}, \dots v_{N-2}}$, then $t^j$ is an internal transition as $j < \firstBroadcast{v_k}{\rho} \leq \firstBroadcast{v}{\rho}$ for all $v\in \set{v_{k}, \dots v_{N-2}}$.
	
	Overall, we get that either $v^j = v_{k-1}$, or $v^j \in \set{v_{k}, \dots, v_{N-2}}$ and $t^j$ is internal. 
	In the first case, all cases ($t^j$ is an internal transition, $t^j$ is a broadcast transition and the message is received by $v_k$, or $t^j$ is a broadcast transition and the message is not received by $v_k$) are covered by our definition of $S$.
	In the latter case, either $v^j \neq v_k$ and $v_k$ and $v_{k-1}$ remain on the same states, or $v^j =v_k$ and this case is covered by our definition of $S$ ($t^j$ is internal).
	
\end{proof}
Hence, as $L_{\firstBroadcast{v_1}{\rho}}(v_1) \in H$ from \cref{lemma:proof:lemma:CoverLine2pb-complete-case0}, by applying inductively \cref{lemma:proof:lemma:CoverLine2pb-complete-case-ind}, we get that $L_{j}(v_{N-2})\in H$ for all $\firstBroadcast{v_{N-3}}{\rho} \leq j \leq \firstBroadcast{v_{N-2}}{\rho}$. As $j_1 = \firstBroadcast{v_{N-3}}{\rho} + 1$, then $L_{j_1}(v_{N-2})\in H$.

With a similar reasoning on $v_{N+2}$ and between $C_{j_1}$ and $C_{j_2}$, we get the following lemma.
\begin{lemma}
	$L_{j_1}(v_{N-2}) \in H$ and $L_{j_2}(v_{N+2}) \in H$.
\end{lemma}

Denote $q_1 =L_{j_1}(v_{N-2})$ and $q_2 = L_{j_2}(v_{N+2}) $.  Denote $C_{j_2} \transup{v^{j_2}, t^{j_2}} C_{j_2+1} \transup{v^{j_2 + 1}, t^{j_2 +1}} \cdots \transup{v^{n-1}, t^{n-1}} C_n$. Recall that from \cref{item:lemma:exec-shape-steps} of \cref{lemma:CoverLine-2pb-inP:execution-shape-2}, $v^{j_2}, \dots v^{n-1} \in \set{v_{N-2}, v_{N-1}, v_N, v_{N+1}, v_{N+2}}$. Denote $C_{q_1, q_2} = (\Gamma_5, L_{q_1, q_2})$, and observe that $L_{q_1, q_2}(v_i) = L_{j_2}(v_{N-3 + i})$ for all $1 \leq i \leq 5$.
Hence, $C_{q_1, q_2}= C'_0  \transup{v_2^0, t_2^0} C'_1 \transup{v_2^1, t_2^1} \cdots \transup{v_2^{n-j_2-1}, t_2^{n-j_2-1}} C'_{n-j_2}$ where for all $ 0 \leq j \leq n-j_2$: $C'_j = (\Gamma_5, L'_j)$ and $L'_j(v_i)=  L_{j_2+j}(v_{N-3 + i})$ for all $1 \leq i \leq 5$. And for all $ 0 \leq j < n-j_2$: $t_2^j = t^{j_2 + j}$ and if $v^{j_2 + j} = v_{N-3+i}$ for some $1 \leq i \leq 5$, then $v_2^j = v_i$.
\lug{ici je n'ai pas prouvé formellement, qu'en pensez vous }\nas{je pense que ça va}
Hence, $C_{q_1, q_2} \trans^\ast C'_{n-j_2}$ with $L'_{n-j_2}(v_3) = L_n(v_N) = q_f$.
This concludes the proof of \cref{lemma:CoverLine-2pb-inP:complete}.

\subsubsection{\Iflong{Proof of \cref{lemma:CoverLine-2pb-inP:correct}}Correctness of the algorithm}\label{appendix:subsec:CoverLine-inP:lemma:correct}
\Ifshort{
	This part is devoted to prove the following lemma.
\begin{lemma}\label{lemma:CoverLine-2pb-inP:correct}
	If there exist $q_1, q_2 \in H$, $v \in \Vert{\Gamma_5}$, such that $C_{q_1, q_2} \trans^\ast C$ and $C = (\Gamma_5, L)$ with $L(v) = q_f$, then there exist $\Gamma \in \Lines$, $C_0 \in \II$, $C' = (\Gamma, L')\in \CC$ such that $C_0 \trans^\ast C'$ and $L'(v) = q_f$.
\end{lemma}
}
We start by proving the following lemma.
\begin{lemma}\label{lemma:CoverLine-inP:inH-coverable}
	For all $q \in H$, there exists $k \in \nat$ such that $q$ is coverable with a line topology $\Gamma = (\set{v_1, \dots, v_k}, \set{\langle v_i, v_{i+1}\rangle \mid 1 \leq i <k})$ and vertex $v_1$.
\end{lemma}
\begin{proof}
	We will in fact prove that for all $(q_1, q_2) \in S$, there exists $\Gamma_{q_1, q_2} \in \Lines$ such that $\Vert{\Gamma_{q_1, q_2}} = \set{v_1, \dots, v_k}$, $\Edges{\Gamma_{q_1,q_2}} = \set{\langle v_i, v_{i+1} \rangle \mid 1 \leq i <k}$ for some $k \in \nat$ with $k \geq 2$ and there exists $C_0  \in \II$, $C = (\Gamma_{q_1, q_2}, L)$ with $C_0 \trans^\ast C$ and $L(v_1) = q_1$ and $L(v_2) = q_2$.
	
	Denote $S_0, \dots, S_K$ the subsets defined by the algorithm to build $S$. We have that $S = S_{N}$.
	We prove the property to be true for all $(q_1, q_2) \in S_i$ by induction on $0 \leq i \leq N$.
	For  $i = 0$, $S_0 = \set{(\qinit, \qinit)}$, this is trivial: define $\Gamma_{\qinit, \qinit} = (\set{v_1,v_2}, \set{\langle v_1, v_2\rangle})$ and take the initial configuration $C = (\Gamma_{\qinit, \qinit}, L_0)$ with $L_0(v_1)=L_0(v_2) = \qinit$.
	Let $0\leq i <K$, and assume we proved the property to be true for all $(q_1,q_2) \in S_i$. 
	Let $(q_1,q_2) \in S_{i+1} \in S_i$.
	\begin{itemize}
		\item if there exists $(q_1, p_2) \in S_i$ such that $(p_2, \tau, q_2) \in \Delta$, let $\Gamma_{q_1, p_2}\in \Lines$, and $C_0$ and $C = (\Gamma_{q_1, p_2}, L)$ obtained from the induction hypothesis such that $C_0 \trans^\ast C$ and $L(v_1) = q_1$ and $L(v_2) = p_2$. Then, $C \transup{v_2, (p_2, \tau, q_2)} C'$ and $C' = (\Gamma_{q_1, p_2}, L')$ with $L'(v_1) = q_1$ and $L'(v_2) = q_2$;
		
		\item if there exists $(p_1, q_2) \in S_i$ such that $(p_1, \tau, q_1) \in \Delta$, the reasoning is the same as in the previous case;
		
		\item if there exists $(p_1, p_2) \in S_i$ and $m \in \Sigma$ such that $(p_1, ?m, q_1) \in \Delta$ and $(p_2, !!m, q_2) \in \Delta$, we use a similar reasoning as in the previous cases;
		
		\item  if there exists $(q_1, p_2) \in S_i$ and $m \in \Sigma$ such that$(p_2, !!m, q_2) \in \Delta$ and $m\nin R(q_1)$, we use a similar reasoning as in the previous cases;
		
		\item if there exists $q_3$ such that $(q_2, q_3) \in S_i$ and $q_1 = \qinit$, then let $\Gamma_{q_2, q_3}$, $C_0$ and $C = (\Gamma_{q_2, q_3}, L)$ obtained from the induction hypothesis such that $C_0 \trans^\ast C$ and $L(v_1) = q_2$ and $L(v_2) = q_3$. Define $\Gamma_{q_1, q_2} = (\set{v'_1} \cup \Vert{\Gamma_{q_2,q_3}}, \set{ \langle v'_1, v_1 \rangle } \cup \Edges{\Gamma_{q_2,q_3}})$. Denote $C_0 \transup{v^0, t^0}\cdots \transup{v^{n-1}, t^{n-1}} C_n = C$. Recall that as $L_n(v_1) \in H$, $L_n(v_1) \in Q_0 \cup Q_1^r$, hence for all $0 \leq i \leq n-1$, if $v^i = v_1$ then $t^i$ is internal. Hence, there exists $C'_0 \dots C'_n$ such that $C'_0 \transup{v^0, t^0} \cdots \transup{v^{n-1}, t^{n-1}} C'_n$ and: for all $0 \leq i < n-1$, with $C'_i = (\Gamma_{q_1, q_2}, L'_i)$ and $C_i = (\Gamma_{q_2, q_3})$, and for all $v \in \Vert{\Gamma_{q_2, q_3}}$, $L'_i(v) = L_i(v)$ and $L'_i(v'_1) = \qinit$. Hence $L'_n(v'_1) = \qinit$ and $L'_n(v_1) = q_2$.
	\end{itemize}
\end{proof}
We are now ready to prove \cref{lemma:CoverLine-2pb-inP:correct}.\\
\begin{proofof}{\cref{lemma:CoverLine-2pb-inP:correct}}
	Let $q_1, q_2 \in H$, from \cref{lemma:CoverLine-inP:inH-coverable}, let $\Gamma_{q_1} = (\set{v_1, \dots, v_k}, \set{\langle v_i, v_{i+1} \rangle \mid 1 \leq i < k})$, and $\Gamma_{q_2}= (\set{v'_1, \dots, v'_n}, \set{\langle v'_i, v'_{i+1} \rangle \mid 1 \leq i < n})$ the two line topologies such that $q_1$ is coverable with $v_1$ and $q_2$ with $v'_1$.
	Denote $C_0 \trans^\ast C$ with $C = (\Gamma_{q_1}, L)$ and $L(v_1)= q_1$ and $C'_0 \trans^\ast C'$ with $C' = (\Gamma_{q_2}, L')$ and $L'(v'_1)= q_2$. As $q_1, q_2 \in H$ and $H \subseteq Q_0 \cup Q_1^r$, $v_1$ does not broadcast anything between $C_0$ and $C$ and $v'_1$ does not broadcast anything between $C'_0$ and $C'$.
	
	Define $\Gamma = (V, E)$ such that $V = \Vert{\Gamma_{q_1}} \cup \Vert{\Gamma_{q_2}} \cup \set{u_2, u_3,u_4}$ and $E = \Edges{\Gamma_{q_1}}\cup \Edges{\Gamma_{q_2}} \cup \set{\langle v_1, u_2\rangle, \langle u_2, u_3 \rangle, \langle u_3, u_4 \rangle, \langle u_4, v'_1 \rangle}$.
	Note that $v_1$ and $v'_1$ are not neighbors, and that each node in $\Vert{\Gamma_{q_1}}$ [resp. in $\Vert{\Gamma_{q_2}}$] has the same neighborhood as before, except for $v_1$, and $v'_1$ which have an additional new neighbor, respectively $u_2$ and $u_4$.
	Denote $C''_0 = (\Gamma, L''_0)$ with $L''_0(v) = \qinit$ for all $v \in \Vert{\Gamma}$. 
	
	Denote $C_0 \transup{v^1, t^1} C_1 \transup{v^2, t^2} \cdots \transup{v^{n_1}, t^{n_1}} C_{n_1} = C$ and $C'_0 \transup{v'^1, t'^1} C'_1 \transup{v'^2, t'^2} \cdots \transup{v'^{n_2}, t'^{n_2}} C'_{n_2} =  C'$.
	Then $C''_0\transup{v^1, t^1} C''_1 \transup{v^2, t^2} \cdots \transup{v^{n_1}, t^{n_1}} C''_{n_1}\transup{v'^1, t'^1} C''_{n_1 +1} \transup{v'^2, t'^2} \cdots \transup{v'^{n_2}, t'^{n_2}} C''_{n_1 + n_2}$ with $C''_{n_1 + n_2} = (\Gamma, L''_{n_1 + n_2})$ and $L''_{n_1 + n_2}(v) = L_{n_1}(v)$ for all $v \in \Vert{\Gamma_{q_1}}$ and $L''_{n_1 + n_2}(v) = L'_{n_2}(v)$ for all $v \in \Vert{\Gamma_{q_2}}$ and $L''_{n_1 + n_2}(v) = \qinit$ for all $v \in \set{u_2, u_3, u_4}$.
	Hence $L''_{n_1 + n_2}(v_1) = q_1$ and $L''_{n_1 +n_2}(v'_1) = q_2$.
	
	Finally, rename $\Gamma_5 = (\set{u_1, u_2, u_3, u_4, u_5}, \set{\langle u_i, u_{i+1} \rangle \mid 1 \leq i < 5})$, and $v_1$ by $u_1$ and $v'_1$ by $u_5$.
	Denote $C^5_{q_1, q_2} \transup{u^1, \delta^1} C^5_1 \transup{u^2, \delta^2} \cdots \transup{u^{n_3}, \delta^{n_3}} C^5_{n_3}$. Denote $C^5_{q_1, q_2} = (\Gamma_5, L_{q_1, q_2})$ and $C^5_{n_3} = (\Gamma_5, L^5_{n_3})$ 
	Note that for all $u \in \Vert{\Gamma_5}$, $L_{q_1, q_2}(u) = L''_{n_1 +n_2}(u)$. Hence, $C''_{n_1 + n_2}\transup{u^1, \delta^1} C''_{n_1 +n_2 +1} \transup{u^2, \delta^2} \cdots \transup{u^{n_3}, \delta^{n_3}} C''_{n_1 + n_2 + n_3}$ with $C''_{n_1 + n_2 +n_3} = (\Gamma, L''_{n_1 + n_2 +n_3})$ and $L''_{n_1 + n_2+n_3}(v) = L_{n_1}(v)$ for all $v \in \set{u_1, u_2, u_3, u_4, u_5}$.
	Hence $L''_{n_1 + n_2+n_3}(u_3) = q_f$.
	\lug{c'est moins formel que d'habitude car je ne pense pas que ce soit compliqué et on est deja à 45 pages, est ce que ça vous va ? }
	

\end{proofof}
\clearpage

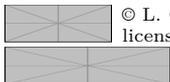
\end{document}